\documentclass[12pt]{article}
\usepackage[margin=1in]{geometry}
\usepackage{amsmath,amssymb,amsthm}
\usepackage[authoryear]{natbib}
\usepackage{titletoc}
\usepackage[dvipsnames]{xcolor}
\usepackage{enumitem}
\usepackage{rotating}
\usepackage{lscape}
\usepackage{todonotes}
\usepackage{float}
\usepackage{graphicx}
\usepackage{subcaption}
\usepackage{booktabs}
\usepackage{epstopdf}
\usepackage{setspace}
\usepackage{subcaption}

\usepackage{multirow}

\newtheorem{lemma}{Lemma}
\theoremstyle{definition}

\theoremstyle{definition}

\theoremstyle{definition}
\newtheorem{remark}{Remark}

\allowdisplaybreaks

\usepackage{titlesec}

\titlespacing*{\section}{0pt}{3pt}{2pt}
\titlespacing*{\subsection}{0pt}{3pt}{2pt}
\titlespacing*{\subsubsection}{0pt}{3pt}{2pt}

\usepackage{hyperref}

\definecolor{darkblue}{rgb}{0, 0, 0.5}
\hypersetup{
     colorlinks = true,
     citecolor = Maroon,
     urlcolor = darkblue,
     linkcolor = darkblue
}

\begin{document}

\title{The Power of Tests for Detecting $p$-Hacking\thanks{Some parts of this paper are based on material in the working paper \citet{elliott2020detecting}, which was not included in the final published version \citep{elliott2022detecting}. This paper was previously circulated as ``(When) Can We Detect $p$-Hacking?'' We are grateful to the Editor (Peter Hull) and three referees for their comments. K.W.\ is also affiliated with CESifo. The usual disclaimer applies. All figures and Monte Carlo results reported in the paper can be replicated using the replication package located at \url{https://github.com/nvkudrin/pHackingPower}.}}

\author{Graham Elliott\thanks{Department of Economics, University of California San Diego, 9500 Gilman Dr. La Jolla, CA 92093; email: \url{grelliott@ucsd.edu}}  \quad \quad Nikolay Kudrin\thanks{Department  of Economics, Queen's University,  94 University Ave.\ Kingston, Ontario K7L 3N9; email: \url{n.kudrin@queensu.ca}}  \quad \quad Kaspar W\"uthrich\thanks{Department  of Economics,  University of Michigan, 238 Lorch Hall,
611 Tappan Ave., Ann Arbor, MI 48109-1220; email: \url{kasparwu@umich.edu}}}

\maketitle

\begin{abstract} 

A flourishing empirical literature investigates the prevalence of $p$-hacking based on the distribution of $p$-values across studies. Interpreting results in this literature requires a careful understanding of the power of methods for detecting $p$-hacking. We theoretically study the implications of likely forms of $p$-hacking on the distribution of $p$-values to understand the power of tests for detecting it. Power can be low and depends crucially on the $p$-hacking strategy and the distribution of true effects. Combined tests for upper bounds and monotonicity and tests for continuity of the $p$-curve tend to have the highest power for detecting $p$-hacking.
\bigskip


\smallskip


\end{abstract}

\newpage

\section{Introduction}

Researchers have a strong incentive to find and report significant results \citep[e.g.,][p.158]{imbens2021statistical}. 
\citet{simonsohn2014p} used the term ``$p$-hacking'' to encompass decisions made by researchers in conducting their work to improve the publication prospects of their results. Their work has generated a flourishing literature that examines empirically the distribution of $p$-values across studies (the ``$p$-curve'') to determine if $p$-hacking is prevalent or not.\footnote{See, e.g., \citet{masicampo2012peculiar,simonsohn2014p,lakens2015what,simonsohn2015better,head2015extent,ulrich2015p} for early applications and further discussions, \citet{ havranek2021does, brodeur2022we,malovana2022borrower, yang2022hedge, decker2023preregistration} for recent applications, and \citet[][]{christensen2018transparency} for a review.} 

Interpreting these empirical results and assessing the usefulness of existing tests for $p$-hacking requires knowing how powerful these tests are, and for this we need to understand how $p$-hacking impacts the $p$-curve.\footnote{We focus on the problem of detecting $p$-hacking based the $p$-curve and do not consider the popular Caliper tests \citep[e.g.,][]{gerber2008do,gerber2008publication,bruns2019reporting,vivalt2019specification,brodeur2020methods}. Caliper tests aim to detect $p$-hacking based on excess mass in the distribution of $z$-statistics right above significance cutoffs. However, these tests do not control size in general \citep{kudrin2022robust}.}  Theoretically and using Monte Carlo analyses, this paper examines the power of tests for detecting likely forms of $p$-hacking, such as selecting controls in regression analyses, selecting among instruments in instrumental variables (IV) regression, and variance estimator selection in regression analyses. A careful study of power is important because the implications of $p$-hacking on the $p$-curve are not clear, and the magnitude of power depends on precisely how the $p$-curve is affected, which in turn depends on the empirical problem and how the $p$-hacking is undertaken. 

We find that tests for $p$-hacking can have low power  (above size, but not by much) and even no power (lower than or equal to size) in some cases, so failure to reject is not strong evidence of a lack of $p$-hacking; that different approaches to $p$-hacking have different effects resulting in different tests having power to detect such $p$-hacking; and that power depends on how easy it is to obtain statistically significant results. To illustrate, consider our first example below, where a researcher is interested in the coefficient on a variable in a linear regression but has a choice over which controls to include in the regression.  

Multiple choices over controls results in the researcher having a choice over a number of $p$-values that can be reported. Selective reporting, or $p$-hacking, could include either reporting the smallest $p$-value (referred to as the ``minimum'' approach) or reporting results from a preferred specification if significant and if not then trying other specifications (referred to as the ``threshold'' approach). The minimum approach pushes the $p$-curve to the left (we get smaller $p$-values) without creating non-monotonicities or discontinuities, so tests that look for these have no power. The minimum approach can however violate bounds on the $p$-curve. Unfortunately, tests for this can have no or low power, even when more than half of the results are $p$-hacked.

The threshold approach creates discontinuities in the $p$-curve as well as violations of upper bounds, so tests for these features have power, and this approach is easier to detect. With half the researchers $p$-hacking in this way power is quite reasonable for the covariate selection problem. However, the threshold approach may not create humps in the $p$-curve, so tests for these may not have power. Publication bias can have similar impacts on the $p$-curve as $p$-hacking, and so contribute to rejections using the tests we examine; however, it can also make it more difficult to detect $p$-hacking.

If the true coefficient in the linear model to be tested for zero is very large, then $p$-values are likely small anyway, and without much motivation to $p$-hack, the impact of $p$-hacking is hard to detect. For hypotheses for which data is informative but the true effects are near popular significance cutoffs, the incentive to search amongst the controls is larger and so is the impact on the distribution of $p$-values, resulting in higher power of the tests we examine. 

The results in this paper have important implications for meta-analysts choosing tests for detecting $p$-hacking. First, for both $p$-hacking approaches, combined tests for violations of upper bounds and monotonicity, such as the CS2B test \citep{elliott2022detecting}, do best. Second, testing for violations of bounds on the $p$-curve is particularly important in settings where $p$-hacking is difficult to detect, for example, when researchers use the minimum approach. Third, tests for discontinuities in the $p$-curve, such as the discontinuity test of \citet{cattaneo2020simple}, can be useful complements to tests based on monotonicity or bounds since the threshold approach can yield pronounced discontinuities. Finally, the classical Binomial test \citep[e.g.,][]{simonsohn2014p,head2015extent} often has lower power, even in settings where other tests have high power for detecting $p$-hacking.

\section{Setup}
\label{sec: setup}

We consider empirical studies where individual researchers provide test results of a hypothesis by reporting a test statistic $T$ with distribution $F_{h}$, where $h \in \mathcal{H}$ indexes parameters of the distribution of $T$.\footnote{The setup and notation here follow \citet{elliott2022detecting}.} Researchers are testing the null hypothesis that $h \in \mathcal{H}_0$ against $h \in \mathcal{H}_1$ with $\mathcal{H}_0 \cap \mathcal{H}_1$ empty. Suppose the test rejects when $T>cv(p)$, where $cv(p)$ is the level $p$ critical value. For any individual study the researcher tests a hypothesis at a particular $h$ (the ``true'' effect).\footnote{Here, $h$ is the local version of the difference between the true coefficient and its value under the null hypothesis, where `local' means it is scaled by the standard error of the estimator.} In what follows, we denote the power function of the test as $\beta\left(p,h \right)=\Pr\left(T>cv(p)\mid h\right)$.

In practice researchers can use different specifications that yield a set of test statistics $\{T_1, T_2, T_3, \dots \}$. This generates a set of $p$-values $\{P_1, P_2, P_3, \dots \}$ that they could report. We assume the researchers have a preferred specification, and in the absence of $p$-hacking would report that $p$-value. If they choose to $p$-hack, the approach they take would then comprise a method of choosing which $p$-value $P_r$ to report, i.e., a selection function $P_r=d(P_1, P_2,P_3, \dots).$ The joint distribution over possible $p$-values will depend on the testing situation and the distribution of the data, including the value for $h$. It is not the case that researchers can select any $p$-value they desire; the available $p$-values will be drawn from this joint distribution (at least, unless they directly fabricate or manipulate data). This relates to \citet{simonsohn2020blog}'s distinction of ``slow $p$-hacking'' and ``fast $p$-hacking'': slow $p$-hacking corresponds to the case where $p$-values change little across analyses; fast $p$-hacking refers to settings where $p$-values change a lot across analyses. Slow $p$-hacking is then when the $p$-values are highly positively correlated, fast $p$-hacking when the $p$-values are less dependent. 

\subsection{The Null Distribution of $p$-Values}
In the absence of $p$-hacking, the researcher reports the $p$-value from the preferred specification.  \citet{elliott2022detecting} provided a theoretical characterization of the distribution of $p$-values across multiple studies in the absence of $p$-hacking for general distributions of true effects.\footnote{See, e.g., \citet{hung1997}, \citet{simonsohn2014p}, and \citet{ulrich2018some} for numerical and analytical examples of $p$-curves for specific tests and/or effect distributions.} Across researchers, there is a distribution $\Pi$ of true effects $h$, which is to say that different researchers testing different hypotheses examine different problems that have different true effects. The resulting CDF of $p$-values across all these studies is then  
\begin{eqnarray*}
G(p)=\int_\mathcal{H}\Pr\left(T>cv(p)\mid h\right)d\Pi(h)=\int_\mathcal{H}\beta\left(p,h\right) d\Pi(h).
\end{eqnarray*}

Under mild regularity assumptions (differentiability of the null and alternative distributions, boundedness and support assumptions; see \citet{elliott2022detecting} for details), we can write the $p$-curve (density of $p$-values) in the absence of $p$-hacking as 
\begin{equation*}
g(p)=\int_\mathcal{H}\frac{\partial \beta\left(p,h\right)}{\partial p} d\Pi(h). 
\end{equation*}

For the purposes of testing for $p$-hacking, the properties of $g$ describe the set of distributions contained in the null hypothesis of no $p$-hacking. Tests can be based on deviations from this set. \citet{elliott2022detecting} provide general sufficient conditions for when the $p$-curve is non-increasing, $g'\le 0$, and continuous when there is no $p$-hacking, allowing for tests of these properties of the distribution to be interpreted as tests of the null hypothesis of no $p$-hacking. These conditions hold for many possible distributions $F_h$ that arise in research, for example, normal, folded normal (relevant for two-sided tests), and $\chi^2$ distributions. 

When $T$ is (asymptotically) normally distributed (for example, tests on means or regression parameters when central limit theorems apply), \citet{elliott2022detecting} show that in addition to being non-increasing, the $p$-curves are completely monotonic (i.e., have derivatives of alternating signs so that $g''\ge 0$, $g'''\le 0$, etc.) and there are testable upper bounds on the $p$-curve and its derivatives. 

We study the power of tests for $p$-hacking that exploit different combinations of these testable restrictions. In our simulation study, we consider four different types of tests: (i) tests for non-increasingness of the $p$-curve, (ii) tests for continuity of the $p$-curve, (iii) tests for upper bounds on the $p$-curve and its derivatives, and (iv) tests for combinations of monotonicity restrictions and upper bounds. We describe the individual tests and how we implement them in detail in Section \ref{sec:tests} and Table \ref{tab:tests}.

\subsection{The Distribution of $p$-Values under $p$-Hacking}
\label{sec: DirectionsPower}

If researchers do $p$-hack, the distribution of the reported $p$-value will depend on the functional form of the selection function $d(\cdot)$ and the joint distribution of $\{P_1, P_2, P_3,\dots \}$ given $h$. To the extent that this differs from the set of distributions under the null, it is then possible that the resulting $p$-curve violates the properties listed above in one way or another, providing the opportunity to test for $p$-hacking. This means that the power of any particular test for detecting $p$-hacking will depend on the functional form of $d(\cdot)$ and the joint distribution of $\{P_1, P_2, P_3,\dots \}$, in the sense that some tests have power against the types of changes to the distribution caused by $p$-hacking but others might not. Understanding the power of tests for $p$-hacking will thus depend on the testing problem as well as how researchers $p$-hack. It is for this reason we provide analytical results in the next section for a variety of testing problems and approaches to $p$-hacking.  

We consider two general approaches to $p$-hacking. The first is the \emph{threshold approach}, where the researcher constructs a test from their preferred model resulting in a $p$-value $P_1$, accepting this test if $P_1$ is below a target value $\alpha$ (for example, $0.05$). If the $p$-value does not achieve this goal value, additional models are considered. This is representative of the ``intuitive'' approach to $p$-hacking that is discussed in much of the literature on testing for $p$-hacking, where humps or discontinuities in the $p$-curve around common critical levels are examined. This results in a choice of $d(\cdot)$ such that  
\begin{equation} \label{eq: Thresholding}
    P_r= P_1 1_{\{P_1\le \alpha\}} + \min\{P_1,P_2,...\}1_{\{P_1 > \alpha\}}.
\end{equation}
Other thresholding approaches could be considered, such as those in Section \ref{sec:mc}. Tests for humps or discontinuities will have some power for detecting $p$-hacking based on the threshold approach, as will tests for violations of the upper bounds on the $p$-curve. 

The second is the \emph{minimum approach}, where researchers take the smallest $p$-value from a set of models. Here the choice of $d(\cdot)$ results in 
\begin{equation} \label{eq: Minimum}
    P_r= \min\{P_1,P_2,...\}.
\end{equation}
Intuitively, we would expect that for this approach the distribution of $p$-values would shift to the left, be monotonically non-increasing, and there would be no expected hump in the distribution of $p$-values near commonly reported significance levels. Tests for non-increasingness and discontinuities will not have power. Instead the bounds derived in \citet{elliott2022detecting} may be violated. 

Ultimately, carefully considering the functions $d(\cdot)$ and distributions of $p$-values allows us to examine power against empirically relevant alternatives. 

\subsection{Impact of Publication Bias}
\label{sec: publication bias}
Our focus is on the power of testing for various types of $p$-hacking. However, empirical studies of $p$-hacking often focus on published $p$-values, which may be subject to publication bias \citep[see, e.g.,][and the references therein]{andrews2019identification}. Publication bias can impact the distribution of $p$-values in ways similar to $p$-hacking.

To see this effect, let $S=1$ if the paper is selected for publication and $S=0$ otherwise. By Bayes' Law, the $p$-curve conditional on publication, $g_{S=1}(p):=g(p\mid S=1)$, is
\begin{equation}
g_{S=1}(p)=\frac{\Pr(S=1\mid p)g^d(p)}{\Pr(S=1)}. \label{eq: g_S=1}
\end{equation}
Here $\Pr(S=1\mid p)$ is the publication probability given $p$-value $P=p$, and $g^d(p)$ refers to the potentially $p$-hacked distribution of $p$-values when there is no publication bias.

Without publication bias, the publication probability does not depend on $p$, $\Pr(S=1\mid p)=\Pr(S=1)$, so that $g_{S=1}=g^d$. With publication bias, $\Pr(S=1\mid p)$ depends on $p$ so that $\Pr(S=1\mid p)$ is not be equal to $\Pr(S=1)$ for some $p\in (0,1)$ and $g_{S=1}\ne g^d$. In this case, one can detect $p$-hacking and/or publication bias if $g_{S=1}$ violates the testable restrictions underlying the statistical tests, so we can regard the tests here as joint tests of the absence of both $p$-hacking and publication bias. It is not possible in general to distinguish $p$-hacking from publication bias without additional assumptions.

It is plausible to assume that papers with smaller $p$-values are more likely to get published so that $\Pr(S=1\mid p)$ is decreasing in $p$. In this case, $g_{S=1}$ is non-increasing in the absence of $p$-hacking. Selection through publication bias that favors smaller $p$-values can result in steeper $p$-curves violating the bounds derived under the null hypothesis of no $p$-hacking. Hence rejections of bounds tests may well be exacerbated by publication bias. Discontinuities in $\Pr(S=1\mid p)$ can generate discontinuities in the absence of $p$-hacking, generating power for discontinuity tests. We examine these effects via Monte Carlo analysis in Section \ref{sec:simulations_publication_bias}.

\section{Implications of $p$-Hacking}
\label{sec:theory}

For a power evaluation to be informative, relevant choices for the joint distribution of the $p$-values and methods for $p$-hacking need to be considered. We deal with each of these through the following choices:

\begin{enumerate}\setlength\itemsep{0em} 

\item For the distribution of $h$, we provide general analytical results for any distribution. For graphical and Monte Carlo purposes, we choose either point masses on a particular value of $h$ or $\hat\Pi$, a distribution calibrated to the data in \citet{brodeur2020methods}.\footnote{Specifically, we model the distribution of $h$ by the Gamma distribution and calibrate its parameters to the randomized controlled trial (RCT) subsample of the data collected by  \cite{brodeur2020methods}, which they use as a benchmark. More specifically, we parameterize the density of the Gamma distribution as $\frac{\beta^\alpha}{\Gamma(\alpha)}x^{\alpha-1}e^{-\beta x}$ and compute maximum likelihood estimates $\hat\alpha=0.8547$ and $\hat\beta=1.8691$.}

\item For the methods employed in $p$-hacking, corresponding to the choice of the function $d(\cdot)$, we consider the two basic approaches to $p$-hacking discussed in Section \ref{sec: DirectionsPower}, the threshold and minimum approaches.

\end{enumerate}

We study the shape of the distribution of $p$-values under four arguably prevalent empirical problems with the potential for $p$-hacking. We focus on selecting control variables, selecting instruments, and variance bandwidth selection in the main text. We discuss selecting across datasets, which formally is a special case of the examples in the main text, in Appendix \ref{app:selecting_across_datasets}. The ability to $p$-hack depends on the distribution of the $p$-values conditional on $h$, and the four examples allow us to study power in these relevant testing situations.\footnote{While we focus on the impact of $p$-hacking on the shape of the $p$-curve and the power of tests for detecting $p$-hacking, explicit models of $p$-hacking are also useful in other contexts. For example, \citet{mccloskey2023critical} use a model of $p$-hacking to construct critical values that are robust to $p$-hacking.} 

For analytical tractability, we focus on the case where researchers use one-sided tests in this section, for a limited number of options for $p$-hacking. Appendix \ref{app:p_curves_two_sided} provides analogous numerical results for two-sided tests. The analytical results provide a clear understanding of the opportunities for tests to have power and what types of situations the tests will have power in. In the simulation study in Section \ref{sec:mc}, we consider generalizations of these analytical examples for two-sided tests, and we also show results for one-sided tests in Appendix \ref{app:additional_simulation_results}.

In addition to analyzing the effects of $p$-hacking on the shape of the $p$-curve, we study its implications for the bias of the estimates and size distortions of the tests reported by researchers engaged in $p$-hacking. We present these results in Appendix \ref{app:theory_bias_size_distortions}. Appendix \ref{app:detailed_derivations} presents the derivations underlying all analytical results.

\subsection{Selecting Control Variables in Linear Regression}
\label{sec:specification_search_regression}

Linear regression has been suggested to be particularly prone to $p$-hacking \citep[e.g.,][]{hendry1980econometrics,leamer1983lets,bruns2016p,bruns2017metaregression}. Researchers usually have available a number of control variables that could be included in a regression along with the variable of interest. Selection of various configurations for the linear model allows multiple chances to obtain a small $p$-value, perhaps below a threshold such as $0.05$. The theoretical results in this section yield a careful understanding of the shape of the $p$-curve when researchers engage in this type of $p$-hacking. 

We construct a stylized model and consider the two approaches to $p$-hacking discussed above in order to provide analytical results that capture the impact of $p$-hacking. Suppose the researchers estimate the impact of a scalar regressor $X_i$ on an outcome $Y_i$. The data are generated as $Y_i = X_i\beta+U_i$, $ i=1,\dots,N,$ where $U_i\overset{iid}\sim\mathcal{N}(0,1)$. For simplicity, we assume that $X_i$ is non-stochastic. The researchers test the hypothesis $H_0:\beta=0$ against $H_1:\beta>0$.

In addition to $X_i$, the researchers have access to two additional non-stochastic control variables, $Z_{1i}$ and $Z_{2i}$.\footnote{For simplicity, we consider a setting where $Z_{1i}$ and $Z_{2i}$ do not enter the true model so that their omission does not lead to omitted variable biases \citep[unlike, e.g., in][]{bruns2016p}. It is straightforward to generalize our results to settings where $Z_{1i}$ and $Z_{2i}$ enter the model: $Y_i = X_i\beta_1+Z_{1i}\beta_2 +Z_{2i}\beta_3 + U_i$.} We assume that $(X_i,Z_{1i},Z_{2i})$ are scale normalized so that $N^{-1}\sum_{i=1}^{N}X^2_i = N^{-1}\sum_{i=1}^{N}Z^2_{1i}=N^{-1}\sum_{i=1}^{N}Z^2_{2i}=1$. To simplify the exposition, we further assume that $N^{-1}\sum_{i=1}^{N}Z_{1i}Z_{2i}=\gamma^2$ and that $N^{-1}\sum_{i=1}^{N}X_iZ_{1i}= N^{-1}\sum_{i=1}^{N}X_iZ_{2i}= \gamma$, where $|\gamma|\in (0, 1)$.\footnote{We omit $\gamma = 0$, i.e., adding control variables that are uncorrelated with $X_i$, because in this case the $t$-statistics and thus $p$-values for each regression are equivalent and hence there is no opportunity for $p$-hacking of this form.} These assumptions are not essential for our analysis and could be relaxed at the expense of a more complicated notation. We let $h := \sqrt{N}\beta \sqrt{1-\gamma^2}$. 

In terms of $p$-hacking, researchers regress $Y_i$ on $X_i$ and $Z_{1i}$ and compute the resulting $p$-value, $P_1$. An additional $p$-value, $P_2$, arises from regressing $Y_i$ on $X_i$ and $Z_{2i}$ instead of $Z_{1i}$. The reported $p$-value $P_r$ is obtained from equation \eqref{eq: Thresholding} under the threshold approach and equation \eqref{eq: Minimum} under the minimum approach. Each approach results in different distributions of $p$-values, and, consequently, tests for $p$-hacking will have different power properties. 

In Appendix \ref{app:derivation_example1}, we show that for the threshold approach the resulting $p$-curve is 
\begin{eqnarray*}
g_1^{t}(p) =\int_\mathcal{H}\exp\left(hz_0(p)-\frac{h^2}{2}\right)\Upsilon^t_1(p; \alpha, h, \rho)d\Pi(h),
\end{eqnarray*}
where $\rho = 1-\gamma^2$, $z_h(p) = \Phi^{-1}(1-p) - h$, $\Phi$ is the standard normal cumulative distribution function (CDF), and 
\begin{equation*}
  \Upsilon^t_1(p;\alpha, h, \rho)=\begin{cases}
    1+ \Phi\left(\frac{z_{h}(\alpha) - \rho z_{h}(p)}{\sqrt{1-\rho^2}}\right), & \text{if $p\le \alpha$},\\
    2\Phi\left(z_h(p)\sqrt{\frac{1-\rho}{1+\rho}}\right), & \text{if $p> \alpha$}.
  \end{cases}
\end{equation*}
In interpreting this result, note that when there is no $p$-hacking, then $\Upsilon_1(p; \alpha, h, \rho)=1$. It follows from the properties of $\Phi$ that the threshold $p$-curve lies above the curve without $p$-hacking for $p \le \alpha$. We can also see that, since $\Phi\left(\frac{z_{h}(\alpha) - \rho z_{h}(p)}{\sqrt{1-\rho^2}}\right)$ is decreasing in $h$,  for larger $h$ the difference between the threshold $p$-curve and the curve without $p$-hacking becomes smaller. This is intuitive since for a larger $h$, the need to $p$-hack diminishes as most of the studies find an effect without resorting to manipulation. Note that a larger $h$, ceteris paribus, corresponds to more ``evidential value'' in the terminology of \citet{simonsohn2014p}.

The distribution of $p$-values for the minimum approach is equal to 
\begin{eqnarray*}
g_1^{m}(p) =2\int_\mathcal{H}\exp\left(hz_0(p)-\frac{h^2}{2}\right)\Phi\left(z_h(p)\sqrt{\frac{1-\rho}{1+\rho}}\right)d\Pi(h).
\end{eqnarray*}
For $p$-hacking of this form, the entire distribution of $p$-values is shifted to the left. For some values of $p$ less than 0.5, the curve lies above the no-$p$-hacking curve. This distribution is monotonically decreasing for all $\Pi$, so does not have a hump and remains continuous. Because of this, only the tests based on upper bounds and higher-order monotonicity have any potential for detecting $p$-hacking. If $\Pi$ is a point mass distribution, there is a range over which $g_1^{m}(p)$ exceeds the upper bound $\exp(z_0(p)^2/2)$ derived in \citet{elliott2022detecting}, the upper end (largest $p$) of which is at $p=1-\Phi(h)$.

Panels (a) and (b) of Figure \ref{fig:p_curves_example1} show the theoretical $p$-curves for various $h$ and $\gamma$.\footnote{Note that the results depend on $\gamma$ via $\rho=1-\gamma^2$. Therefore, the results do not depend on the sign of $\gamma$, and we only show results for positive values of $\gamma$.} In terms of violating the condition that the $p$-curve is monotonically decreasing, violations for the threshold case can occur but only for $h$ small enough. For $p<\alpha$, the derivative is
\begin{eqnarray*}
{g^{t}_1}'(p) = \int_\mathcal{H}\frac{\phi(z_h(p))\left[\frac{\rho}{\sqrt{1-\rho^2}}\phi\left(\frac{z_h(\alpha) - \rho z_h(p)}{\sqrt{1-\rho^2}}\right) - h\left(1+\Phi\left(\frac{z_{h}(\alpha) - \rho z_{h}(p)}{\sqrt{1-\rho^2}}\right)\right)\right]}{\phi^2(z_0(p))}d\Pi(h),
\end{eqnarray*}
where $\phi$ is the standard normal probability density function (PDF). Note that $\rho$ is always positive and, when all nulls are true (i.e., when $\Pi$ assigns probability one to $h=0$), ${g^t}'(p)$ is positive for all $p\in (0,\alpha)$.\footnote{For $p>\alpha$, ${g^{t}_1}'(p)$ is negative and equal to ${g^t_1}'(p)=-\int_{\mathcal{H}}\frac{\phi(z_h(p))}{\phi^2(z_0(p))}\left(h+ \sqrt{\frac{1-\rho}{1+\rho}}\phi\left(z_h(p)\sqrt{\frac{1-\rho}{1+\rho}}\right)\right)d\Pi(h).$} This can be seen for the dashed line in the Panel (a) of Figure \ref{fig:p_curves_example1}. However, at $h=1$, this effect no longer holds, and the $p$-curve is downward sloping. From Panel (b) in Figure \ref{fig:p_curves_example1}, we see that violations of monotonicity are larger for smaller $\gamma$. When $\gamma=0.1$, the $p$-curve even becomes bimodal (one mode at 0 and one mode at 0.05). 

Figure \ref{fig:p_curves_example1} indicates that the threshold approach to $p$-hacking implies a discontinuity at $\alpha$. The size of the discontinuity is larger for larger $h$ and remains for each $\gamma$, although how that translates to power of tests for discontinuity also depends on the shape of the rest of the curve. We examine this in Monte Carlo experiments in Section \ref{sec:mc}. \textbf{[FIGURE \ref{fig:p_curves_example1} HERE.]}

Panel (c) of Figure \ref{fig:p_curves_example1} compares the threshold and minimum approach to $p$-hacking. Results are presented for $h=1$ and $\gamma=0.5$, with respect to the bounds under no $p$-hacking. We also report the no-$p$-hacking distribution. Simply taking the minimum $p$-value as a method of $p$-hacking results in a curve that remains downward sloping and has no discontinuity --- tests for these two features will have no power against such $p$-hacking. But as Panel (c) of Figure \ref{fig:p_curves_example1} shows, the upper bounds on the $p$-curve are violated for both methods of $p$-hacking. The violation in the thresholding case is pronounced. By contrast, the violation in the minimum case is barely visible, so that the power of tests for detecting this type of $p$-hacking based on violations of the upper bounds will be low, except in large samples.

\subsection{Selecting amongst Instruments in IV Regression}
\label{sec: iv_example}

Suppose that the researchers use an IV regression to estimate the causal effect of a scalar regressor $X_i$ on an outcome $Y_i$. The data are generated as
\begin{align*}
Y_i &= X_i\beta+U_i, \\
X_i &= Z_{1i}\gamma_1 + Z_{2i}\gamma_2 + V_{i},
\end{align*}
where $(U_i,V_i)'\overset{iid}\sim\mathcal{N}(0,\Omega)$ with $\Omega_{12}\ne 0$. The instruments are generated as $Z_i\overset{iid}\sim \mathcal{N}(0, I_2)$ and independent of $(U_i, V_i)$. The researchers test the hypothesis 
$H_0:\beta=0$ against $H_1:\beta>0.$
To simplify the exposition, suppose that ${\Omega_{11}=\Omega_{22}=1}$ and $\gamma_1=\gamma_2=\gamma$. We let $h := \sqrt{N}\beta|\gamma|$.

For $p$-hacking, researchers start by using all the available information (both instruments) to compute $P_1$ but can also generate $p$-values by only using the first instrument to compute $P_2$ or only the second instrument for $P_3$. The threshold and minimum choices for $P_r$ follow from equations \eqref{eq: Thresholding} and \eqref{eq: Minimum}.\footnote{In practice, it is likely that researchers also consider the first stage $F$-statistic when selecting instruments \citep[e.g.,][]{andrews2019weak,brodeur2020methods}. We explore this in our Monte Carlo simulations.} 

For the threshold approach, the $p$-curve (see Appendix \ref{app:derivation_example2} for derivations) is 
\begin{equation*}
    g_2 ^t(p) = \int_{\mathcal{H}}\exp\left(hz_0(p) - \frac{h^2}{2}\right)\Upsilon^t_2(p; \alpha, h)d\Pi(h),
\end{equation*}
where 
\begin{equation*}
  \Upsilon^t_2(p;\alpha, h)=\begin{cases}
    \frac{\phi(z_{\sqrt{2}h}(p))}{\phi(z_h(p))} + 2\Phi(D_h(\alpha)-z_h(p)), & \text{if $0<p\le\alpha$},\\
    \frac{\phi(z_{\sqrt{2}h}(p))}{\phi(z_h(p))}\zeta(p) + 2\Phi(D_h(p)-z_h(p)), & \text{if $\alpha<p\le1/2$},\\
     2\Phi(z_h(p)), & \text{if $1/2<p<1$},\\
  \end{cases}
\end{equation*}
and $\zeta(p) = 1 - 2\Phi((1-\sqrt{2})z_0(p))$ and $D_h(p)=\sqrt{2}z_0(p)-2h$.

In Figure \ref{fig:p_curves_example2}, the $p$-curves for $h\in \{0,1,2\}$ are shown for the threshold approach in Panel (a). As in the covariate selection example, it is only for small values of $h$ that we see upward sloping curves and a hump below size. For $h=1$ and $h=2$, no such violation of non-increasingness occurs, and tests aimed at detecting such a violation will have no power. The reason is similar to that of the covariate selection problem --- when $h$ becomes larger many tests reject anyway, so whilst there is still a possibility to $p$-hack the actual rejections overwhelm the ``hump'' building of the $p$-hacking. For all $h$, there is still a discontinuity in the $p$-curve arising from $p$-hacking, so tests for a discontinuity at $\alpha$ will still have power.  \textbf{[FIGURE \ref{fig:p_curves_example2} HERE]}

For the minimum approach in equation \eqref{eq: Minimum}, the distribution of $p$-values is equal to 
\begin{equation*}
    g^m_2(p) = \int_{\mathcal{H}}\exp\left(hz_0(p) - \frac{h^2}{2}\right){\Upsilon}^m_2(p; \alpha, h)d\Pi(h),
\end{equation*}
where 
\begin{equation*}
  {\Upsilon}^m_2(p;\alpha, h)=\begin{cases}
    \frac{\phi(z_{\sqrt{2}h}(p))}{\phi(z_h(p))}\zeta(p) + 2\Phi(D_h(p)-z_h(p)), & \text{if $0<p\le1/2$},\\
     2\Phi(z_h(p)), & \text{if $1/2<p<1$}.\\
  \end{cases}
\end{equation*}

Panel (b) of Figure \ref{fig:p_curves_example2} displays the $p$-curves for $h\in \{0,1,2\}$ for the minimum approach. There is no hump, as expected, and all the curves are non-increasing. Only tests based on upper bounds for and higher-order monotonicity of the $p$-curve have the possibility of rejecting the null hypothesis of no $p$-hacking in this situation. The upper bound is violated around the 0.05 threshold for $h = 0$ and at lower $p$-values for larger $h$.

Panel (c) in Figure \ref{fig:p_curves_example2} shows the comparable figure for the IV problem as Panel (c) of Figure \ref{fig:p_curves_example1} shows for the covariate selection example. The results are qualitatively similar across these examples, although quantitatively the $p$-hacked curves in the IV problem are closer to and more likely to exceed  the bounds than in the covariates problem.

Overall, as with the case of covariate selection, both the relevant tests for $p$-hacking and their power will depend strongly on the range of $h$ relevant to the studies underlying the data employed for the tests.

\subsection{Variance Bandwidth Selection}
\label{sec:variance_selection_analytical}

In time series regression, sums of random variables such as means or regression coefficients are standardized by an estimate of the spectral density of the relevant series at frequency zero. A number of estimators exist; the most popular in practice is a nonparametric estimator that takes a weighted average of covariances of the data. With this method, researchers are confronted with a choice of the bandwidth for estimation. Different bandwidth choices allow for multiple chances at constructing $p$-values, hence allowing for the potential for $p$-hacking.

To examine this analytically, consider the model $Y_{t}=\beta+U_{t}$, $t=1,\dots,N$, where we assume that $U_{t}\overset{iid}\sim \mathcal{N}(0,1)$. Researchers consider two statistics for testing $H_0:\beta=0$ against $H_1:\beta>0$. First, the usual $t$-statistic, $T_1=\sqrt{N}\bar{Y}_{N}$, which generates a $p$-value $P_1$. Secondly, $T_2=(\sqrt{N}\bar{Y}_{N})/\hat\omega$, which generates a $p$-value $P_2$, where $\bar{Y}_{N}=N^{-1} \sum_{t=1}^N Y_t$, $\hat\omega^2:={\omega}^2(\hat\rho):=1+2\kappa \hat{\rho}$ and $\hat{\rho} = (N-1)^{-1} \sum_{t=2}^N\hat{U}_t \hat{U}_{t-1}$. Here ${\kappa}$ is the weight in the spectral density estimator. For example, in the 
\citet{newey1987simple} estimator with one lag, $\kappa=1/2$. We again consider both the threshold approach to $p$-hacking \eqref{eq: Thresholding} as well as the minimum approach \eqref{eq: Minimum}.\footnote{If $\hat\rho$ is such that $\hat\omega^2$ is negative, the researcher always reports the initial result.}

In Appendix \ref{app:derivation_example4}, we show that the distribution of $p$-values has the form 
\begin{eqnarray*}
g_3^{t}(p) =\int_\mathcal{H}\exp\left(h z_0(p)-\frac{h^2}{2}\right)\Upsilon^t_3(p; \alpha, h, \kappa)d\Pi(h),
\end{eqnarray*}
with $\Upsilon_3(p; \alpha, h, \kappa)$ taking different forms over different parts of the support of the distribution. Define $l(p)=(2\kappa)^{-1} \left( \left(\frac{z_0(\alpha)}{z_0(p)} \right)^2-1 \right)$ and let $H_N$ and $\eta_N$ be the CDF and PDF of $\hat\rho$, respectively. Then we have
\begin{equation*}
  \Upsilon^t_3= \begin{cases}
    1+\frac{1}{\phi(z_h(p))}\int_{-(2\kappa)^{-1}}^{l(p)}\omega(r)\phi(z_0(p)\omega(r)-h)\eta_N(r)dr, & \text{if $0<p\le\alpha$},\\ 
       1-H_N(0)+H_N(-(2\kappa)^{-1})+\frac{1}{\phi(z_h(p))}\int_{-(2\kappa)^{-1}}^{0}\omega(r)\phi(z_0(p)\omega(r)-h)\eta_N(r)dr, & \text{if $\alpha<p\le 1/2$},\\
    H_N(0)+\frac{1}{\phi(z_h(p))}\int_{0}^{\infty}\omega(r)\phi(z_0(p)\omega(r)-h)\eta_N(r)dr, & \text{if $1/2<p<1$}.
    \end{cases}
\end{equation*}

Panel (a) in Figure \ref{fig:pcurves_lag} presents the $p$-curves for the thresholding case. Notice that, unlike the earlier examples, thresholding creates the intuitive hump in the $p$-curve at the chosen size (here $0.05$) for all of the values for $h$. Thus tests that attempt to find such humps may have power. Discontinuities at the chosen size and violations of the upper bounds also occur. \textbf{[FIGURE \ref{fig:pcurves_lag} HERE]}

When the minimum over the two $p$-values is chosen, the $p$-curve is given by 
\begin{eqnarray*}
g_3^{m}(p) =\int_\mathcal{H}\exp\left(h z_0(p)-\frac{h^2}{2}\right){\Upsilon}^m_3(p; \alpha, h, \kappa)d\Pi(h),
\end{eqnarray*}
where 
\begin{equation*}
  \Upsilon^m_3= \begin{cases}
       1-H_N(0)+H_N(-(2\kappa)^{-1})+\frac{1}{\phi(z_h(p))}\int_{-(2\kappa)^{-1}}^{0}\omega(r)\phi(z_0(p)\omega(r)-h)\eta_N(r)dr, & \text{if $0<p\le 1/2$},\\
    H_N(0)+\frac{1}{\phi(z_h(p))}\int_{0}^{\infty}\omega(r)\phi(z_0(p)\omega(r)-h)\eta_N(r)dr, & \text{if $1/2<p<1$}.
    \end{cases}
\end{equation*}

Panel (b) in Figure \ref{fig:pcurves_lag} presents the $p$-curves for the minimum case. When $p$-hacking works through taking the minimum $p$-value, the impact is to move the distributions towards the left, making the $p$-curves fall more steeply. Of interest is what happens at $p=0.5$, where taking the minimum (this effect is also apparent in the thresholding case) results in a discontinuity. The reason for this is that choices over the denominator of the $t$-statistic used to test the hypothesis cannot change the sign of the $t$-test. Within each side, the effect is to push the distribution to the left, so this results in a (small) discontinuity at $p=0.5$. This effect will extend to all methods where $p$-hacking is based on searching over different choices of variance-covariance matrices --- for example, different choices in estimators, different choices in the number of clusters (as we consider in the Monte Carlo simulations), etc. Panel (b) of Figure \ref{fig:pcurves_lag} shows that for $h=1,2$, the bound is not reached, and any discontinuity at $p=0.5$ is very small. For $h=0$, the bound is slightly below the $p$-curve after the discontinuity.

\section{Statistical Tests for $p$-Hacking}
\label{sec:tests}

In this section, we discuss several statistical tests for the null hypothesis of no $p$-hacking based on a sample of $n$ $p$-values, $\{P_i\}_{i=1}^n$.

\smallskip
\noindent \textbf{Histogram-based Tests for Combinations of Restrictions.}
Histogram-based tests \citep{elliott2022detecting} provide a flexible framework for constructing tests for different combinations of testable restrictions. Let $0=x_0<x_1<\cdots< x_J=1$ be an equidistant partition of $[0,1]$ and define the population proportions $\pi_j = \int_{x_{j-1}}^{x_j}g(p)dp$, $j=1,\dots, J.$
The main idea of histogram-based tests is to express the testable implications of $p$-hacking in terms of restrictions on the population proportions $(\pi_{1},\dots,\pi_J)$. For instance, non-increasingness of the $p$-curve implies that $\pi_{j} - \pi_{j-1}\le 0$ for $j=2,\dots,J$. More generally, \citet{elliott2022detecting} show that $K$-monotonicity (i.e., derivatives with alternating signs up to order $K$) and upper bounds on the $p$-curve and its derivatives can be expressed as $H_0:~ A\pi_{-J}\le b$,
for a matrix $A$ and vector $b$, where $\pi_{-J} := (\pi_1,\dots, \pi_{J-1})'$.\footnote{Here we incorporate the adding up constraint $\sum_{j=1}^J\pi_j=1$ into the definition of $A$ and $b$ and express the testable implications in terms of the ``core moments'' $(\pi_1,\dots, \pi_{J-1})$ instead of $(\pi_1,\dots, \pi_{J})$.}

To test this hypothesis, we estimate $\pi_{-J}$ by the vector of sample proportions $\hat\pi_{-J}$. The estimator $\hat\pi_{-J}$ is asymptotically normal with mean $\pi_{-J}$ so that the testing problem can be recast as the problem of testing affine inequalities about the mean of a multivariate normal distribution \citep[e.g.,][]{kudo1963,wolak1987,cox2022simple}. Following \citet{elliott2022detecting}, we use the conditional chi-squared test of \citet{cox2022simple}, which is easy to implement and remains computationally tractable when $J$ is moderate or large.

\smallskip
\noindent \textbf{Tests for Non-Increasingness of the $p$-Curve.} 
A popular test for non-increasingness of the $p$-curve is the Binomial test \citep[e.g.,][]{simonsohn2014p,head2015extent}, where researchers compare the number of $p$-values in two adjacent bins right below significance cutoffs. Under the null of no $p$-hacking, the fraction of $p$-values in the bin closer to the cutoff should be weakly smaller than the fraction in the bin farther away. Implementation is typically based on an exact Binomial test. Binomial tests are ``local'' tests that ignore information about the shape of the $p$-curve farther away from the cutoff, which often leads to no or low power in our simulations.\footnote{A ``global'' alternative is Fisher's test \citep[e.g.,][]{simonsohn2014p}. We do not report results for Fisher's test since we found that this test has essentially no power for detecting the types of $p$-hacking we consider.}

In addition to the classical Binomial test, we consider tests based on the least concave majorant (LCM) \citep{elliott2022detecting}.\footnote{LCM tests have been successfully applied in many different contexts \citep[e.g.,][]{carolan2005,beare2015nonparametric,fang2019refinements}.} LCM tests are based on the observation that non-increasingness of $g$ implies that the CDF $G$ is concave. Concavity can be assessed by comparing the empirical CDF of $p$-values, $\hat{G}$, to its LCM $\mathcal{M}\hat{G}$, where $\mathcal{M}$ is the LCM operator. We choose the test statistic $\sqrt{n}\|\hat{G}-\mathcal{M}\hat{G}\|_\infty$. The uniform distribution is least favorable for this test \citep{kulikov2008distribution,beare2021least}, and critical values can be obtained via simulations.

\smallskip
\noindent \textbf{Tests for Continuity of the $p$-Curve.} 
Continuity of the $p$-curve at pre-specified cutoffs $p=\alpha$ can be assessed using standard density discontinuity tests \citep[e.g.,][]{mccrary2008manipulation,cattaneo2020simple}. Following \citet{elliott2022detecting}, we use the approach by \citet{cattaneo2020simple} with the automatic bandwidth selection implemented in the \texttt{R}-package \texttt{rddensity} \citep{cattaneo2021rddensity}. 

\section{Monte Carlo Simulations}
\label{sec:mc}

In this section, we investigate the finite sample properties of the tests in Section \ref{sec:tests} using a Monte Carlo simulation study based on generalizations of the analytical examples of $p$-hacking in Section \ref{sec:theory}. 

\subsection{Generalized $p$-Hacking Examples}
In all examples, researchers test a hypothesis about a scalar parameter $\beta$:
\begin{equation}
H_0:\beta=0\qquad  \text{against}\qquad H_{1}:\beta\ne 0.
\label{eq:H0_simul}
\end{equation}
The results for one-sided tests of  $H_0:\beta=0$ against $H_{1}:\beta> 0$ are similar. See Figure \ref{fig:power_cov_K3}.

Researchers may $p$-hack their initial results by exploring additional model specifications or estimators and report a different result of their choice. Specifically, we consider the two general approaches to $p$-hacking discussed in Sections \ref{sec: setup} and \ref{sec:theory}: the threshold and the minimum approach. In what follows, we discuss the generalized examples of $p$-hacking in more detail.

\subsubsection{Selecting Control Variables in Linear Regression}
\label{sec:covariate_selection_MC}

Researchers have access to a random sample with $N=200$ observations generated as $Y_i=X_{i}\beta+u_i$,
where $X_i\sim \mathcal{N}(0,1)$ and $u_i\sim \mathcal{N}(0,1)$ are independent of each other. There are $K$ additional control variables, $Z_i:=(Z_{1i},\dots, Z_{Ki})'$,  generated as $Z_{ki}= \gamma_kX_{i}+\sqrt{1-\gamma_k^2} \epsilon_{Z_k,i}$ for $k=1,\dots,K$, where $\epsilon_{Z_k,i} \sim \mathcal{N}(0,1)$ and $\gamma_k\sim U[-0.8, 0.8]$. 
We set $\beta=h/\sqrt{N}$ with $h\in \{0,1,2\}$, and we also show results for $h\sim \hat\Pi$, where $\hat\Pi$ is the Gamma distribution fitted to the RCT subsample of the \cite{brodeur2020methods} data.\footnote{The normalization $\beta=h/\sqrt{N}$ is motivated by $h$ being a local effect in our notation and the asymptotic variance of the OLS estimator from a regression of $Y_i$ on $X_i$ being equal to $1$ in this simple model.}

Researchers use one of two threshold approaches or a minimum approach to $p$-hacking. 

\noindent \textit{Threshold approach (general-to-specific):} Researchers regress $Y_i$ on $X_i$ and $Z_i$, test \eqref{eq:H0_simul}, and obtain $P$. If $P\le 0.05$, they report it. If $P>0.05$, they regress $Y_i$ on $X_i$, trying all $(K-1)\times 1$ subvectors of $Z_i$. They report the smallest $p$-value if it is smaller than $0.05$. If it is larger than $0.05$, they continue and explore all $(K-2)\times 1$ subvectors of $Z_i$ etc. If all results are insignificant, they report the smallest $p$-value. 

\noindent\textit{Threshold approach (specific-to-general):} Researchers start by regressing $Y_i$ on $X_i$ only, test \eqref{eq:H0_simul}, and obtain $p$-value $P$. If $P\le 0.05$, they report it. If $P>0.05$, they regress $Y_i$ on $X_i$, trying every component of $Z_i$ as control variable and select the result with the smallest $p$-value. If the smallest $p$-value is larger than $0.05$, they continue and explore all $2\times 1$ subvectors of $Z_i$ etc. If all results are insignificant, they report the smallest $p$-value.

\noindent \textit{Minimum approach:} Researchers run regressions of $Y_i$ on $X_i$ and each possible configuration of covariates $Z_i$ and report the minimum $p$-value.

Figure \ref{fig:hists_covar_selection} shows the null distributions ($p$-curves without $p$-hacking), including estimates of the average power of the underlying studies, and the alternative distributions ($p$-curves with $p$-hacking).\footnote{To generate these distributions, we run the algorithm one million times and collect $p$-hacked and non-$p$-hacked results.} The threshold approach leads to a discontinuity in the $p$-curve and may lead to non-increasing $p$-curves and humps below significance thresholds. By contrast, the minimum approach generally leads to continuous and non-increasing $p$-curves. The distribution of $h$ is an important determinant of the shape of the $p$-curve. It affects the joint distribution of $p$-values under both $p$-hacking approaches and the probability that the researchers find significant results earlier under the threshold approaches. 
Finally, as expected, the violations of the testable restrictions are more pronounced for large $K$ (i.e., when researchers have many degrees of freedom).

\subsubsection{Selecting amongst Instruments in IV Regression}
\label{sec:iv_selection_MC}
Researchers have access to a random sample with $N=200$ observations generated as
\begin{equation*}
\begin{split}
Y_i&=X_{i}\beta+U_i \\
X_{i} &= Z_i'\pi + V_i
\end{split}~~, 
\qquad \begin{pmatrix}U_i\\V_i \end{pmatrix}\sim 
\mathcal{N}\left( \begin{pmatrix} 0 \\ 0 \end{pmatrix},\begin{pmatrix}1 & 0.5 \\ 0.5 & 1 \end{pmatrix}\right).
\end{equation*}
The instruments $Z_i:=(Z_{1i},\dots, Z_{Ki})'$ are generated as $Z_{ki}= \gamma_k\xi_{i}+\sqrt{1-\gamma_k^2} \epsilon_{Z_k,i}$ for $k=1,\dots,K$, where
 $\xi_i\sim \mathcal{N}(0,1)$, $\epsilon_{Z_k,i} \sim \mathcal{N}(0,1)$, and $\gamma_k\sim U[-0.8, 0.8]$. Here $\xi_i, \epsilon_{Z_k, i}$, and $\gamma_k$ are independent for all $k$. Also, $\pi_k\overset{iid}\sim U[1, 3]$ for $k=1,\dots, K$. We set $\beta=h/(3 \sqrt{N})$ with $h\in \{0,1,2\}$ and also show results for $h\sim \hat\Pi$.\footnote{We additionally scale $h$ by 3 to make the average power of studies more comparable to the other examples.}

Researchers use either a threshold or a minimum approach to $p$-hacking.

\noindent \textit{Threshold approach:} Researchers estimate the model using all instruments $Z_i$, test \eqref{eq:H0_simul}, and obtain the $p$-value $P$. If $P\le 0.05$, they report it. If $P>0.05$, they try all $(K-1)\times 1$ subvectors of $Z_i$ as instruments and select the result corresponding to the smallest $p$-value. If the smallest $p$-value is larger than $0.05$, they explore all $(K-2)\times 1$ subvectors of $Z_i$ etc. If all results are insignificant, they report the smallest $p$-value. 

\noindent \textit{Minimum approach:} The researchers run IV regressions of $Y_i$ on $X_i$ using each  possible configuration of instruments and report the minimum $p$-value.

Figures \ref{fig:hists_IV3} and \ref{fig:hists_IV5} display the null distributions, including estimates of the average power of the underlying studies, and the $p$-hacked distributions.\footnote{Unlike in the covariate selection example, we do not show results for $K=7$ since there is a very high concentration of $p$-values at zero in this case.} We also show these distributions for a scenario where the researchers screen out specifications with first-stage $F$-statistics below 10. In this case, the researchers ignore such specifications while doing thresholding or minimum type searches described above. As with covariate selection, the threshold approach yields discontinuous $p$-curves and may lead to non-increasingness and humps, whereas reporting the minimum $p$-value leads to continuous and decreasing $p$-curves. The distribution of $h$ and $K$ are important determinants of the shape of the $p$-curve.

\subsubsection{Standard Error Selection: Lag Length and Clustering}
\label{sec:se_selection_MC}

We consider two different types of standard error selection: lag length selection as in Section \ref{sec:variance_selection_analytical} and selecting the level of clustering, given the prevalence of clustered standard errors in empirical research \cite[e.g.,][]{cameron2015practitioner,mackinnon2023cluster}.

\smallskip
\noindent \textbf{Lag length selection.} Researchers have access to a random sample with $N=200$ observations from $Y_t=X_t\beta+U_t$, where $X_t\sim \mathcal{N}(0,1)$ and $U_t\sim \mathcal{N}(0,1)$ are independent. We set $\beta=h/\sqrt{N}$ with $h\in \{0,1,2\}$ and also show results for $h\sim \hat\Pi$.

Researchers use either a threshold or a minimum approach to $p$-hacking.

\noindent\textit{Threshold approach.} Researchers first regress $Y_t$ on $X_t$ and calculate the standard error using the classical Newey-West estimator with the number of lags selected using the Bayesian Information Criterion (they only choose up to $4$ lags). They then use a $t$-test to test \eqref{eq:H0_simul} and calculate the $p$-value $P$. If $P\le 0.05$, they report it. If $P>0.05$, they try the Newey-West estimator with one extra lag. If the result is not significant, they try two extra lags etc. If all results are insignificant, they report the smallest $p$-value. 

\noindent \textit{Minimum approach.}  Researchers regress $Y_t$ on $X_t$, calculate the standard error using Newey-West with $0$ to $4$ lags and report the minimum $p$-value.

The null distributions, including estimates of the average power of the underlying studies, and the $p$-hacked distributions are displayed in Figure \ref{fig:hists_var}. The threshold approach induces a sharp spike right below 0.05. This is because $p$-hacking via lag selection does not lead to huge improvements in terms of $p$-value. 

\smallskip
\noindent \textbf{Cluster level selection.} Consider the same model as for the lag length selection. The researchers calculate cluster-robust standard errors from grouping data into 20, 40, 50, and  100 clusters. They also calculate standard errors without clustering (or equivalently with 200 clusters). They then report either the minimum $p$-value across clustering levels (\textit{minimum approach}) or the first significant $p$-value found while searching through clustering levels starting from 20 clusters (\textit{threshold approach}). The null distributions, including estimates of the average power of the underlying studies, and the $p$-hacked distributions are displayed in Figure \ref{fig:hists_clust}.

\subsection{Simulations}

\subsubsection{Setup}\label{sec:simulations_setup}
 We model the distribution of reported $p$-values as a mixture, $g^{o}(p) = \tau \cdot g^d(p)+(1-\tau)\cdot g^{np}(p)$. Here, $g^d$ is the distribution under the different $p$-hacking approaches described above; $g^{np}$ is the distribution in the absence of $p$-hacking (i.e., the distribution of the $p$-value from the preferred specification). The parameter $\tau\in [0,1]$ captures the fraction of researchers who engage in $p$-hacking.

To generate the data, we first simulate the $p$-hacking algorithms one million times to obtain samples corresponding to $g^d$ and $g^{np}$. Then, to construct samples in every Monte Carlo iteration, we draw $n=5000$ $p$-values with replacement from a mixture of those samples. Appendix \ref{app:sample_size} presents results for smaller sample sizes. Following \citet{elliott2022detecting}, we apply the tests to the subinterval $(0, 0.15]$. Therefore, the effective sample size depends on the $p$-hacking strategy, the distribution of $h$, and the fraction of $p$-hackers $\tau$.\footnote{Note that the data-generating processes in Sections~\ref{sec:covariate_selection_MC}--\ref{sec:se_selection_MC} imply that, for any $p > \alpha$, the number of observations in the interval $(0, p]$ is identical under both $p$-hacking approaches for covariate, IV, and cluster selection.}

We compare the finite sample performance of the tests described in Section \ref{sec:tests}. See Table \ref{tab:tests} for more details.\footnote{For CS1, CSUB, and CS2B, the optimization routine fails to converge for some realizations of the data due to nearly singular covariance matrix estimates. We count these cases as non-rejections of the null in our Monte Carlo simulations.}  The simulations are implemented using  \texttt{MATLAB} \citep{MATLAB:2023} and \texttt{R} \citep{R2023}. \textbf{[TABLE \ref{tab:tests} HERE]}

\subsubsection{Power Curves} 
\label{sec:power_curves}
In this section, we present figures showing how power varies with the fraction of $p$-hackers, $\tau$, referred to as power curves, for the different data generating processes (DGPs). For covariate and instrument selection, we focus on the results for $K=3$ in the main text and present the results for larger values of $K$ in Appendix \ref{app:additional_simulation_results}. We focus on two-sided tests in the main text. Figure \ref{fig:power_cov_K3} in Appendix \ref{app:additional_simulation_results} presents the results for covariate selection with one-sided tests. The nominal level is  5\%. All results are based on 5000 simulation draws. Figures \ref{fig:power_cov_combined}--\ref{fig:power_se_combined} present the results. \textbf{[FIGURES \ref{fig:power_cov_combined}--\ref{fig:power_se_combined} HERE.]}

The power for detecting $p$-hacking crucially depends on whether the researchers use a thresholding or a minimum approach to $p$-hacking, the econometric method, the fraction of $p$-hackers, $\tau$, and the distribution of $h$. When researchers $p$-hack using a threshold approach, the $p$-curves are discontinuous at the threshold, may violate the upper bounds, and may be non-monotonic. Thus, tests exploiting these testable restrictions may have power when the fraction of $p$-hackers is large enough. 

CS2B, which exploits monotonicity restrictions and bounds, has the highest power overall. Among the tests that exploit monotonicity of the entire $p$-curve, CS1 typically exhibits higher power than LCM. LCM can exhibit non-monotonic power curves because the test statistic converges to zero in probability for strictly decreasing $p$-curves \citep{beare2015nonparametric}. 

The Binomial test often exhibits no or low power. The reason is that the $p$-hacking approaches we consider do not lead to isolated humps or spikes near $0.05$, even if researchers use a threshold $p$-hacking approach. There is one notable exception. When researchers engage in lag length selection, $p$-hacking based on the threshold approach can yield isolated humps right below the cutoff. By construction, the Binomial test is well-suited for detecting this type of $p$-hacking and is among the most powerful tests in this case. Our results for the Binomial test demonstrate the inherent disadvantage of using tests that only exploit testable implications locally. Such tests only have power against very specific forms of $p$-hacking, which limits their usefulness in practice. 

Discontinuity tests are a useful complement to tests based on monotonicity and upper bounds because $p$-hacking based on threshold approaches often yields pronounced discontinuities. These tests are particularly powerful for detecting $p$-hacking based on lag length selection, which leads to spikes and pronounced discontinuities at $0.05$, as discussed above.

When researchers always report the minimum $p$-value, the power of the tests is much lower than when they use a threshold approach. The minimum approach to $p$-hacking does not lead to violations of monotonicity and continuity over $p\in (0,0.15]$. Therefore, by construction, tests based on these restrictions have no power, irrespective of the fraction of researchers who are $p$-hacking. 

The minimum approach may yield violations of the upper bounds. The range over which the upper bounds are violated and the extent of these violation depend on the distribution of $h$ and the econometric method used. The simulations show power only for the tests based on upper bounds (CSUB and CS2B) and for covariate and IV selection when a sufficiently large fraction of researchers $p$-hacks. It is noteworthy that no test has power under the minimum approach when $h\sim \hat\Pi$. This is because $\hat\Pi$ puts mass on large values of $h$, which as explained below, can lead to no or low power.

Under the minimum approach, the power curves of CSUB and CS2B are quite similar, suggesting that the power of CS2B comes mainly from using upper bounds. This finding demonstrates the importance of exploiting upper bounds in addition to monotonicity and continuity restrictions in practice.

The relationship between the power of the tests and the value of $h$ need not be monotonic. The value of $h$ affects the shape of the $p$-curve in complicated ways, so that the power can be non-monotonic depending on the setting, testable restrictions, and specific test. Under the minimum approach, large values of $h$ lead to $p$-values close to zero, where the upper bounds are more difficult to violate. This is the reason why tests based on those upper bounds have very low or no power for $h=2$ (and $h\sim \hat\Pi$).

Finally, the results in Appendix \ref{app:additional_simulation_results} show that the larger $K$ (i.e., the more degrees of freedom the researchers have when $p$-hacking)  the higher the power of CSUB and CS2B.

Overall, the tests' ability to detect $p$-hacking is highly context-specific and can be low in some cases or lacking entirely. This is because $p$-hacking may not lead to violations of the testable restrictions used by the statistical tests for $p$-hacking. Moreover, even if $p$-hacking leads to violations of the testable restrictions, these violations may be small and can thus only be detected based on large samples of $p$-values. Regarding the choice of testable restrictions, the simulations demonstrate the importance of exploiting upper bounds in addition to monotonicity and continuity for constructing powerful tests against plausible $p$-hacking alternatives. 

\begin{remark}[Relationship between Power and Costs of $p$-Hacking] In Appendix \ref{app:power_bias}, we study the relationship between the power of the tests and the costs of $p$-hacking, measured by the relative bias of the underlying estimates. We find that the relationship between the relative bias and power is hump-shaped for the most powerful tests, with the maximum power being achieved between 10\% and 20\% relative bias under thresholding and between 30\% and 35\% relative bias under the minimum approach. See Appendix \ref{app:power_bias} for a more detailed discussion. \qed
\end{remark}

\subsubsection{The Impact of Publication Bias}
\label{sec:simulations_publication_bias}

Here we investigate the impact of publication bias on the power of the tests for testing the joint null hypothesis of no $p$-hacking and no publication bias. We generate a sample of $n=5000$ $p$-values as in Section \ref{sec:simulations_setup} and keep each $p$-value, $P_i$, with probability $\Pr(S=1\mid P_i)$. 

We consider two types of publication bias that differ with respect to how $\Pr(S=1\mid p)$ varies with $p$: sharp publication bias and smooth publication bias. Under sharp publication bias, $\Pr(S=1\mid p)$ is a step function: $\Pr(S=1\mid p)=1_{\{p\le 0.05\}}+0.1\times 1_{\{p> 0.05\}}$. Hence, significant results are 10 times more likely to be published than insignificant ones. Under smooth publication bias, we set $\Pr(S=1\mid p)=\exp(-A\cdot p)$, where we choose $A=8.45$ to make results comparable across both types of publication bias.\footnote{When $A=8.45$, the ratio between $\int_{0}^{0.05}\Pr(S=1\mid p)dp$ and $\int_{0.05}^{1}\Pr(S=1\mid p)dp$ is the same for both types of publication bias.}

Table \ref{tab:publication_bias_1} shows the power of the tests for three levels of $p$-hacking ($\tau \in \{0,0.5,1\}$) with no publication bias, sharp publication bias, and smooth publication bias for covariate selection with $K=3$ and $h=0$. We show results for $h\in \{1,2\}$ and $h\sim \hat{\Pi}$ in Appendix \ref{app:additional_simulation_results}. The impact of publication bias on power depends on the testable restrictions that the tests exploit. The CSUB and CS2B tests have high power for detecting publication bias in the absence of $p$-hacking (when $\tau=0$), and publication bias substantially increases their power for rejecting the joint null when $p$-hacking alone is difficult to detect. This is expected since both forms of publication bias favor small $p$-values, which leads to $p$-curves that are more likely to violate the upper bounds, as discussed in Section \ref{sec: publication bias}. \textbf{[TABLE \ref{tab:publication_bias_1} HERE.]}

For the tests based on monotonicity of the entire $p$-curve (CS1 and LCM), the results depend on the type of publication bias. Sharp publication bias tends to increase power, whereas smooth publication bias can substantially lower power. Due to the local nature of the Binomial test, sharp publication bias does not increase its power. This again demonstrates the advantages of using ``global'' tests.

Sharp publication bias accentuates existing discontinuities and leads to discontinuities in otherwise smooth $p$-curves and thus increases the power of the discontinuity test. By contrast, smooth publication bias can decrease the power of the discontinuity test.

Overall, our results suggest that the presence of publication bias leads to higher power for tests based on upper bounds, especially when $p$-hacking is difficult to detect. For the other tests, whether publication bias increases or decreases the power depends on the exact nature of the publication bias: sharp publication bias tends to increase power, whereas smooth publication bias can decrease power.

\section{Empirical illustration}
\label{sec:illustration}

In this section, we illustrate how our results can help guide the interpretation of empirical results from tests for $p$-hacking by reanalyzing the data in \citet{brodeur2020methods}.\footnote{The data are form \citep{brodeur2022data}. \citet{kudrin2022robust,kudrin2024jmp} analyze some of the subsamples in \citet{brodeur2020methods} using overlapping sets of tests. The differences between our results and theirs based on the same subsamples and tests are due to differences between V1 and V2 of the replication package.} \citeauthor{brodeur2020methods} study how $p$-hacking and publication bias vary by causal inference methods (difference-in-differences, RCT, regression discontinuity design, IV), strength of the instrument (whether the first-stage $F$-statistic is below or above 30), journal rank, and over time. We apply the tests for $p$-hacking in Table \ref{tab:tests} to 13 different subsamples (incl.\ the overall sample) analyzed by \citet{brodeur2020methods}.\footnote{We do not consider the comparison between working and published papers in Section IV.B because the data on working papers only contain the $t$-statistics, so that we cannot apply our de-rounding strategy.} See Table \ref{tab:application} for a list of the different subsamples. 

An important practical issue is that the test results reported in research papers are typically rounded \citep[e.g.,][]{elliott2022detecting,kranz2022methods}. Failure to de-round leads the tests to overreject. The reason is that rounding induces violations of the testable restrictions, such as discontinuities and non-monotonicities. The issue is particularly pronounced for the Binomial and the discontinuity test because it induces a mass point at $t=2$ ($p=0.046$) \citep{elliott2022detecting,kranz2022methods}. For example, in the full sample, there are 258 observations with $t=2$, making up more than 37\% of the 693 observations with $p\in [0.04,0.05]$ used by the Binomial test. 

To document and correct for the impact of rounding, we present results based on the rounded original data in Panel (a) of Table \ref{tab:application} and based on de-rounded data in Panel (b) of Table \ref{tab:application}. The de-rounded data are obtained by adding uniformly distributed noise from the interval $[-0.5, 0.5] \cdot 10^{dp}$ to each statistic (coefficient estimate, standard error, $t$-statistic, or $p$-value) recorded with $dp$ decimal places.\footnote{This de-rounding strategy is similar to \cite{brodeur2016}. De-rounding the $p$-values preserves the monotonicity of the $p$-curve \citep{elliott2020detecting} but could lead to violations of higher-order monotonicity that induce some power for tests targeting these restrictions. \citet{kranz2022methods} propose an improved de-rounding method designed for Caliper tests. We leave an extension of their method to the tests we analyze here for future research.} To mitigate the impact of the added randomness from de-rounding, we present the results in terms of the average rejection rates over 1,000 draws of the de-rounding. In addition to reporting the results for each subsample and the overall sample, we also compute an overall rejection rate across all samples to get an estimate of the ``empirical power'' of the tests.\footnote{The dataset contains multiple tests per paper. We therefore adjust for clustering for the tests for which cluster-robust versions are available (i.e., CS2B, CSUB, and CS1).} \textbf{[TABLE \ref{tab:application} HERE.]}

The results based on the rounded data (Panel (a)) suggest that there is substantial heterogeneity across the different subsamples. For example, all tests except CS1 reject at the 5\% level in the DID subsample, whereas no test rejects in the $F\ge 30$ subsample. 
Overall, the Binomial test rejects in 77\% of the subsamples, followed by CS2B, which rejects in 54\% of the subsamples, whereas the other tests reject in less than 50\% of the subsamples. De-rounding drastically changes the results and empirical conclusions. The overall rejection rates drop substantially for all tests. The overall rejection rate of the Binomial test drops from 75\% to 0\% after de-rounding, and only CSUB and CS2B have overall rejection rates larger than 10\%.

The results in this paper can help explain these findings and clarify their interpretation. First, test results based on rounded data are misleading because the tests overreject. This is because the tests are often not able to distinguish a mass point at $p=0.046$ due to rounding only from $p$-hacking based on thresholding. As we show in this paper, many tests have some power to detect this type of $p$-hacking, which helps explain the rejections based on the rounded data. Second, once we focus on the de-rounded data, there is only limited empirical evidence for $p$-hacking: the average rejection rates are relatively low, even for the best tests. However, given the limited ability of tests to detect $p$-hacking we document in this paper, such findings do not imply that there is only limited $p$-hacking. Finally, there are substantial differences between the results of the different tests. These results are broadly consistent with our simulation results. For example, the most powerful tests (CS2B and CSUB) have the highest average rejection rates after de-rounding.

\section{Conclusion}
\label{sec:conclusion}

Concerns about $p$-hacking have motivated a fast-growing literature testing for $p$-hacking based on $p$-curves. Interpreting empirical work based on these tests requires a careful understanding of their ability to detect $p$-hacking. We examine how well existing tests are able to detect $p$-hacking in practice. Threshold approaches to $p$-hacking (where a predetermined significance level is targeted) result in $p$-curves that typically have discontinuities, $p$-curves that exceed upper bounds under no $p$-hacking, and less often violations of monotonicity restrictions. Many tests have some power to find such $p$-hacking, and the best tests are those exploiting both monotonicity and upper bounds and those based on testing for discontinuities. $p$-Hacking based on reporting the minimum $p$-value does not result in $p$-curves exhibiting discontinuities or monotonicity violations. While tests based on bound violations have some power, $p$-hacking based on mininum approaches is generally much harder to detect than $p$-hacking based on thresholding. The presence of publication bias (in addition to $p$-hacking) often increases but can also decrease the power of the tests for rejecting the joint null hypothesis of no $p$-hacking and no publication bias.

Overall, we find that tests for $p$-hacking can have low power in many settings and even no power in some cases. Therefore, failure to reject using tests for $p$-hacking based on $p$-values does not really indicate the absence of $p$-hacking. 
The lack of power we document has already motivated the development of new tests \citep[e.g.,][]{kudrin2024jmp}, and we believe that this is a promising area for future research.

In this paper, we focus on tests of the null of no $p$-hacking (and no publication bias). From an empirical perspective, testing this null hypothesis is a useful starting point, and rejecting it can then help motivate further analyses and attempts to correct for selective reporting.\footnote{We view the role of this null hypothesis as similar to the role of sharp nulls in experiments. See, for example, the discussion in Chapter 5.11 of \citet{imbens2015causal}.} These tests can also be useful to evaluate the impact of interventions to curb $p$-hacking and publication bias, such as mandating pre-registration, editorial statements \citep{blancoperez2020publication}, or data-sharing policies \citep{brodeur2024phacking}. From an econometric perspective, the advantage of focusing on the null of no $p$-hacking is that it allows for developing tests that control size absent any restrictions on how researchers $p$-hack. Methods that go beyond testing for the existence of $p$-hacking typically rely on specific models of $p$-hacking \citep{mccloskey2023critical} or additional assumptions \citep{andrews2019identification,brodeur2020methods}. A key finding of this paper is that even the power of the best tests can be low, despite the null of no $p$-hacking being strong and potentially unrealistic in some applications. This suggests that considering less stringent nulls could lead to tests with even lower power, absent strong additional assumptions.\footnote{An example would be the null of there being less than a certain fraction of $p$-hacking \citep{kudrin2024jmp}.}

We end with two final remarks. First, there are important practical issues when testing $p$-hacking that our theoretical and Monte Carlo results do not directly speak to. Examples include how to optimally deal with rounding, how to define and select tests of interest, and how to ensure quality control when collecting large samples of $p$-values \citep[e.g.,][]{simonsohn2015blog}.\footnote{In Appendix \ref{app:sample_selection}, we show that mistakes in selecting the $p$-values that researchers target when $p$-hacking can substantially lower the power of tests for detecting $p$-hacking.} Second, we examine situations where the model is correctly specified, so estimators are consistent for their true values. For poorly specified models, for example, the omission of important variables that leads to omitted variables (confounding) effects, it is possible to generate a larger variation in $p$-values. Such problems with empirical studies are well understood and perhaps best found through theory and replication than meta-studies. 

\bibliographystyle{apalike}
\bibliography{references}

\newpage

\section*{Figures}

\begin{figure}[H]
     \begin{center}
                   \caption{$p$-Curves from covariate selection} 
              \label{fig:p_curves_example1}
        \includegraphics[width=0.325\textwidth]{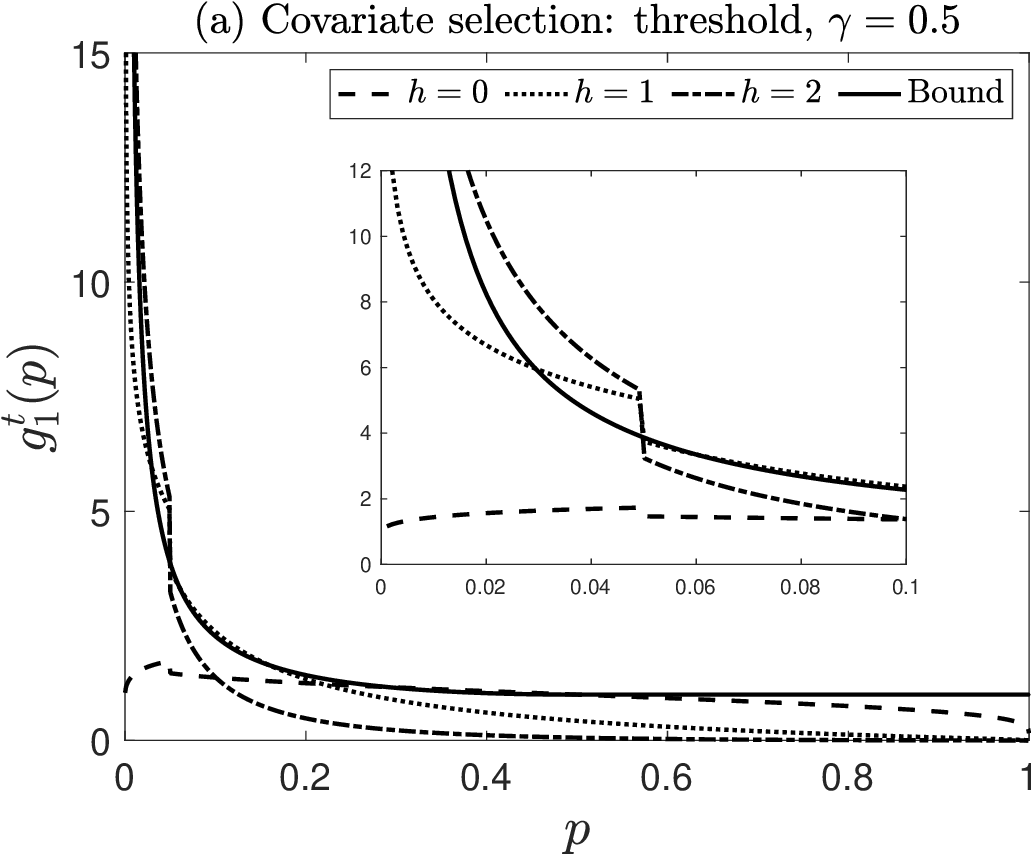}
\includegraphics[width=0.325\textwidth]{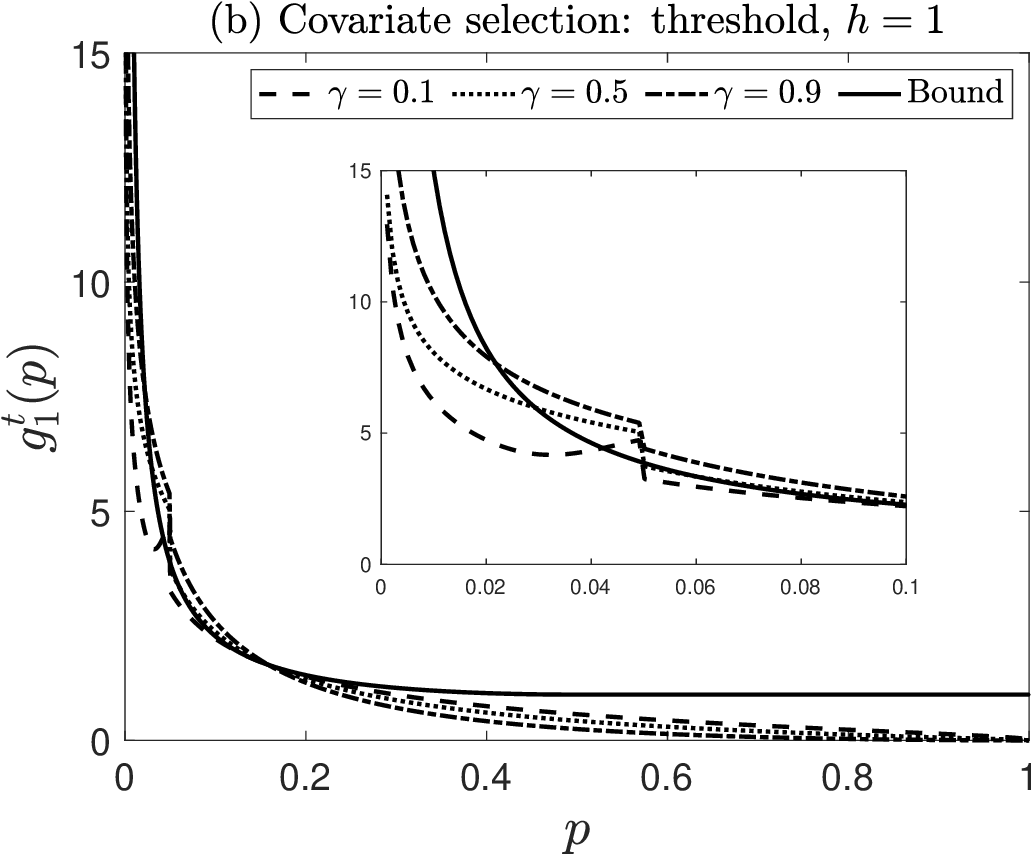}
 \includegraphics[width=0.325\textwidth]{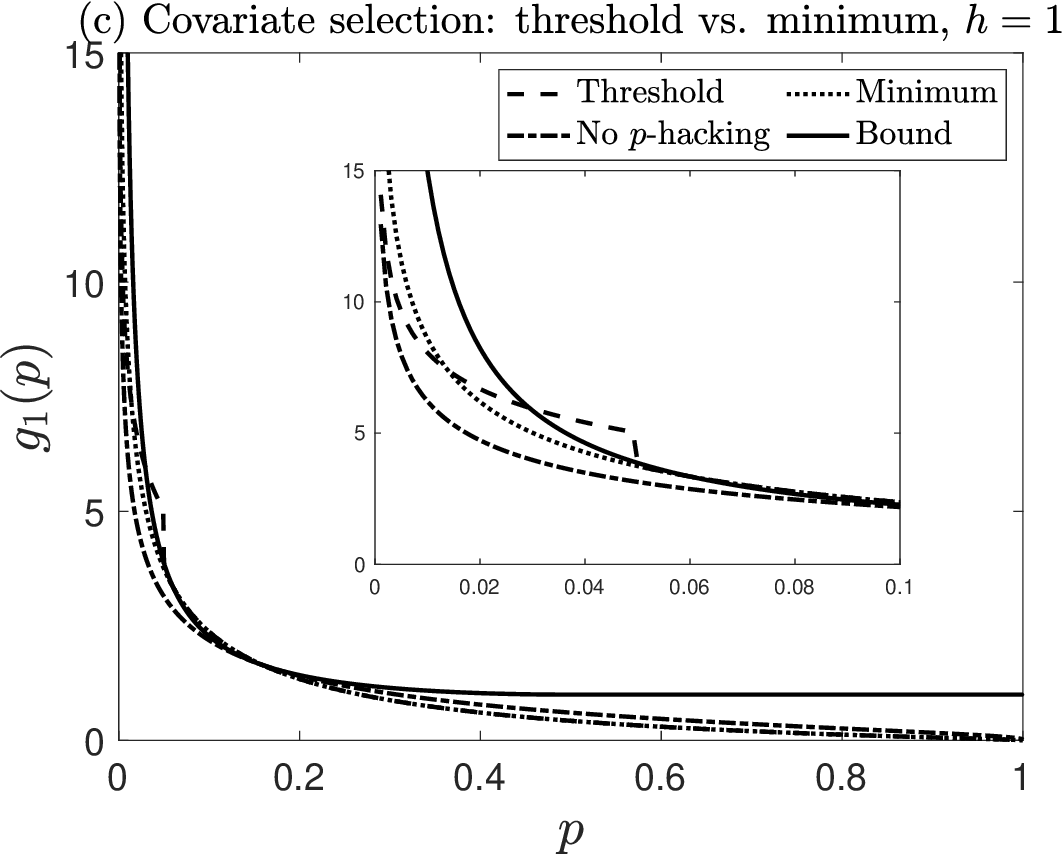}

\end{center}

\vspace*{-3mm}
\doublespacing
\textit{Notes:} Figure shows the $p$-curves from $p$-hacking based on the simple model of covariate selection described in this section and the thresholding and minimum approach. In Panel (a), we set $\gamma=0.5$ and vary $h$. In Panel (b), we set $h=1$ and vary $\gamma$. Panel (c) compares the threshold and minimum approaches.

\end{figure}

\newpage

\begin{figure}[H]
              \caption{$p$-Curves from IV selection} 
              \label{fig:p_curves_example2}
     \begin{center}
        \includegraphics[width=0.325\textwidth]{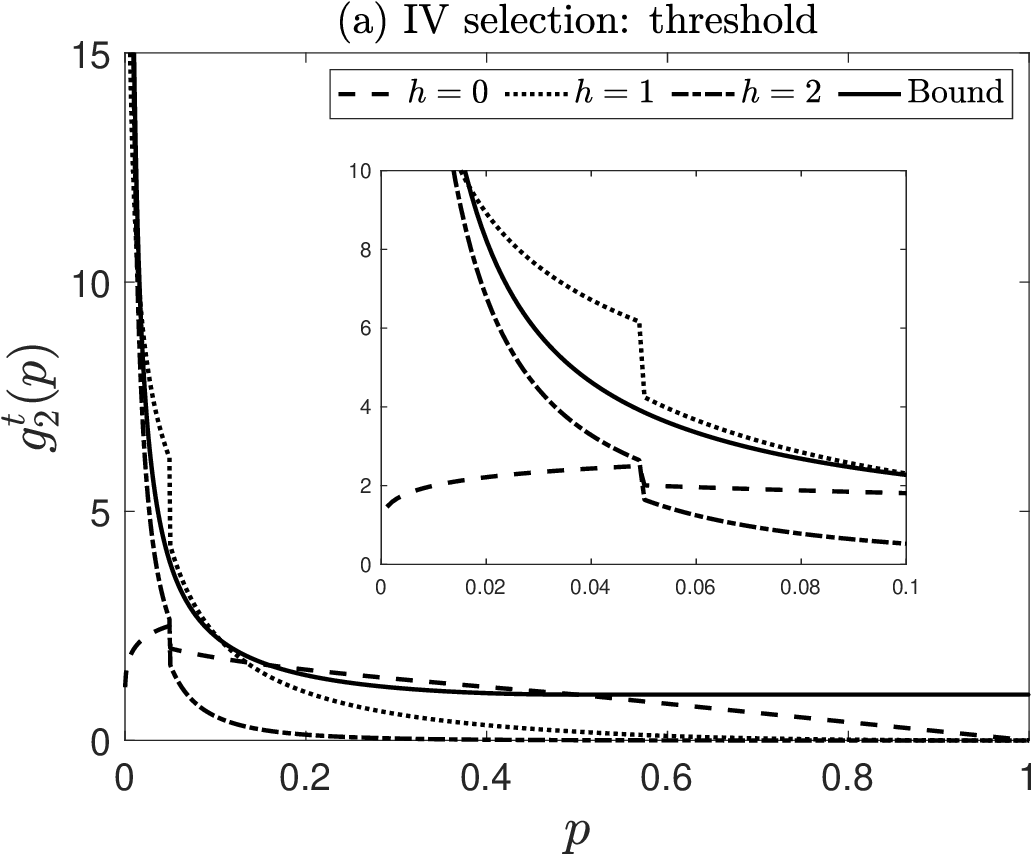}
\includegraphics[width=0.325\textwidth]{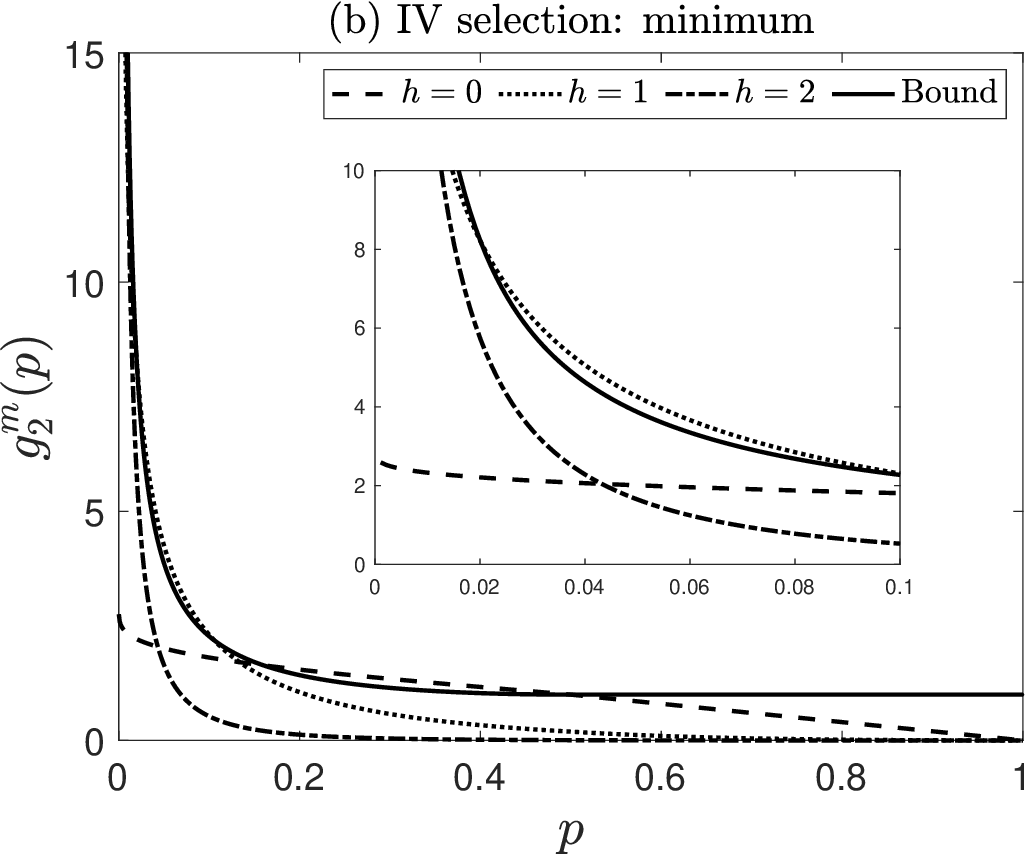}
\includegraphics[width=0.325\textwidth]{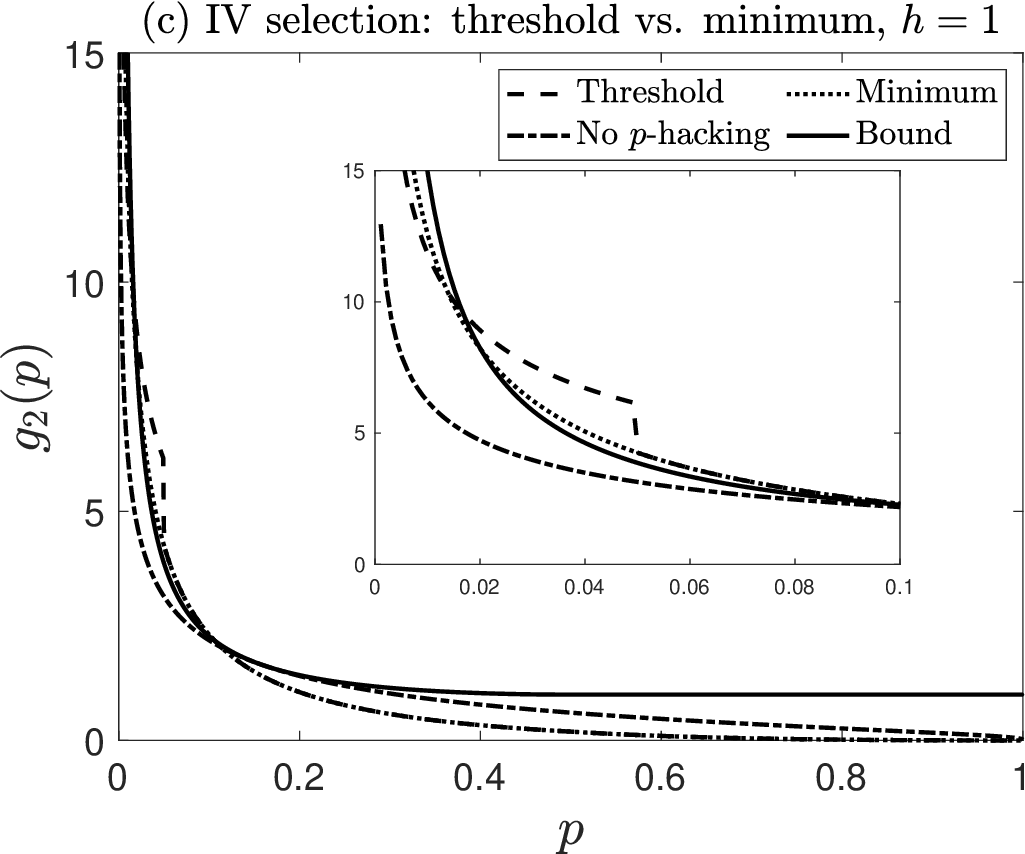}

\end{center}

\vspace*{-3mm}
\doublespacing
\textit{Notes:} Figure shows the $p$-curves from $p$-hacking based on the simple model of IV selection described in this section. Panels (a) and (b) show the results for different values of $h$. Panel (c) compares the threshold and minimum approach for $h=1$.

\end{figure}

\newpage

\begin{figure}[H]
     \begin{center}
\caption{$p$-Curves from lag length selection}              \label{fig:pcurves_lag}
 \includegraphics[width=0.43\textwidth]{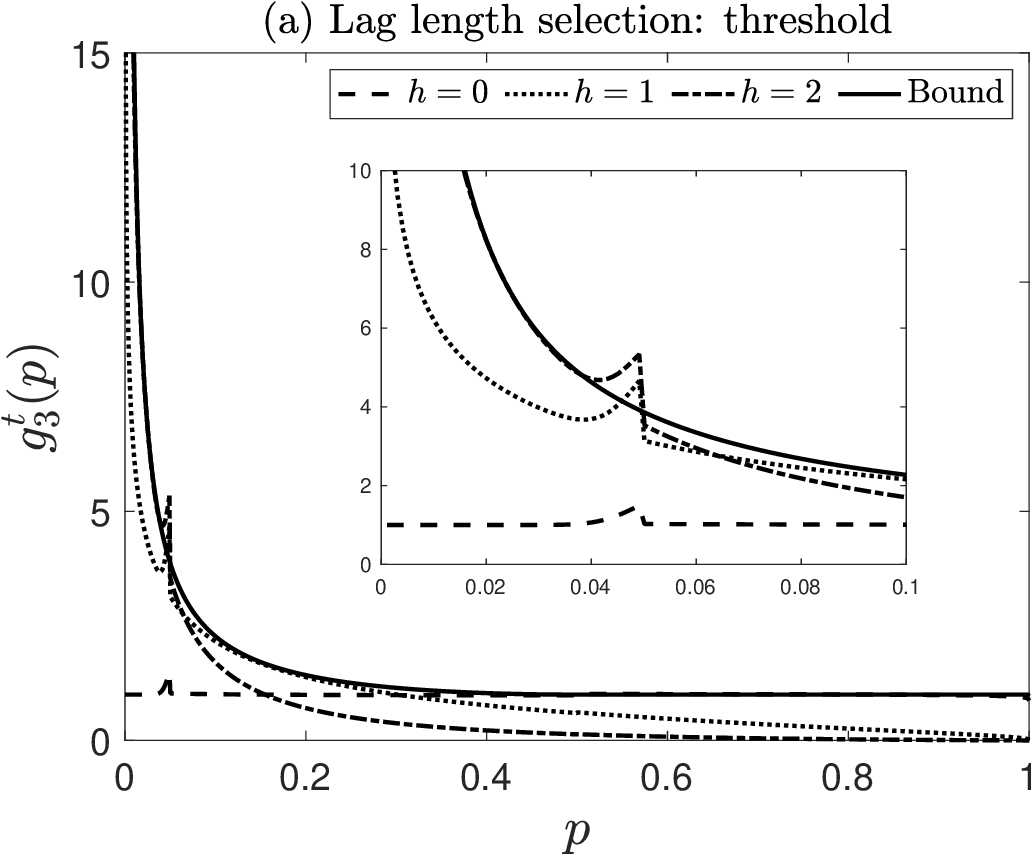}
\includegraphics[width=0.43\textwidth]{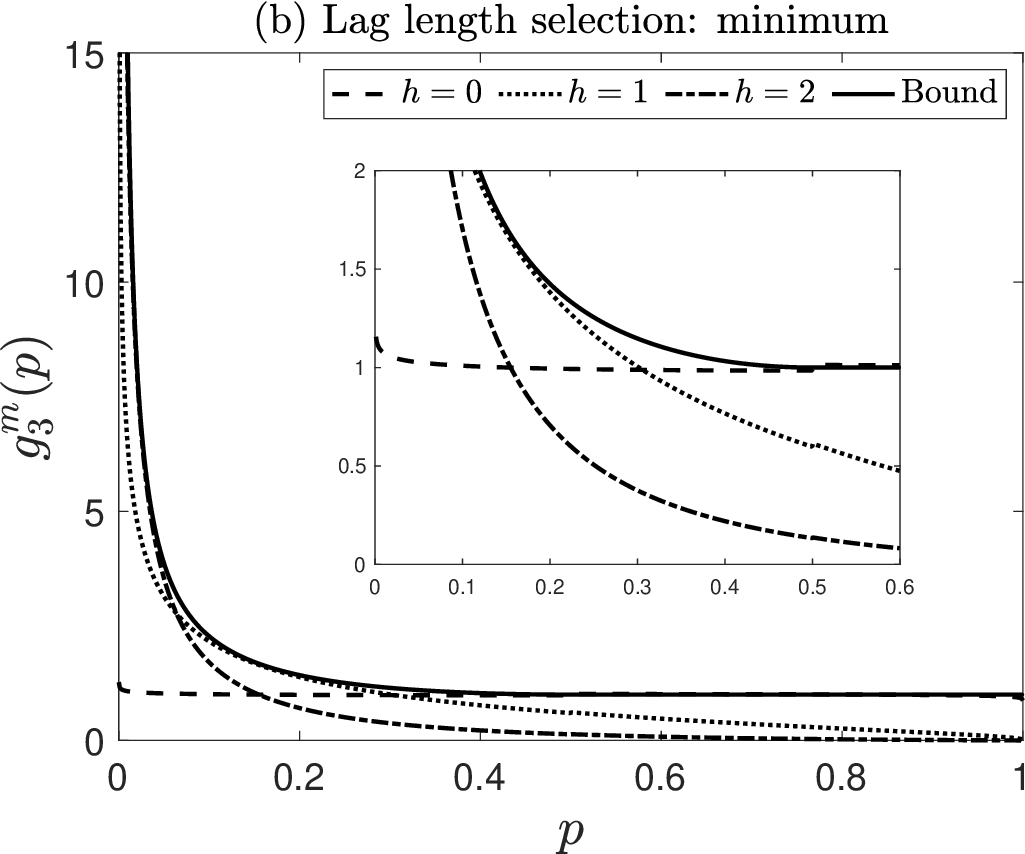}
\end{center}

\vspace*{-3mm}

\doublespacing

\textit{Notes:} Figures show the $p$-curves from lag length selection with $N=200$ and $\kappa=1/2$. Panel (a) shows the $p$-curves based on the threshold approach. Panel (b) shows the $p$-curves based on the minimum approach.

\end{figure}

\newpage

\begin{figure}[H]
\caption{Power curves for covariate selection with $K=3$}
\label{fig:power_cov_combined}

\vspace{-5mm}

\begin{center}

\begin{subfigure}[b]{\textwidth} 
\caption{General-to-specific} \label{fig:power_cov_combined_a}
\centering
\textbf{\small Thresholding}

\smallskip

\includegraphics[width=0.24\textwidth]{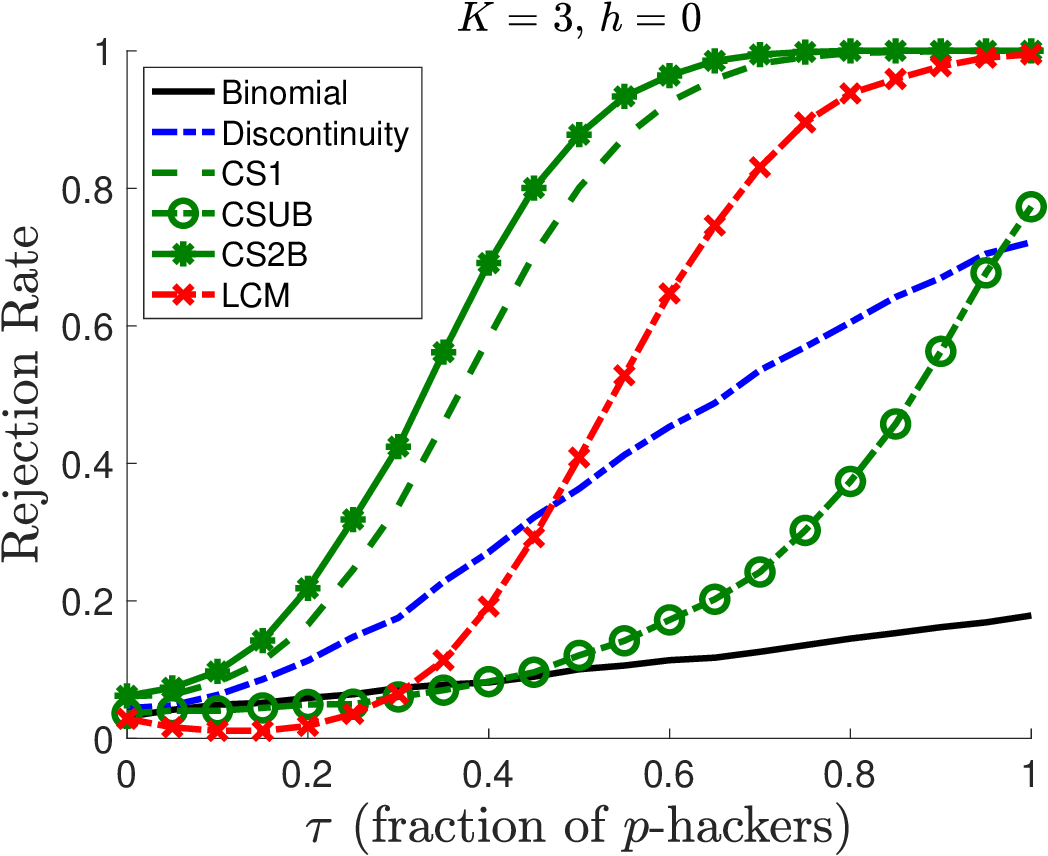}
\includegraphics[width=0.24\textwidth]{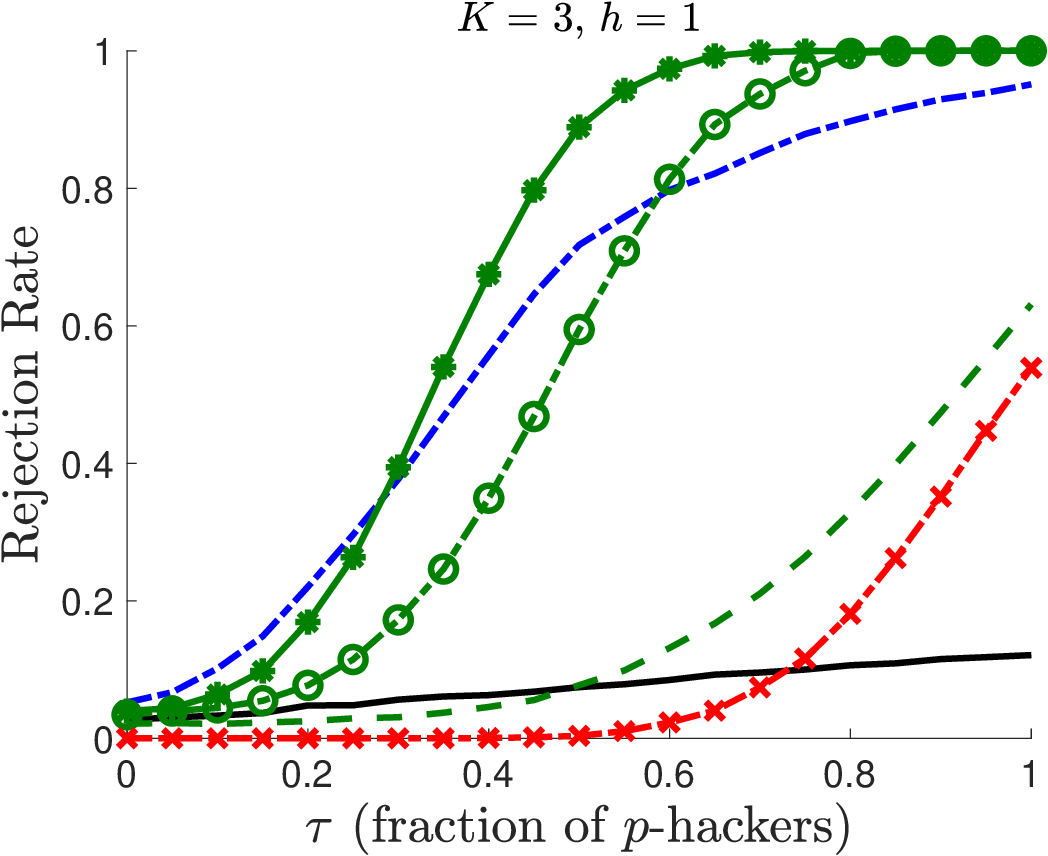}
\includegraphics[width=0.24\textwidth]{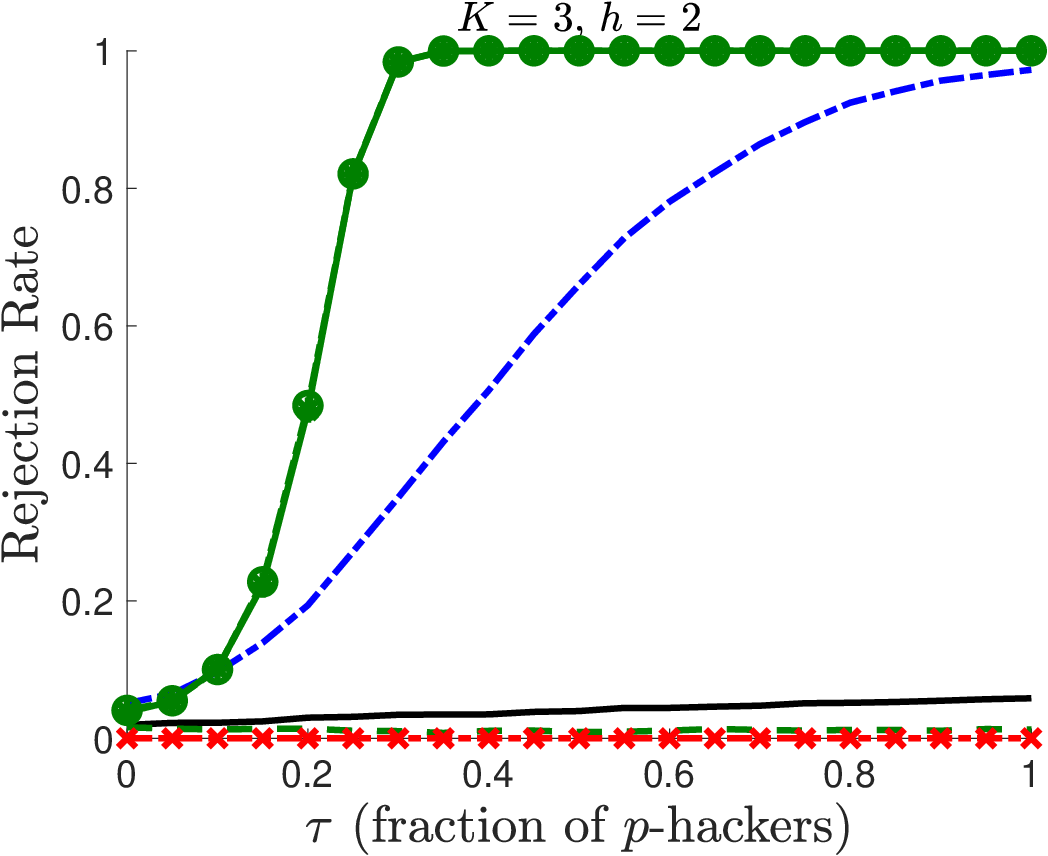}
\includegraphics[width=0.24\textwidth]{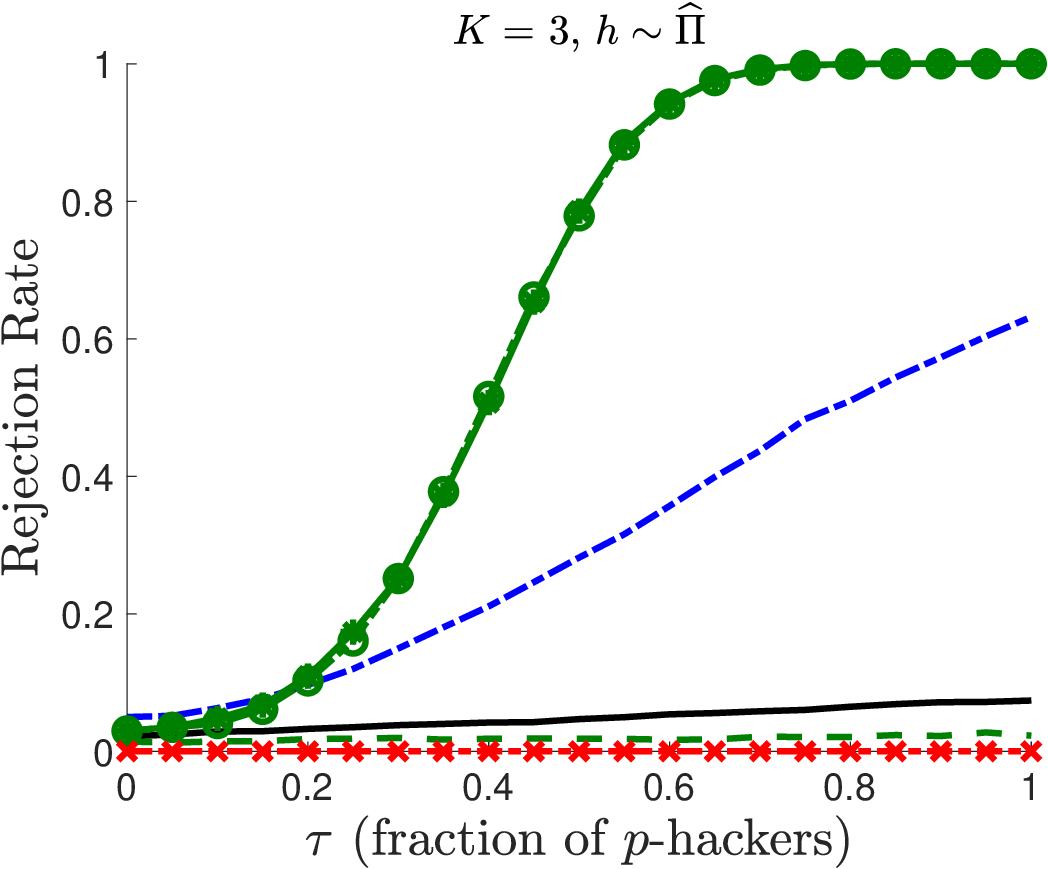}

\textbf{\small Minimum}

\smallskip

\includegraphics[width=0.24\textwidth]{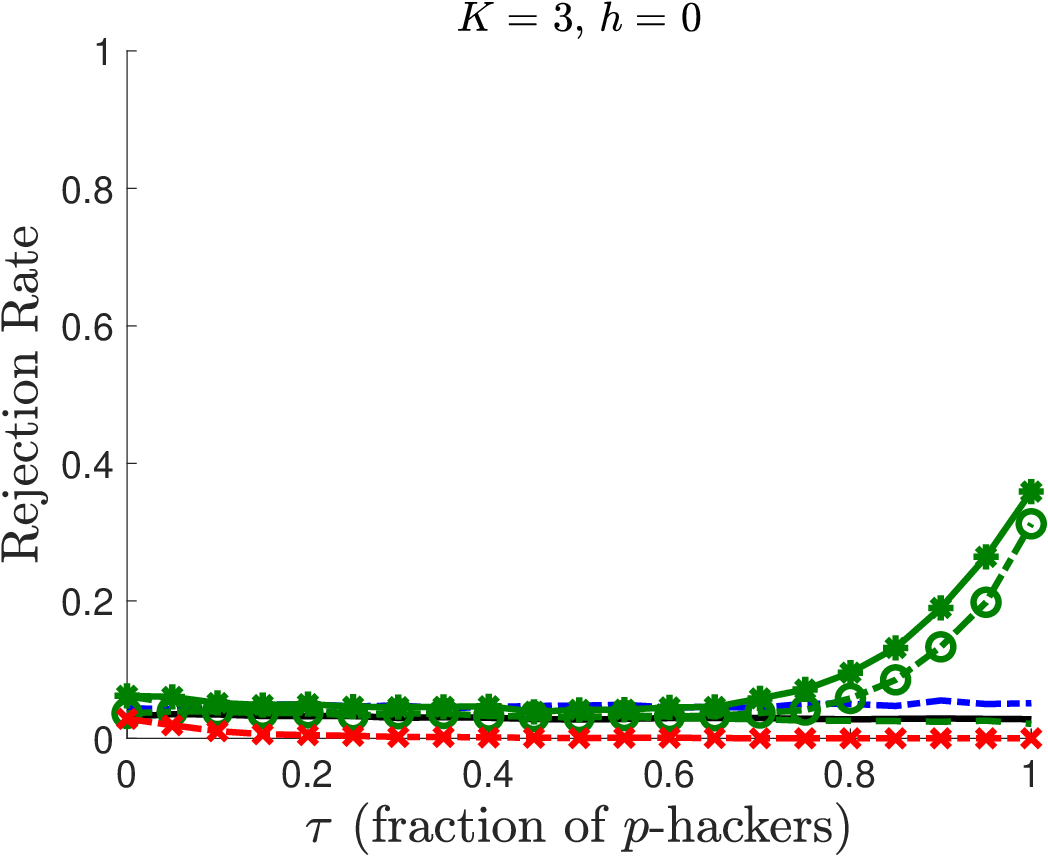}
\includegraphics[width=0.24\textwidth]{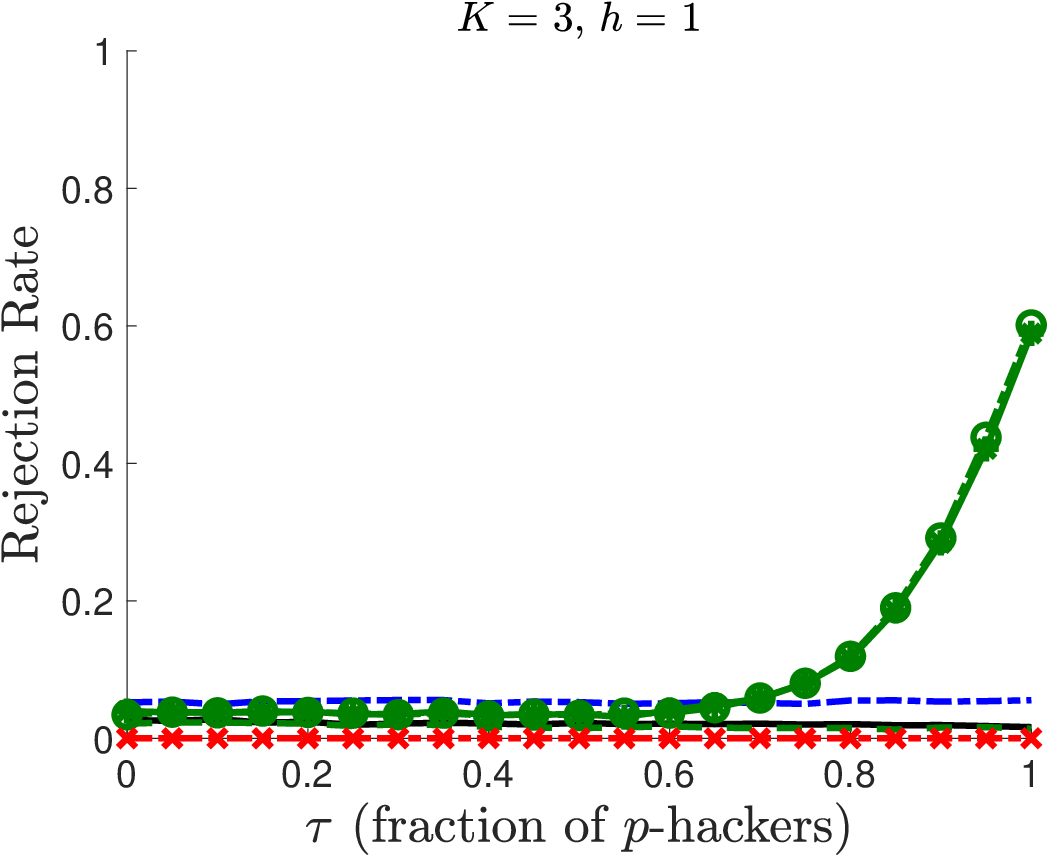}
\includegraphics[width=0.24\textwidth]{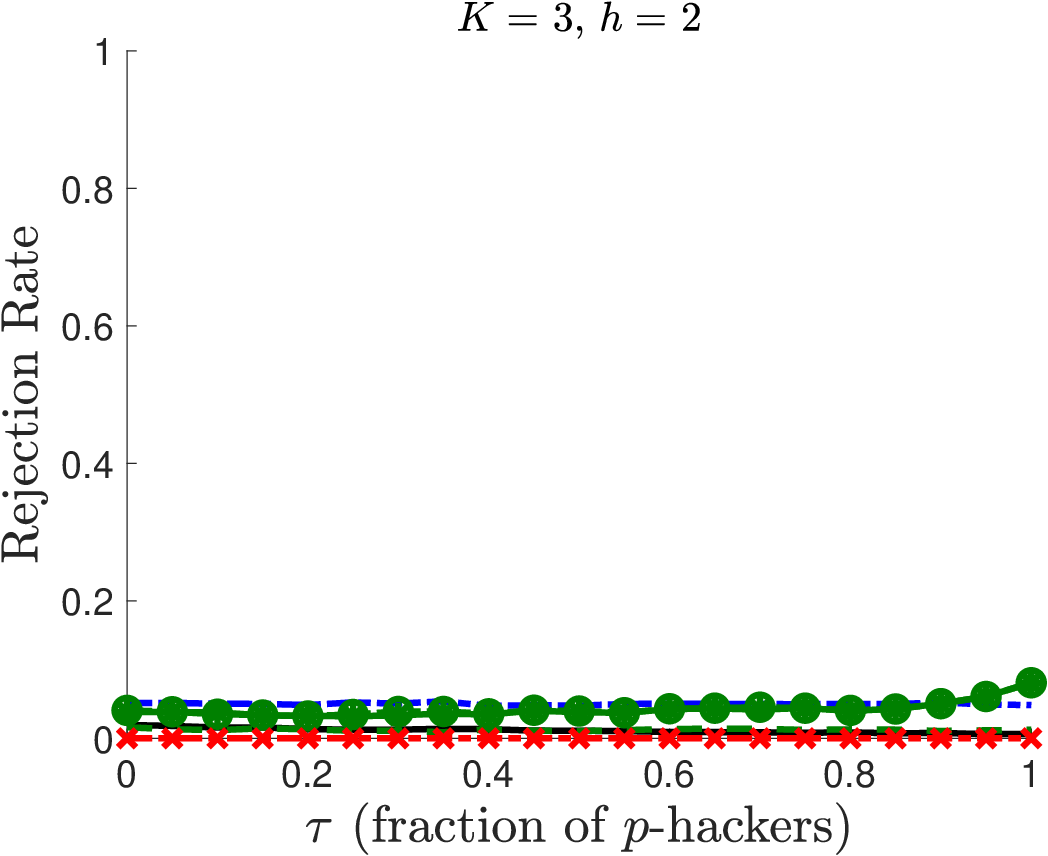}
\includegraphics[width=0.24\textwidth]{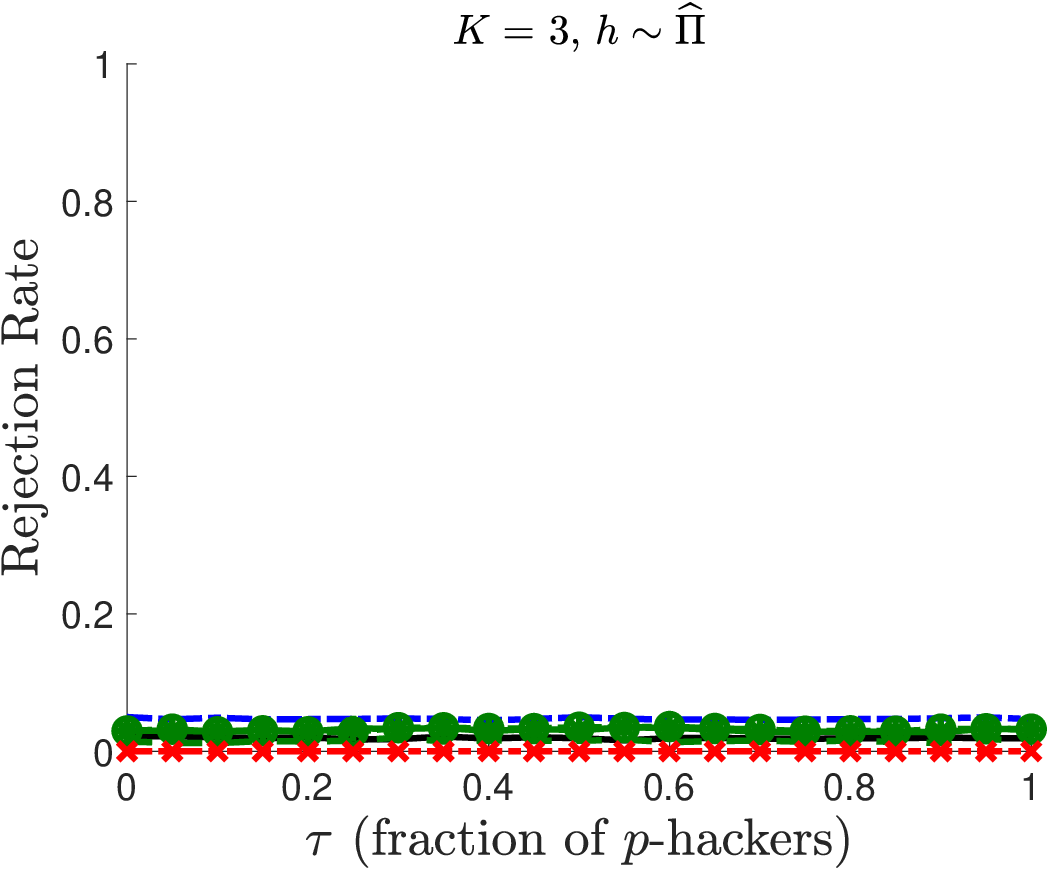}

\end{subfigure}

\vspace{3mm}

\begin{subfigure}[b]{\textwidth}
\caption{Specific-to-general}
\label{fig:power_cov_combined_b}
\centering
\textbf{\small Thresholding}

\smallskip

\includegraphics[width=0.24\textwidth]{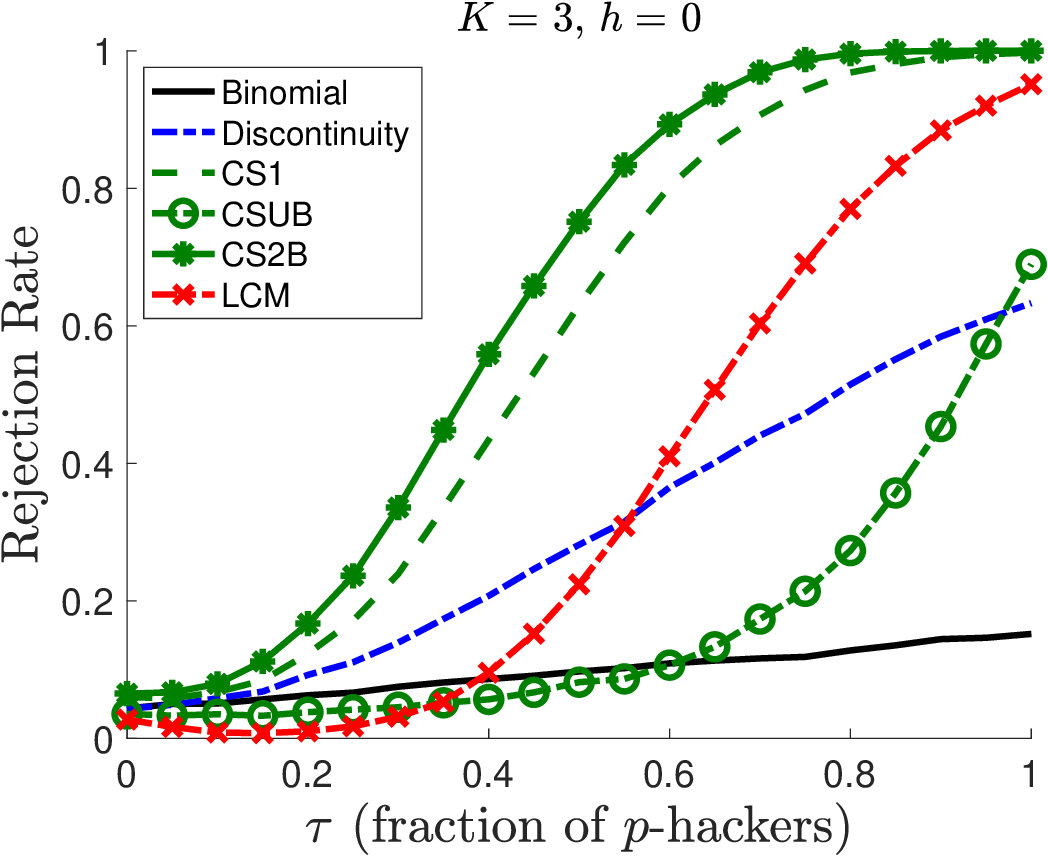}
\includegraphics[width=0.24\textwidth]{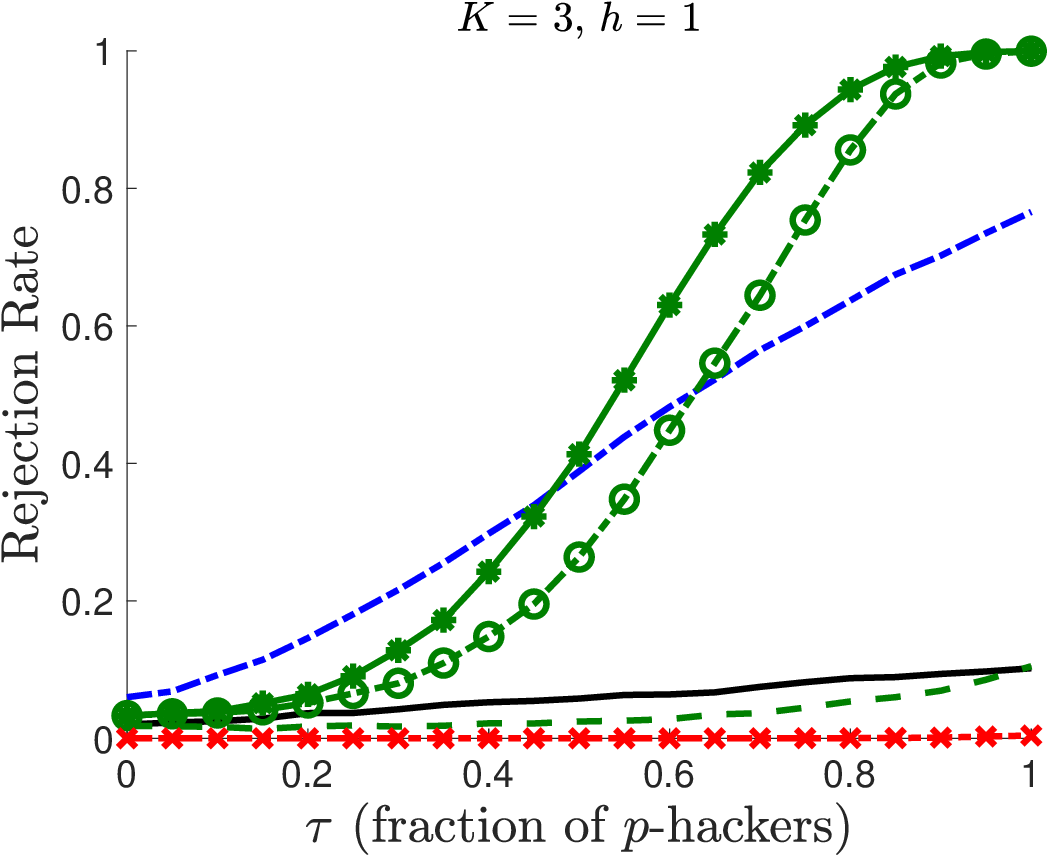}
\includegraphics[width=0.24\textwidth]{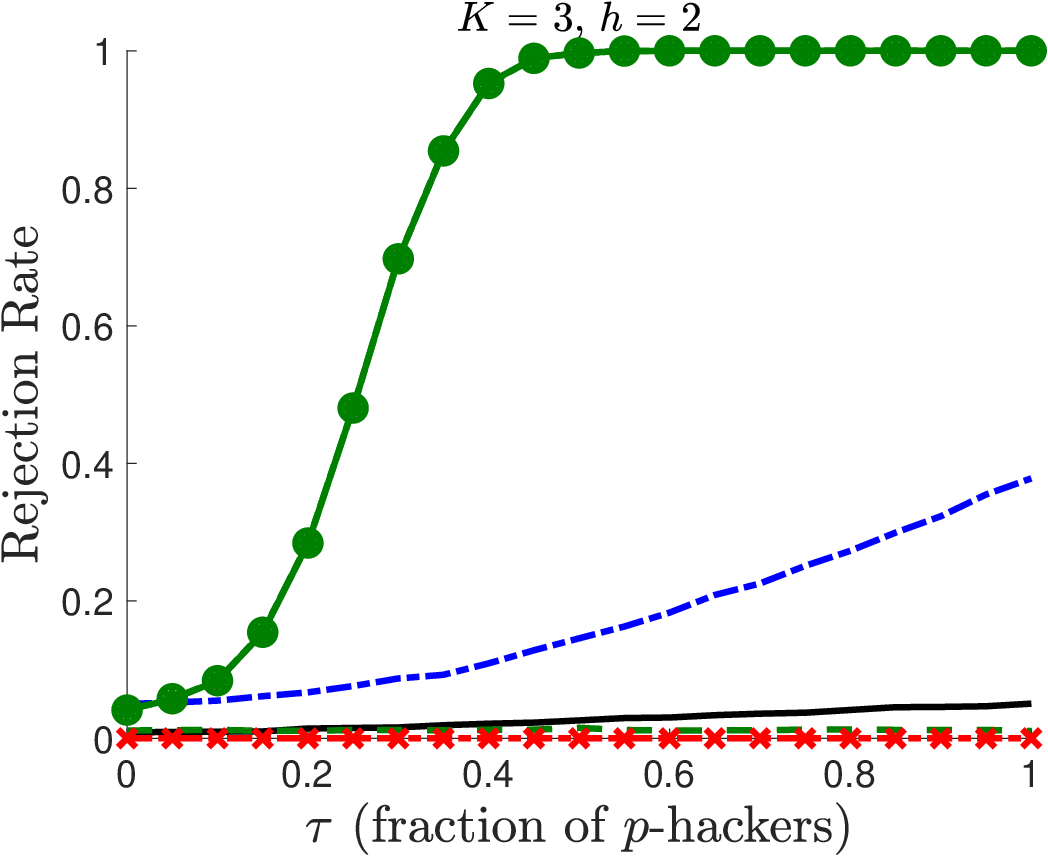}
\includegraphics[width=0.24\textwidth]{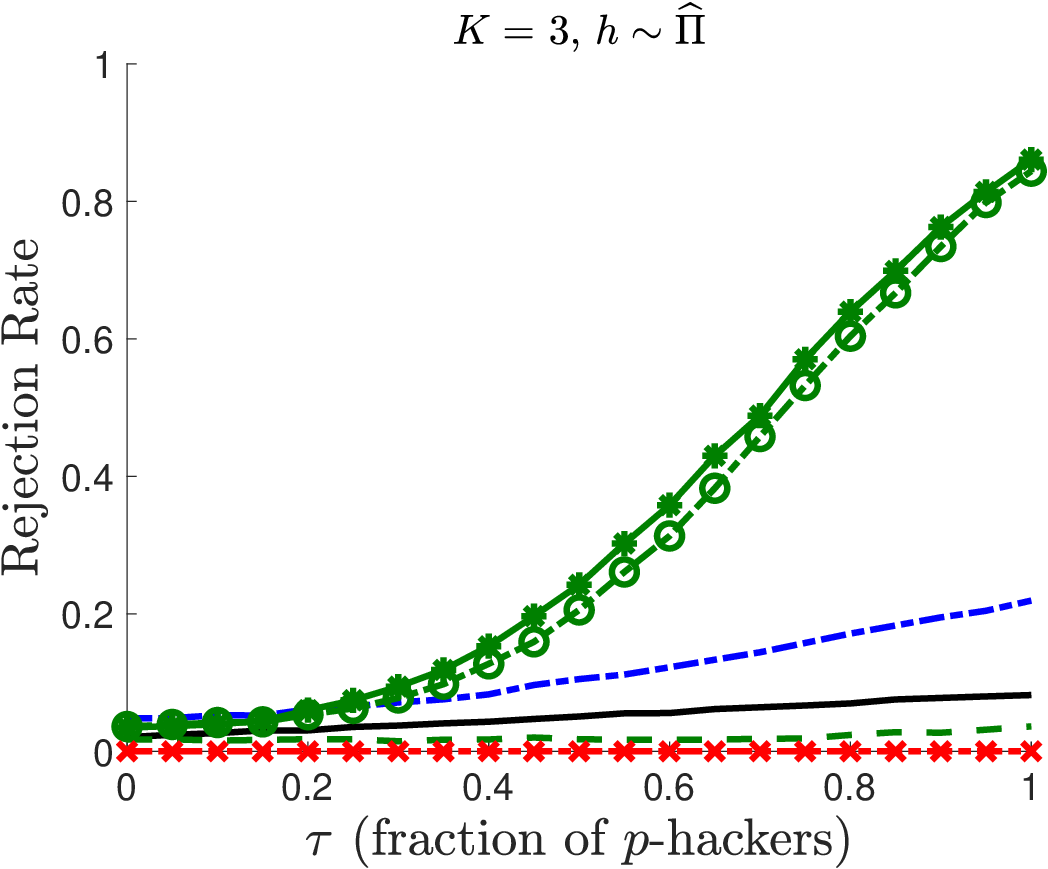}

\textbf{\small Minimum}

\smallskip

\includegraphics[width=0.24\textwidth]{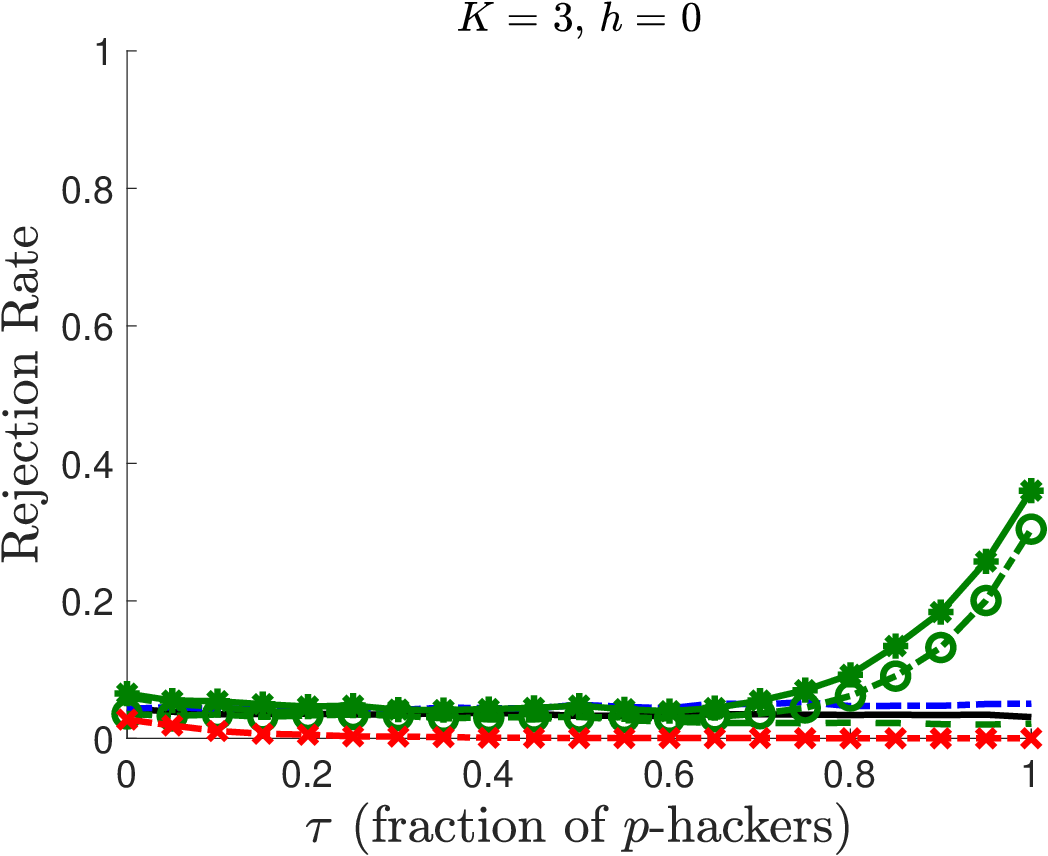}
\includegraphics[width=0.24\textwidth]{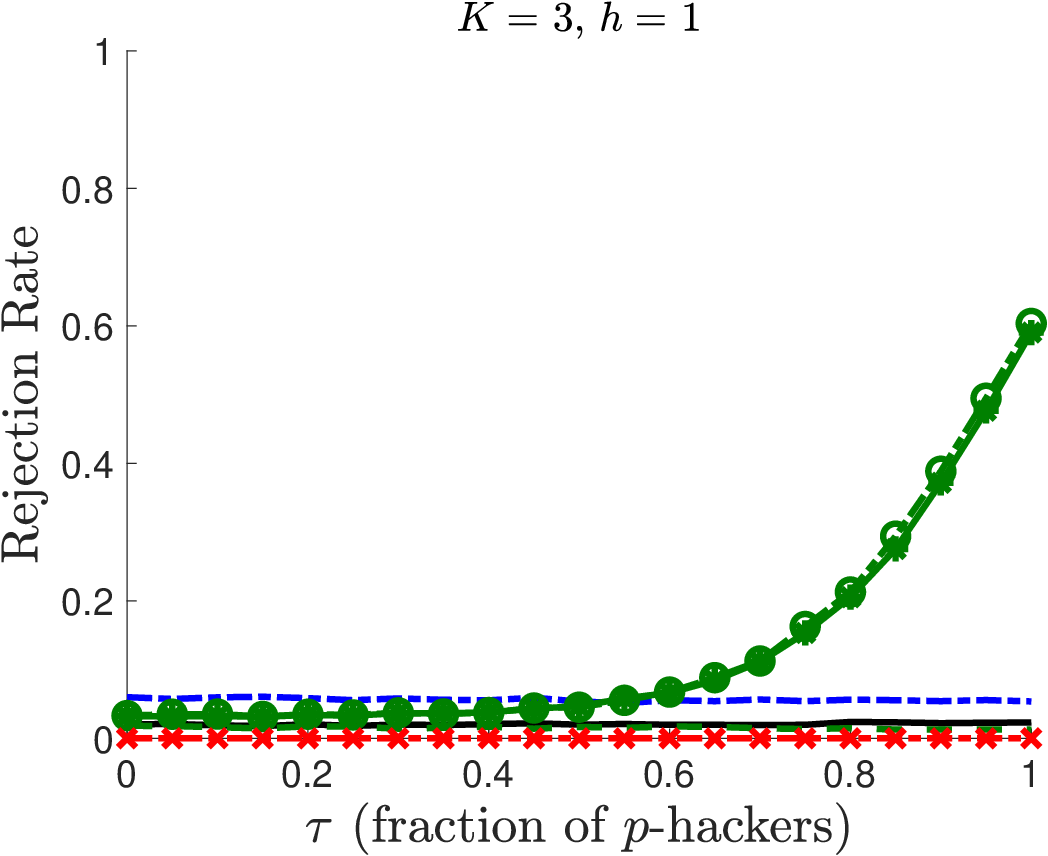}
\includegraphics[width=0.24\textwidth]{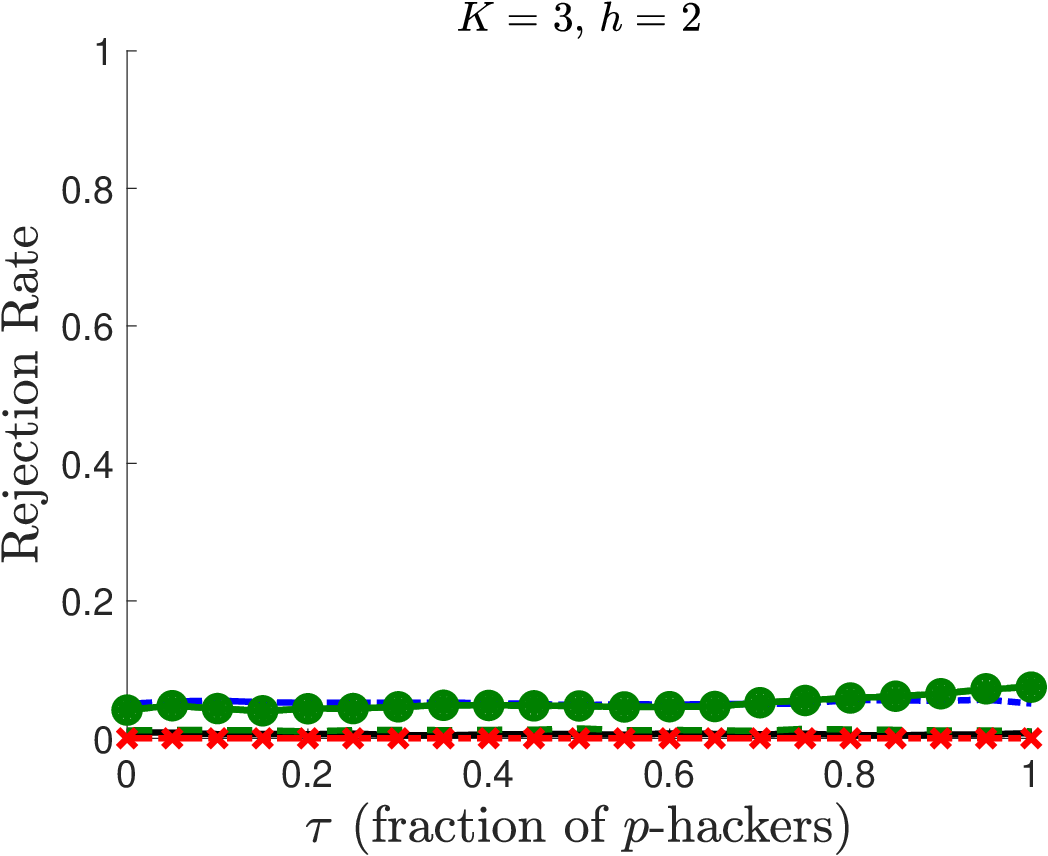}
\includegraphics[width=0.24\textwidth]{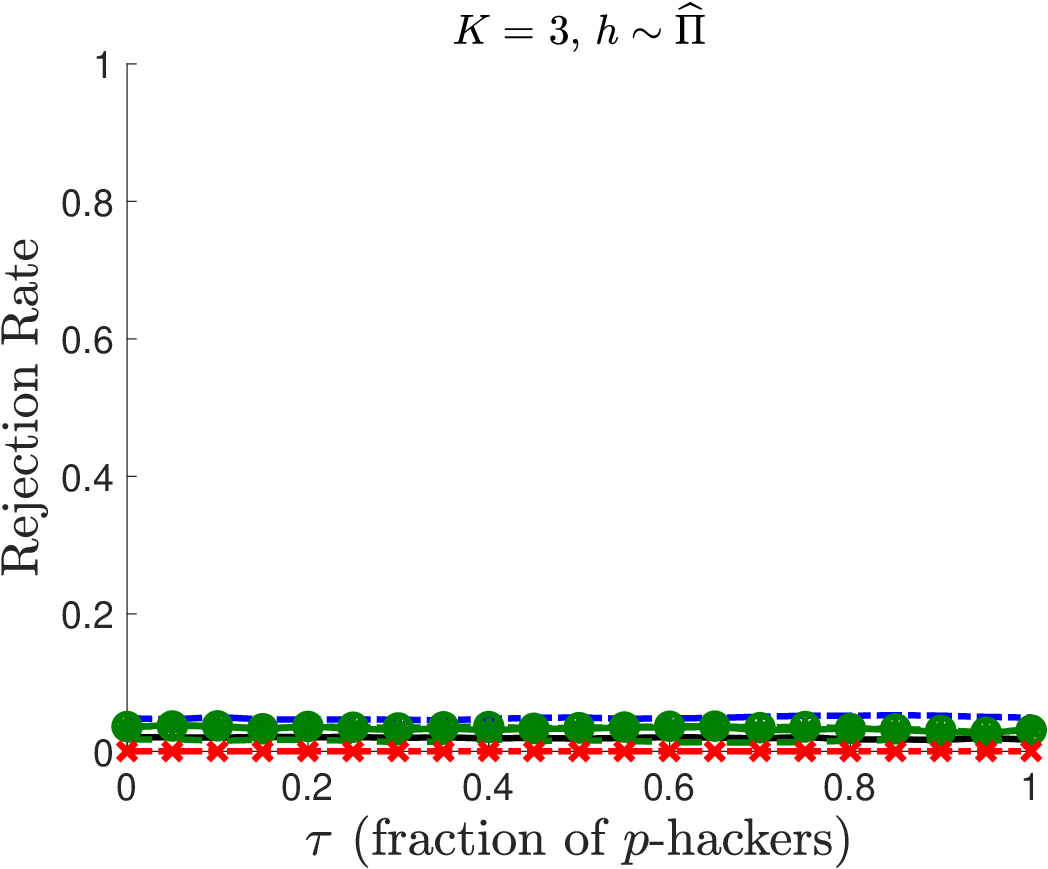}

\end{subfigure}
\end{center}
\vspace{-2mm}

\doublespacing
 \textit{Notes:} Figures show rejection rates of the tests in Table \ref{tab:tests} as a function of $\tau$ for the threshold and minimum approach. The simulation design is described in Sections \ref{sec:covariate_selection_MC} and \ref{sec:simulations_setup}. Figures \ref{fig:power_cov_combined_a} and \ref{fig:power_cov_combined_b} show the results for the general-to-specific and specific-to-general approach to $p$-hacking with two-sided tests, respectively. The results are based on 5,000 simulation repetitions.
\end{figure}

\newpage

\begin{figure}[H]
\caption{Power curves for IV selection with $K=3$}
\label{fig:power_iv_combined}

\vspace{-5mm}

\begin{center}

\begin{subfigure}[b]{\textwidth} 
\caption{All specifications} \label{fig:power_iv_combined_a}
\centering
\textbf{\small Thresholding}

\smallskip
\includegraphics[width=0.24\textwidth]{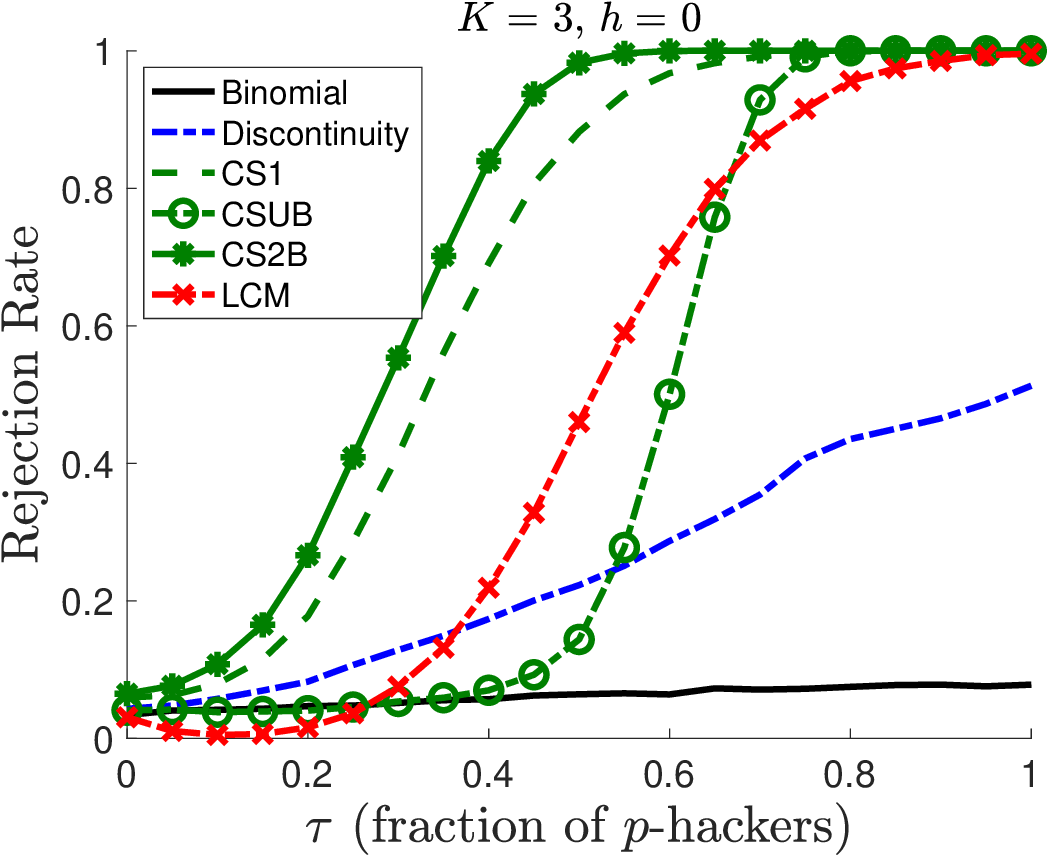}
\includegraphics[width=0.24\textwidth]{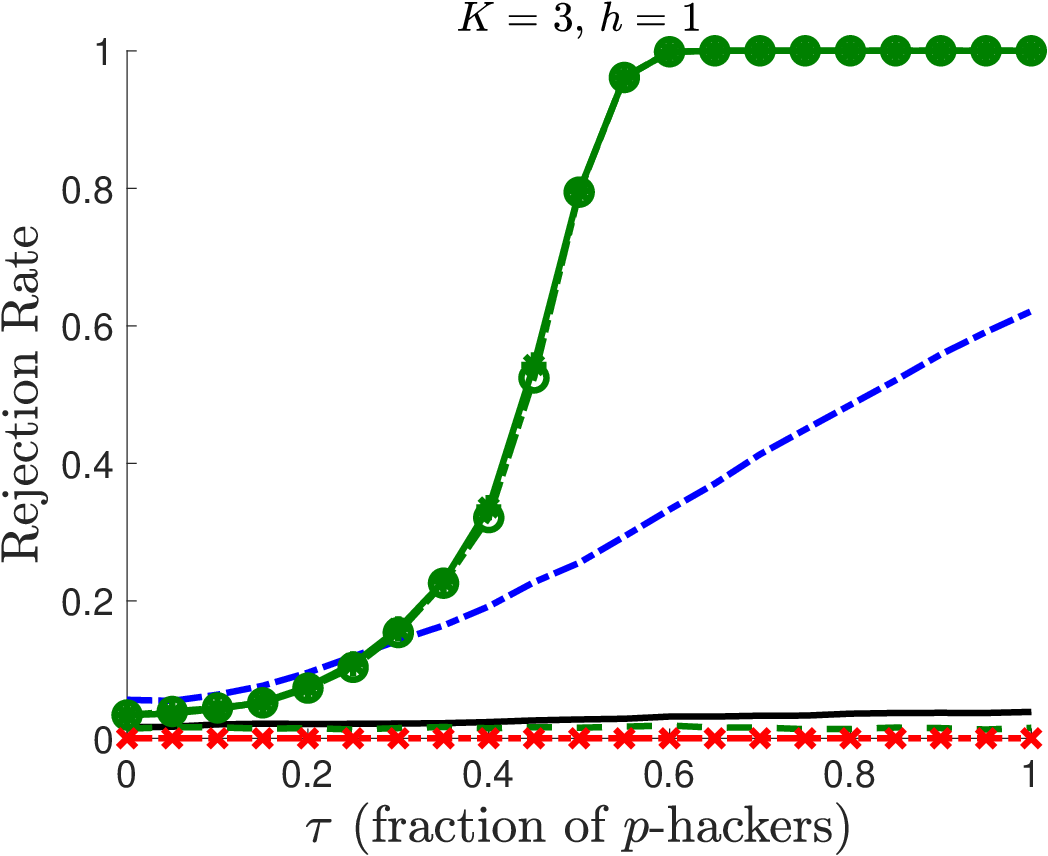}
\includegraphics[width=0.24\textwidth]{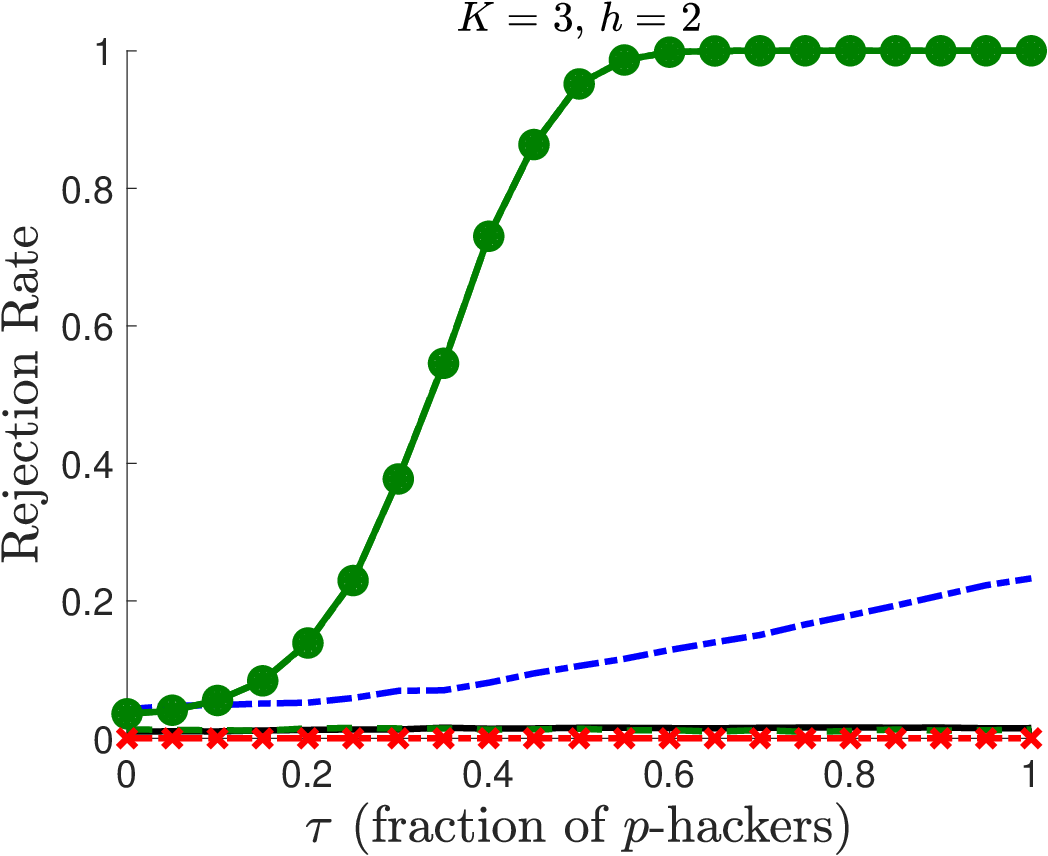}
\includegraphics[width=0.24\textwidth]{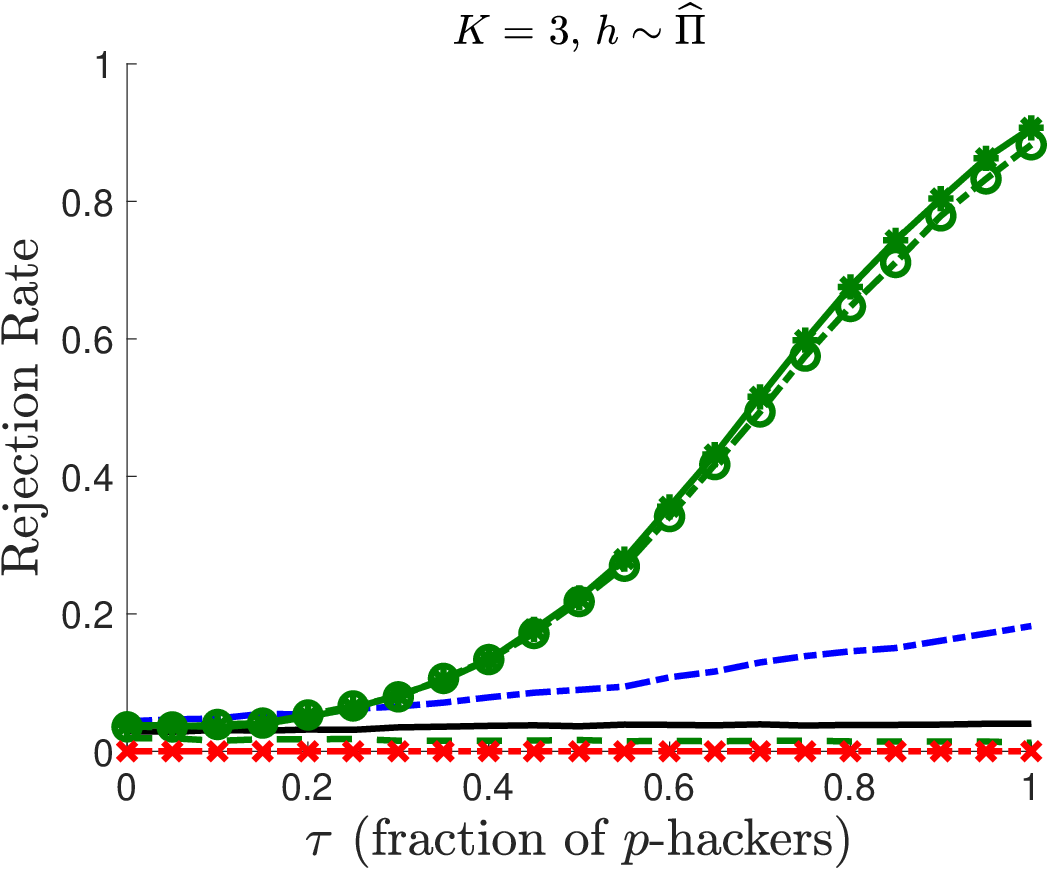}

\textbf{\small Minimum}

\smallskip
\includegraphics[width=0.24\textwidth]{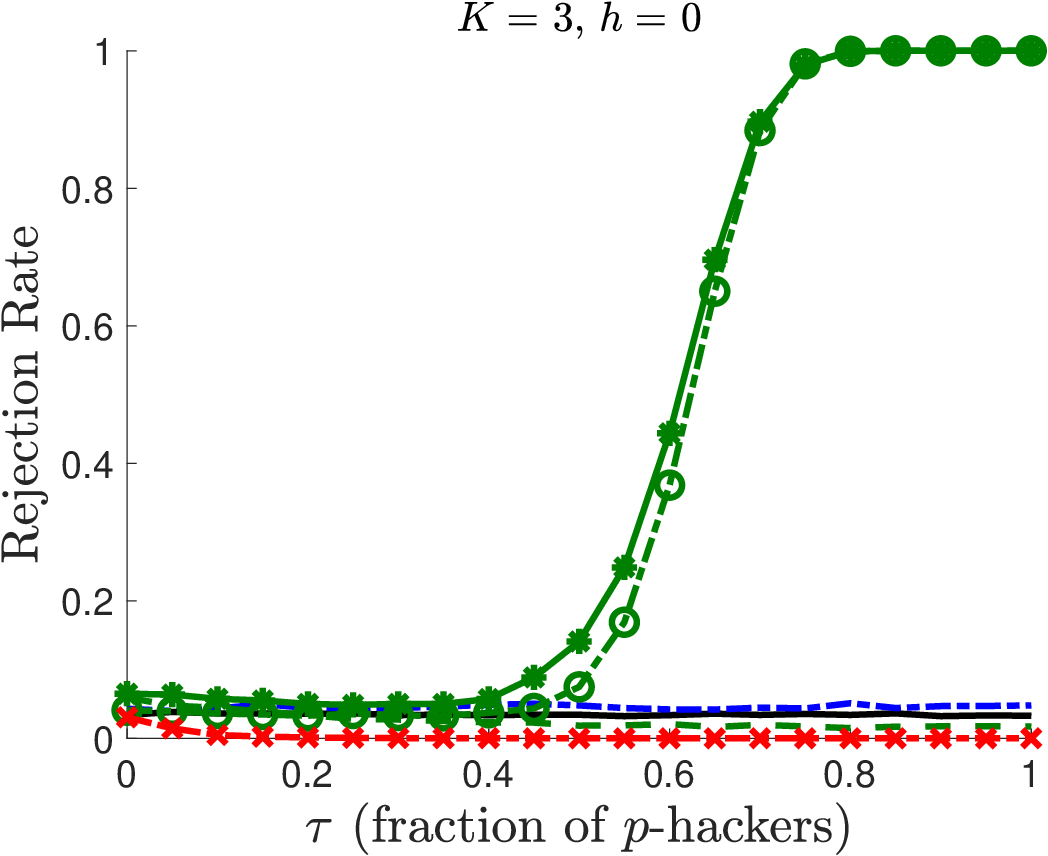}
\includegraphics[width=0.24\textwidth]{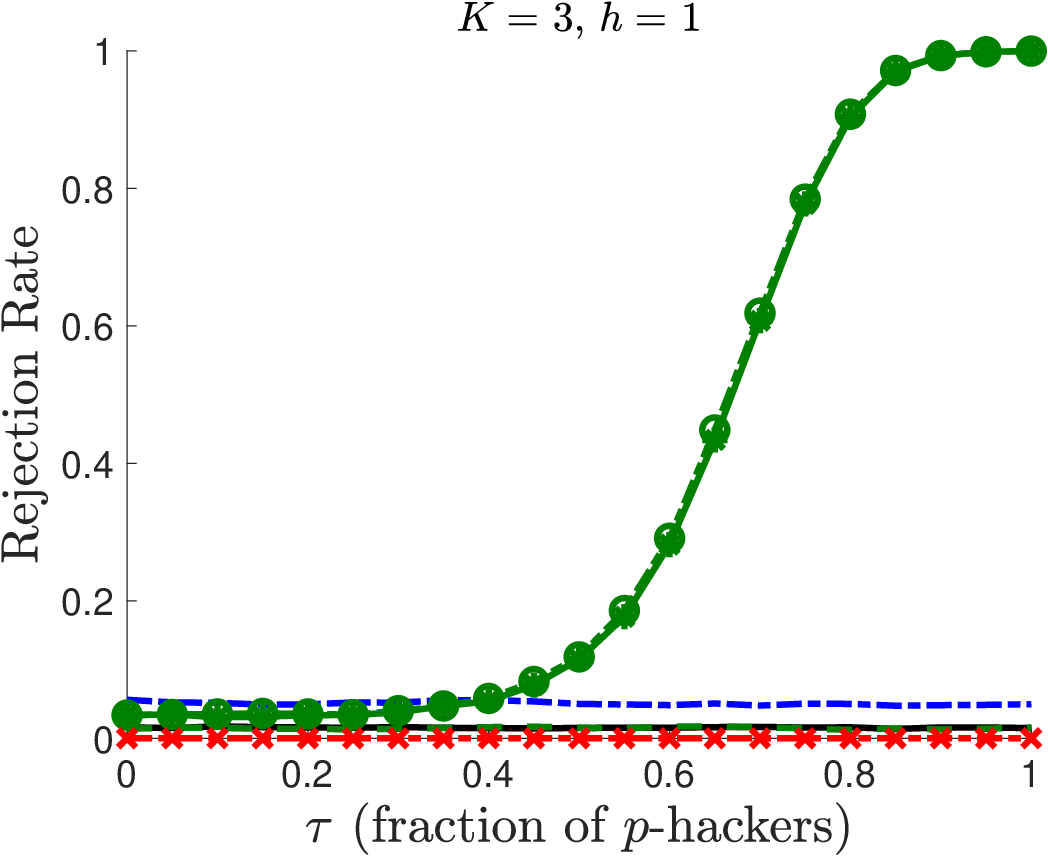}
\includegraphics[width=0.24\textwidth]{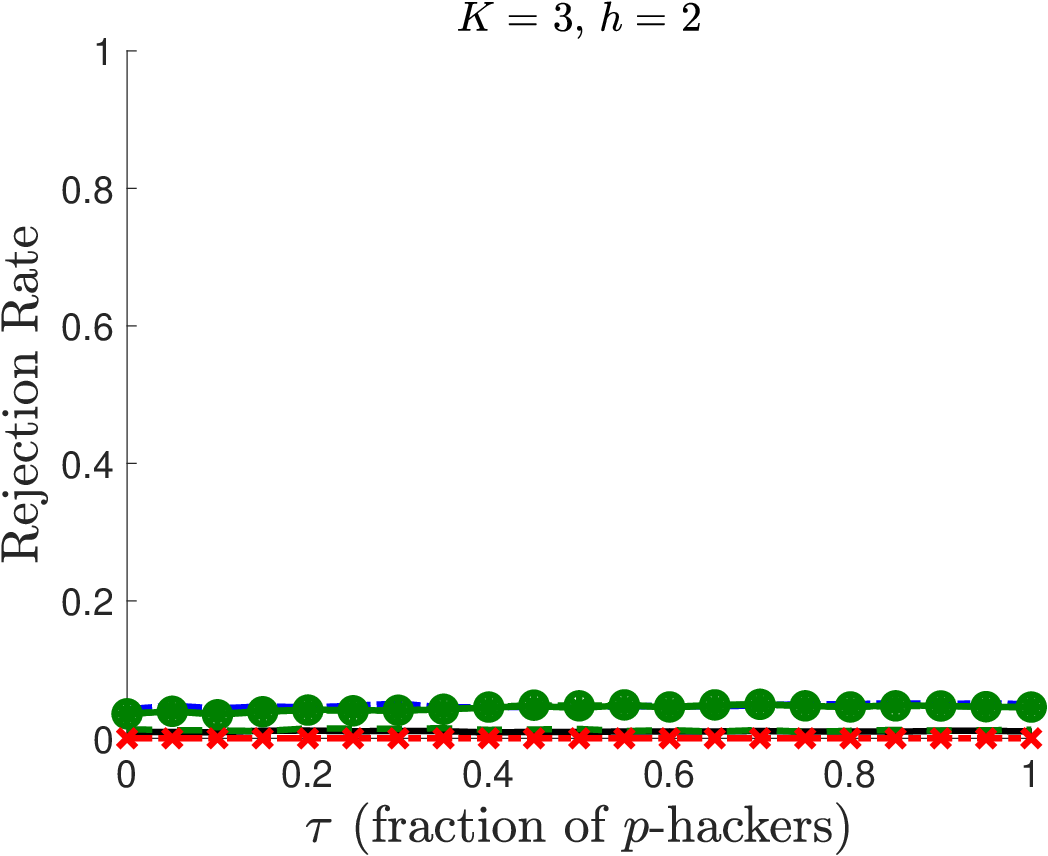}
\includegraphics[width=0.24\textwidth]{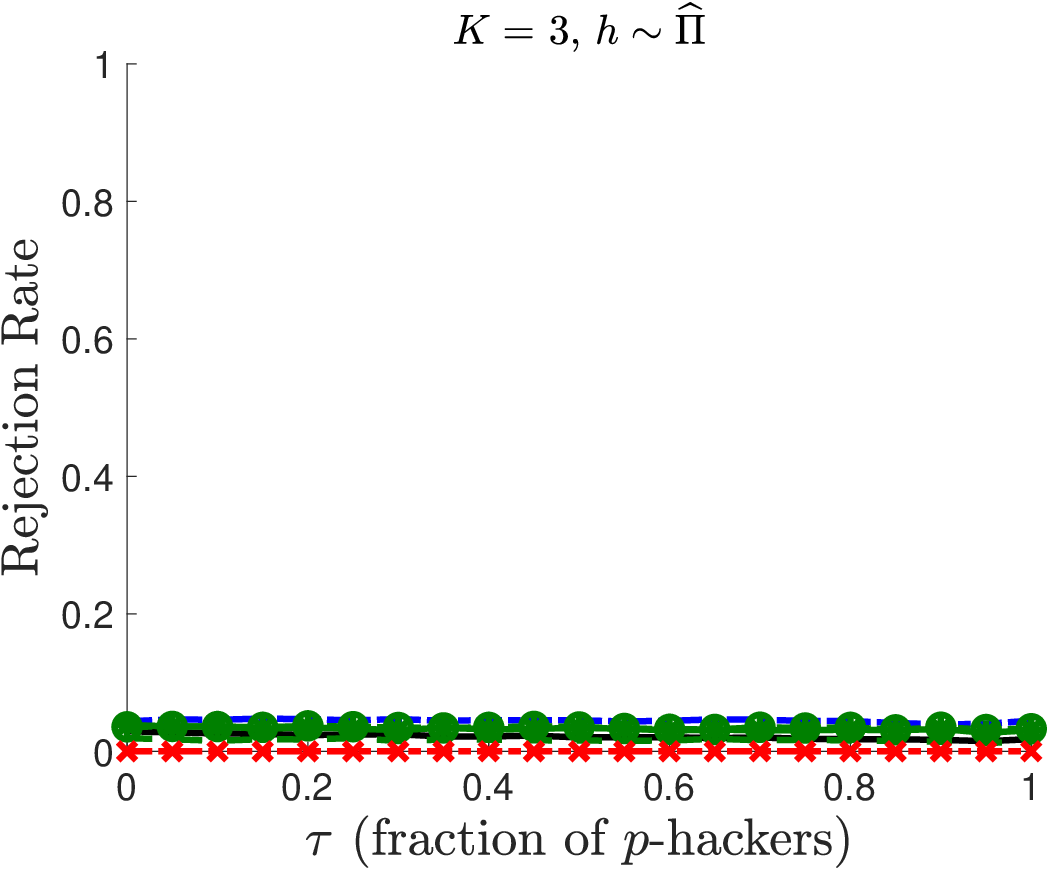}

\end{subfigure}

\vspace{3mm}

\begin{subfigure}[b]{\textwidth}
\caption{Specifications with $F>10$}
\label{fig:power_iv_combined_b}
\centering
\textbf{\small Thresholding}

\smallskip
\includegraphics[width=0.24\textwidth]{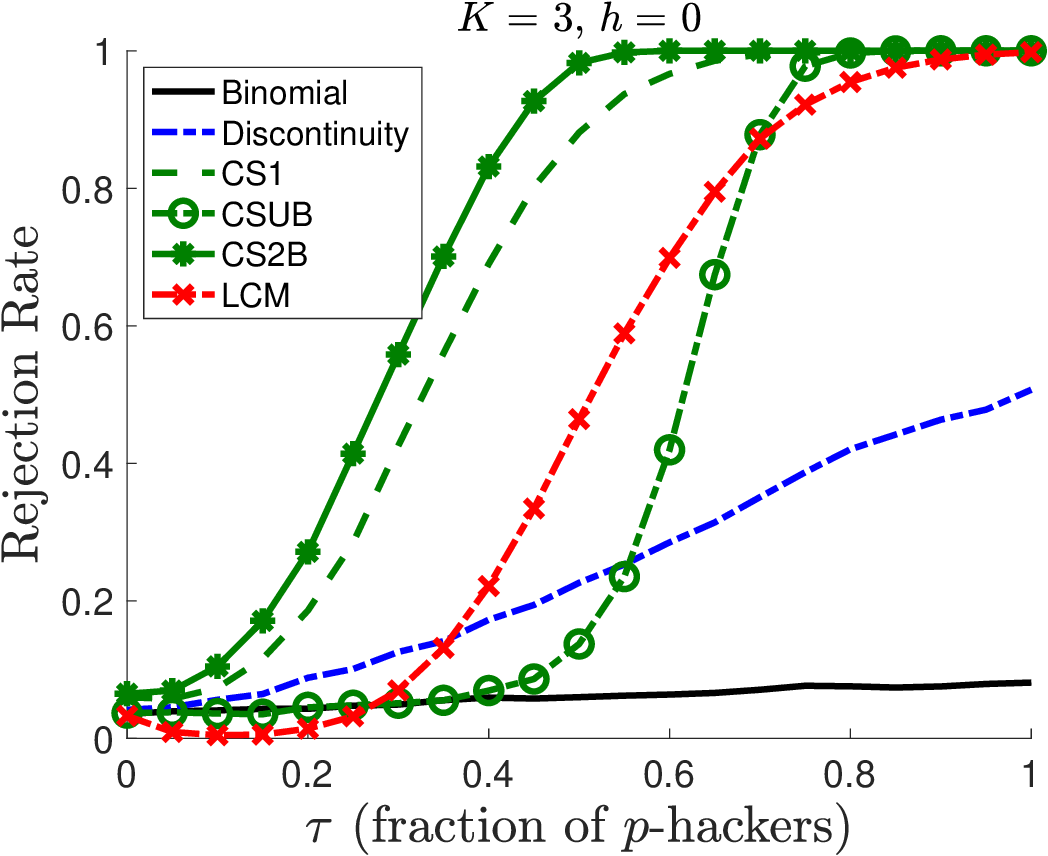}
\includegraphics[width=0.24\textwidth]{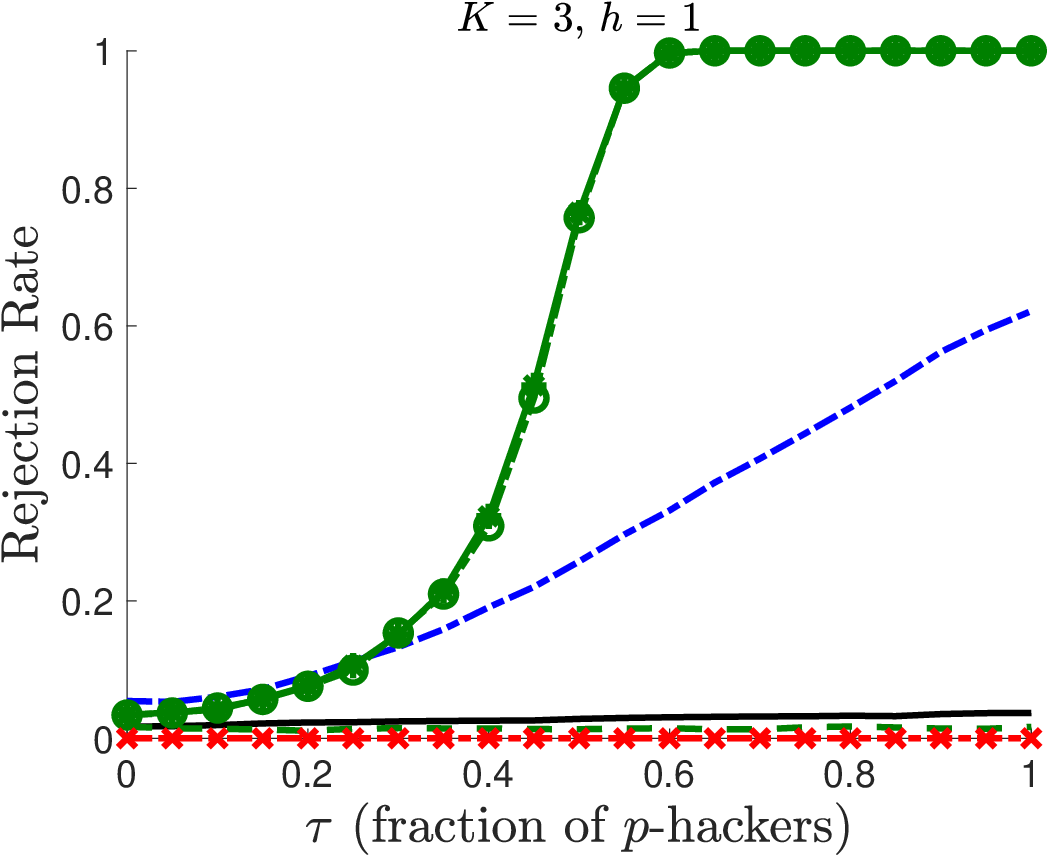}
\includegraphics[width=0.24\textwidth]{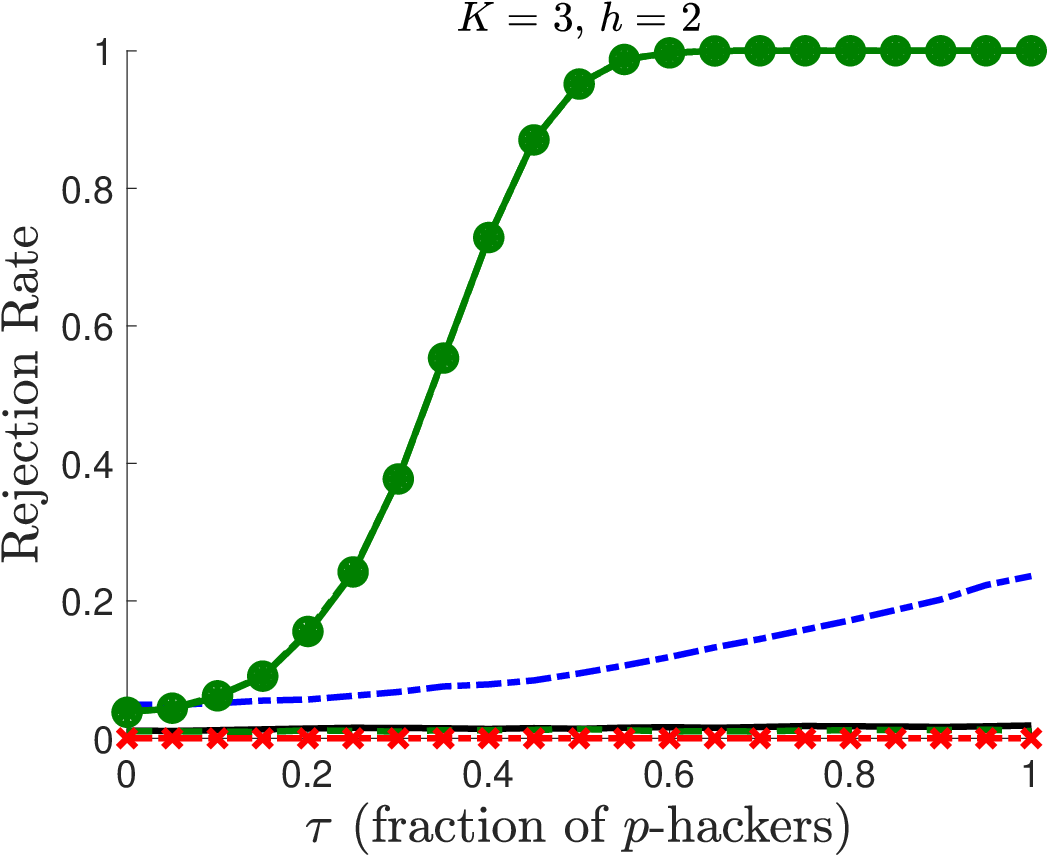}
\includegraphics[width=0.24\textwidth]{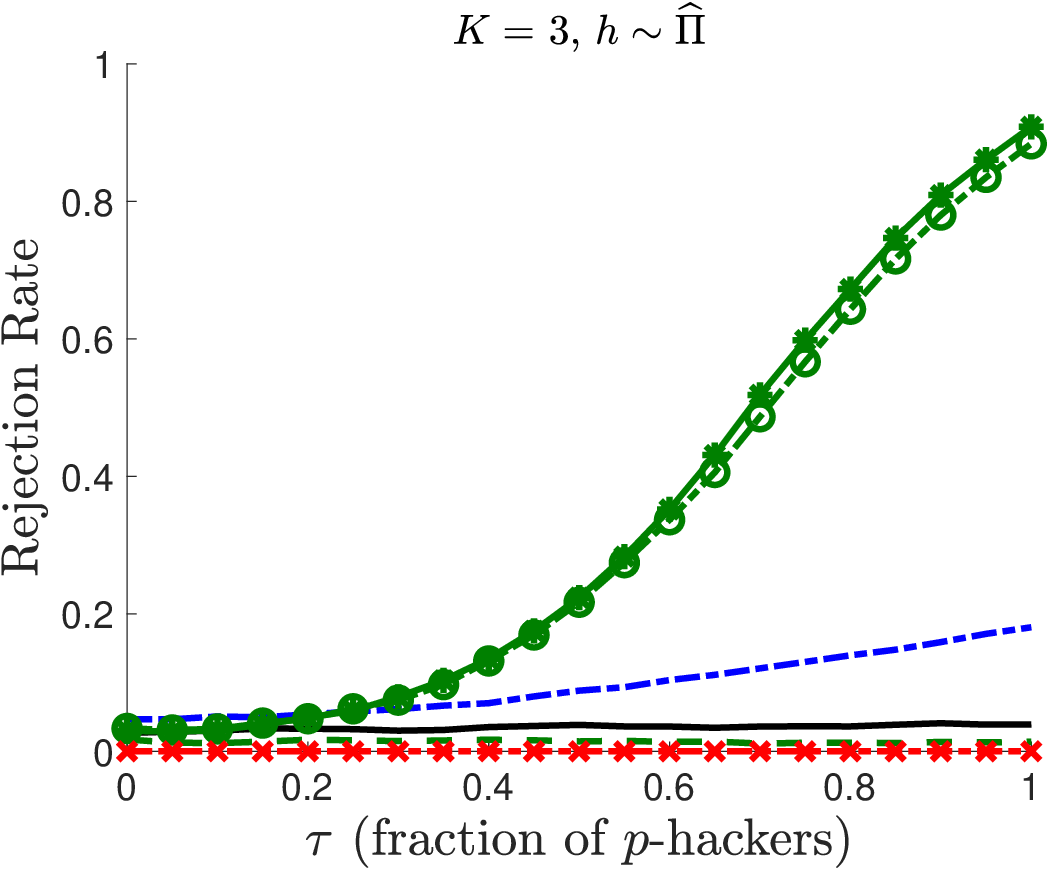}

\textbf{\small Minimum}

\smallskip

\includegraphics[width=0.24\textwidth]{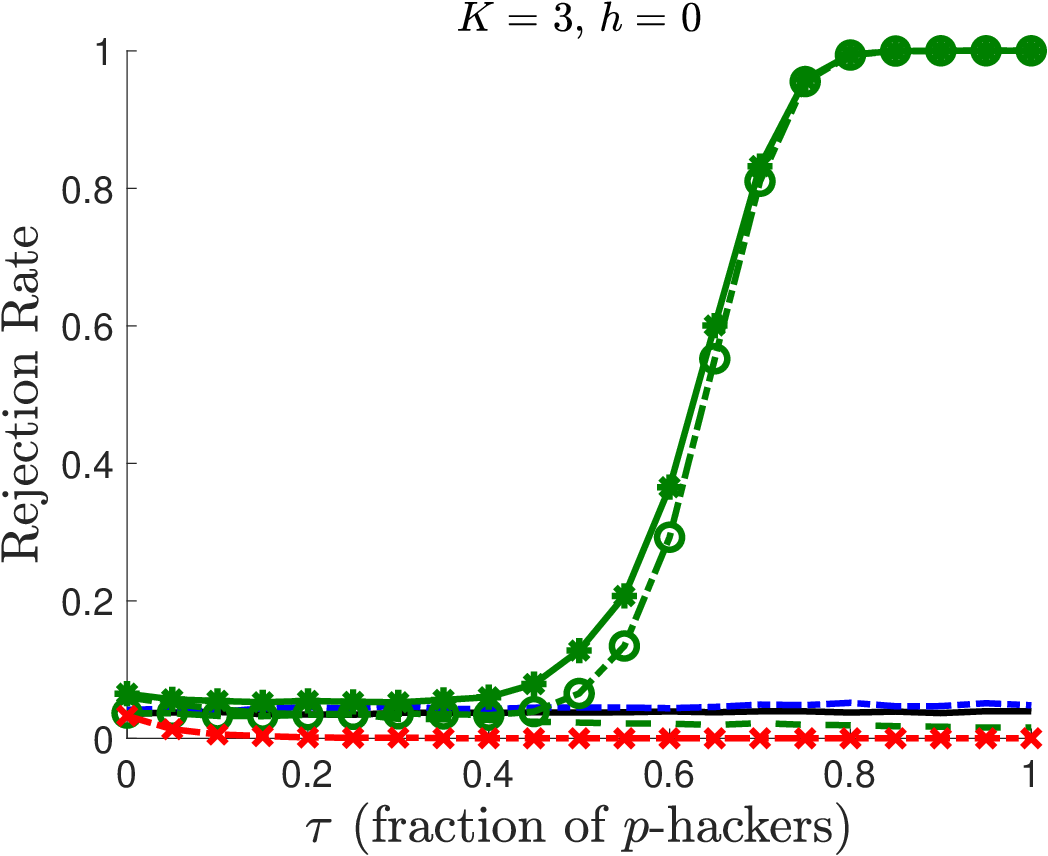}
\includegraphics[width=0.24\textwidth]{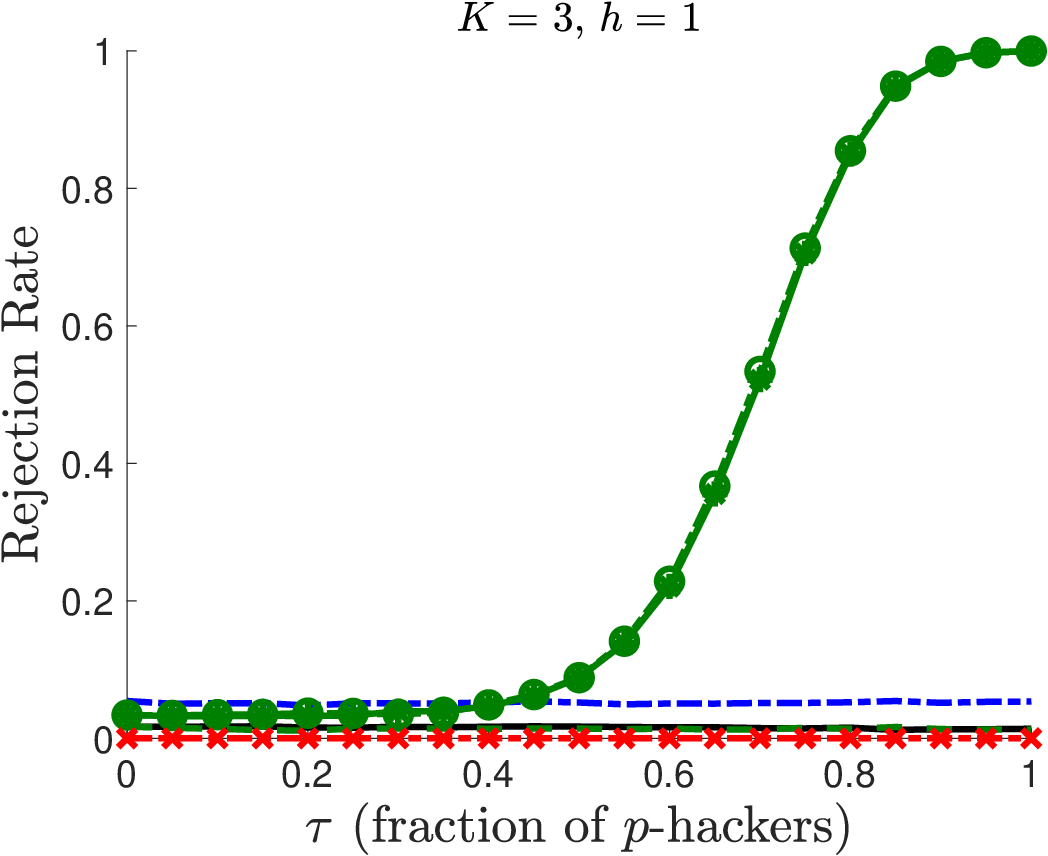}
\includegraphics[width=0.24\textwidth]{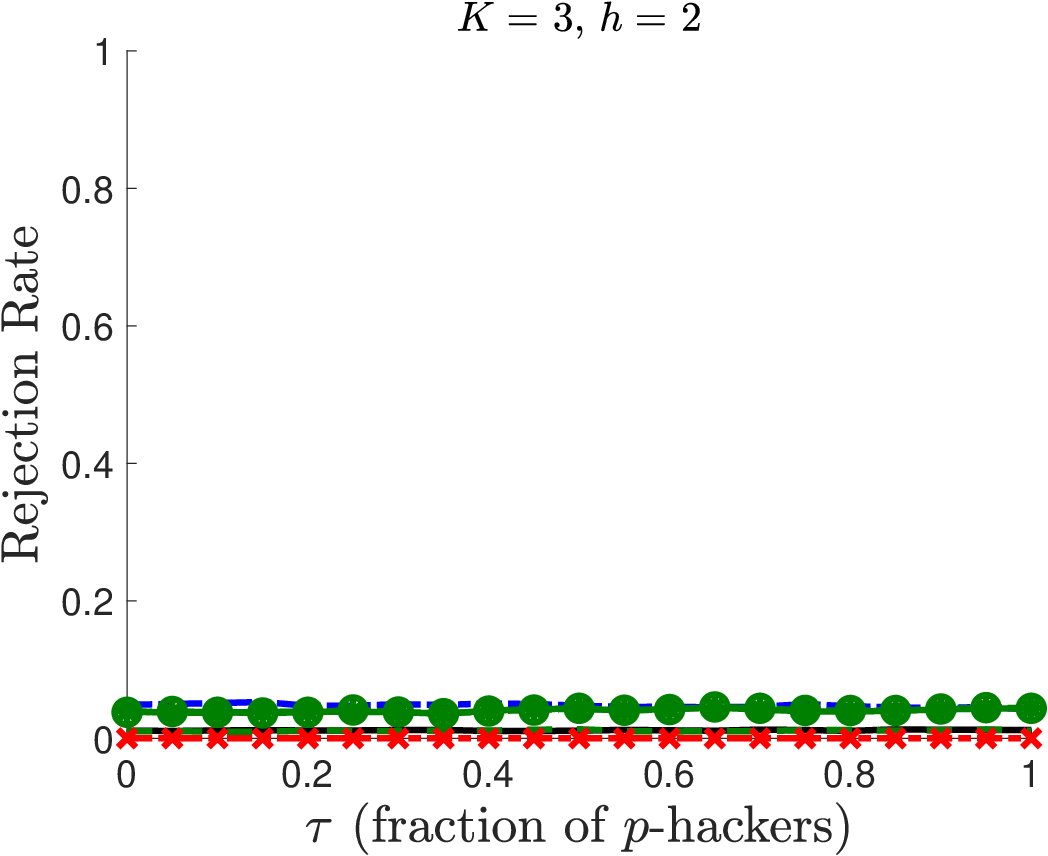}
\includegraphics[width=0.24\textwidth]{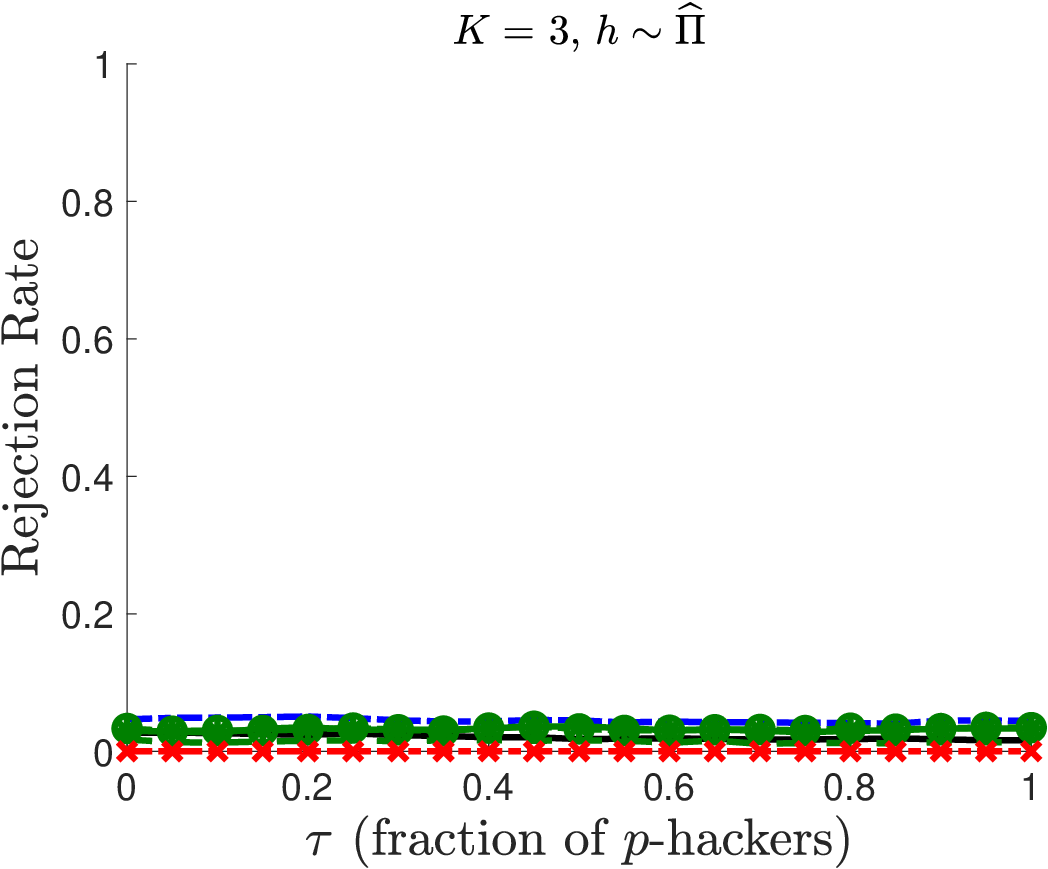}

\end{subfigure}
\end{center}
\vspace{-2mm}

\doublespacing

\textit{Notes:} Figures show rejection rates of the tests in Table \ref{tab:tests} as a function of $\tau$ for the threshold and minimum approach. The simulation design is described in Sections \ref{sec:iv_selection_MC} and \ref{sec:simulations_setup}. Figure \ref{fig:power_iv_combined_a} shows the results for all specifications. Figure \ref{fig:power_iv_combined_b} shows the results for specifications with $F>10$. The results are based on 5,000 simulation repetitions.
\end{figure}

\newpage

\begin{figure}[H]
\caption{Power curves for standard error selection}
\label{fig:power_se_combined}

\vspace{-5mm}

\begin{center}

\begin{subfigure}[b]{\textwidth} 
\caption{Lag length selection} \label{fig:power_se_combined_a}
\centering
\textbf{\small Thresholding}

\smallskip

\includegraphics[width=0.24\textwidth]{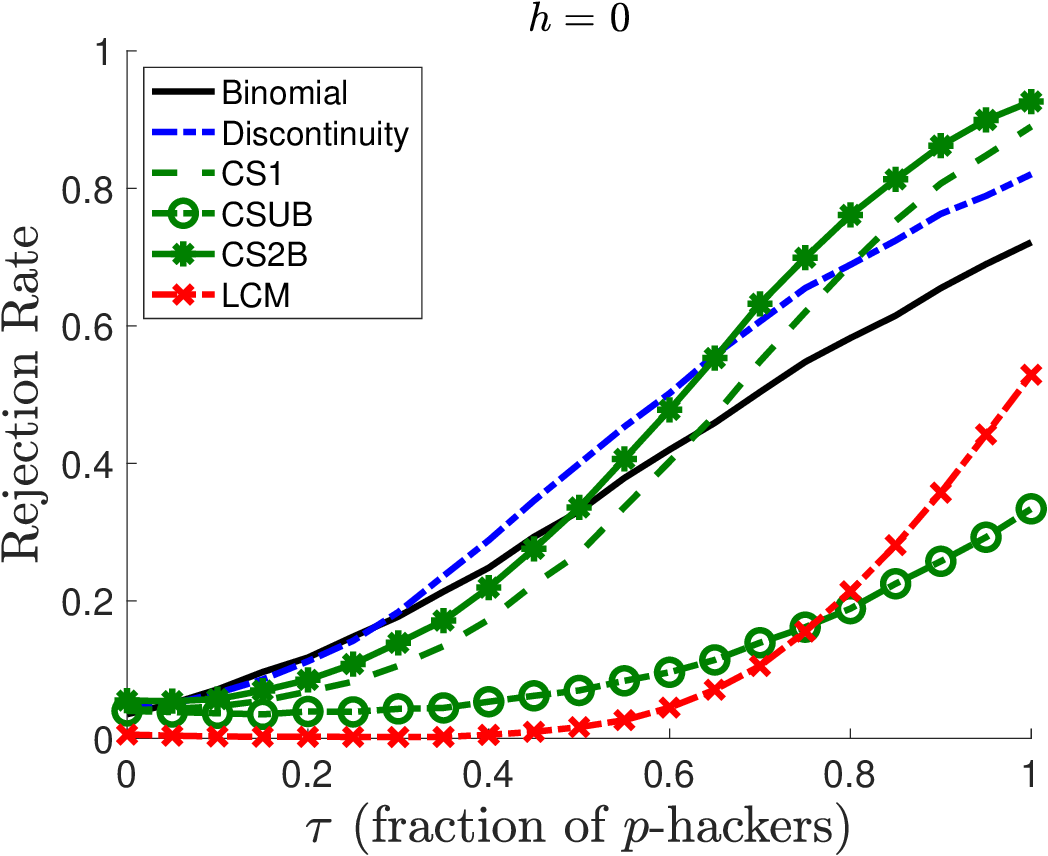}
\includegraphics[width=0.24\textwidth]{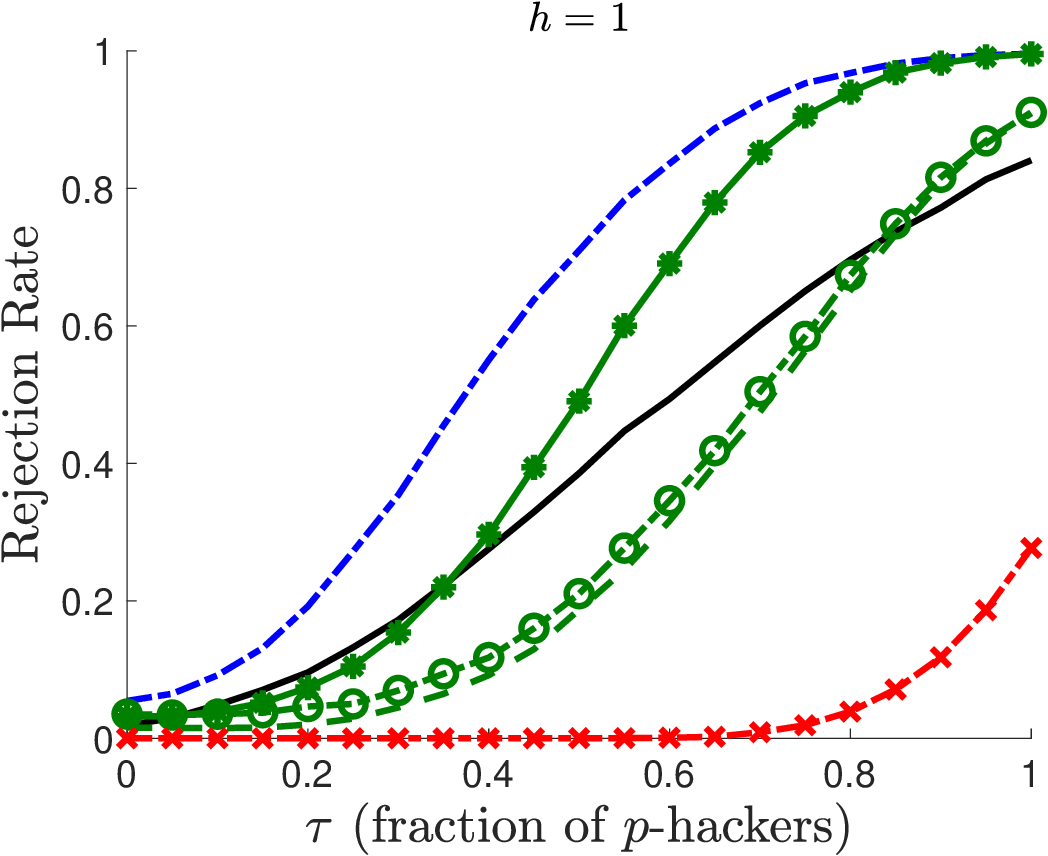}
\includegraphics[width=0.24\textwidth]{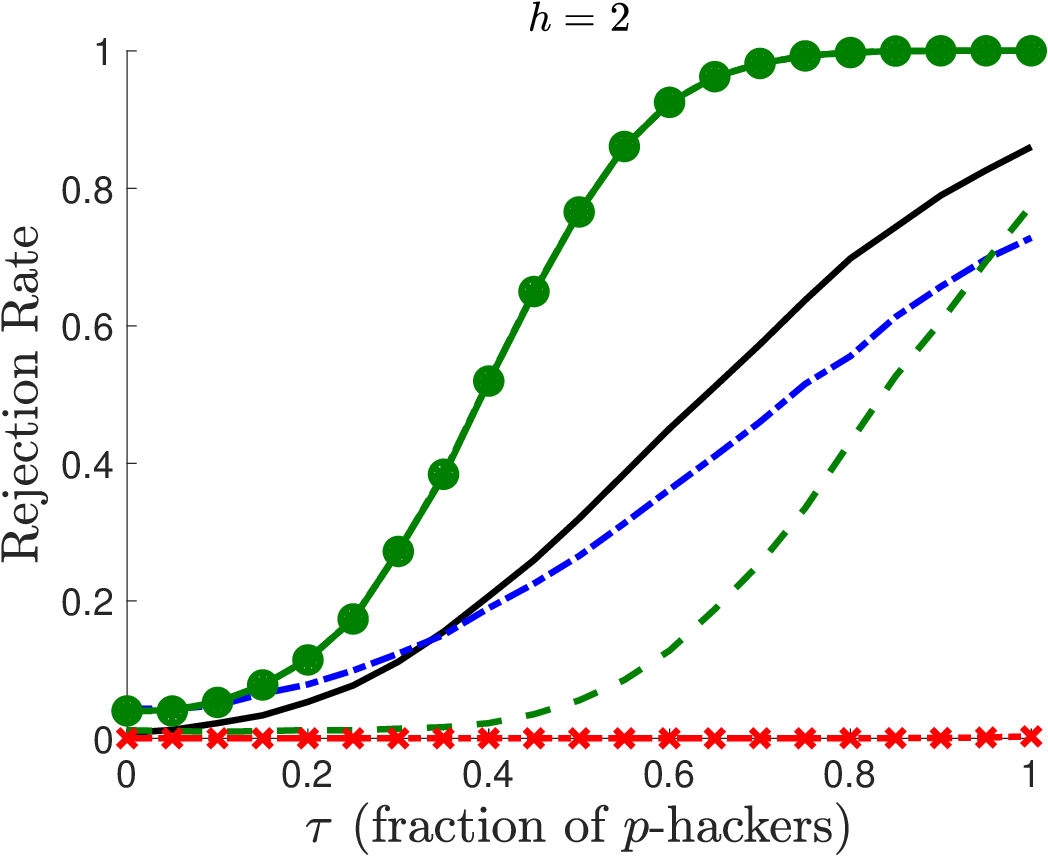}
\includegraphics[width=0.24\textwidth]{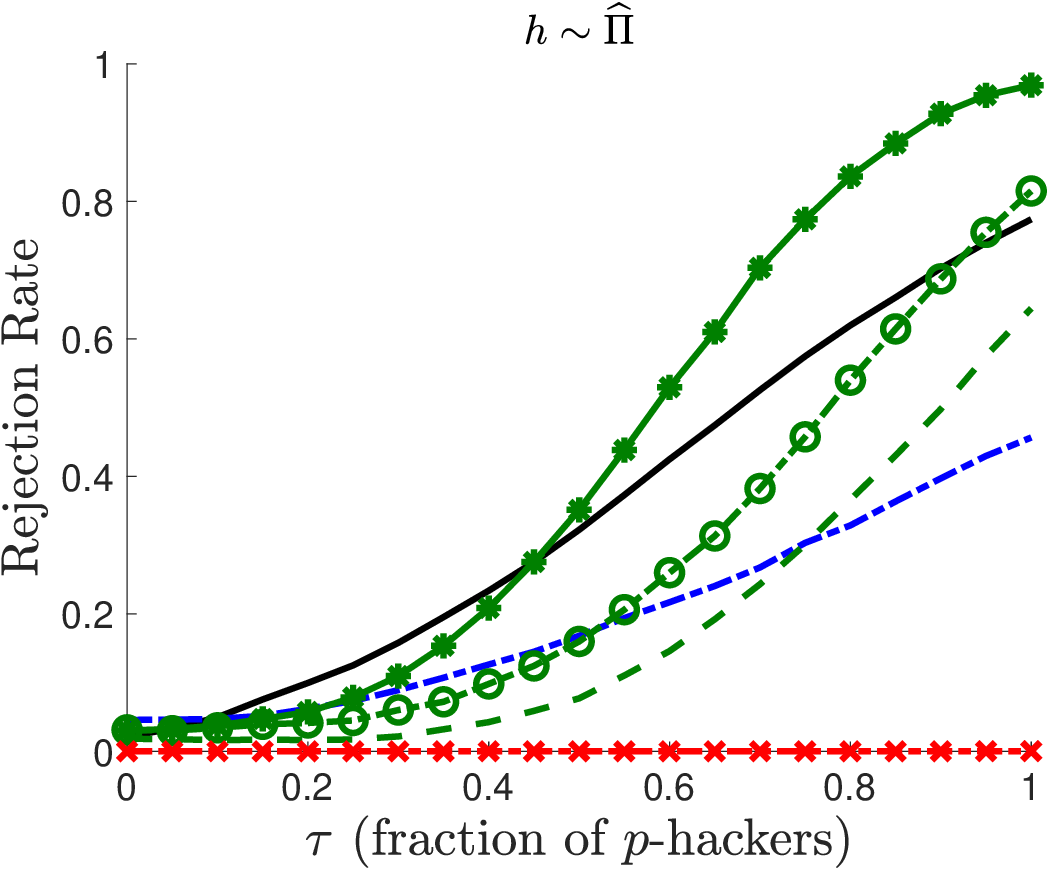}

\textbf{\small Minimum}

\smallskip
\includegraphics[width=0.24\textwidth]{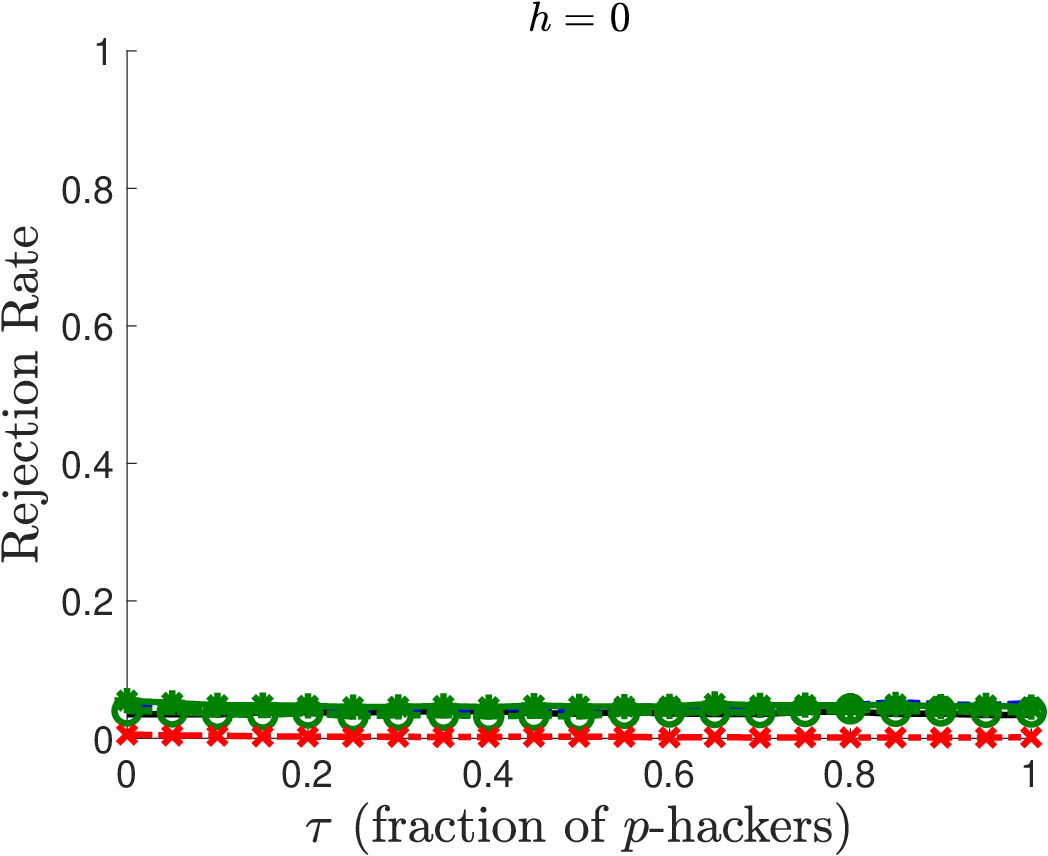}
\includegraphics[width=0.24\textwidth]{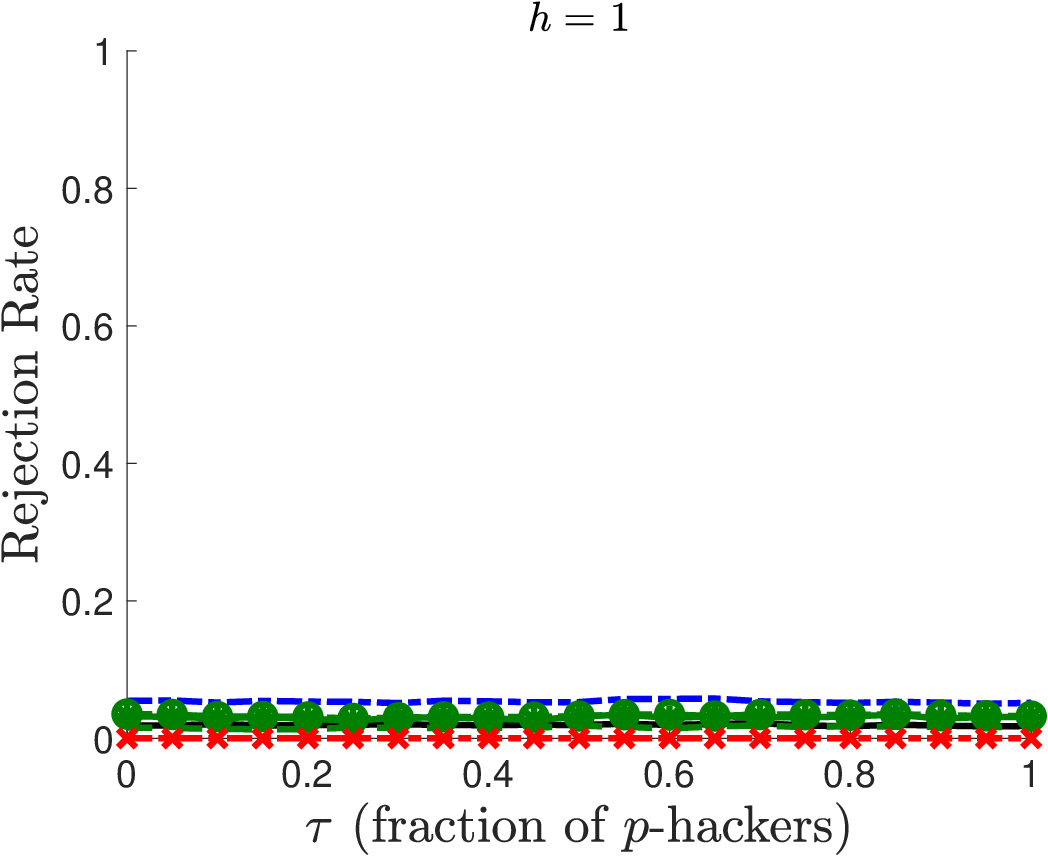}
\includegraphics[width=0.24\textwidth]{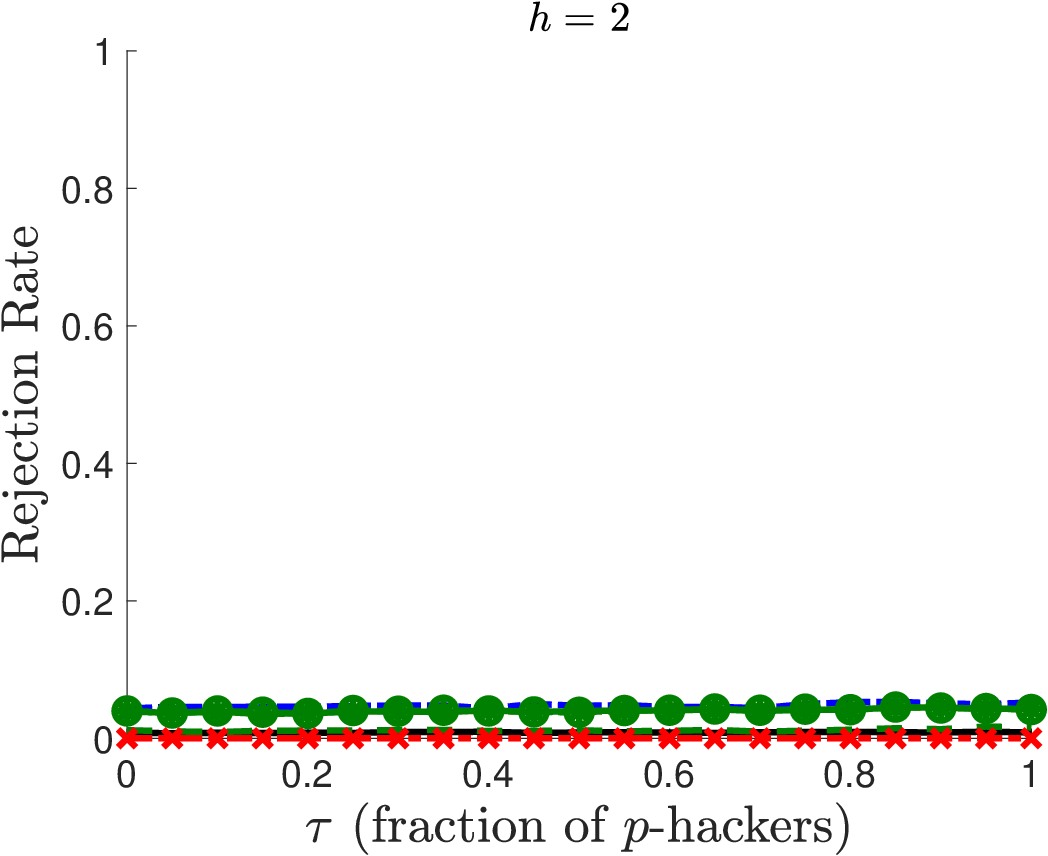}
\includegraphics[width=0.24\textwidth]{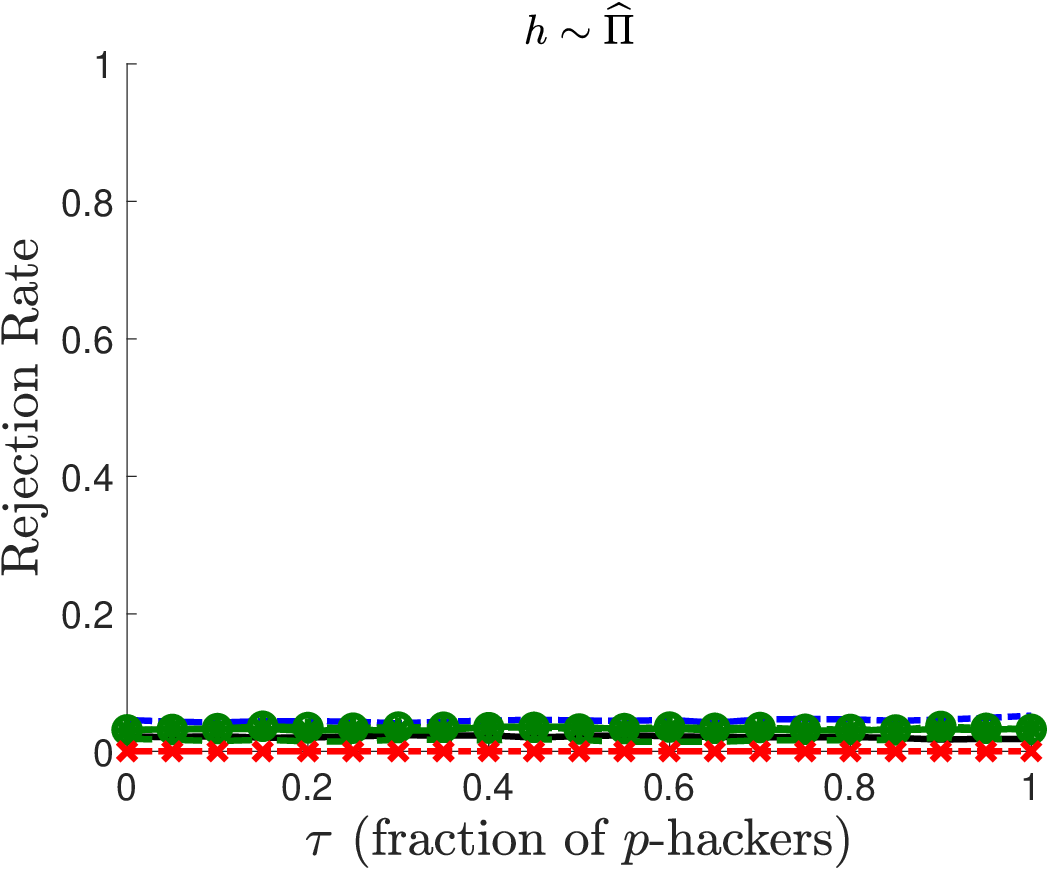}

\end{subfigure}

\vspace{3mm}

\begin{subfigure}[b]{\textwidth}
\caption{Cluster level selection}
\label{fig:power_se_combined_b}
\centering
\textbf{\small Thresholding}

\smallskip

\includegraphics[width=0.24\textwidth]{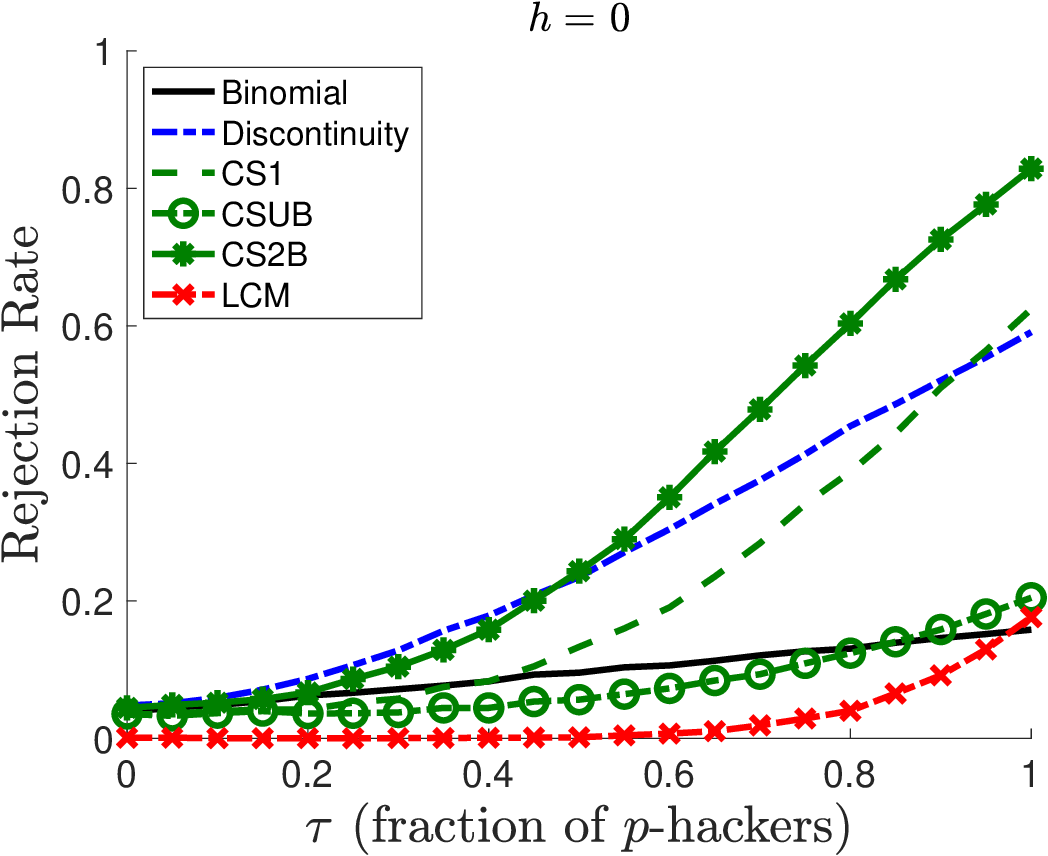}
\includegraphics[width=0.24\textwidth]{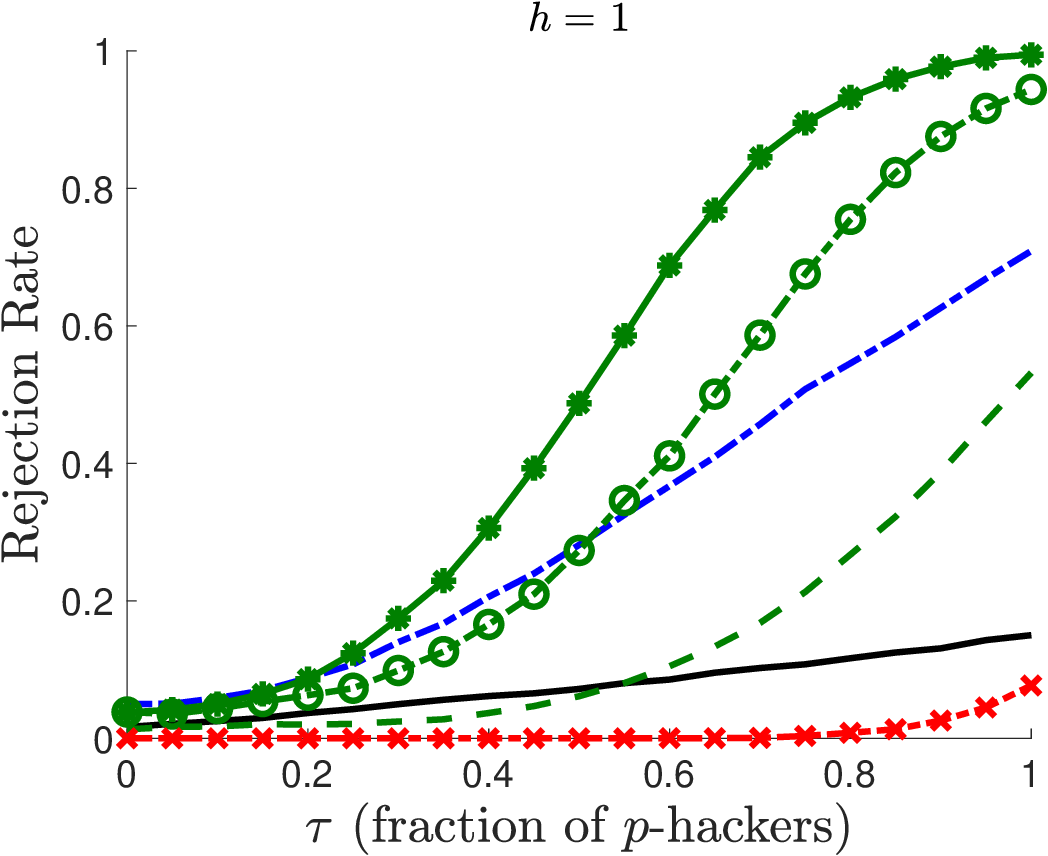}
\includegraphics[width=0.24\textwidth]{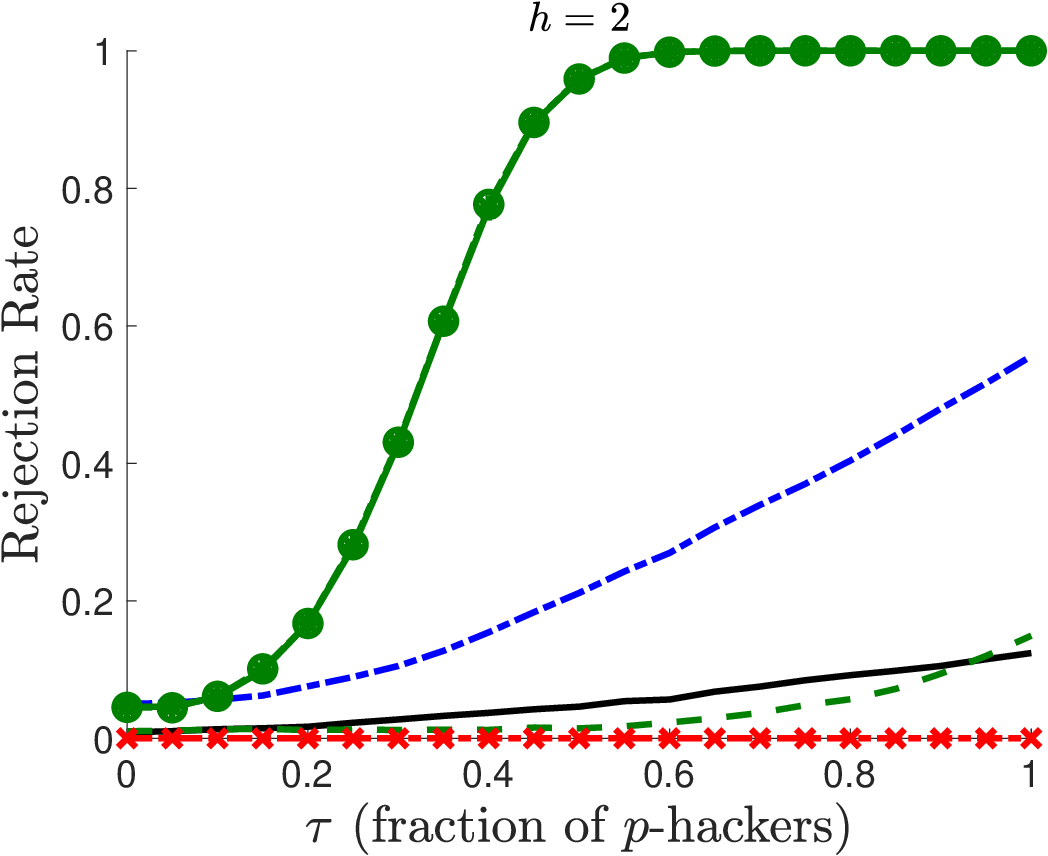}
\includegraphics[width=0.24\textwidth]{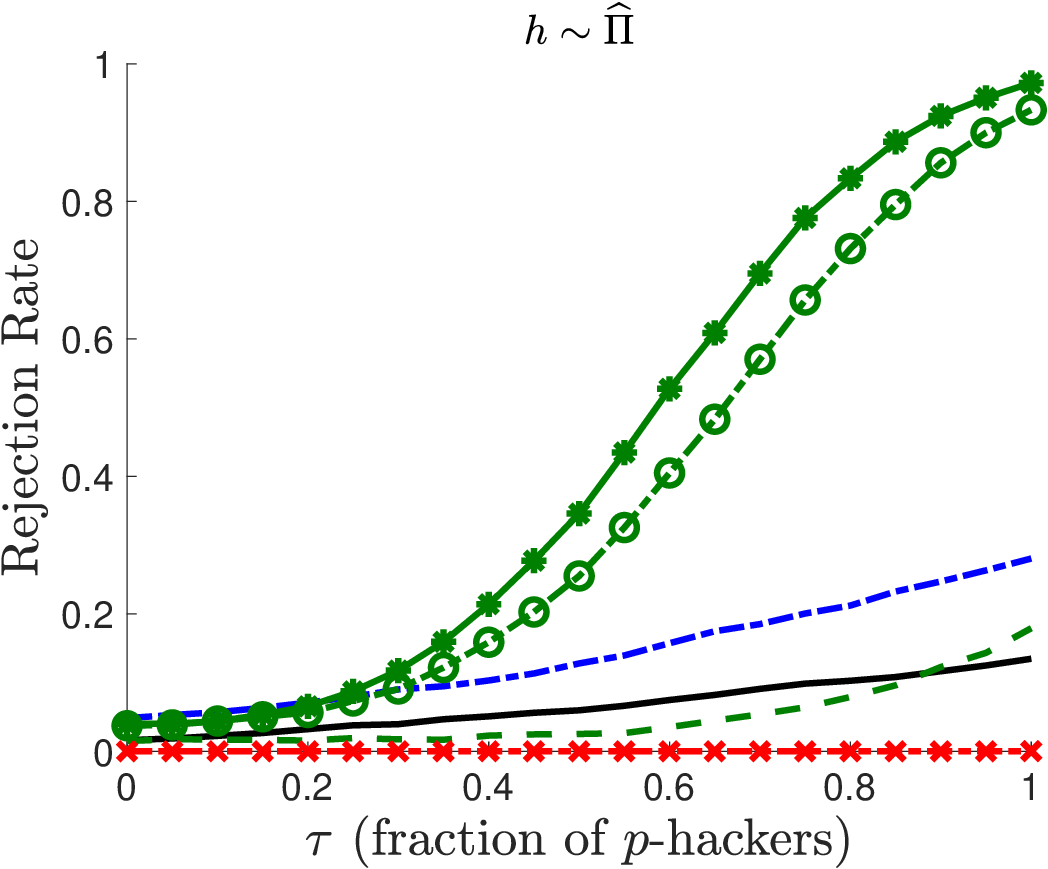}

\textbf{\small Minimum}

\smallskip

\includegraphics[width=0.24\textwidth]{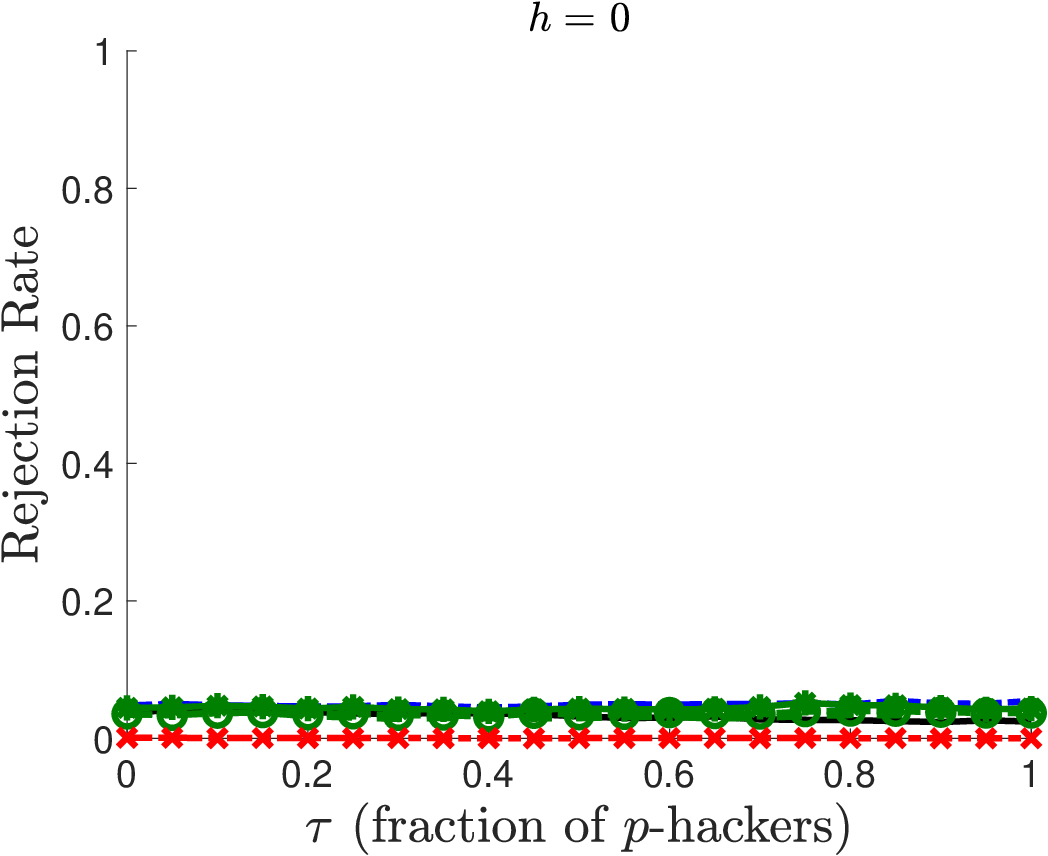}
\includegraphics[width=0.24\textwidth]{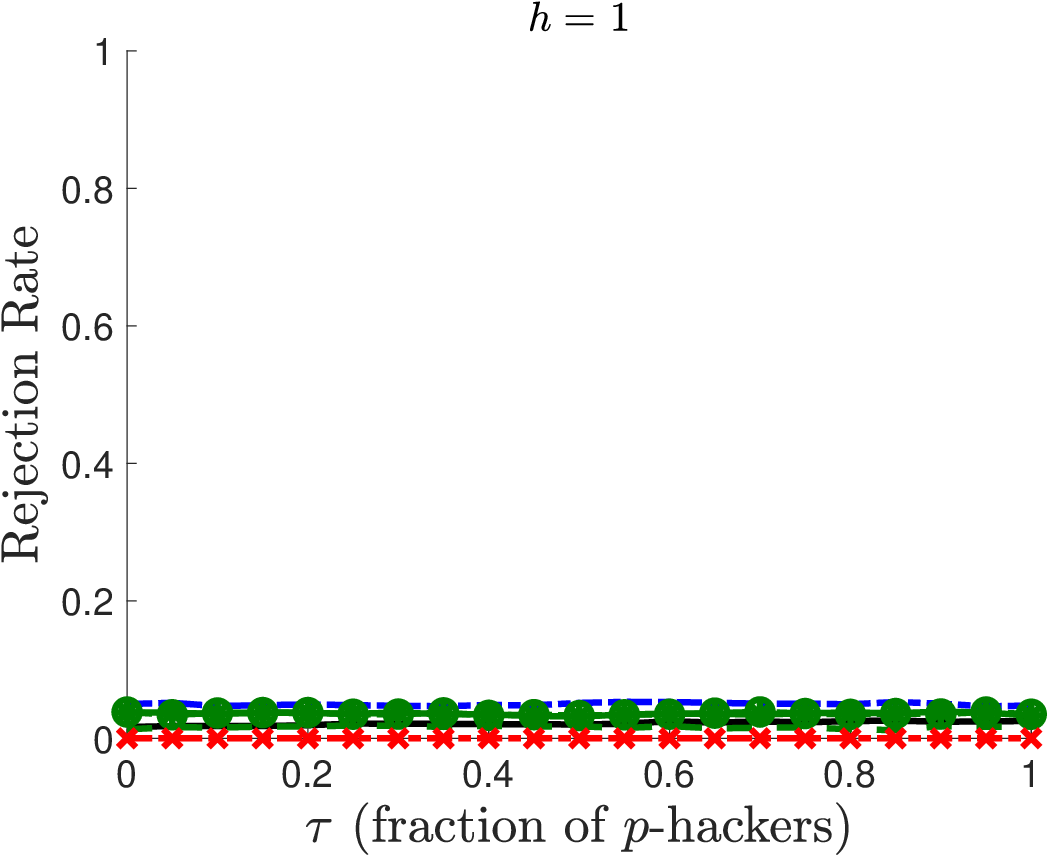}
\includegraphics[width=0.24\textwidth]{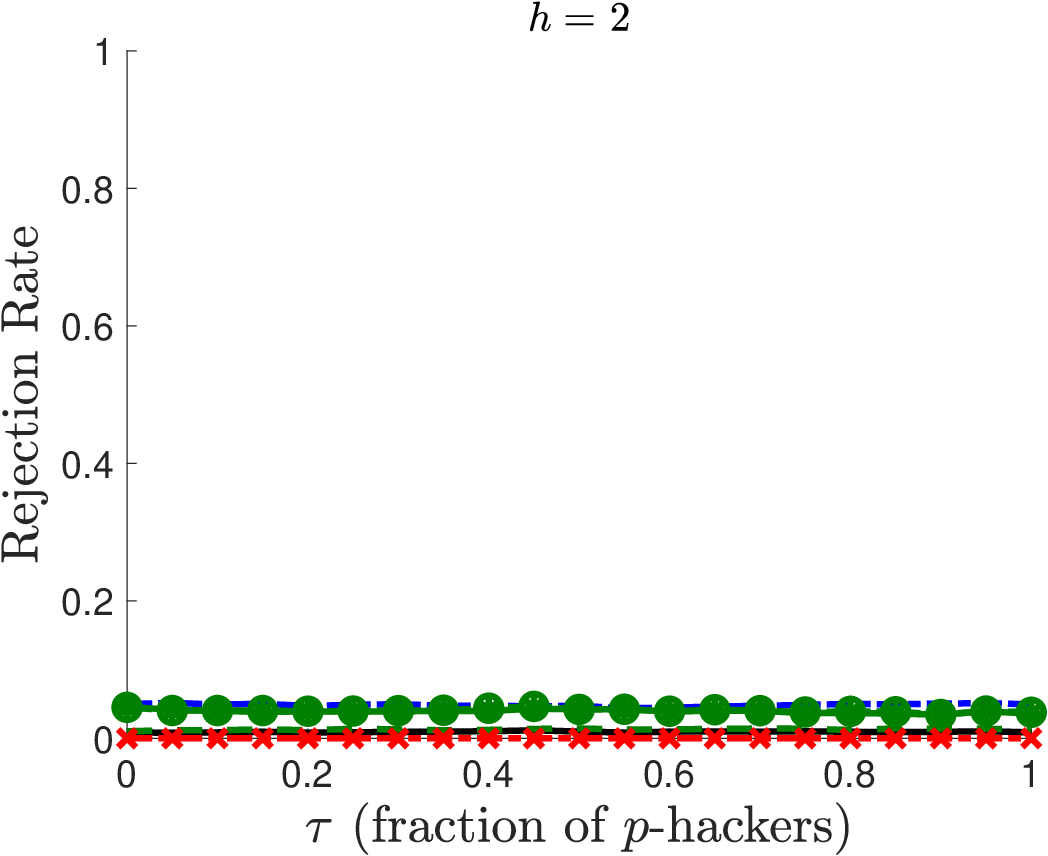}
\includegraphics[width=0.24\textwidth]{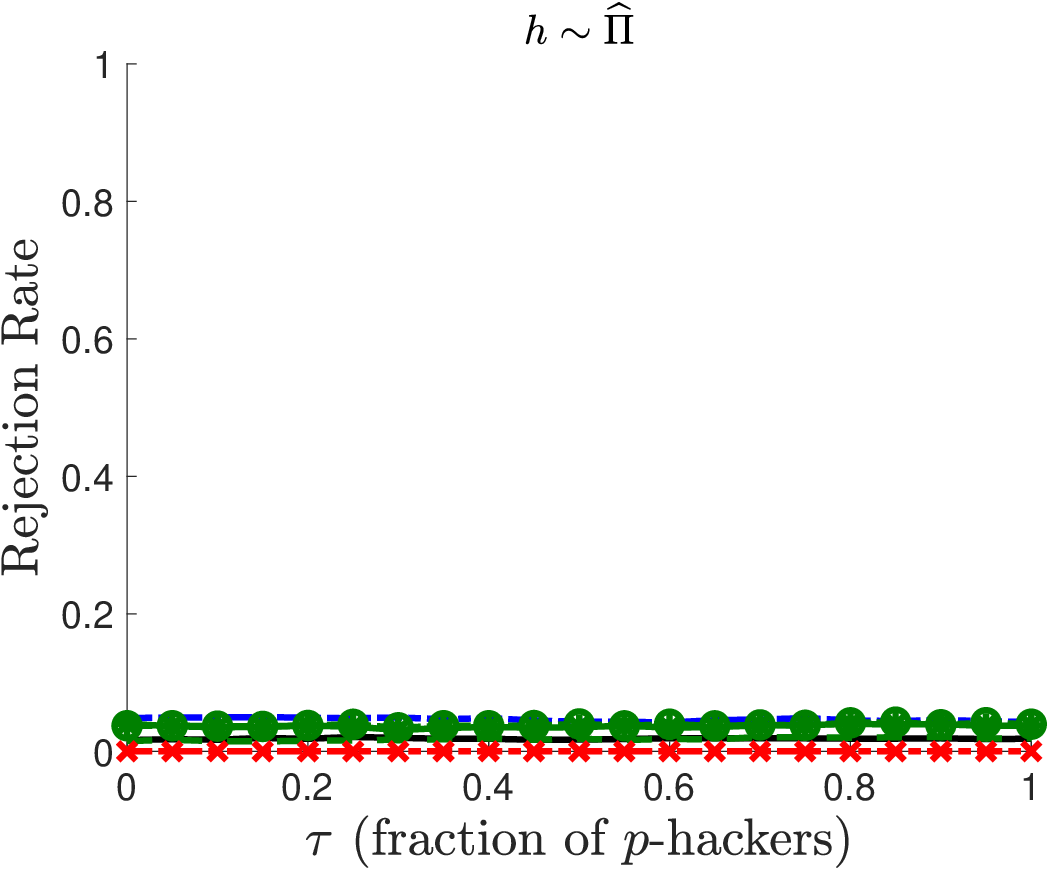}

\end{subfigure}
\end{center}
\vspace{-2mm}

\doublespacing

\textit{Notes:} Figures show rejection rates of the tests in Table \ref{tab:tests} as a function of $\tau$ for the threshold and minimum approach. The simulation design is described in Sections \ref{sec:se_selection_MC} and \ref{sec:simulations_setup}. Figure \ref{fig:power_se_combined_a} and Figure \ref{fig:power_se_combined_b} show the results from lag length selection and cluster level selection, respectively. The results are based on 5,000 simulation repetitions.
\end{figure}

\newpage

\section*{Tables}

 \begin{table}[H]

\centering
\caption{Tests for $p$-hacking}

\small

\onehalfspacing

\begin{tabular}{l l}
\toprule
\midrule

\multicolumn{2}{c}{\textbf{Testable restriction:} non-increasingness  of $p$-curve}\\
\cmidrule(l{5pt}r{5pt}){1-2}
CS1 & Histogram-based test based on \citet{cox2022simple} with $J=15$\\
LCM & LCM test \\ 
Binomial & Binomial test with bins $[0.040,0.045)$ and $[0.045,0.050]$ \\ 
\midrule
\multicolumn{2}{c}{\textbf{Testable restriction:} continuity of $p$-curve}\\
\cmidrule(l{5pt}r{5pt}){1-2}
Discontinuity & Density discontinuity test  \citep{cattaneo2021rddensity}\\
\midrule
\multicolumn{2}{c}{\textbf{Testable restriction:} upper bounds on $p$-curve, 1st, and 2nd derivative}\\
\cmidrule(l{5pt}r{5pt}){1-2}
CSUB  & Histogram-based test based on \citet{cox2022simple} with $J=15$\\
\midrule
\multicolumn{2}{c}{\textbf{Testable restriction:} 2-monotonicity and upper bounds on $p$-curve, 1st, and 2nd derivative}\\
\cmidrule(l{5pt}r{5pt}){1-2}
CS2B &  Histogram-based test based on \citet{cox2022simple} with $J=15$\\ 

\midrule
\bottomrule
\end{tabular}
\label{tab:tests}
\end{table}
\normalsize

\newpage
\begin{table}[H]
\begin{center}
\caption{The effect of publication bias}
\label{tab:publication_bias_1}
\onehalfspacing
\footnotesize
\begin{tabular}{lccccccccccccccccccccccc}
\toprule
 & \multicolumn{18}{c}{Test} \\ \cline{2-19}
 & \multicolumn{3}{c}{Binomial} & \multicolumn{3}{c}{Discontinuity} & \multicolumn{3}{c}{CS1} & \multicolumn{3}{c}{CSUB} & \multicolumn{3}{c}{CS2B} & \multicolumn{3}{c}{LCM} \\
Frac. of $p$-hackers & 0 & 0.5 & 1 & 0 & 0.5 & 1 & 0 & 0.5 & 1 & 0 & 0.5 & 1 & 0 & 0.5 & 1 & 0 & 0.5 & 1 \\ \hline
\textit{} & \multicolumn{18}{c}{Thresholding} \\ \cline{2-19}
No Pub Bias & 0.03 & 0.10 & 0.18 & 0.04 & 0.36 & 0.72 & 0.06 & 0.80 & 1.00 & 0.04 & 0.12 & 0.77 & 0.06 & 0.88 & 1.00 & 0.03 & 0.41 & 0.99 \\
Sharp Pub Bias & 0.03 & 0.10 & 0.18 & 0.92 & 1.00 & 1.00 & 0.13 & 0.83 & 1.00 & 0.99 & 1.00 & 1.00 & 0.99 & 1.00 & 1.00 & 0.01 & 0.93 & 1.00 \\
Smooth Pub Bias & 0.02 & 0.05 & 0.09 & 0.04 & 0.22 & 0.48 & 0.02 & 0.12 & 0.79 & 1.00 & 1.00 & 1.00 & 1.00 & 1.00 & 1.00 & 0.00 & 0.03 & 0.58 \\
\textit{} & \multicolumn{18}{c}{Minimum} \\ \cline{2-19}
No Pub Bias & 0.03 & 0.03 & 0.03 & 0.04 & 0.05 & 0.05 & 0.06 & 0.03 & 0.02 & 0.04 & 0.03 & 0.31 & 0.06 & 0.04 & 0.36 & 0.03 & 0.00 & 0.00 \\
Sharp Pub Bias & 0.03 & 0.03 & 0.03 & 0.92 & 0.95 & 0.96 & 0.13 & 0.07 & 0.06 & 0.99 & 1.00 & 1.00 & 0.99 & 1.00 & 1.00 & 0.01 & 0.00 & 0.00 \\
Smooth Pub Bias & 0.02 & 0.02 & 0.02 & 0.04 & 0.04 & 0.05 & 0.02 & 0.01 & 0.01 & 1.00 & 1.00 & 1.00 & 1.00 & 1.00 & 1.00 & 0.00 & 0.00 & 0.00 \\
\bottomrule
\end{tabular}
\end{center}
\doublespacing\textit{Notes:} Table shows the impact of publication bias on the power of the tests when $p$-hacking is based on covariate selection with $K=3$ (general-to-specific, two-sided tests) and $h=0$. The results are based on 5,000 simulation repetitions.
\end{table}

\newpage

\begin{table}[H]
\footnotesize
\caption{Results empirical application}
\label{tab:application}

\onehalfspacing
\begin{tabular}{lcccccccccccccc}
\toprule

\multicolumn{1}{c}{\multirow{2}{*}{Test}} & \multicolumn{13}{c}{Subsample} & \multicolumn{1}{c}{\multirow{2}{*}{Overall RR}} \\
\cline{2-14}
\multicolumn{1}{c}{} &DID &RCT &RDD &IV &F<30 &F$\geq$30 &Top 5 &!Top 5 &2015 &2018 &AJQ &SW &All& \multicolumn{1}{c}{} \\
\midrule
& \multicolumn{14}{c}{(a) Original (rounded) data: $p$-values} \\
\cline{2-14}
LocBin     & 0.000 & 0.000 & 0.002 & 0.324 & 0.791 & 0.997 & 0.003 & 0.000 & 0.000 & 0.004 & 0.002 & 0.000 & 0.000 & 0.769 \\
Discont.   & 0.000 & 0.019 & 0.008 & 0.951 & 0.500 & 0.753 & 0.000 & 0.270 & 0.757 & 0.172 & 0.005 & 0.679 & 0.962 & 0.385 \\
CS1        & 0.076 & 0.098 & 0.293 & 0.931 & 0.347 & 0.474 & 0.932 & 0.045 & 0.003 & 0.143 & 0.981 & 0.020 & 0.046 & 0.308 \\
CSUB       & 0.025 & 0.167 & 0.165 & 0.077 & 0.073 & 0.448 & 0.453 & 0.000 & 0.007 & 0.002 & 0.353 & 0.027 & 0.000 & 0.462 \\
CS2B       & 0.010 & 0.149 & 0.107 & 0.077 & 0.024 & 0.376 & 0.659 & 0.000 & 0.003 & 0.001 & 0.230 & 0.018 & 0.000 & 0.538 \\
LCM        & 0.007 & 0.130 & 0.807 & 1.000 & 1.000 & 1.000 & 0.671 & 0.001 & 0.004 & 0.192 & 0.530 & 0.049 & 0.000 & 0.385 \\
\midrule
& \multicolumn{14}{c}{(b) Average rejection rates across 1,000 deroundings, 5\% significance level} \\
\cline{2-14}
LocBin     & 0.000 & 0.000 & 0.000 & 0.000 & 0.000 & 0.000 & 0.000 & 0.000 & 0.000 & 0.000 & 0.000 & 0.000 & 0.000 & 0.000 \\
Discont.   & 0.066 & 0.045 & 0.028 & 0.005 & 0.002 & 0.100 & 0.026 & 0.016 & 0.013 & 0.009 & 0.023 & 0.021 & 0.015 & 0.028 \\
CS1        & 0.010 & 0.053 & 0.043 & 0.086 & 0.007 & 0.222 & 0.007 & 0.031 & 0.008 & 0.201 & 0.005 & 0.001 & 0.046 & 0.055 \\
CSUB       & 0.071 & 0.080 & 0.081 & 0.267 & 0.271 & 0.217 & 0.025 & 0.103 & 0.011 & 0.351 & 0.030 & 0.002 & 0.137 & 0.127 \\
CS2B       & 0.054 & 0.056 & 0.097 & 0.221 & 0.248 & 0.187 & 0.021 & 0.111 & 0.012 & 0.415 & 0.032 & 0.000 & 0.133 & 0.122 \\
LCM        & 0.000 & 0.000 & 0.000 & 0.000 & 0.000 & 0.000 & 0.000 & 0.000 & 0.000 & 0.000 & 0.000 & 0.000 & 0.000 & 0.000 \\
\midrule
Obs & 5853 & 7569 & 3148 & 5170 & 679 & 760 & 3954 & 17786 & 11211 & 10529 & 3341 & 17772 & 21740 & \\
\bottomrule
\end{tabular}

\vspace{3mm}

\normalsize
\doublespacing
\textit{Notes:} The table reports $p$-values and rejection rates for tests for $p$-hacking applied to various subsets of \cite{brodeur2016data} and \cite{brodeur2022data} datasets. ``DID'', ``RCT'', ``RDD'', and ``IV'' represent subsets of the data corresponding to results obtain from difference-in-differences analysis, randomized control trials, regression discontinuity design analysis, and instrumental variables estimation, respectively. ``F$<$30'' and ``F$\geq$30'' are subsamples of IV studies for which the first-stage F-statistics are available and exceed or do not exceed 30, respectively. ``Top 5'' corresponds to the subsample of results from top-5 economics journals and ``!Top 5'' corresponds to the rest of the sample. Columns ``2015'' and ``2018'' represent papers published in years 2015 and 2018, respectively. Column ``AJQ'' represents paper published in \textit{The American Economic Review}, \textit{The Journal of Political Economy}, or \textit{The Quarterly Journal of Economics} during the sample period (years 2015 and 2018). ``SW'' is a set of finding published in the same three journals and collected by \cite{brodeur2016data} for the 2005--2011 period. ``All'' is the full sample of observations collected by \cite{brodeur2022data}. ``Overall RR'' is the fraction of rejections at the 5\% level across all subsamples  (rounded data) and the average rejection rate across all subsamples and de-rounding draws (de-rounded data).
\end{table}

\newpage

\begin{center}
{\LARGE Online Appendix to `The Power of Tests for Detecting $p$-Hacking'}
\end{center}

\appendix

\setcounter{page}{1}

\section*{Table of Contents}
\linespread{1}
\setlength{\parskip}{0in}

\startcontents[sections]
\printcontents[sections]{l}{1}{\setcounter{tocdepth}{2}}

\linespread{1.25}
\setlength{\parskip}{0.05in}

\newpage

\section{Selecting across Datasets}
\label{app:selecting_across_datasets}
Consider a setting where a researcher conducts a finite number of $K>1$ independent tests over which they can choose the best results. In each case, the researcher uses a $t$-test to test their hypothesis, with test statistic $T_i \sim \mathcal{N}(h,1)$. We assume that the true local effect $h$ is the same across datasets. This gives the researcher $K$ possible $p$-values to consider, enabling the possibility of $p$-hacking. For example, a researcher conducting experiments with students, as is common in experimental economics, could have several independent sets of students on which to test a hypothesis. As with the other examples, researchers could simply search over all datasets and report the smallest $p$-value or engage in a strategy of searching for a low $p$-value. 

Let $K=2$ and consider a search where first the researchers construct a dataset for their study and compute a $p$-value $P_1$ for their hypothesis on this dataset. $P_2$ follows from constructing a new dataset. 

For the threshold approach (where the reported $p$-value follows from equation \eqref{eq: Thresholding}), the $p$-curve is given by 
\begin{eqnarray*}
g_4^{t}(p) =\int_\mathcal{H}\exp\left(hz_0(p)-\frac{h^2}{2}\right)\Upsilon^t_4(p; \alpha, h)d\Pi(h),
\end{eqnarray*}
where  
\begin{equation*}
  \Upsilon^t_4(p;\alpha, h)=\begin{cases}
    1+ \Phi\left(z_{h}(\alpha) \right), & \text{if $p\le \alpha$},\\
    2\Phi\left(z_h(p) \right), & \text{if $p> \alpha$}.
  \end{cases}
\end{equation*}

This is a special case of the results in Section \ref{sec:specification_search_regression} where $\rho=0$ because of the independence assumption across datasets. If the $t$-statistics were correlated through dependence between the datasets, then setting $\rho$ equal to that correlation and using the results in Section \ref{sec:specification_search_regression} would yield the correct distribution of $p$-values. 

Figure \ref{fig:p_curves_example3} shows $p$-curves for $h\in \{0,1,2\}$. For all values of $h$, no upward sloping $p$-curves are induced over any range of $p$. So for this type of $p$-hacking, even with thresholds such as in this example, tests that look for $p$-hacking through a lack of monotonically downward sloping $p$-curves will not have power. This method does suggest that tests for discontinuities in the distribution will have power, but likely only if studies have moderately large $h$'s, that is, $h$'s sufficient to induce a pronounced discontinuity (e.g., $h=1$ or $h=2$) while generating a non-negligible amount of insignificant results to $p$-hack. 

\begin{figure}[H]
    \begin{center}
    \caption{$p$-Curves from dataset selection based on the threshold approach} 
    \label{fig:p_curves_example3}
         \includegraphics[width=0.43\textwidth]{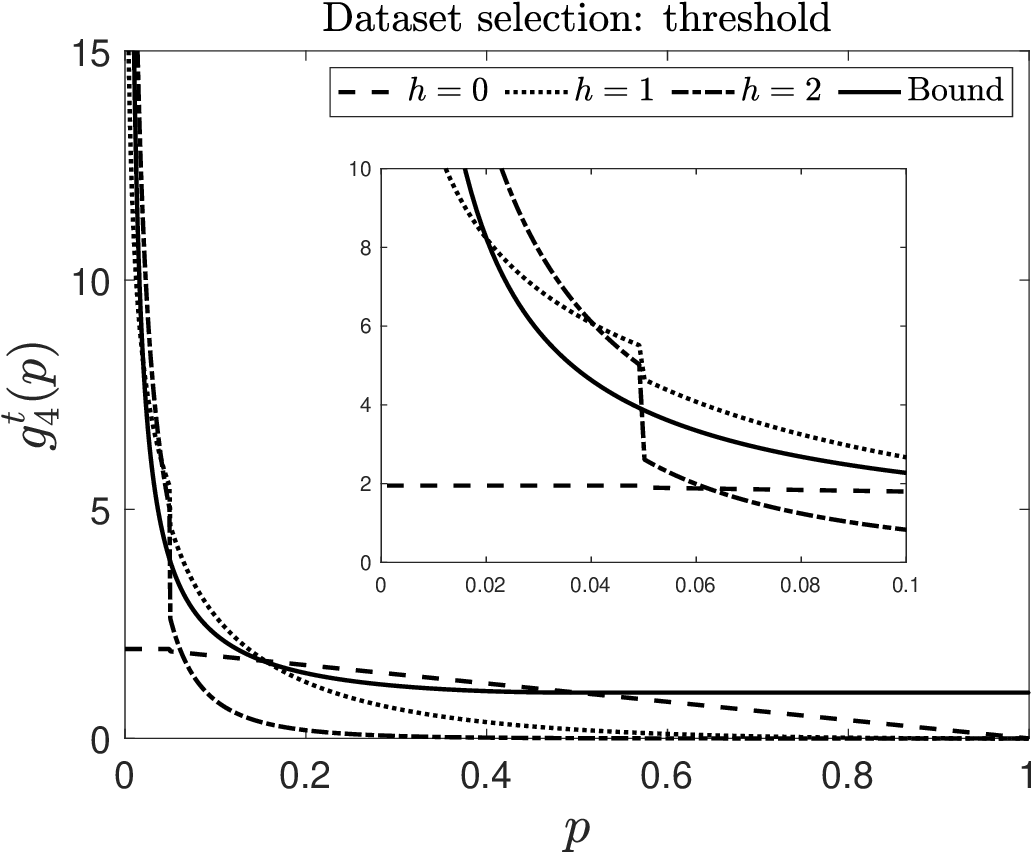}
\end{center}

\vspace{-2mm}

\footnotesize{\textit{Notes:} Figure shows the $p$-curves from dataset selection based on the threshold approach with $\gamma=0.5$}
              
\end{figure}

We also consider the minimum approach (Equation \eqref{eq: Minimum}) with the $p$-values from across all datasets or subsamples \citep[e.g.,][]{ulrich2015p,elliott2022detecting}. For general $K$, the $p$-curve is given by 
\begin{equation}
g_4^{m}(p;K) = K\int_\mathcal{H}\exp\left(hz_0(p)-\frac{h^2}{2}\right)\Phi(z_h(p))^{K-1}d\Pi(h).\label{eq:generalized_um}
\end{equation}

\begin{figure}[H]

\begin{center}
              \caption{$p$-Curves from dataset selection based on the minimum approach} 
              \label{fig:p_curves_example3bnd}
     
 \includegraphics[width=0.43\textwidth]{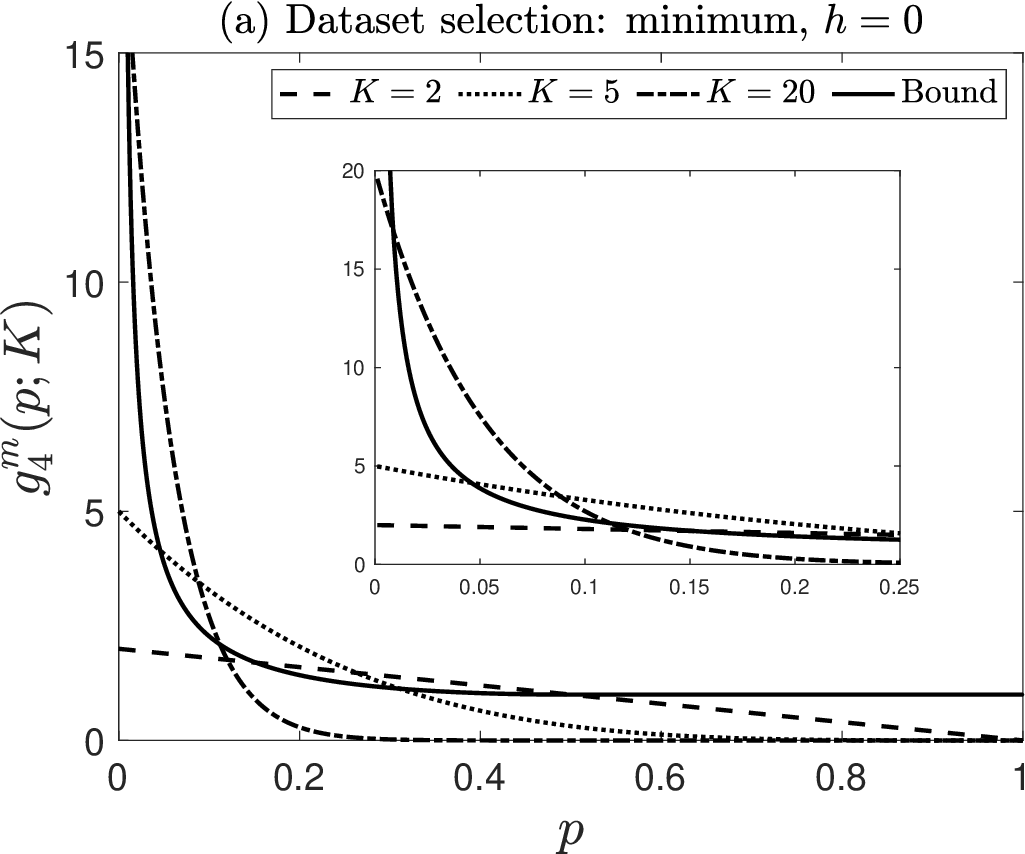}
\includegraphics[width=0.43\textwidth]{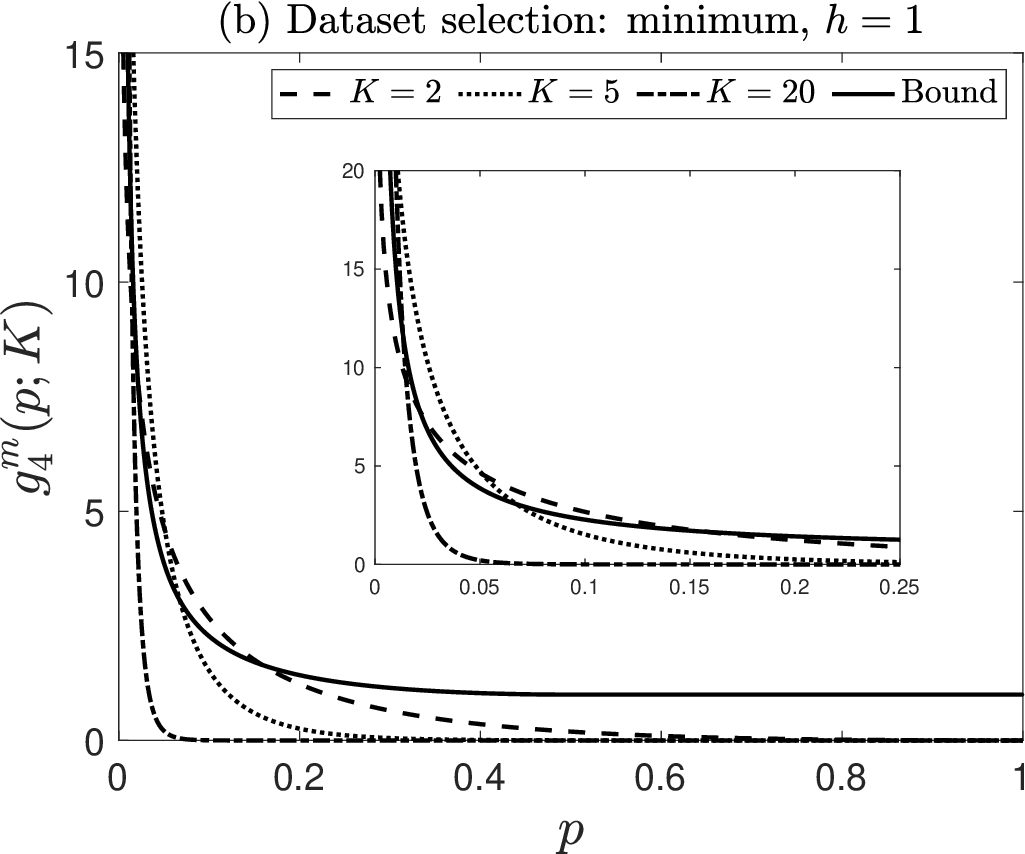}

\end{center}

\vspace{-2mm}

\footnotesize{\textit{Notes:} Figures show the $p$-curves from dataset selection based on the minimum approach. Panel (a): $h=0$. Panel (b): $h=1$.}

\end{figure}

The $p$-curve under $p$-hacking, $g^m_4$, is non-increasing and completely monotone for $p\le 0.5$ \citep{elliott2022detecting}. This can be seen in Figure  \ref{fig:p_curves_example3bnd} where for various $K$ and $h$ each of the curves are decreasing. Tests for violations of monotonicity will have no power. Similarly, tests for discontinuities will also not have power. Figure  \ref{fig:p_curves_example3bnd} also shows (solid line) the bounds under the null hypothesis of no $p$-hacking. Clearly, each of the curves violates the bounds for some range of $p$; see also Figure 2 in \citet{elliott2022detecting}. 
 
Alternatively, the researcher could consider the threshold strategy of first using both datasets, choosing to report this $p$-value if it is below a threshold and, otherwise, choosing the best of the available $p$-values. For $K=2$, this gives three potential $p$-values to choose between. For many such testing problems (for example, testing a regression coefficient in a linear regression), $T_k \sim \mathcal{N}(h,1)$, $k=1,2$, approximately so that the $t$-statistic from the combined samples is $T_{12} \simeq (T_1+T_2)/\sqrt{2}$. This is precisely the same setup asymptotically as in the IV case presented above, so those results apply directly to this problem. As such, we refer to the discussion there rather than re-present the results.

\section{Theoretical Analysis of Bias and Size Distortions due to $p$-Hacking}
\label{app:theory_bias_size_distortions}
There are two important implications of $p$-hacking. The first is that the empirical size (i.e., the size under $p$-hacking) is larger than the nominal size. The second implication is that the reported estimates will be larger in magnitude and hence biased. Here we study the bias and size distortions theoretically in the context of the analytical examples in Section \ref{sec:theory}. We consider both the bias $E[\hat\beta_r - \beta]$, where $\hat\beta_r$ is the estimate reported by the researchers, and the more easily interpretable relative bias, that is, the bias scaled by the true coefficient of interest $\beta$. To keep the exposition short, we only provide theoretical formulas for the relative bias. The corresponding formulas for the bias can be obtained by multiplying the relative bias by $\beta$.

Appendix \ref{app:detailed_derivations} presents the analytical derivations underlying the results. We do not present results for the case where researchers select across datasets since this is a special case of selecting control variables in linear regression.

\subsection{Selecting Control Variables in Linear Regression}

The magnitude of size distortions follows from the derived CDF for the $p$-hacked curve evaluated at $h=0$. The size distortion is the same for both the thresholding case and the situation where the researcher simply reports the minimum $p$-value, since in either case, if there is a rejection at the desired size, each method of $p$-hacking will use it. Empirical size for any nominal size is given by \begin{equation*}
  G_0(\alpha)= 1 - \Phi_2(z_0(\alpha), z_0(\alpha); \rho),
\end{equation*}
where $\Phi_2(\cdot, \cdot;\rho)$ is the CDF of the bivariate normal distribution with standard marginals and correlation $\rho$. Figure \ref{fig:Covsize} shows the difference between empirical and nominal size. Panel (a) shows, for nominal size $\alpha = 0.05$, how the empirical size varies with $\gamma$. For small $\gamma$, the tests are highly correlated ($\rho$ is close to one), leaving little room for effective $p$-hacking, and hence there is only a small effect on size. As $\gamma$ becomes larger, so does the size distortion as it moves towards having an empirical size double that of the nominal size. Panel (b) shows, for three choices of $\gamma$ ($\gamma \in\{ 0.1,0.5,0.9\}$), how the empirical size exceeds nominal size. The lower line is nominal size; the empirical size is larger for each value of $\gamma$. Essentially, the result is somewhat uniform over this empirical size range, with size coming close to double empirical size for the largest value of $\gamma$.   

\begin{figure}[H]
\begin{center}
              \caption{Size distortions under $p$-hacking}
              \label{fig:Covsize}
\includegraphics[width=0.43\textwidth]{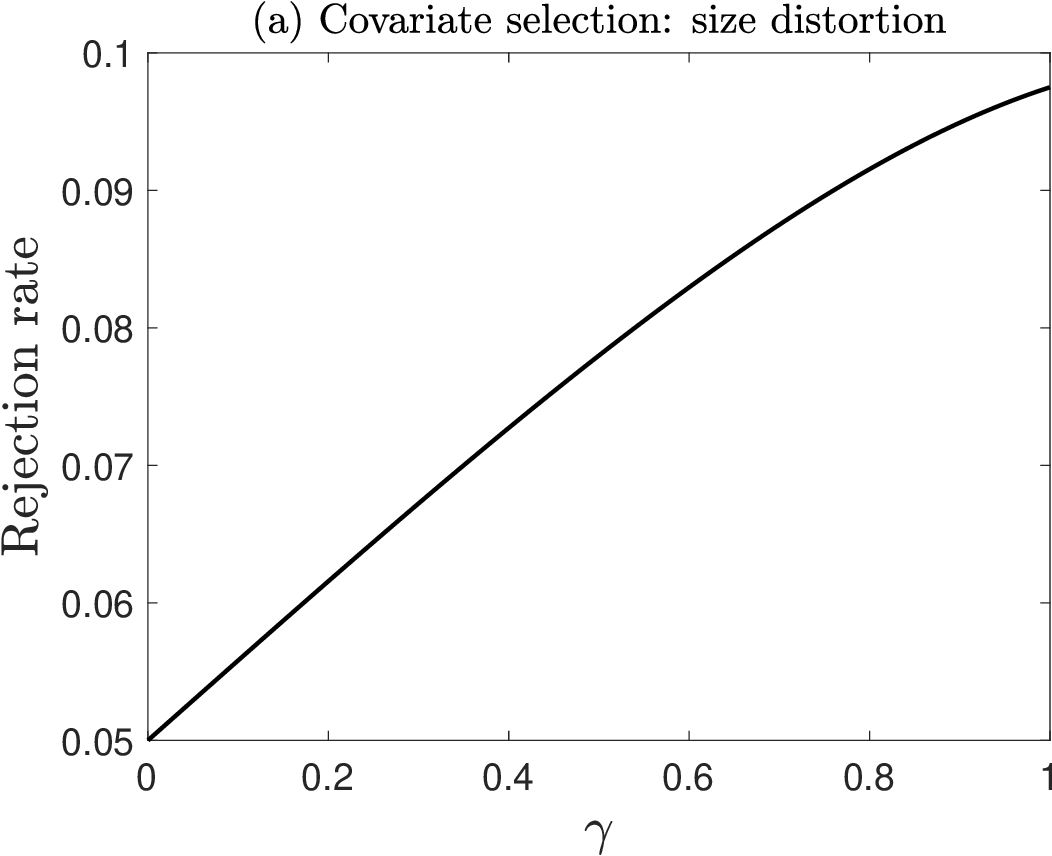}
\includegraphics[width=0.45\textwidth]{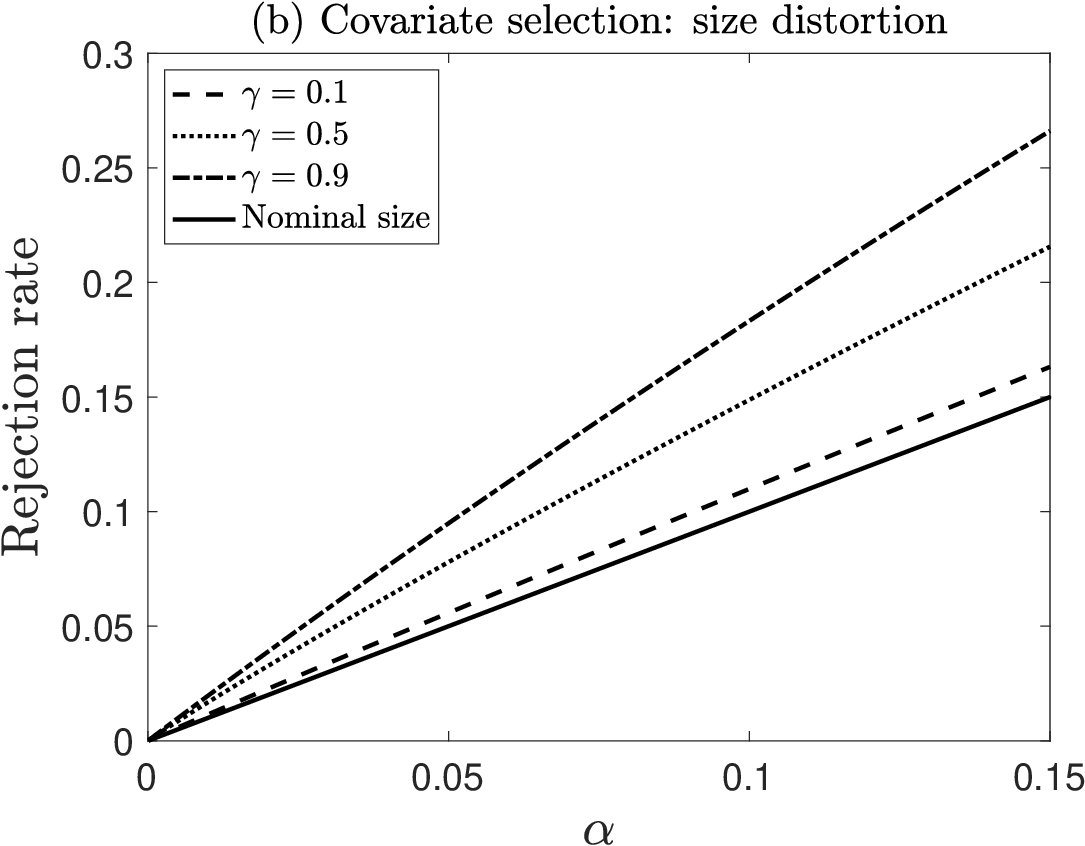}

\end{center}

\vspace{-2mm}

\footnotesize{\textit{Notes:} Figures show the rejection rates under $p$-hacking. Panel (a): rejection rate as a function of $\gamma$ for $\alpha=0.05$ and $h=0$. Panel (b): rejection rate as a function of $\alpha$ for $\gamma\in \{0.1,0.5,0.9\}$ and $h=0$.}

\end{figure}

Selectively choosing larger $t$-statistics results in selectively choosing larger estimated effects.
The relative biases for the threshold and minimum approaches, $\mathcal{B}_{1}^{t}$ and $\mathcal{B}_{1}^{m}$, respectively, are given by
\begin{eqnarray*}
&\mathcal{B}_{1}^{t} =& \mathcal{B}_{1}^{m}\Phi\left(\sqrt{\frac{2}{1+\rho}}z_h(\alpha)\right) +h^{-1}(1-\rho)\phi(z_h(\alpha))\Phi\left(-\sqrt{\frac{1-\rho}{1+\rho}}z_h(\alpha)\right),\\
&\mathcal{B}_{1}^{m} =& h^{-1}\sqrt{2(1-\rho)}\phi(0).
\end{eqnarray*}

The bias as a function of $h$ can be seen in Figure \ref{fig:CovBias}(a). For the thresholding case, most $p$-hacking occurs when $h$ is small. As a consequence, the bias is larger for small $h$. A larger $\gamma$ means a smaller $\rho$, and hence draws of the estimate and the $p$-value are less correlated, allowing for larger impacts. For the minimum approach, the bias does not depend on $h$ and is larger than that for the threshold approach. The reason is that the minimum approach always chooses the largest effect since in our simple setting the standard errors are the same in both regressions. 

Figure \ref{fig:CovBias}(b) shows the relative bias as a function of $h$. Unlike for the bias, the relative bias decreases in $h$ for the minimum approach. While the graphs of the relative bias as a function of $h$ have different shapes than those for the bias, the qualitative conclusions regarding the impact of $\gamma$ remain the same.

\begin{figure}[H]
     \begin{center}              
     \caption{Bias from covariate selection} 
     \label{fig:CovBias}

     \includegraphics[width=0.43\textwidth]{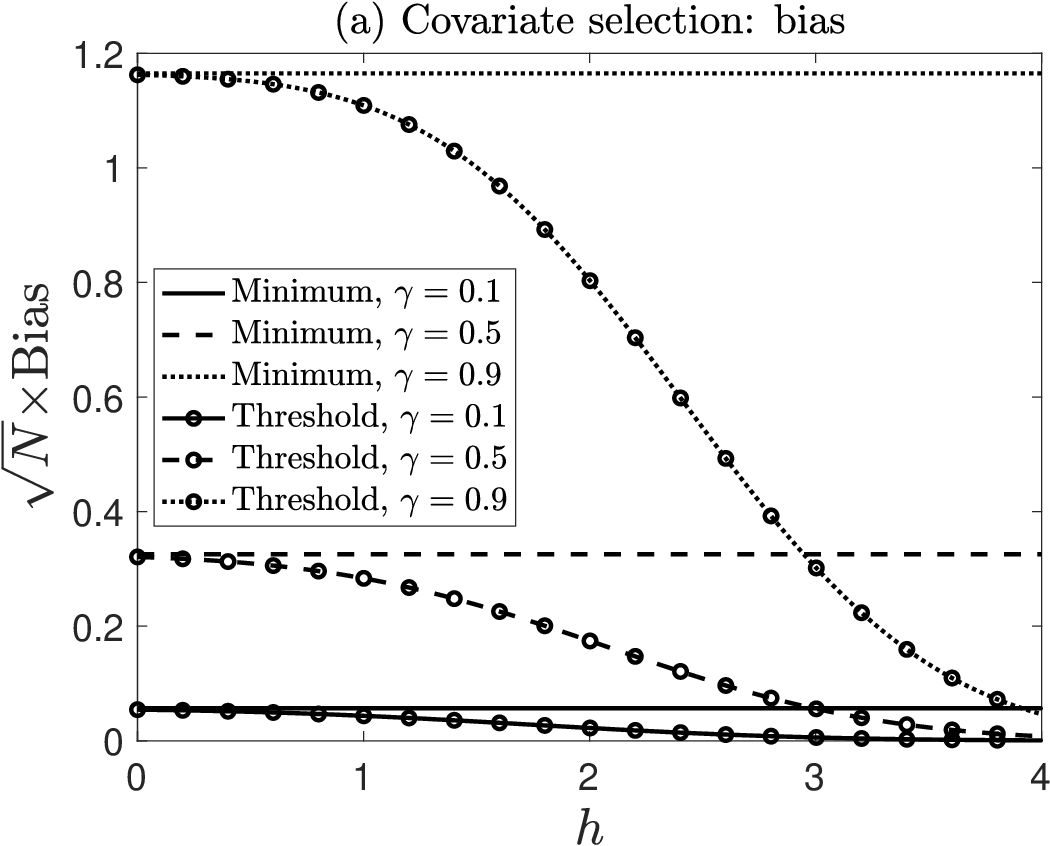}
     \includegraphics[width=0.43\textwidth]{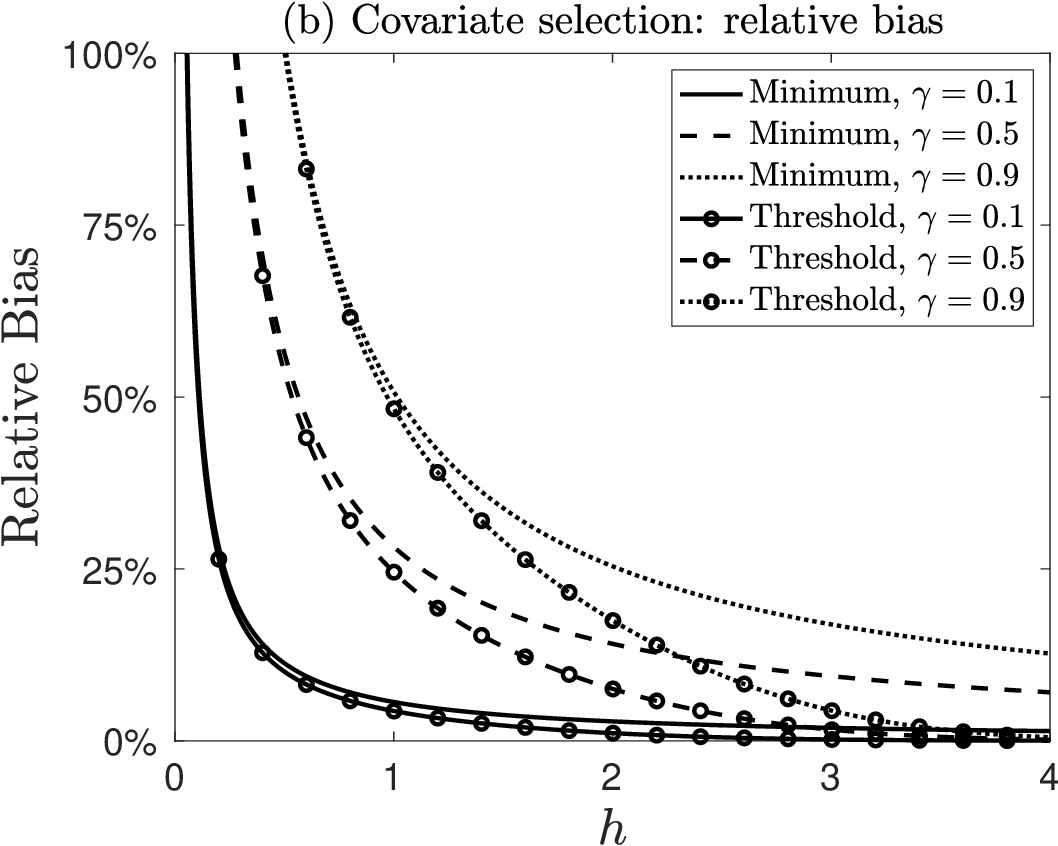}

    \end{center}
\vspace{-2mm}

\footnotesize{\textit{Notes:} Figure shows the bias from covariate selection for $\gamma\in\{0.1, 0.5, 0.9\}$ and $\alpha=0.05$. Panel (a): bias. Panel (b): relative bias.}
    
\end{figure}

\subsection{Selecting amongst Instruments in IV Regression}
Empirical size for any nominal size $\alpha$ is given by 
\begin{equation*}
  G_0(\alpha)= 1 - \Phi(z_0(\alpha))\Phi((\sqrt{2}-1)z_0(\alpha))-\int_{(\sqrt{2}-1)z_0(\alpha)}^{z_0(\alpha)} \phi(x) \Phi(\sqrt{2}z_0(\alpha)-x)dx.
\end{equation*}
The expression is the same for both the threshold approach and taking the minimum, for the same reason as in the case of covariate selection. The magnitude of the size distortion is given in Figure \ref{fig:IVsize}(a). Empirical size is essentially double the nominal size, with the $p$-hacked size at 11\% when nominal size is 5\%.  

\begin{figure}[H]

\begin{center}
\caption{Size distortions and relative bias under $p$-hacking} 
\label{fig:IVsize}

        \includegraphics[width=0.43\textwidth]{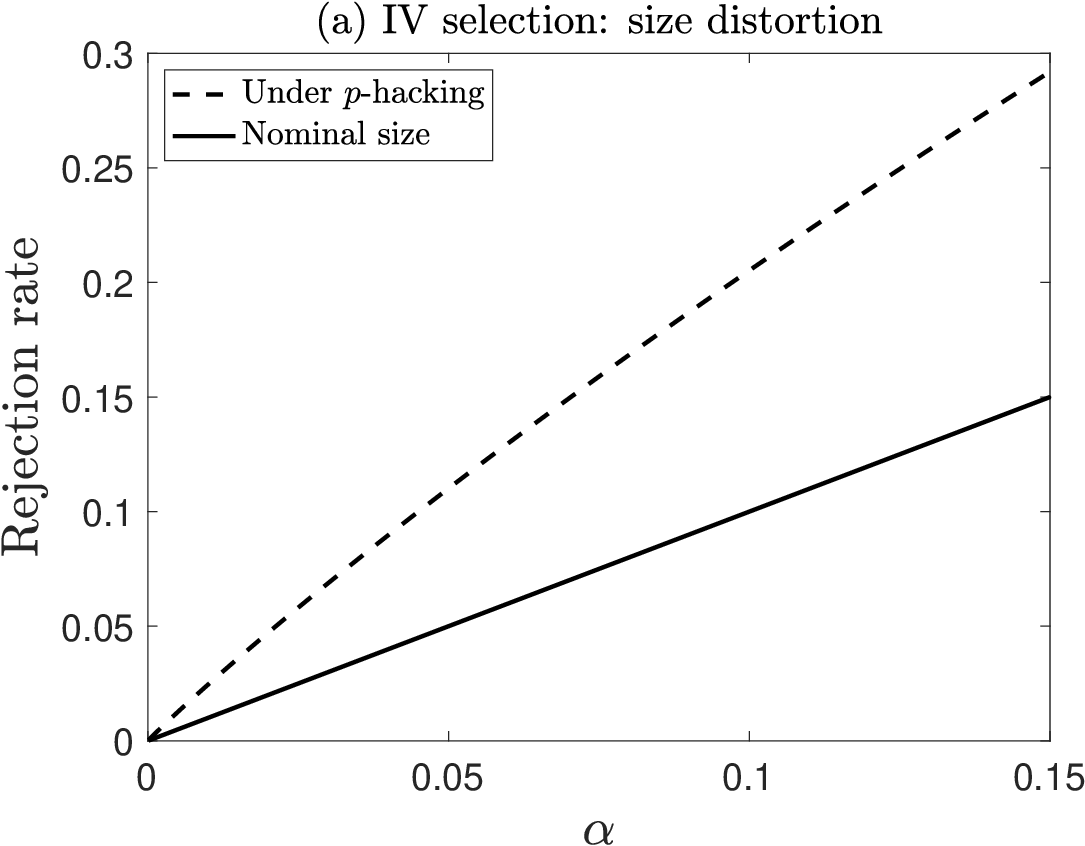}
\includegraphics[width=0.43\textwidth]{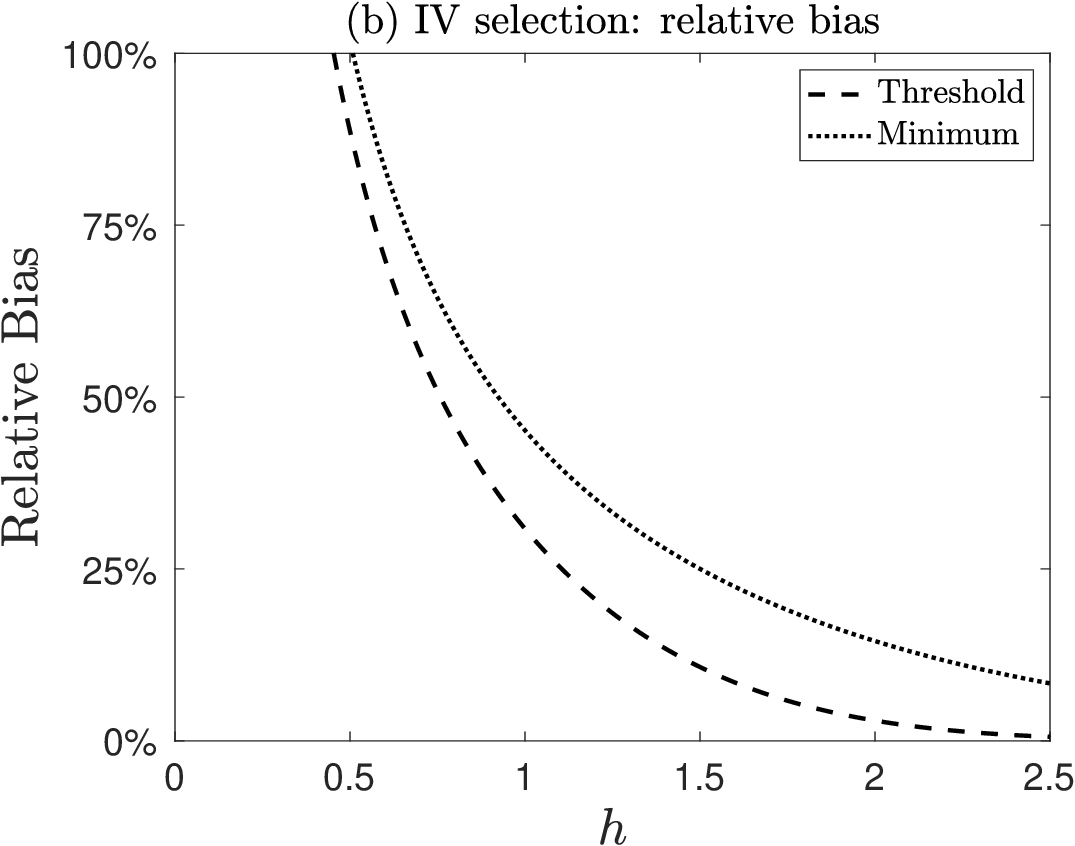}

\end{center}

\vspace{-2mm}

\footnotesize{\textit{Notes:} Figures show the rejection rate under $p$-hacking. Panel (a): rejection rate as a function of $\alpha$ for $h=0$. Panel (b): relative bias from $p$-hacking for different values of $h>0$, $\alpha=0.05$.}

\end{figure}

In terms of the bias induced by $p$-hacking, distributions over $h$ will induce distributions over the biases since the bias for any study depends on the true model. We report here the bias for different $h$ rather than choose a (non-degenerate) distribution. For the threshold and minimum approaches with $\alpha\le 1/2$, the (first-order) relative biases, $\mathcal{B}_2^t$ and $\mathcal{B}_2^m$, respectively,\footnote{Since the first moments of the IV estimators do not exist in just-identified cases \citep{kinal1980existence}, we define $\mathcal{B}_2^j$ to be the mean of the asymptotic distribution of $(\hat\beta^j_r-\beta)/\beta$, where $\hat\beta^j_r$ is the $p$-hacked estimate and $j\in\{m,t\}$.} are
\begin{eqnarray*}
  &\mathcal{B}_2^t =& \mathcal{B}_2^m -{\frac{h^{-1}}{\sqrt{2-\sqrt{2}}}}\phi\left(\sqrt{\frac{\sqrt{2}-1}{\sqrt{2}}}h\right)\Phi\left(\frac{h}{\sqrt{2-\sqrt{2}}}-\sqrt{4-2\sqrt{2}}z_0(\alpha) \right),\\
  &\mathcal{B}_2^m =& {\frac{h^{-1}}{\sqrt{2-\sqrt{2}}}}\phi\left(\sqrt{\frac{\sqrt{2}-1}{\sqrt{2}}}h\right)\Phi\left(\frac{h}{\sqrt{2-\sqrt{2}}}\right)+h^{-1}\sqrt{2} \phi(0) (1-\Phi(\sqrt{2}h)).
\end{eqnarray*}

Figure \ref{fig:IVsize}(b) shows the relative bias as a function of $h$ for both approaches to $p$-hacking. The relative bias decreases in $h$, as in the covariate selection example. It is higher for the minimum approach than for the threshold approach since taking the minimum results in the researcher being better able to find a smaller $p$-value.

\subsection{Variance Bandwidth Selection}
Variance bandwidth selection does not induce bias. We therefore focus on size distortions.
For size at $h=0$ and $p=\alpha$ we have 
\begin{eqnarray*}
G_0(\alpha) &=& \alpha+(1-\alpha)(H_N(0)-H_N(-(2\kappa)^{-1}))-\int_{-(2\kappa)^{-1}}^{0}\Phi(z_0 (\alpha)\omega(r)-h) \eta_N(r)dr 
\end{eqnarray*}
Figure \ref{fig:p_curves_example4size} shows that the size distortions through this example of $p$-hacking are quite modest. The reason is that for a reasonable sample size, the estimated first-order autocorrelation is very close to zero. Thus the estimated standard errors when an additional lag is included are very close to one, meaning that the two $t$-statistics are quite similar and highly correlated. This means that there is not much room to have size distortions due to this $p$-hacking. 
\begin{figure}[H]
     \begin{center}
     \caption[]{Size distortions under $p$-hacking}              \label{fig:p_curves_example4size}
\includegraphics[width=0.43\textwidth]{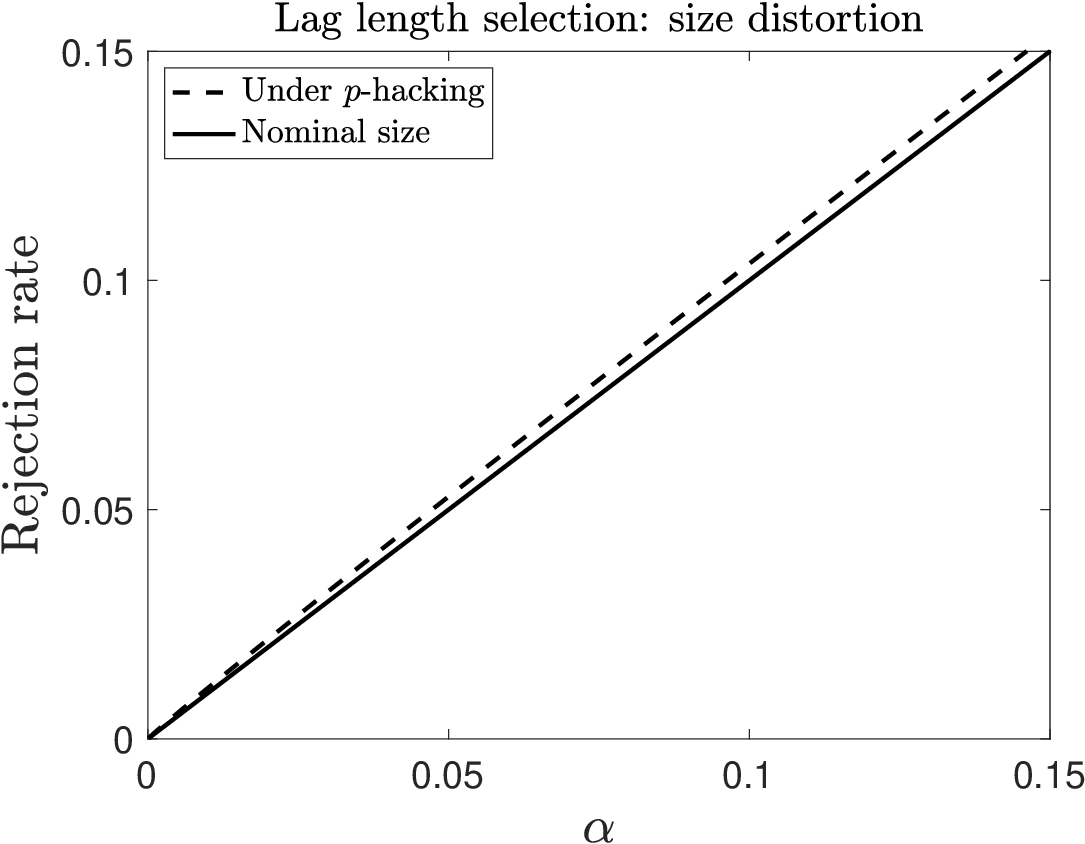} 
\end{center}

\footnotesize{\textit{Notes:} Figure shows rejection rate under $p$-hacking from lag length selection with $N=200$, $\kappa=1/2$, and $h=0$.}

\end{figure}

\section{Detailed Derivations Section \ref{sec:theory} and Appendix \ref{app:theory_bias_size_distortions}} 
\label{app:detailed_derivations}
\subsection{Selecting Control Variables in Linear Regression}
\label{app:derivation_example1}

\subsubsection{$p$-Curve under $p$-Hacking}

We denote by $\hat{\sigma}_j$ the standard error of the estimator of $\beta$ when using $Z_j$ as the control variable ($j=1,2$). Under our assumptions, because the variance of $U$ is known, we have
$$\hat{\sigma}^2_j=\frac{1}{1-\gamma^2},\quad j=1,2.$$
Therefore, the $t$-statistic for testing $H_0: \beta=0$ is distributed as follows
$$T_j = \frac{\sqrt{N}\hat{\beta}_j}{\hat{\sigma}_j}\overset{d}{=}{h} + \frac{W_{xu} - \gamma W_{z_ju}}{\sqrt{1-\gamma^2}},\quad  j=1,2,$$
where  
$$\begin{pmatrix}W_{xu}\\ W_{z_1u}\\W_{z_2u} \end{pmatrix}\sim \mathcal{N}\left(\begin{pmatrix}0\\ 0\\0 \end{pmatrix}, \begin{pmatrix}1&\gamma&\gamma\\ \gamma&1& \gamma^2\\ \gamma&\gamma^2&1 \end{pmatrix}\right).$$
Thus, conditional on $h$, 
$$\begin{pmatrix}T_{1}\\T_{2} \end{pmatrix}\sim \mathcal{N}\left(\begin{pmatrix}h\\ h \end{pmatrix}, \begin{pmatrix}1&\rho\\  \rho&1 \end{pmatrix}\right),$$
where the correlation is $\rho = 1 - \gamma^2$. As the control variables and $X_i$ become more correlated (larger $\gamma$), $\rho$ becomes smaller. 

The CDF of $P_r$ on $(0,1)$ for the threshold case is
\begin{eqnarray*}
G^t_h(p) &=& \Pr(P_r\le p)\\
&=& \Pr(P_1\le p\mid P_1\le\alpha)\Pr(P_1\le\alpha) \\
&&+ \Pr(\min\{P_1,P_2\}\le p\mid P_1>\alpha)\Pr(P_1>\alpha)\\
&=&\Pr(P_1\le \min\{p, \alpha\}) + (1-\Pr(P_1>p, P_2>p \mid P_1>\alpha))\Pr(P_1>\alpha)\\ 
&=&\Pr(T_1\ge z_0(\min\{p, \alpha\}))+\Pr(T_1<z_0(\alpha)) - \Pr(T_1<z_0(\max\{p,\alpha\}), T_2< z_0(p))\\
&=& 1 - \Phi(z_h(\min\{p, \alpha\}))+\Phi(z_h(\alpha)) - \int_{-\infty}^{z_h(p)}\int_{-\infty}^{z_h(\max\{p,\alpha\})}f(x,y; \rho)dxdy,
\end{eqnarray*}
where $f(x,y;\rho) = \frac{1}{2\pi \sqrt{1-\rho^2}}\exp\{-\frac{x^2-2\rho xy +y^2}{2(1-\rho^2)}\}$.

For $p\in (0, \alpha)$, differentiating $G^t_h(p)$ with respect to $p$ yields: 
\begin{eqnarray*}
\frac{dG^t_h(p)}{dp} &=& \frac{dz_h(p)}{dp}\left[- \phi(z_h(p)) - \int_{-\infty}^{z_h(\alpha)}f(z_h(p),y; \rho)dy\right]\\
&=&\frac{ \phi(z_h(p))\left[1 +  \Phi\left(\frac{z_h(\alpha) - \rho z_h(p)}{\sqrt{1-\rho^2}}\right)\right]}{\phi(z_0(p))}.
\end{eqnarray*}
For $p\in (\alpha, 1)$, the derivative is
\begin{eqnarray*}
\frac{dG^t_h(p)}{dp} &=& \frac{2\phi(z_h(p))\Phi\left(\frac{z_h(p) - \rho z_h(p)}{\sqrt{1-\rho^2}}\right)}{\phi(z_0(p))}.
\end{eqnarray*}
It follows that the PDF of $p$-values is
$$g^t_1(p) = \int_{\mathcal{H}}\frac{dG^t_h(p)}{dp}d\Pi(h) =\int_{\mathcal{H}}\frac{\phi(z_h(p))\Upsilon^t_1(p;\alpha, h, \rho)}{\phi(z_0(p))}d\Pi(h),$$
where $\Upsilon^t_1(p;\alpha, h, \rho)=1_{\{p\le\alpha\}}\left[1 +  \Phi\left(\frac{z_h(\alpha) - \rho z_h(p)}{\sqrt{1-\rho^2}}\right)\right]+1_{\{p>\alpha\}}2\Phi\left(\frac{z_h(p) - \rho z_h(p)}{\sqrt{1-\rho^2}}\right)$. The final expression follows because $\phi(z_h(p))/\phi(z_0(p))=\exp\left(hz_0(p) - \frac{h^2}{2}\right)$.

For the case when the researchers report the minimum of two $p$-values, $P_r = \min\{P_1, P_2\}$, we have
\begin{eqnarray*}
G^m_h(p) &=& \Pr(P_r\le p)\\
&=& \Pr(P_1\le p, P_1\le P_2) + \Pr(P_2\le p, P_2< P_1) \\
&=&\Pr(T_1\ge z_0(p), T_1\ge T_2) + \Pr(T_2\ge z_0(p), T_2> T_1)\\
&=&2 \Pr(\xi_1\ge z_h(p), \xi_1\ge \xi_2)\\ 
&=&2\int_{-\infty}^{z_h(p)}\int_{z_h(p)}^{\infty}f(x,y;\rho)dxdy + 2\int_{z_h(p)}^{\infty}\int_y^{\infty}f(x,y;\rho)dxdy,
\end{eqnarray*}
where $\xi_j=T_j-h$, $j=1,2$.

The derivative of $G^m_h(p)$ with respect to $p$ is
\begin{eqnarray*}
\frac{dG^m_h(p)}{dp} &=& 2\frac{dz_h(p)}{dp}\left[\int_{z_h(p)}^{\infty}f(x,z_h(p);\rho)dx - \int_{-\infty}^{z_h(p)}f(z_h(p),y;\rho)dy - \int_{z_h(p)}^{\infty}f(x,z_h(p);\rho)dx \right]\\
&=&2\frac{\phi(z_h(p))}{\phi(z_0(p))}\Phi\left(z_h(p)\sqrt{\frac{1-\rho}{1+\rho}}\right).
\end{eqnarray*}
Therefore, the PDF of $p$-values is 
\begin{equation*}
g_1^{m}(p) =2\int_\mathcal{H}\exp\left(hz_0(p)-\frac{h^2}{2}\right)\Phi\left(z_h(p)\sqrt{\frac{1-\rho}{1+\rho}}\right)d\Pi(h).
\end{equation*}

\subsubsection{Bias of the $p$-Hacked Estimator}
Fix $h$ for now. We have $\hat\beta_r^t=\hat{\beta_1}+(\hat{\beta_2}-\hat{\beta_1})1_{\{T_2>T_1,T_1<z_0(\alpha)\}}$. The bias in the $p$-hacked estimator is given by 
\begin{eqnarray*}
E\hat\beta_r^t-\beta &=& E \left[ \hat{\beta_1}-\beta+(\hat{\beta_2}-\hat{\beta_1})1_{\{T_2>T_1,T_1<z_0(\alpha)\}} \right]\\
&=& E \left[ (\hat{\beta_2}-\hat{\beta_1})1_{\{T_2>T_1,T_1<z_0(\alpha)\}} \right]\\
&=&\frac{1}{\sqrt{N}\sqrt{\rho}}E \left[ (\xi_{2}-\xi_{1})1_{\{\xi_2>\xi_1,\xi_1<z_h(\alpha)\}} \right]\\
&=& \frac{1}{\sqrt{N}\sqrt{\rho}}E \left[ V1_{\{V>0,\xi_1<z_h(\alpha)\}} \right],
\end{eqnarray*}
where $V=\xi_2-\xi_1 \sim \mathcal{N}(0,2(1-\rho))$ and $E[V\xi_1]=-(1-\rho)$. 
Now
\begin{eqnarray*}
E \left[ V1_{\{V>0_1,\xi_1<z_h(\alpha)\}}\right] &=& \int_{0}^{\infty} \int_{-\infty}^{z_h(\alpha)} vf_{V,\xi_1}(v,x)dxdv\\
&=& \int_{0}^{\infty} \int_{-\infty}^{z_h(\alpha)} vf_{\xi_1|V}(x|v)f_V(v)dxdv\\
&=&\int_{0}^{\infty}  v f_V(v) \left( \int_{-\infty}^{z_h(\alpha)}f_{\xi_1|V}(x|v)dx \right) dv,
\end{eqnarray*}
and 
\begin{eqnarray*}
\int_{-\infty}^{z_h(\alpha)}f_{\xi_1|V}(x|v)dx &=& \Pr \left( \xi_1 <z_h(\alpha) \mid V=v\right)\\
&=& \Pr \left( \frac{\xi_1+v/2}{\sqrt{\frac{1+\rho}{2}}} < \frac{z_h(\alpha)+v/2}{\sqrt{\frac{1+\rho}{2}}}  \right)\\
&=& \Phi \left( \frac{z_h(\alpha)+v/2}{\sqrt{\frac{1+\rho}{2}}} \right).
\end{eqnarray*}
So now we have 
$$ E\hat\beta_r^t-\beta = \frac{1}{\sqrt{N}\sqrt{\rho}}\int_{0}^{\infty}  v \Phi \left( \frac{z_h(\alpha)+v/2}{\sqrt{\frac{1+\rho}{2}}} \right)  f_V(v)  dv, $$
and the final expression follows by direct integration.

For the minimum approach, $\hat\beta^m_r=\hat{\beta_1}+(\hat{\beta_2}-\hat{\beta_1})1_{\{T_2>T_1\}}$. The bias in the $p$-hacked estimator is given by 
\begin{eqnarray*}
E\hat\beta_r^m-\beta &=& E \left[ \hat{\beta_1}-\beta +(\hat{\beta_2}-\hat{\beta_1})1_{\{T_2>T_1\}} \right]\\
&=& \frac{1}{\sqrt{N}\sqrt{\rho}}E \left( V1_{\{V>0\}} \right).
\end{eqnarray*}
Now $ E \left[ V1_{\{V>0\}}\right] = \sqrt{2(1-\rho)} \phi(0) $ so 
$$ E\hat\beta^m_r-\beta = \frac{1}{\sqrt{N}} \sqrt{\frac{2(1-\rho)}{\rho}} \phi(0). $$
It follows that $E\hat\beta^t_r \le E\hat\beta^m_r$ because 
\begin{eqnarray*}
EV 1_{\{V>0, \xi_1 < z_h(\alpha)\}}   &=& \int_0^{\infty} vf_V(v) \left( \int_{-\infty}^{z_h(\alpha)} f_{\xi_1|V}(x|v)dx \right) dv\\
&\le& \int_0^{\infty} vf_V(v)dv.  
\end{eqnarray*}

The expressions for $\mathcal{B}^{t}_{1}$ and $\mathcal{B}^{m}_{1}$ can be obtained by dividing the biases by $\beta$.

\subsection{Selecting amongst Instruments in IV Regression}
\label{app:derivation_example2}

\subsubsection{$p$-Curve under $p$-Hacking}

Since $Z_1$ and $Z_2$ are assumed to be uncorrelated, the IV estimator with 2 instruments is
$$\hat{\beta}_{12} = \beta +\left[\frac{(\sum X_iZ_{1i})^2}{\sum Z_{1i}^2}+ \frac{(\sum X_iZ_{2i})^2}{\sum Z_{2i}^2}+o_P(1)\right]^{-1}\left[\frac{\sum X_iZ_{1i}\sum U_iZ_{1i}}{\sum Z_{1i}^2}+\frac{\sum X_iZ_{2i}\sum U_iZ_{2i}}{\sum Z_{2i}^2}\right]$$
with asymptotic variance ${1}/{2\gamma^2}$. Therefore, the $t$-statistic is
$$T_{12} = \sqrt{N}\hat\beta_{12}\sqrt{2}|\gamma|\overset{d}\to \sqrt{2}h + \frac{W_1+W_2}{\sqrt{2}},$$
where $(W_1, W_2)'\sim \mathcal{N}(0,I_2)$. 
With one instrument,
$$\hat\beta_j = \beta+\frac{\sum Z_{ji}U_i}{\sum X_iZ_{ji}}, \quad j=1,2,$$
and the asymptotic variance is ${1}/{\gamma^2}$. Moreover, we have that
$$T_j=\sqrt{N}\hat\beta_j|\gamma|\overset{d}\to h + W_j.$$
Note that $T_{12}$ is asymptotically equivalent to $\frac{T_1+T_2}{\sqrt{2}}$. Define $P_1$ as the $p$-value using $T_{12}$, and $P_2$ and $P_3$ arise from using $T_1$ and $T_2$ respectively.

For now, we fix $h$. Define $z_h(p) = z_0(p) - {h}$ and $D_h(p) = \sqrt{2}z_{\sqrt{2}h}(p)$, where $z_0(p) = \Phi^{-1}(1-p)$. The (asymptotic) CDF of $P_r$ on $(0, 1/2]$ is 
\begin{eqnarray*}
G^t_h(p) &=& \Pr(P_r\le p)\\
&=& \Pr(P_1 \le p\mid P_1\le\alpha)\Pr(P_1\le\alpha) \\
&&+ \Pr(\min\{P_1,P_2,P_3\}\le p\mid P_1>\alpha)\Pr(P_1>\alpha)\\
&=&\Pr(P_{1}\le \min\{p,\alpha\}) + \Pr(P_{1}>\alpha)\\
&& - \Pr(P_1>p, P_2>p, P_3>p|P_{1}>\alpha)\Pr(P_{1}>\alpha)\\
&=&\Pr(T_{12}\ge z_0(\min\{p,\alpha\})) + \Pr(T_{12}<z_0(\alpha))\\
&&- \Pr(T_1< z_0(p), T_2< z_0(p), T_{12}<z_0(\max\{p,\alpha\}))\\
&=& 1 - \Phi(z_{\sqrt{2}h}(\min\{p,\alpha\})) + \Phi(z_{\sqrt{2}h}(\alpha))- \Phi(z_h(p))\Phi(D_h(\max\{p,\alpha\})-z_h(p))\\
&& - \int_{D_h(\max\{p,\alpha\}) - z_h(p)}^{z_h(p)}\phi(x)\Phi(D_h(\max\{p,\alpha\})-x)dx.
\end{eqnarray*}
The last equality follows because for $p\le 1/2$ we have $2z_h(p)>D_h(\max\{p,\alpha\})$, $\Pr(T_1< z_0(p), T_2< z_0(p), T_{12}<z_0(\max\{p,\alpha\}))=\Pr(W_1< z_h(p), W_2<z_h(p), W_1+W_2< D_h(\max\{p,\alpha\}))$ and 
\begin{eqnarray*}
 \Pr(W_1< a, W_2<a, W_1+W_2< b)&=&  \int_{-\infty}^{a}\int_{-\infty}^{b-a}\phi(x)\phi(y)dxdy\\
 &&+\int_{b-a}^{a}\int_{-\infty}^{b - x}\phi(x)\phi(y)dydx\\
 &=&\Phi(a)\Phi(b-a)+\int_{b - a}^{a}\phi(x)\Phi(b-x)dx
 \end{eqnarray*}
for $2a>b$.

The derivative of $G^t_h(p)$ with respect to $p$ on $(0,\alpha)$ is 
\begin{eqnarray*}
\frac{dG^t_h(p)}{dp} &=& \frac{\phi(z_{\sqrt{2}h}(p)) + 2C(z_h(p), D_h(\alpha))}{\phi(z_0(p))},\quad p\in (0, \alpha),
\end{eqnarray*}
where $C(x,y) := \phi(x)\Phi(y-x)$.

For $p\in(\alpha, 1/2)$ the derivative is
\begin{equation*}
    \frac{dG^t_h(p)}{dp} = \frac{\phi(z_{\sqrt{2}h}(p))(1-2\Phi((1-\sqrt{2})z_0(p)))+2C(z_h(p), D_h(p)) }{\phi(z_0(p))}, \quad p\in (\alpha, 1/2).
\end{equation*}

For $p>1/2$, we have $2z_h(p)<D_h(\max\{p, \alpha\})$, and similar arguments yield 
\begin{eqnarray*}
G^t_h(p) &=& 1 - \Pr(W_1< z_h(p), W_2<z_h(p), W_1+W_2< D_h(\max\{p, \alpha\}))\\
&=& 1 - \int_{-\infty}^{z_h(p)}\int_{-\infty}^{z_h(p)}\phi(x)\phi(y)dxdy\\
&=&1 - \Phi^2(z_h(p))
\end{eqnarray*}
and
\begin{equation*}
    \frac{dG^t(p)}{dp}=\frac{2\phi(z_h(p)\Phi(z_h(p))}{\phi(z_0(p)}, \quad p>1/2.
\end{equation*}
Since $g^t(p) = \int_{\mathcal{H}}\frac{G^t_h(p)}{dp}d\Pi(h)$, we have
\begin{equation*}
    g^t_2(p) = \int_{\mathcal{H}}\exp\left(hz_0(p) - \frac{h^2}{2}\right)\Upsilon^t_2(p; \alpha, h)d\Pi(h),
\end{equation*}
where $\zeta(p) = 1 - 2\Phi((1-\sqrt{2})z_0(p))$, $cv_1(p)=z_0(p)$ and

\begin{equation*}
  \Upsilon^t_2(p; \alpha, h)=\begin{cases}
    \frac{\phi(z_{\sqrt{2}h}(p))}{\phi(z_h(p))} + 2\Phi(D_h(\alpha)-z_h(p)), & \text{if $0<p\le\alpha$},\\
    \frac{\phi(z_{\sqrt{2}h}(p))}{\phi(z_h(p))}\zeta(p) + 2\Phi(D_h(p)-z_h(p)), & \text{if $\alpha<p\le1/2$},\\
     2\Phi(z_h(p)), & \text{if $1/2<p<1$}.\\
  \end{cases}
\end{equation*}

The $p$-curve for the minimum approach arises as a corollary of the above results.

\subsubsection{Bias of the $p$-Hacked Estimator}

For the bias, consider the estimator for the causal effect given by the threshold approach. The $p$-hacked estimator is given by 
\begin{equation*}
    \hat{\beta}^t_r = \hat{\beta}_{12}  + (\hat{\beta_1}-\hat{\beta}_{12})1_{\{A_N\}} + (\hat{\beta}_1-\hat{\beta}_{12})1_{\{B_N\}},
\end{equation*}
where we define sets $A_N=\{T_{12}<z_0(\alpha),T_1>T_{12},T_1>T_2\}$ and $B_N=\{T_{12}<z_0(\alpha),T_2>T_{12},T_2>T_1\}$. Using the same standard 2SLS results used above to generate the approximate distributions of the $t$-statistics and by the continuous mapping theorem, 
$$\sqrt{N}|\gamma|(\hat\beta^t_r-\beta) \overset{d}\to \xi^t:= \frac{W_1+W_2}{2} + \frac{W_1-W_2}{2}1_{\{A\}} + \frac{W_2-W_1}{2}1_{\{B\}},$$
where $A = \{\sqrt{2}h+\frac{W_1+W_2}{\sqrt{2}}<z_0(\alpha), h+W_1> \sqrt{2}h+\frac{W_1+W_2}{\sqrt{2}}, W_1>W_2\}$ and $B = \{\sqrt{2}h+\frac{W_1+W_2}{\sqrt{2}}<z_0(\alpha), h+W_2> \sqrt{2}h+\frac{W_1+W_2}{\sqrt{2}}, W_2>W_1\}$. 
By the symmetry of the problem,
$$E[\xi^t]= E[(W_1-W_2)1_{\{A\}}]=E[W_11_{\{A\}}]-E[W_21_{\{A\}}].$$
To compute the expectations, we need to calculate $P(A|W_1)$ and $P(A|W_2)$. Note that
\begin{equation*}
  A=\begin{cases}
    \sqrt{2}h+\frac{W_1+W_2}{\sqrt{2}}<z_0(\alpha) & \\
    h+W_1> \sqrt{2}h+\frac{W_1+W_2}{\sqrt{2}} & \\
    W_1>W_2 & 
  \end{cases}
  \Leftrightarrow
    \begin{cases}
   W_2<\sqrt{2}z_0(\alpha)-2h-W_1 & \\
    W_2< (\sqrt{2}-1)(W_1-\sqrt{2}h)& \\
    W_2<W_1 & 
  \end{cases}
    \Leftrightarrow
    \begin{cases}
   W_1<\sqrt{2}z_0(\alpha)-2h-W_2 & \\
    W_1> \sqrt{2}h+\frac{W_2}{\sqrt{2}-1}& \\
    W_1>W_2
  \end{cases}
  \eqno{(*)}
\end{equation*}

From equation $(*)$, we have
$$\Pr(A|W_1)=\Phi(\min\{W_1, \sqrt{2}z_0(\alpha)-2h-W_1, (\sqrt{2}-1)(W_1-\sqrt{2}h)\}),$$
where
\begin{equation*}
  \min\{W_1, \sqrt{2}z_0(\alpha)-2h-W_1, (\sqrt{2}-1)(W_1-\sqrt{2}h)\}=\begin{cases}
    W_1, \text{ if $W_1<-h$,} & \\
    \sqrt{2}z_0(\alpha)-2h-W_1, \text{ if $W_1>z_0(\alpha)-h$,}& \\
    (\sqrt{2}-1)(W_1-\sqrt{2}h), \text{ if $W_1\in (-h, z_0(\alpha)-h)$.} & 
  \end{cases}
 \end{equation*}
Also, the last system of inequalities in equation $(*)$ is equivalent to 
  $$W_1\in  \begin{cases}
   \left(W_2, \sqrt{2}z_0(\alpha)-2h-W_2\right), \text{ if $W_2<-h$,} & \\
    \left(\sqrt{2}h+\frac{W_2}{\sqrt{2}-1}, \sqrt{2}z_0(\alpha)-2h-W_2\right), \text{ if $W_2\ge -h$.}& 
  \end{cases}$$
Note that $\left(\sqrt{2}h+\frac{W_2}{\sqrt{2}-1}, \sqrt{2}z_0(\alpha)-2h-W_2\right)$ is non-empty when $W_2\le (\sqrt{2}-1)z_0(\alpha)-h$, and $ \left(W_2, \sqrt{2}z_0(\alpha)-2h-W_2\right)$ is non-empty for all values of $W_2<-h$. Therefore,
\begin{eqnarray*}
\Pr(A|W_2)&=&\Pr\left(W_1\in \left(W_2, \sqrt{2}z_0(\alpha)-2h-W_2\right)|W_2\right) \cdot1_{\left\{W_2<-h\right\}}\\
&&+\Pr\left(W_1\in \left(\sqrt{2}h+\frac{W_2}{\sqrt{2}-1}, \sqrt{2}z_0(\alpha)-2h-W_2\right)|W_2\right) \cdot1_{\left\{W_2\ge-h\right\}}\\
&=&\left[\Phi\left(\sqrt{2}z_0(\alpha)-2h-W_2\right) - \Phi(W_2)\right]\cdot1_{\left\{W_2<-h\right\}}\\
&&+\left[\Phi\left(\sqrt{2}z_0(\alpha)-2h-W_2\right) - \Phi\left(\sqrt{2}h+\frac{W_2}{\sqrt{2}-1}\right)\right]\cdot 1_{\left\{-h\le W_2\le(\sqrt{2}-1)z_0(\alpha)-h   \right\}}\\
&=&\Phi\left(\sqrt{2}z_0(\alpha)-2h-W_2\right) \cdot 1_{\left\{W_2\le(\sqrt{2}-1)z_0(\alpha)-h   \right\}} -\Phi(W_2)\cdot1_{\left\{W_2<-h\right\}} \\
&&- \Phi\left(\sqrt{2}h+\frac{W_2}{\sqrt{2}-1}\right)\cdot 1_{\left\{-h\le W_2\le(\sqrt{2}-1)z_0(\alpha)-h   \right\}}
\end{eqnarray*}

To finish the calculation of expectations, we will need to calculate several integrals of the form 
$$\int_{L}^{U}w\phi(w)\Phi(aw+b)dw.$$

The following result is therefore useful.
\begin{lemma}\label{lem:auxil}
\begin{eqnarray*}
\int_{L}^{U}w\phi(w)\Phi(aw+b)dw &=& \Phi(aL+b)\phi(L)-\Phi(aU+b)\phi(U)+\frac{a}{\sqrt{1+a^2}}\phi\left(\frac{b}{\sqrt{1+a^2}}\right)\\&&\times\left[\Phi\left(\sqrt{1+a^2}U+\frac{ab}{\sqrt{1+a^2}}\right)-\Phi\left(\sqrt{1+a^2}L+\frac{ab}{\sqrt{1+a^2}}\right)\right]
\end{eqnarray*}
\end{lemma}
\begin{proof}
\begin{eqnarray*}
\int_{L}^{U}w\phi(w)\Phi(aw+b)dw &=& -\int_{L}^{U}\Phi(aw+b)d\phi(w) \\
&=&\Phi(aL+b)\phi(L)-\Phi(aU+b)\phi(U) + \int_{L}^{U}\phi(w)d\Phi(aw+b)\\
&=&\Phi(aL+b)\phi(L)-\Phi(aU+b)\phi(U) + a\int_{L}^{U}\phi(w)\phi(aw+b)dw\\
&=&\Phi(aL+b)\phi(L)-\Phi(aU+b)\phi(U) + a\phi\left(\frac{b}{\sqrt{1+a^2}}\right)J,
\end{eqnarray*}
where
\begin{eqnarray*}J:&=&\int_{L}^{U}\phi\left(\sqrt{1+a^2}w+\frac{ab}{\sqrt{1+a^2}}\right)\\ &=& \frac{1}{\sqrt{1+a^2}}\left[\Phi\left(\sqrt{1+a^2}U+\frac{ab}{\sqrt{1+a^2}}\right)-\Phi\left(\sqrt{1+a^2}L+\frac{ab}{\sqrt{1+a^2}}\right)\right].
\end{eqnarray*}
\end{proof}

Finally, consider
\begin{eqnarray*}
E[W_11_{\{A\}}]&=&E[W_1P(A|W_1)]\\
&=&E[W_1\Phi(\min\{W_1, \sqrt{2}z_0(\alpha)-2h-W_1, (\sqrt{2}-1)(W_1-\sqrt{2}h)\})] \\
&=&\int_{-\infty}^{-h}w\phi(w)\Phi(w)dw\\
&&+\int_{-h}^{z_0(\alpha)-h}w\phi(w)\Phi((\sqrt{2}-1)w - \sqrt{2}(\sqrt{2}-1)h)dw\\
&&+\int_{z_0(\alpha)-h}^{+\infty}w\phi(w)\Phi(\sqrt{2}z_0(\alpha)-2h-w)dw\\
&=&-(1-\Phi(h))\phi(h)+\frac{1}{\sqrt{2}}\phi(0)(1-\Phi(\sqrt{2}h))\\
&&+(1-\Phi(h))\phi(h) - \Phi((\sqrt{2}-1)z_0(\alpha)-h)\phi(z_0(\alpha)-h)\\
&&+{\frac{\sqrt{2}-1}{\sqrt{4-2\sqrt{2}}}}\phi\left({\frac{\sqrt{2}-1}{\sqrt{2-\sqrt{2}}}}h\right)\left[\Phi\left(\frac{h}{\sqrt{2-\sqrt{2}}}\right) -\Phi\left(\frac{h}{\sqrt{2-\sqrt{2}}}-\sqrt{4-2\sqrt{2}}z_0(\alpha) \right)\right]\\
&&+\Phi((\sqrt{2}-1)z_0(\alpha)-h)\phi(z_0(\alpha)-h)-\frac{1}{\sqrt{2}}\phi(z_0(\alpha)-\sqrt{2}h)[1 - \Phi((\sqrt{2}-1)z_0(\alpha))]\\
&=&\frac{1}{\sqrt{2}}\phi(0)(1-\Phi(\sqrt{2}h)) -\frac{1}{\sqrt{2}}\phi(z_0(\alpha)-\sqrt{2}h)[1 - \Phi((\sqrt{2}-1)z_0(\alpha))]\\
&&+{\frac{\sqrt{2}-1}{\sqrt{4-2\sqrt{2}}}}\phi\left({\frac{\sqrt{2}-1}{\sqrt{2-\sqrt{2}}}}h\right)\left[\Phi\left(\frac{h}{\sqrt{2-\sqrt{2}}}\right) -\Phi\left(\frac{h}{\sqrt{2-\sqrt{2}}}-\sqrt{4-2\sqrt{2}}z_0(\alpha) \right)\right],
\end{eqnarray*}
where the fourth equality follows from the direct application of Lemma \ref{lem:auxil}. Moreover, we have that
\begin{eqnarray*}
E[W_21_{\{A\}}]&=&E[W_2\Pr(A|W_2)]\\
&=&E\left[W_2\Phi\left(\sqrt{2}z_0(\alpha)-2h-W_2\right) \cdot 1_{\left\{W_2\le(\sqrt{2}-1)z_0(\alpha)-h \right\}}\right]\\
&&-E\left[W_2\Phi(W_2)\cdot1_{\left\{W_2<-h\right\}}\right]\\
&&-E\left[W_2\Phi\left(\sqrt{2}h+\frac{W_2}{\sqrt{2}-1}\right)\cdot 1_{\left\{-h\le W_2\le(\sqrt{2}-1)z_0(\alpha)-h   \right\}}\right]\\
&=&\int_{-\infty}^{(\sqrt{2}-1)z_0(\alpha)-h}w\phi(w)\Phi\left(\sqrt{2}z_0(\alpha)-2h-w\right)dw\\
&&-\int_{-\infty}^{-h}w\phi(w)\Phi(w)dw\\
&&-\int_{-h}^{(\sqrt{2}-1)z_0(\alpha)-h}w\phi(w)\Phi\left(\sqrt{2}h+\frac{w}{\sqrt{2}-1}\right)dw\\
&=& -\Phi(z_0(\alpha)-h)\phi((\sqrt{2}-1)z_0(\alpha)-h)-\frac{1}{\sqrt{2}}\phi(z_0(\alpha)-\sqrt{2}h)(1-\Phi((\sqrt{2}-1)z_0(\alpha)))\\
&&+\left(1 - \Phi\left(h\right)\right)\phi\left(h\right)-\frac{1}{\sqrt{2}}\phi(0)\left(1-\Phi\left(\sqrt{2}h\right)\right)\\
&&-(1-\Phi\left(h\right))\phi\left(h\right)+\Phi(z_0(\alpha)-h)\phi((\sqrt{2}-1)z_0(\alpha)-h)\\
&&-\frac{1}{\sqrt{4-2\sqrt{2}}}\phi\left(\sqrt{\frac{\sqrt{2}-1}{\sqrt{2}}}h\right)\left[\Phi\left(\frac{h}{\sqrt{2-\sqrt{2}}}\right) -\Phi\left(\frac{h}{\sqrt{2-\sqrt{2}}}-\sqrt{4-2\sqrt{2}}z_0(\alpha) \right)\right]\\
&=& -\frac{1}{\sqrt{2}}\phi(z_0(\alpha)-\sqrt{2}h)(1-\Phi((\sqrt{2}-1)z_0(\alpha)))-\frac{1}{\sqrt{2}}\phi(0)\left(1-\Phi\left(\sqrt{2}h\right)\right)\\
&&-\frac{1}{\sqrt{4-2\sqrt{2}}}\phi\left(\sqrt{\frac{\sqrt{2}-1}{\sqrt{2}}}h\right)\left[\Phi\left(\frac{h}{\sqrt{2-\sqrt{2}}}\right) -\Phi\left(\frac{h}{\sqrt{2-\sqrt{2}}}-\sqrt{4-2\sqrt{2}}z_0(\alpha) \right)\right].
\end{eqnarray*}
Combining these results gives us the first-order bias of $\hat\beta^t_r$ as $|\gamma|^{-1}E[\xi^t]$, where
\begin{eqnarray*}
E[\xi^t] &=& {\frac{1}{\sqrt{2-\sqrt{2}}}}\phi\left(\sqrt{\frac{\sqrt{2}-1}{\sqrt{2}}}h\right)\left(\Phi\left(\frac{h}{\sqrt{2-\sqrt{2}}}\right) -\Phi\left(\frac{h}{\sqrt{2-\sqrt{2}}}-\sqrt{4-2\sqrt{2}}z_0(\alpha) \right)\right)\\
&&+\sqrt{2}\phi(0)\left(1-\Phi\left(\sqrt{2}h\right)\right).
\end{eqnarray*}

Finally, for the minimum approach, the $p$-hacked estimator is given by 
\begin{equation*}
    \hat{\beta}^m_r = \hat{\beta}_{12}  + (\hat{\beta_1}-\hat{\beta}_{12})1_{\{A_N\}} + (\hat{\beta}_1-\hat{\beta}_{12})1_{\{B_N\}},
\end{equation*}
where now we define sets $A_N=\{T_{1}>T_{12},T_1>T_{2}\}$ and $B_N=\{T_2>T_{12},T_2>T_1\}$. 

Thus
$$\sqrt{N}|\gamma|(\hat\beta^m_r-\beta) \overset{d}\to \xi^m:= \frac{W_1+W_2}{2} + \frac{W_1-W_2}{2}1_{\{A\}} + \frac{W_2-W_1}{2}1_{\{B\}},$$
where $A = \{W_1>\max\{W_2, \sqrt{2}h+W_2/(\sqrt{2}-1)\}\}$ and $B = \{W_2>\max\{W_1, \sqrt{2}h+W_1/(\sqrt{2}-1)\}\}$.

By the symmetry of the problem,
$$E[\xi^m] = E[(W_1-W_2)1_{\{A\}}]=E[W_11_{\{A\}}]-E[W_21_{\{A\}}].$$

Note that $\Pr(A|W_2) = 1 - \Phi(\max\{W_2, \sqrt{2}h+W_2/(\sqrt{2}-1)\})$ and $\Pr(A|W_1) = \Phi(\min\{W_1, W_1(\sqrt{2}-1)-\sqrt{2}(\sqrt{2}-1)h\})$. Therefore,
\begin{eqnarray*}
E[W_11_{\{A\}}]&=&E[W_1\Pr(A|W_1)]\\
&=&E\left[W_1\Phi(\min\{W_1, W_1(\sqrt{2}-1)-\sqrt{2}(\sqrt{2}-1)h\})\right]\\
&=&\frac{1}{\sqrt{2}}\phi(0)(1-\Phi(\sqrt{2}h))+\frac{\sqrt{2}-1}{\sqrt{4-2\sqrt{2}}}\phi\left(\sqrt{\frac{\sqrt{2}-1}{\sqrt{2}}}h\right)\Phi\left(\frac{h}{\sqrt{2-\sqrt{2}}}\right)
\end{eqnarray*}
and
\begin{eqnarray*}
E[W_21_{\{A\}}]&=&E[W_2\Pr(A|W_2)]\\
&=&-E\left[W_2\Phi(\max\{W_2, \sqrt{2}h+W_2/(\sqrt{2}-1)\})\right]\\
&=&-\frac{1}{\sqrt{2}}\phi(0)(1-\Phi(\sqrt{2}h))-\frac{1}{\sqrt{4-2\sqrt{2}}}\phi\left(\sqrt{\frac{\sqrt{2}-1}{\sqrt{2}}}h\right)\Phi\left(\frac{h}{\sqrt{2-\sqrt{2}}}\right)
\end{eqnarray*}
Putting these together gives the asymptotic bias of $\hat\beta^m_r$ as $|\gamma|^{-1}E[\xi^m]$, where
$$E[\xi^m]=\frac{1}{\sqrt{2-\sqrt{2}}}\phi\left(\sqrt{\frac{\sqrt{2}-1}{\sqrt{2}}}h\right)\Phi\left(\frac{h}{\sqrt{2-\sqrt{2}}}\right) +\sqrt{2}\phi(0)(1-\Phi(\sqrt{2}h)).$$

The expressions for $\mathcal{B}^{t}_{2}$ and $\mathcal{B}^{m}_{2}$ can be obtained by dividing the biases by $\beta$.

\subsection{Selecting across Datasets}
\label{app:derivation_example3}

With the assumption that the datasets are independent, we have that $(T_1,\dots,T_K)' \sim \mathcal{N}(h \iota_K , I_K)$, where $\iota_K$ is a $K\times1$ vector of ones. The assumption that each dataset tests for the same effect mathematically appears as $h$ being the mean for all $t$-statistics.  For $K=2$, the definition for $P_r$ in this example is the same as that in Appendix \ref{app:derivation_example1} with $\rho=0$. Hence the result for the thresholding case is $g_1^t$ evaluated at $\rho=0$ and for the minimum case is $g_1^m$ also evaluated at $\rho=0$.

For general $K$, note that for the minimum case 
\begin{eqnarray*}
G^m(p;K) &=& \Pr(\max(T_1,T_2,\dots,T_K) \ge z_0(p))\\
&=& 1 - \Pr(T_1 \le z_0(\alpha),T_2 \le z_0(\alpha),\dots,T_K \le z_0(\alpha)) \\
&=& 1 - (\Phi(z_h(p))^K.
\end{eqnarray*}

Setting $p=\alpha$ gives the size after $p$-hacking for a nominal value of $\alpha$. Differentiating with respect to $p$ generates the $p$-curve 
\begin{eqnarray*}
g_4^m(p;K) &=& \frac{d}{d p}\left( 1 - (\Phi(z_h(p))^K\right) \\
&=& -K \Phi(z_h(p))^{K-1} \frac{d}{d p} \Phi(z_h(p))  \\
&=& K \Phi(z_h(p))^{K-1} \frac{ \phi(z_h(p))}{\phi(z_0(p))}.
\end{eqnarray*}
The expression in the text follows directly from integrating over $h$.

\subsection{Variance Bandwidth Selection}
\label{app:derivation_example4}

Note that $T_1\sim \mathcal{N}(h,1)$ and $\hat\rho$ are independent.\footnote{The independence follows from the fact that $\hat\rho$ is a function of $V:=(U_2-\bar{U}, \dots, U_N-\bar{U})'$, $T_1 = h+\sqrt{N}\bar{U}$ and that $V$ and $\bar{U}$ are independent.} Also note that $T_2\ge T_1$ happens in the following cases: (i) $T_1 \ge 0$ and $0<\omega^2(\hat\rho)\le 1$, equivalent to $-(2\kappa)^{-1}<\hat{\rho}\le 0$; (ii) $T_1 < 0$ and $\hat{\rho}>0$. The researchers  report the $p$-value corresponding to $T_1$ if the result is significant at level $\alpha$ or if $\hat\omega^2<0$, otherwise they report the $p$-value associated with the largest $t$-statistic. Fixing $h$, we have 
\begin{eqnarray*}
G^t_h(p) &=& \Pr(P_r\le p)\\
&=& \Pr(T_1\ge z_0(p), T_1\ge z_0(\alpha)) \\
&& +\Pr(T_1\ge z_0(p), T_1< z_0(\alpha), -\infty<\hat\rho\le -(2\kappa)^{-1}) \\
&&+ \Pr(T_1 \ge z_0(p), T_1\ge 0, T_1< z_0(\alpha), \hat\rho > 0)\\
&&+ \Pr(T_2 \ge z_0(p), T_1\ge 0, T_1< z_0(\alpha), -(2\kappa)^{-1}<\hat\rho \le 0)\\
&&+ \Pr(T_2 \ge z_0(p), T_1\le 0, T_1 < z_0(\alpha), \hat\rho > 0)\\
&&+ \Pr(T_1 \ge z_0(p), T_1\le 0, T_1 < z_0(\alpha), -(2\kappa)^{-1}<\hat\rho \le 0).
\end{eqnarray*}

We can rewrite these expressions using the independence of $T_1$ and $\hat\rho$. For $p \le \alpha\le 1/2$, this is 
\begin{eqnarray*}
G^t_h(p) &=& 1- \Phi(z_h(p))+ \int_{-(2\kappa)^{-1}}^{l(p)} (\Phi(z_h(\alpha)) - \Phi(z_0(p)\omega(r)-h)) \eta_N(r)dr.  \\
\end{eqnarray*}
The last term follows since $\Pr(T_2 \ge z_0(p) ,0 \le T_1\le z_0(\alpha), -(2\kappa)^{-1}<\hat{\rho} \le 0)$ can be written as $\Pr(z_0(p)\omega(\hat\rho) \le T_1\le z_0(\alpha), -(2\kappa)^{-1}<\hat{\rho} \le l(p))$ and 
$$ l(p) = \frac{1}{2 \kappa} \left( \left(\frac{z_0(\alpha)}{z_0(p)} \right)^2-1\right), $$ which is the largest value for $ \hat{\rho}$ at each $p\le \alpha$ for which the interval for $T_1$ is nonempty.

For $\alpha < p \le 1/2$, we have
\begin{eqnarray*}
G^t_h(p) &=& 1-\Phi(z_h(p))(1-H_N(0)+H_N(-(2\kappa)^{-1}))-\int_{-(2\kappa)^{-1}}^{0}\Phi(z_0 (p)\omega(r)-h) \eta_N(r)dr . 
\end{eqnarray*}
For $\alpha<p$ and $p > 1/2$, we obtain 
\begin{eqnarray*}
G^t_h(p) &=& 1-\Phi(z_h(p))H_N(0)- \int_{0}^{\infty} \Phi(z_0(p)\omega(r)-h) \eta_N(r)dr  .
\end{eqnarray*}

Differentiating with respect to $p$ and integrating over the distribution of $h$ gives the density 
\begin{equation*}
    g^t_3(p) = \int_{\mathcal{H}}\exp\left(h z_0(p) - \frac{h^2}{2}\right) \Upsilon^t_4(p; \alpha, h, \kappa)d\Pi(h),
\end{equation*}
where  
\begin{equation*}
  \Upsilon^t_3= \begin{cases}
    1+\frac{1}{\phi(z_h(p))}\int_{-(2\kappa)^{-1}}^{l(p)}\omega(r)\phi(z_0(p)\omega(r)-h)\eta_N(r)dr, & \text{if $0<p\le\alpha$},\\ 
       1-H_N(0)+H_N(-(2\kappa)^{-1})+\frac{1}{\phi(z_h(p))}\int_{-(2\kappa)^{-1}}^{0}\omega(r)\phi(z_0(p)\omega(r)-h)\eta_N(r)dr, & \text{if $\alpha<p\le 1/2$},\\
    H_N(0)+\frac{1}{\phi(z_h(p))}\int_{0}^{\infty}\omega(r)\phi(z_0(p)\omega(r)-h)\eta_N(r)dr, & \text{if $1/2<p<1$}.
    \end{cases}
\end{equation*}

The derivations for the minimum $p$-value approach are analogous and presented below. Note that
\begin{eqnarray*}
G^m_h(p) &=& \Pr(P_r\le p)\\
&=&\Pr(T_1\ge z_0(p), -\infty<\hat\rho\le -(2\kappa)^{-1}) \\
&&+ \Pr(T_1 \ge z_0(p), T_1\ge 0, \hat\rho > 0)\\
&&+ \Pr(T_2 \ge z_0(p), T_1\ge 0, -(2\kappa)^{-1}<\hat\rho \le 0)\\
&&+ \Pr(T_2 \ge z_0(p), T_1\le 0, \hat\rho > 0)\\
&&+ \Pr(T_1 \ge z_0(p), T_1\le 0, -(2\kappa)^{-1}<\hat\rho \le 0)
\end{eqnarray*}
For $p\le 1/2$, we have
\begin{eqnarray*}
G^m_h(p) &=& H_N(0) - H_N(-(2\kappa)^{-1})+(1-\Phi(z_h(p)))(1-H_N(0)+H_N(-(2\kappa)^{-1}))\\
&&- \int_{-(2\kappa)^{-1}}^{0} \Phi(z_0(p)\omega(r)-h) \eta_N(r)dr  
\end{eqnarray*}
and, for $p>1/2$, we have
\begin{eqnarray*}
G^m_h(p) &=& 1-H_N(0) +(1-\Phi(z_h(p)))H_N(0)- \int_{0}^{\infty} \Phi(z_0(p)\omega(r)-h) \eta_N(r)dr. 
\end{eqnarray*}
Differentiating with respect to $p$ and integrating over the distribution of $h$ gives the density 
\begin{equation*}
    g_3^m(p) = \int_{\mathcal{H}}\exp\left(h z_0(p) - \frac{h^2}{2}\right) \Upsilon^m_3(p; \alpha, h)d\Pi(h),
\end{equation*}
where  
\begin{equation*}
  \Upsilon^m_3= \begin{cases}
       1-H_N(0)+H_N(-(2\kappa)^{-1})+\frac{1}{\phi(z_h(p))}\int_{-(2\kappa)^{-1}}^{0}\omega(r)\phi(z_0(p)\omega(r)-h)\eta_N(r)dr, & \text{if $0<p\le 1/2$},\\
    H_N(0)+\frac{1}{\phi(z_h(p))}\int_{0}^{\infty}\omega(r)\phi(z_0(p)\omega(r)-h)\eta_N(r)dr, & \text{if $1/2<p<1$}.
    \end{cases}
\end{equation*}

\section{Numerical $p$-Curves for Two-sided Tests} 
\label{app:p_curves_two_sided}

In this section, we show numerically computed $p$-curves based on two-sided tests corresponding to theoretical results based on one-sided tests in Section \ref{sec:theory} and Appendix \ref{app:selecting_across_datasets}.

\begin{figure}[H]
\begin{center}
\caption{$p$-Curves from covariate selection with two-sided tests} 
\label{fig:p_curves_example1_2s}
        \includegraphics[width=0.325\textwidth]{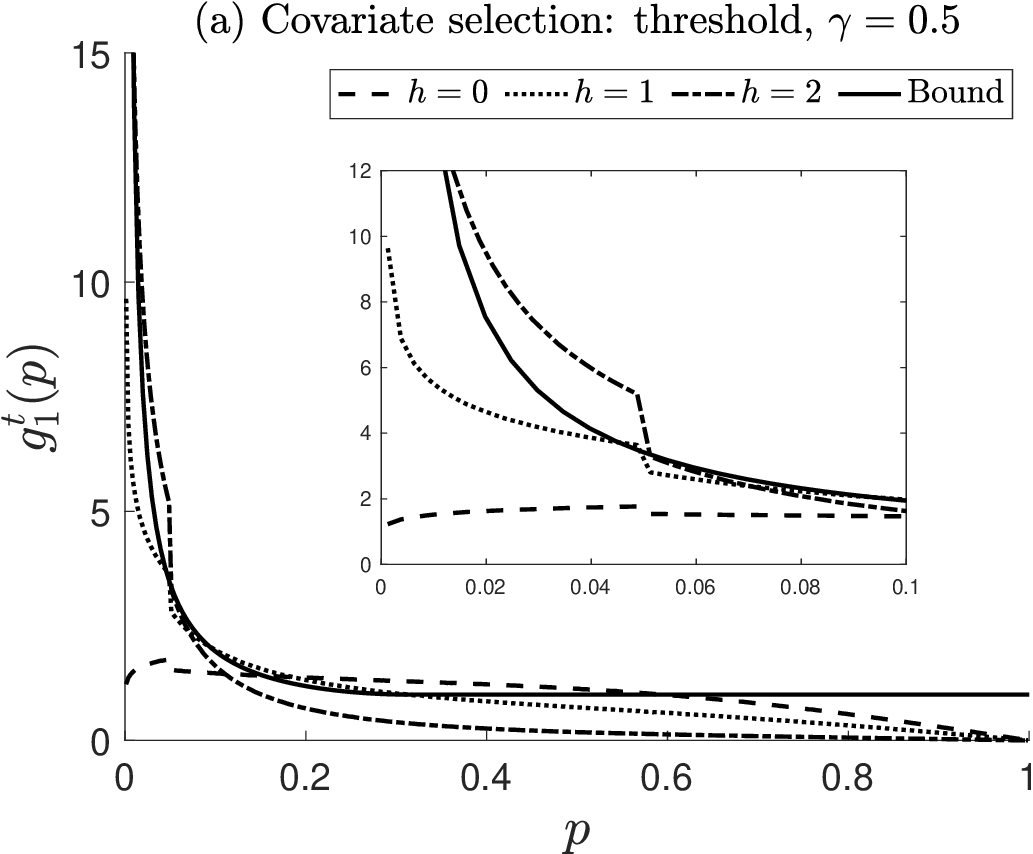}
\includegraphics[width=0.325\textwidth]{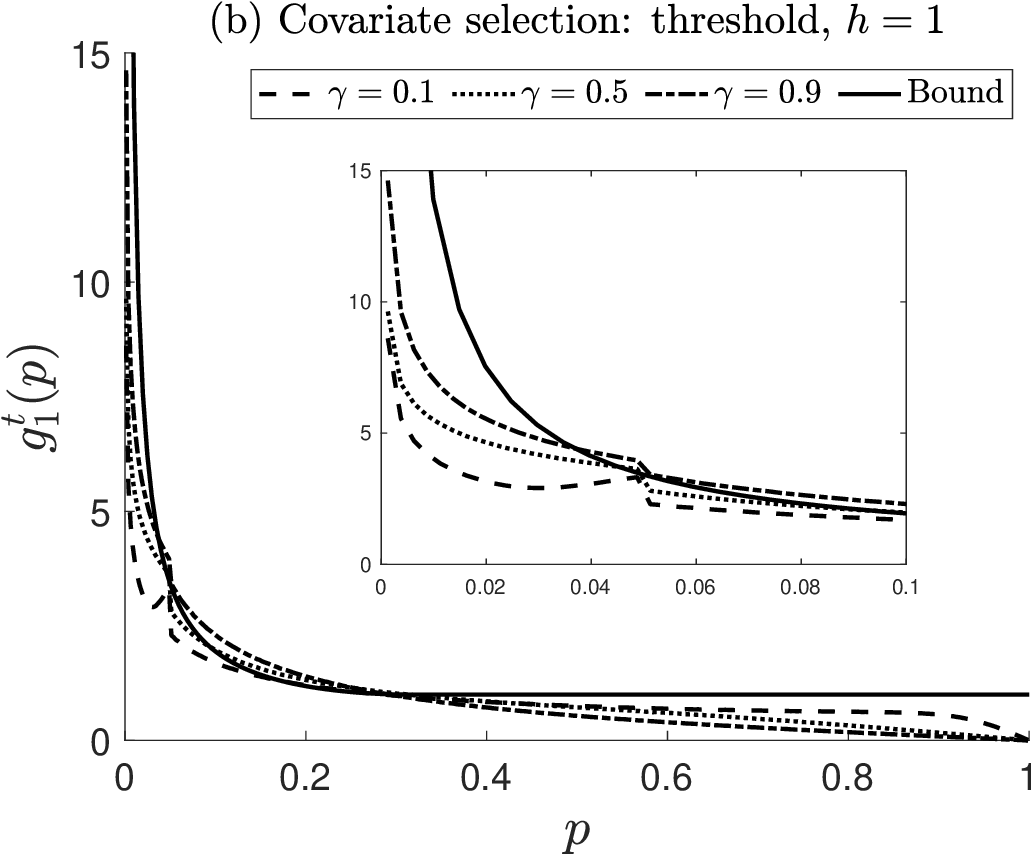}
 \includegraphics[width=0.325\textwidth]{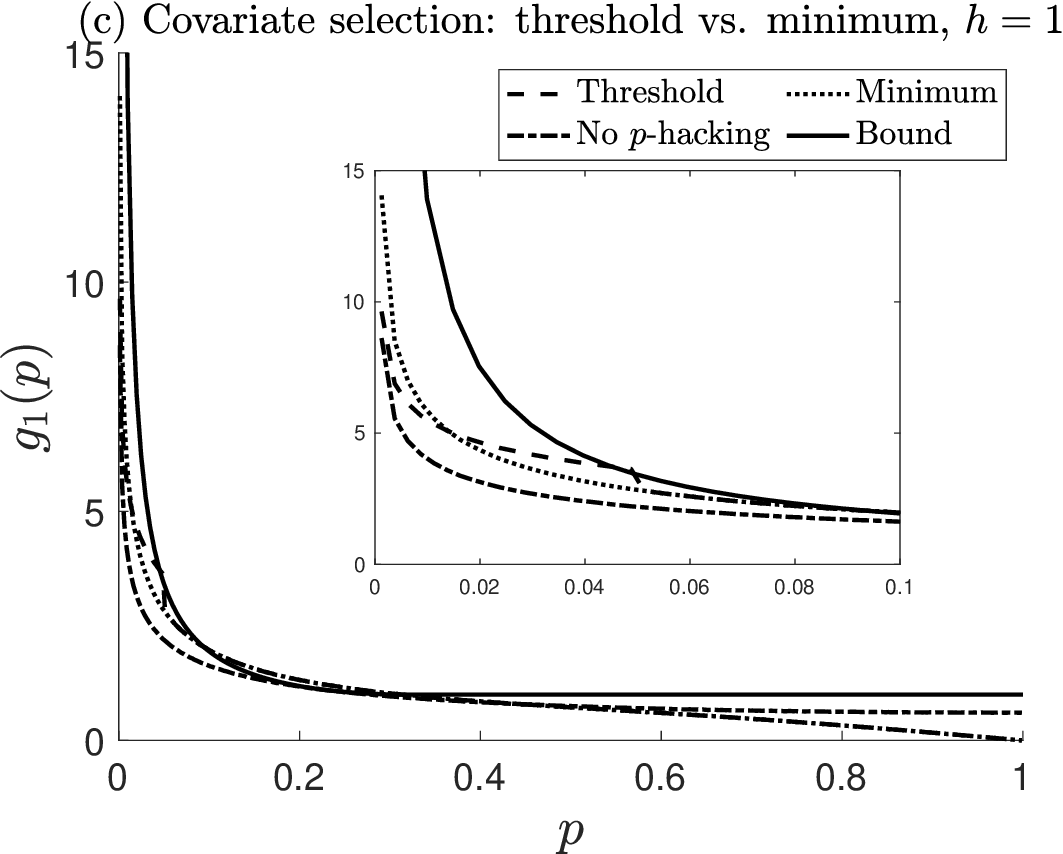}
\end{center}
\vspace{-2mm}
\footnotesize{\textit{Notes:} Figures show $p$-curves from covariate selection with two-sided tests. In Panel (a), we set $\gamma=0.5$ and vary $h$. In Panel (b), we set $h=1$ and vary $\gamma$. In Panel (c), we compare the threshold and minimum approach.}
\end{figure}

\begin{figure}[H]
\begin{center}
              \caption{$p$-Curves from IV selection with two-sided tests} 
              \label{fig:p_curves_example2_2s}
        \includegraphics[width=0.325\textwidth]{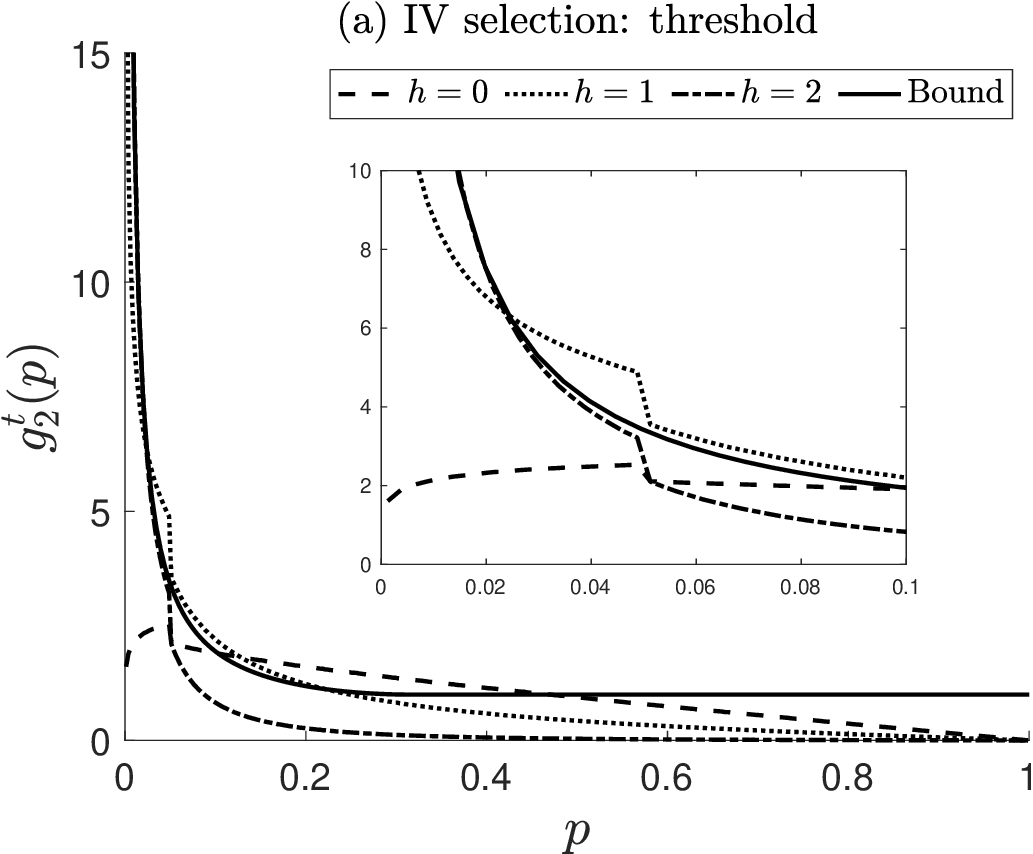}
\includegraphics[width=0.325\textwidth]{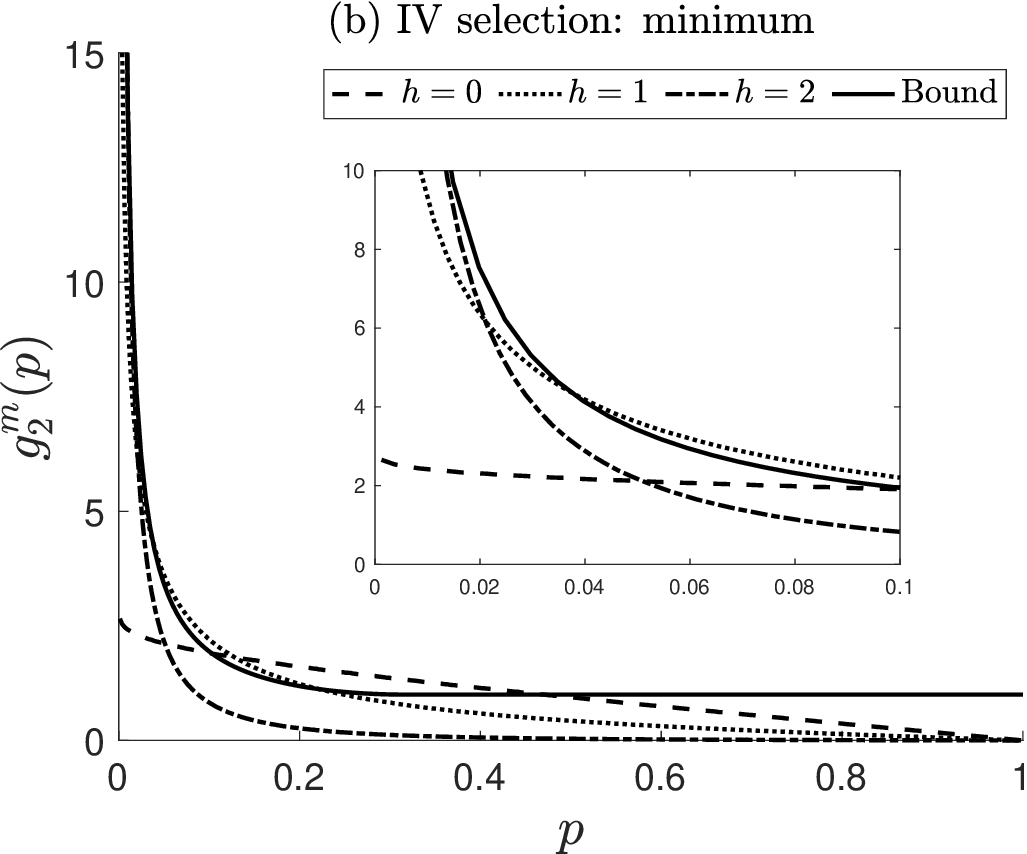}
         \includegraphics[width=0.325\textwidth]{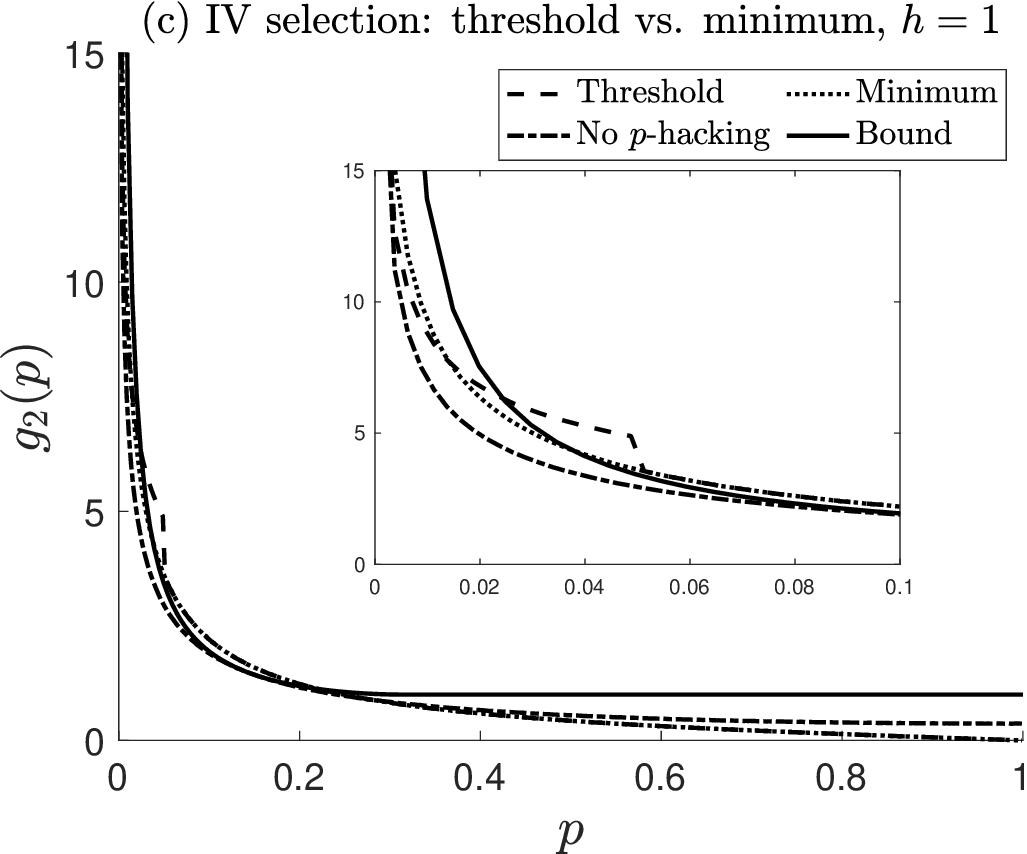}
\end{center}
\vspace{-2mm}
\footnotesize{\textit{Notes:} Figure shows the $p$-curves from $p$-hacking based on IV selection with two-sided tests. Panels (a) and (b) show the results for different values of $h$. Panel (c) compares the threshold and minimum approach for $h=1$.}

\end{figure}

\begin{figure}[H]

\begin{center}
\caption{$p$-Curves from dataset selection based on the threshold approach with two-sided tests}
        \label{fig:p_curves_example3_2s}

         \includegraphics[width=0.43\textwidth]{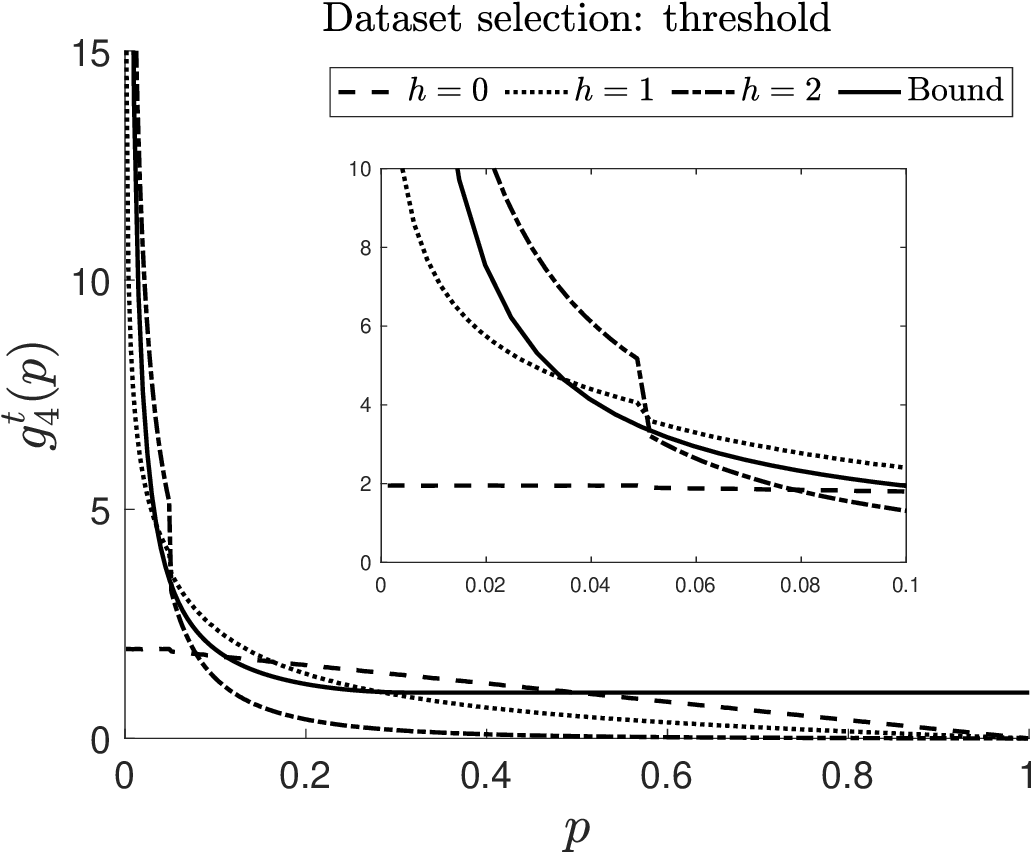}

\end{center}

\vspace{-2mm}

\footnotesize{\textit{Notes:} Figure shows the $p$-curves from dataset selection with two-sided tests based on the threshold approach with $\gamma=0.5$.}

\end{figure}

\begin{figure}[H]
 \begin{center}
\caption{$p$-Curves from dataset selection based on the minimum approach with two-sided tests} 
\label{fig:p_curves_example3bnd_2s}
 \includegraphics[width=0.43\textwidth]{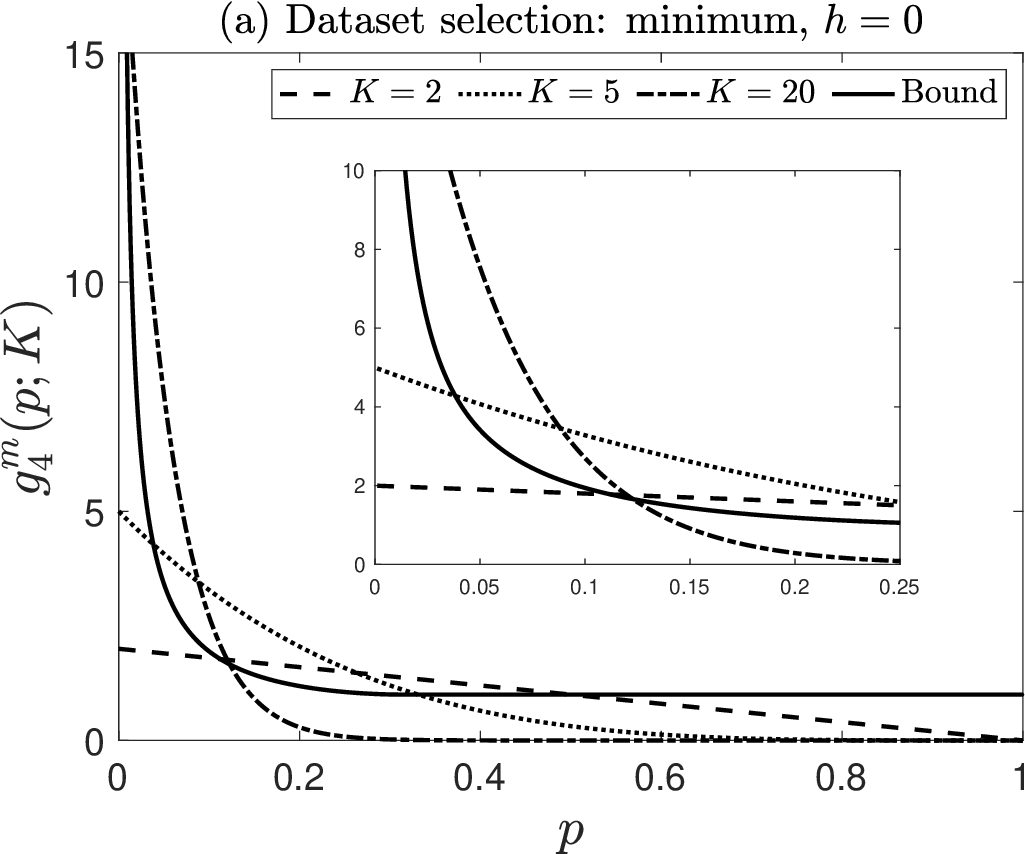}
\includegraphics[width=0.43\textwidth]{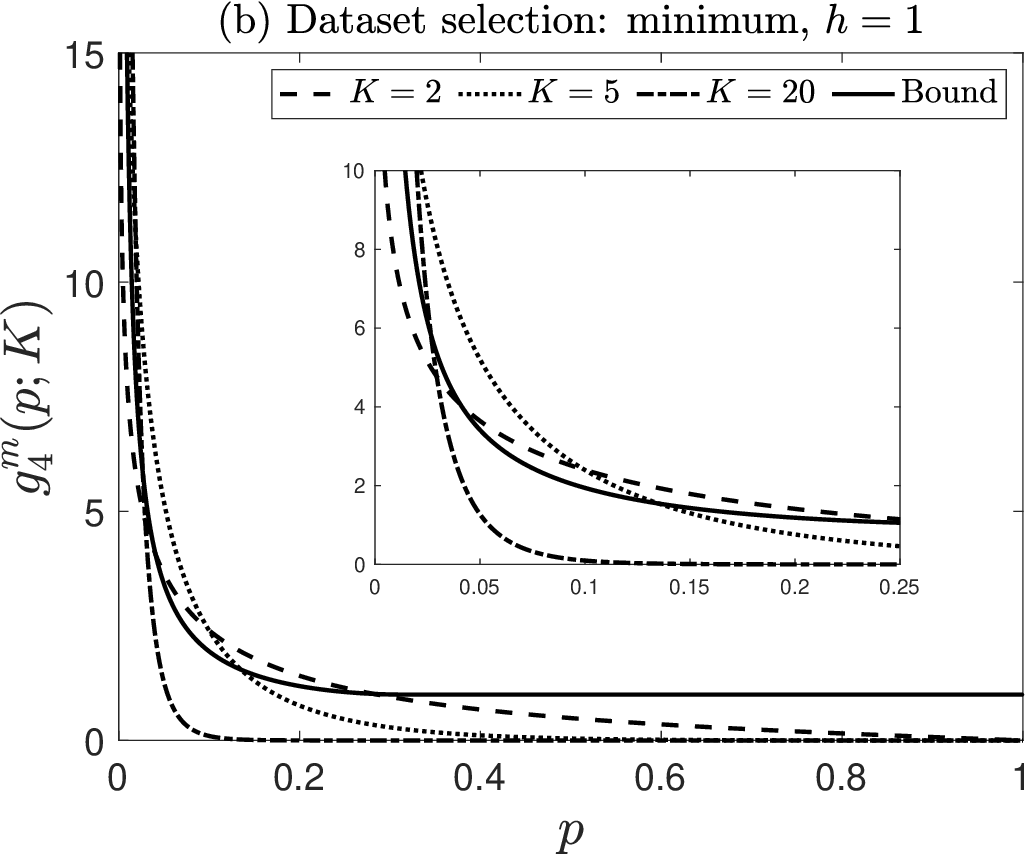}

\end{center}
\vspace{-2mm}

\footnotesize{\textit{Notes:} Figures show the $p$-curves from dataset selection with two-sided tests based on the minimum approach. Panel (a): $h = 0$. Panel (b): $h = 1$.}
             
\end{figure}

\begin{figure}[H]
    \begin{center}

\caption{$p$-Curves from lag length selection with two-sided tests} 
\label{fig:pcurves_lag_2s}
     
 \includegraphics[width=0.43\textwidth]{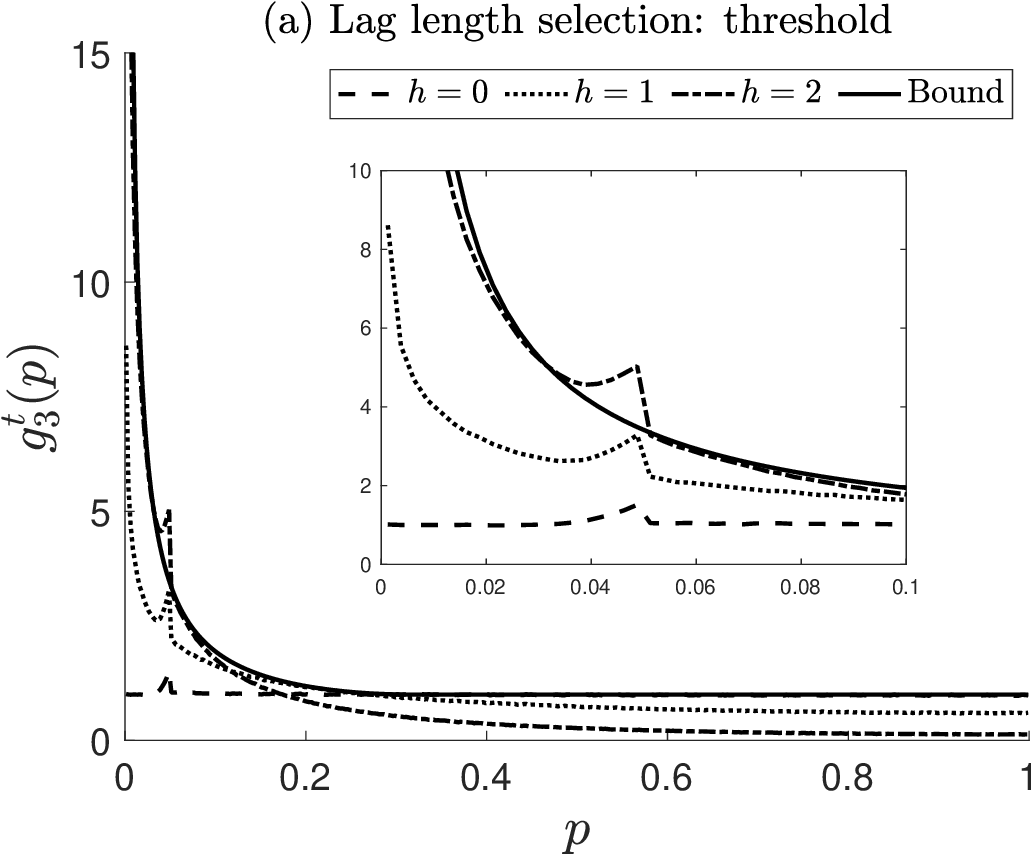}
\includegraphics[width=0.43\textwidth]{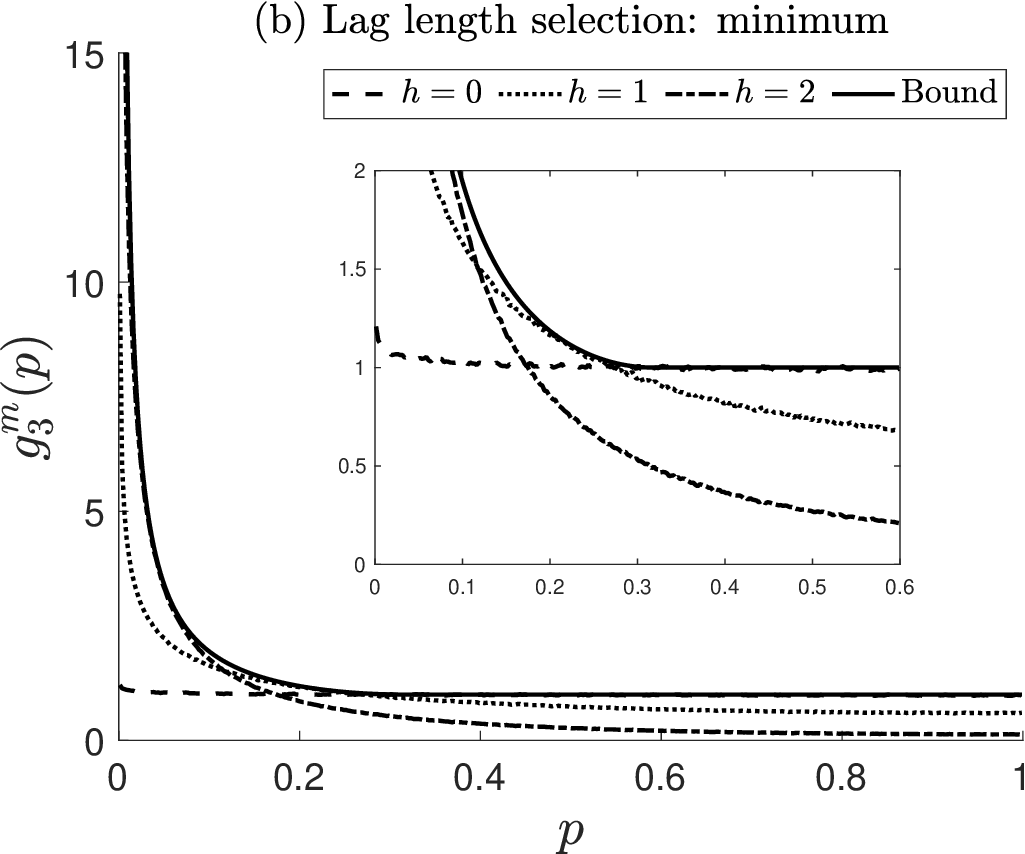}
\end{center}
\vspace{-2mm}

\footnotesize{\textit{Notes:} Figures show the $p$-curves from lag length selection with two-sided tests, $N=200$, and $\kappa=1/2$. Panel (a) shows the $p$-curves based on the threshold approach. Panel (b) shows the $p$-curves based on the minimum approach.}

\end{figure}


\section{Null and Alternative Distributions MC Study}
\label{app: histograms}
\enlargethispage{1\baselineskip}

\begin{figure}[H]

\caption{Covariate selection: null and p-hacked distributions}
				\label{fig:hists_covar_selection}

\begin{center}

\begin{center}
\vspace{-12mm}
$$K=3$$

\vspace{-5mm}
{\small{Two-sided, general-to-specific}}

\includegraphics[width=0.24\textwidth]{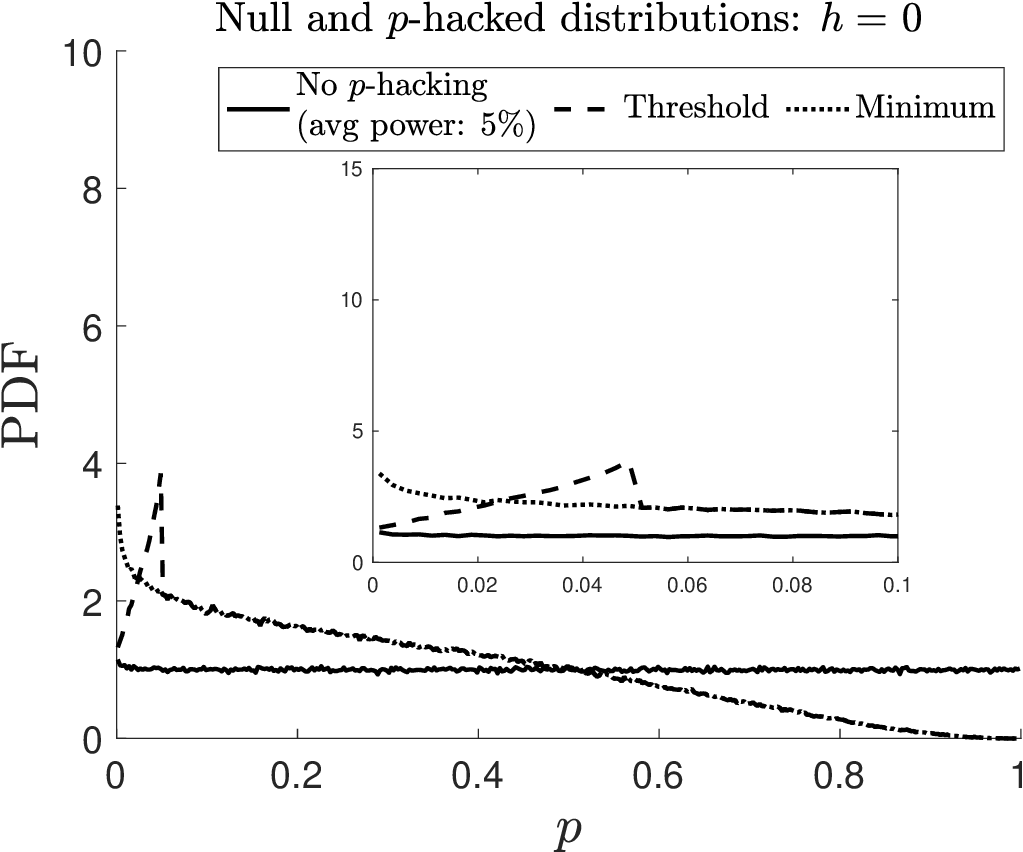}
\includegraphics[width=0.24\textwidth]{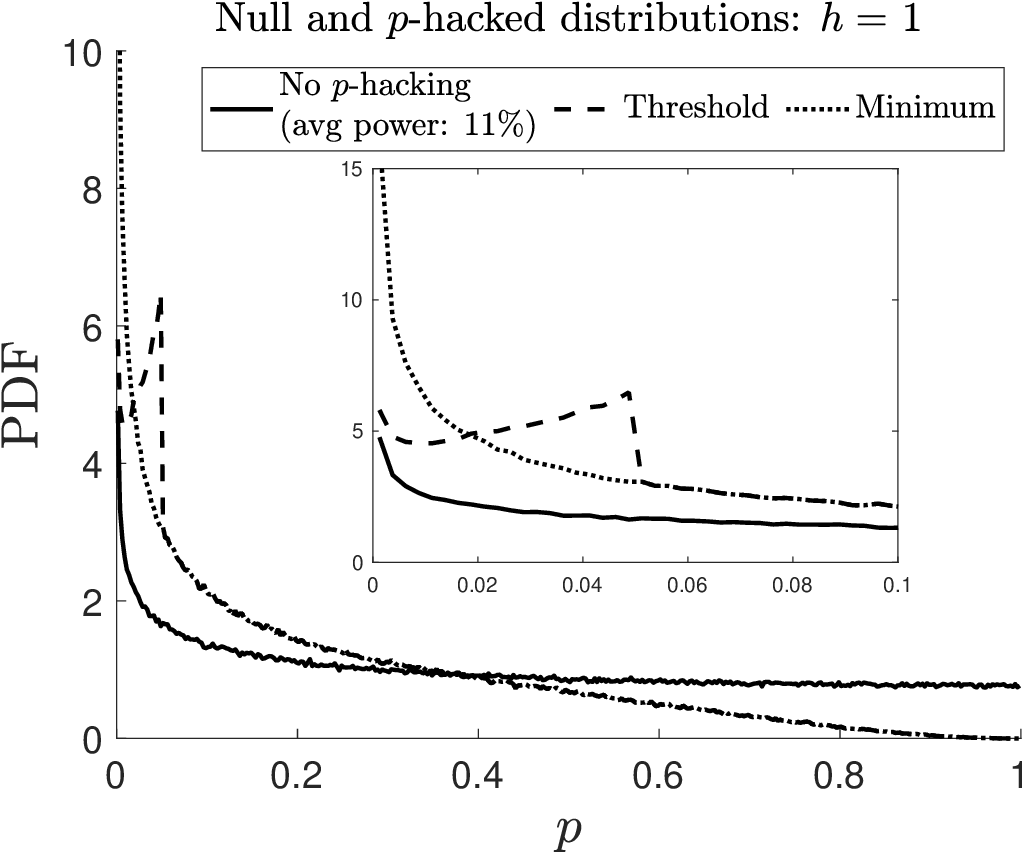}
\includegraphics[width=0.24\textwidth]{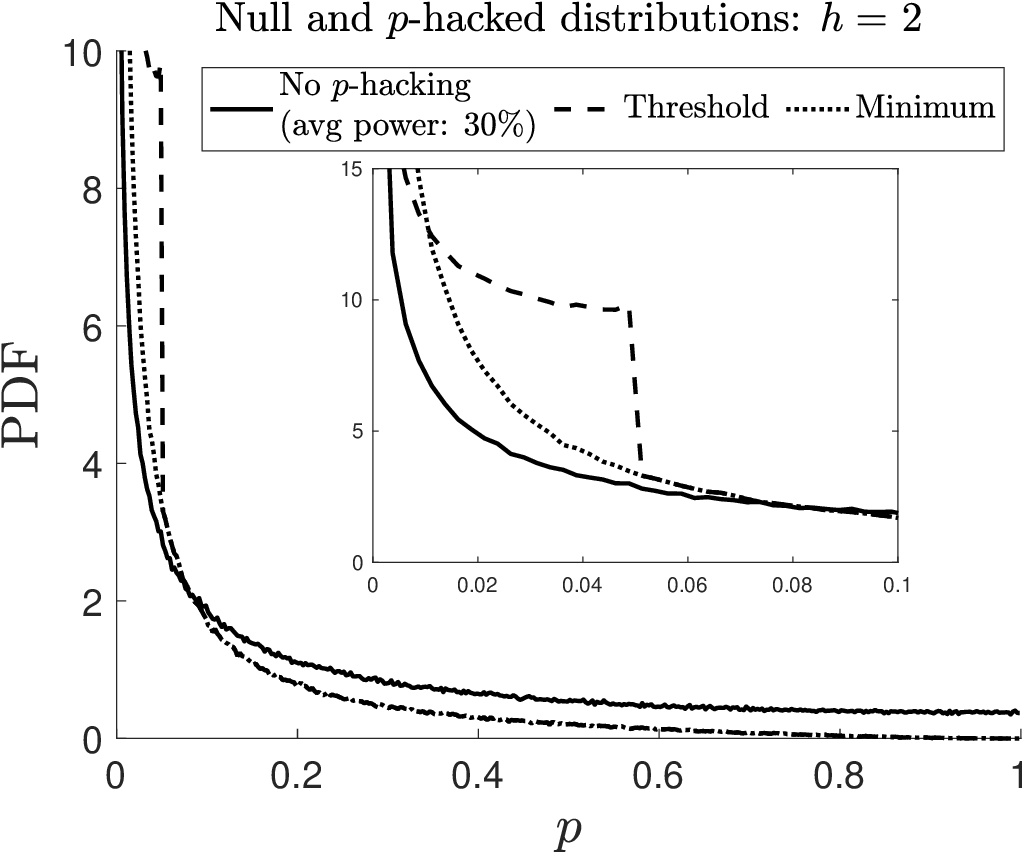}
\includegraphics[width=0.24\textwidth]{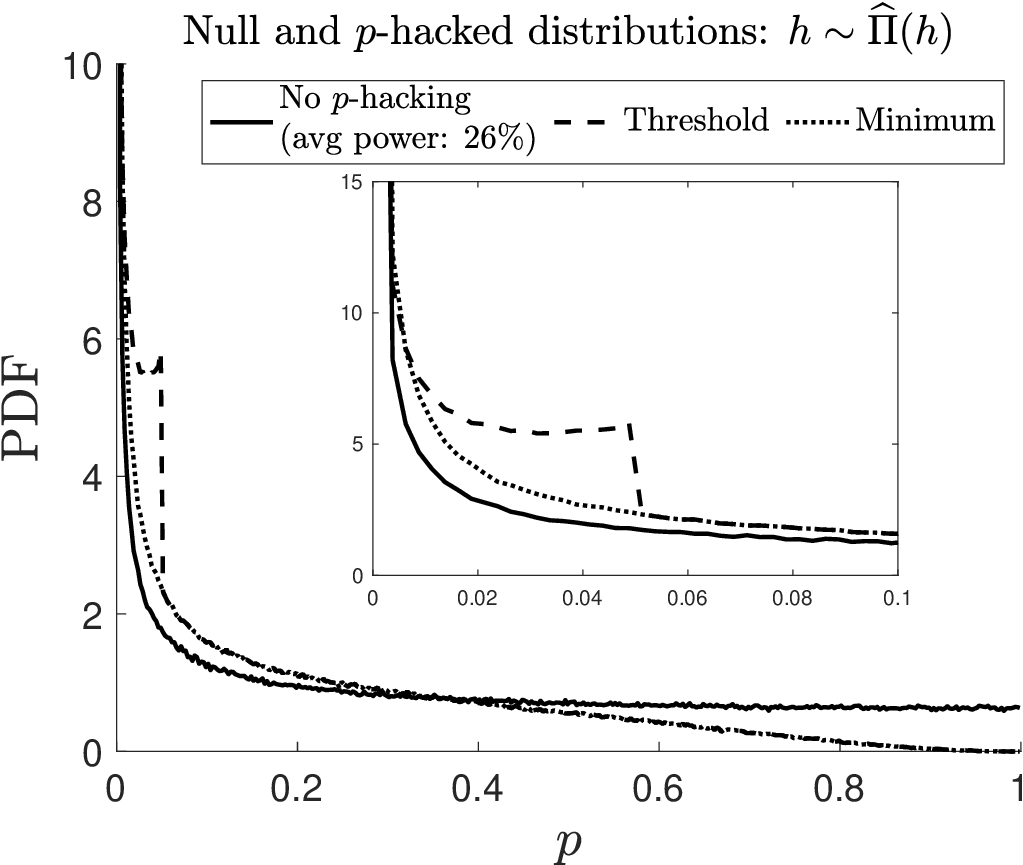}

\vspace{-2.5mm}
{\small{Two-sided, specific-to-general}}

\includegraphics[width=0.24\textwidth]{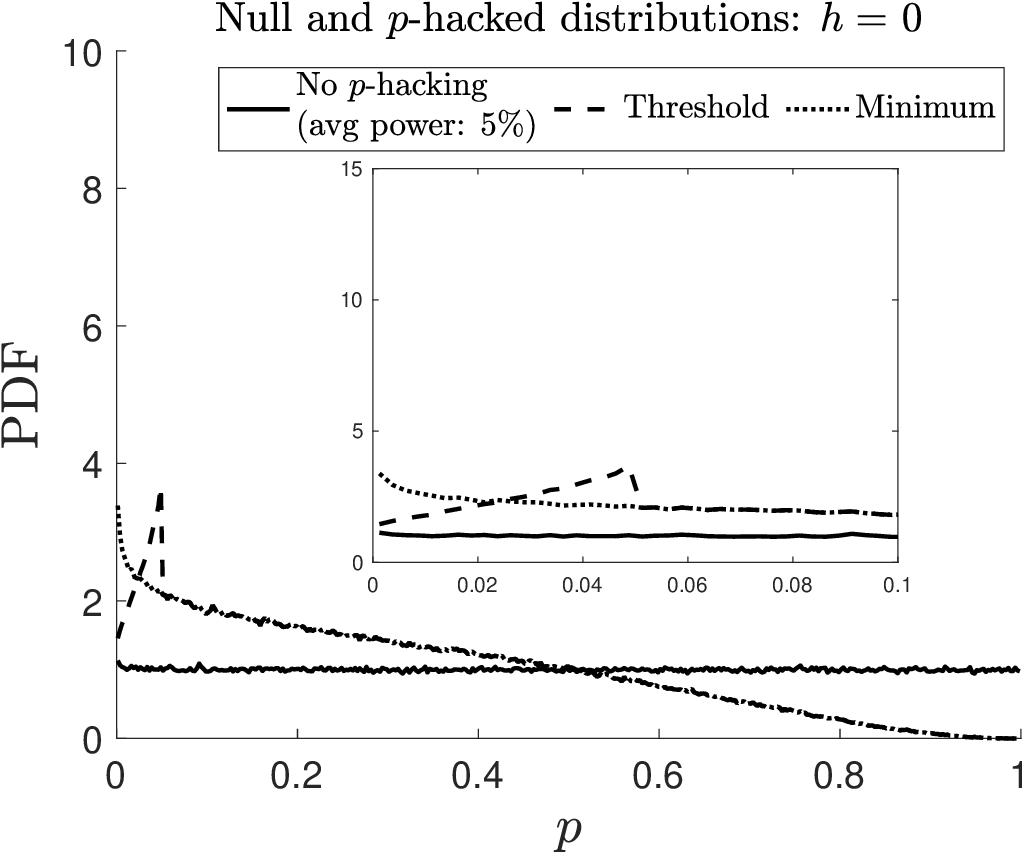}
\includegraphics[width=0.24\textwidth]{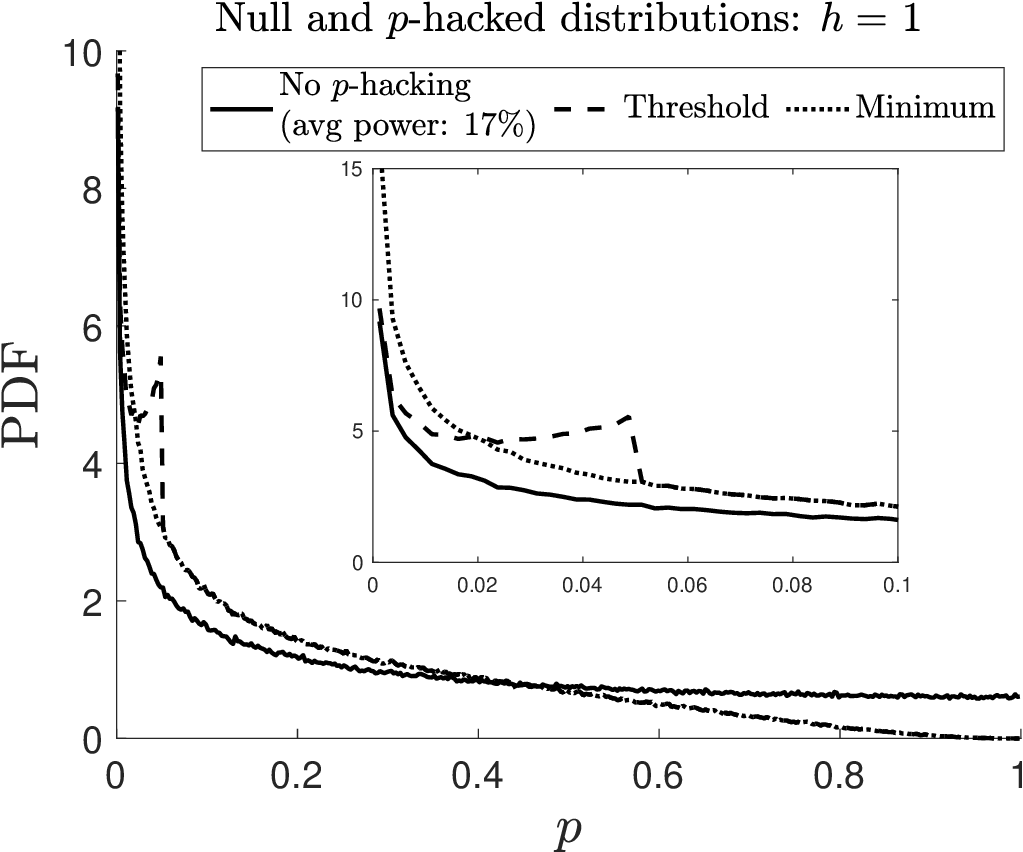}
\includegraphics[width=0.24\textwidth]{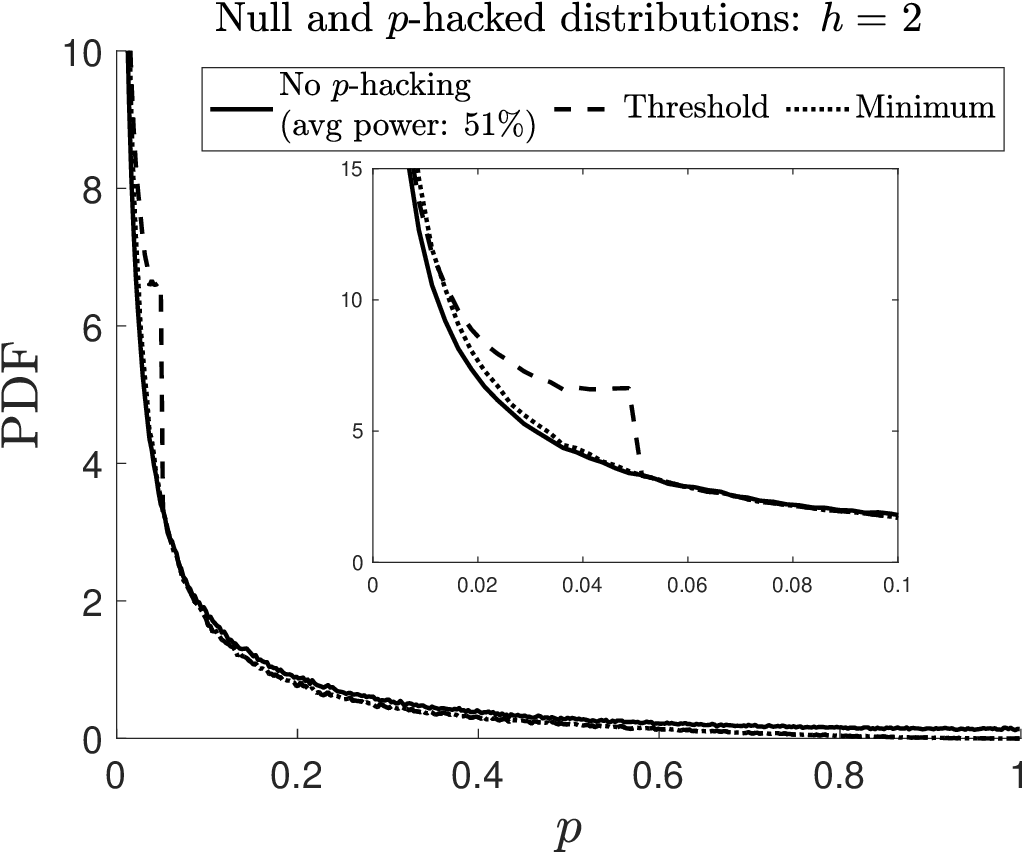}
\includegraphics[width=0.24\textwidth]{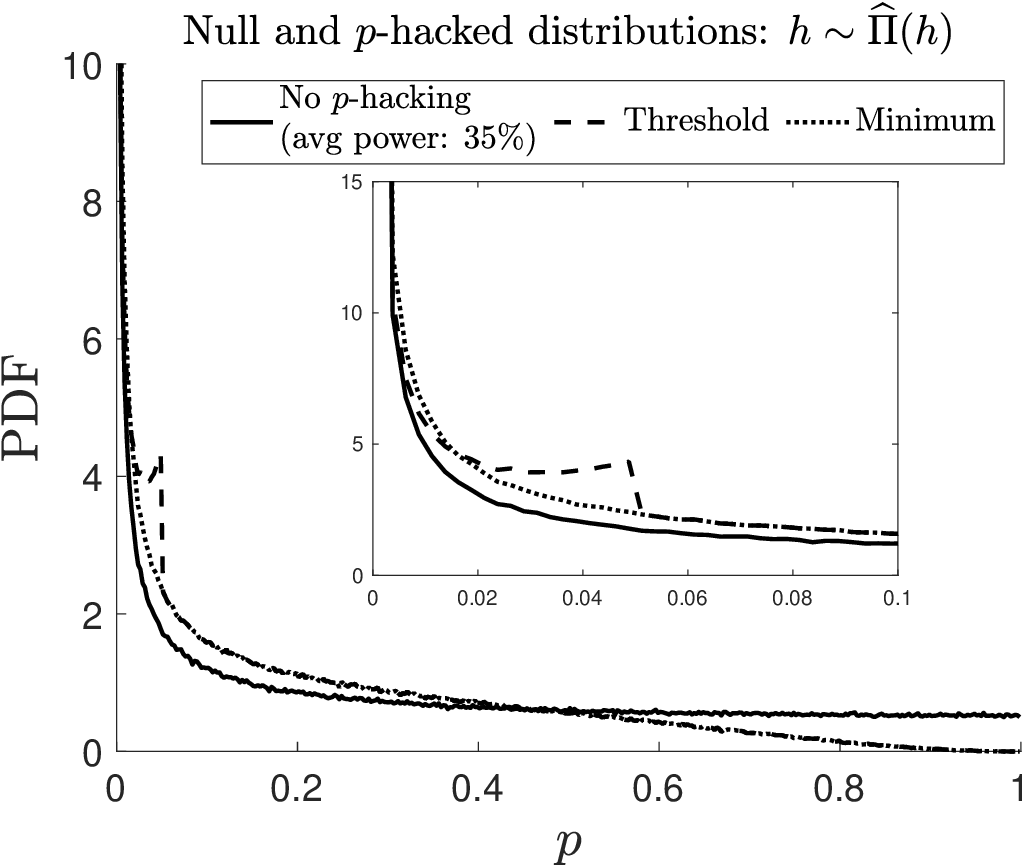}

\vspace{-2.5mm}
{\small{One-sided, general-to-specific}}

\includegraphics[width=0.24\textwidth]{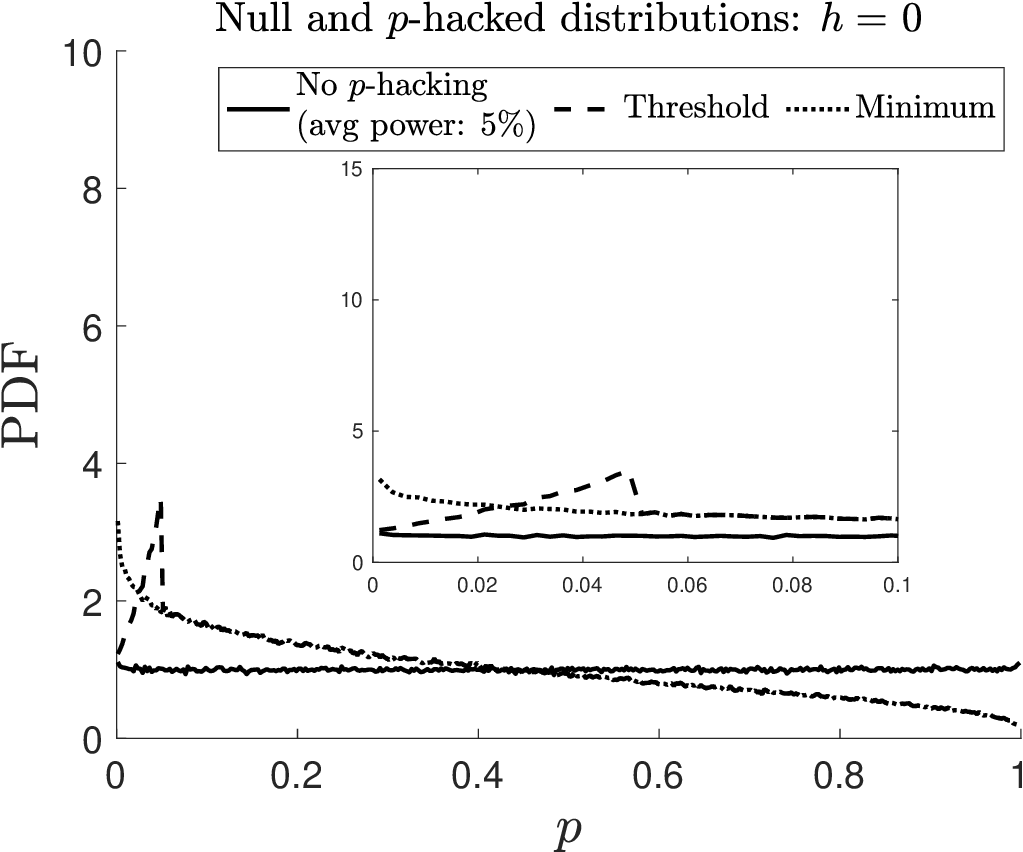}
\includegraphics[width=0.24\textwidth]{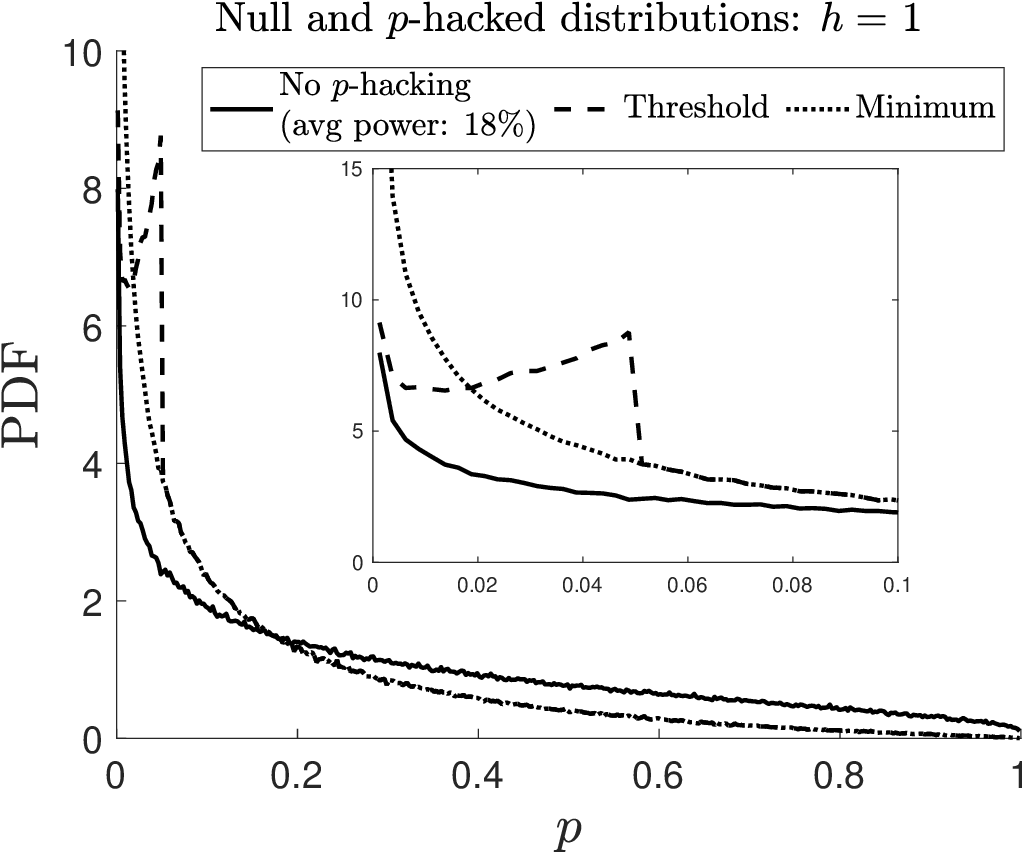}
\includegraphics[width=0.24\textwidth]{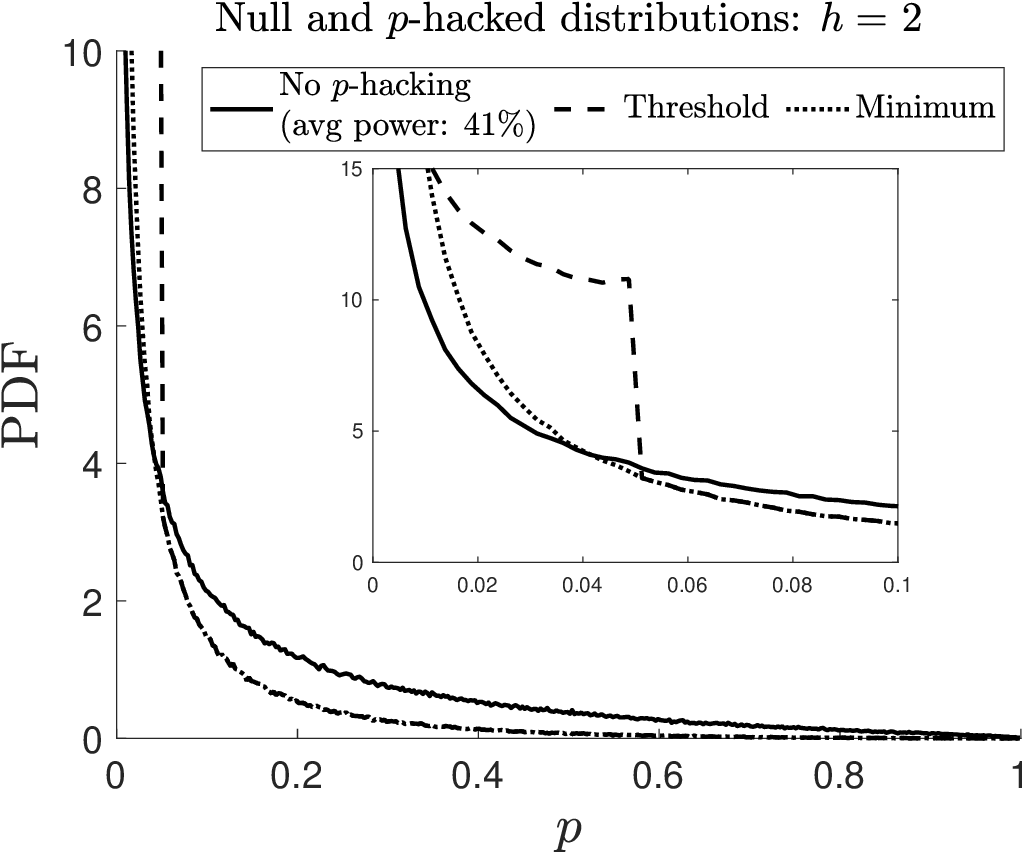}
\includegraphics[width=0.24\textwidth]{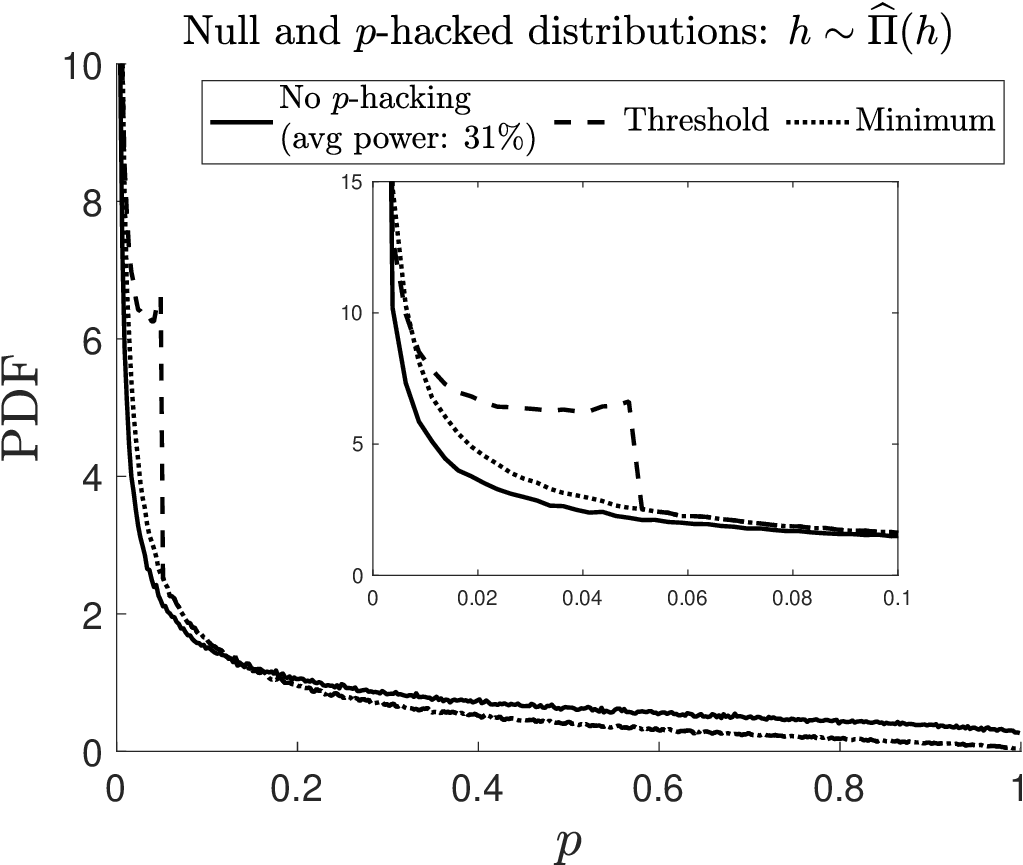}
\vspace{-7.5mm}
$$K=5$$

\vspace{-5mm}
{\small{Two-sided, general-to-specific}}

\includegraphics[width=0.24\textwidth]{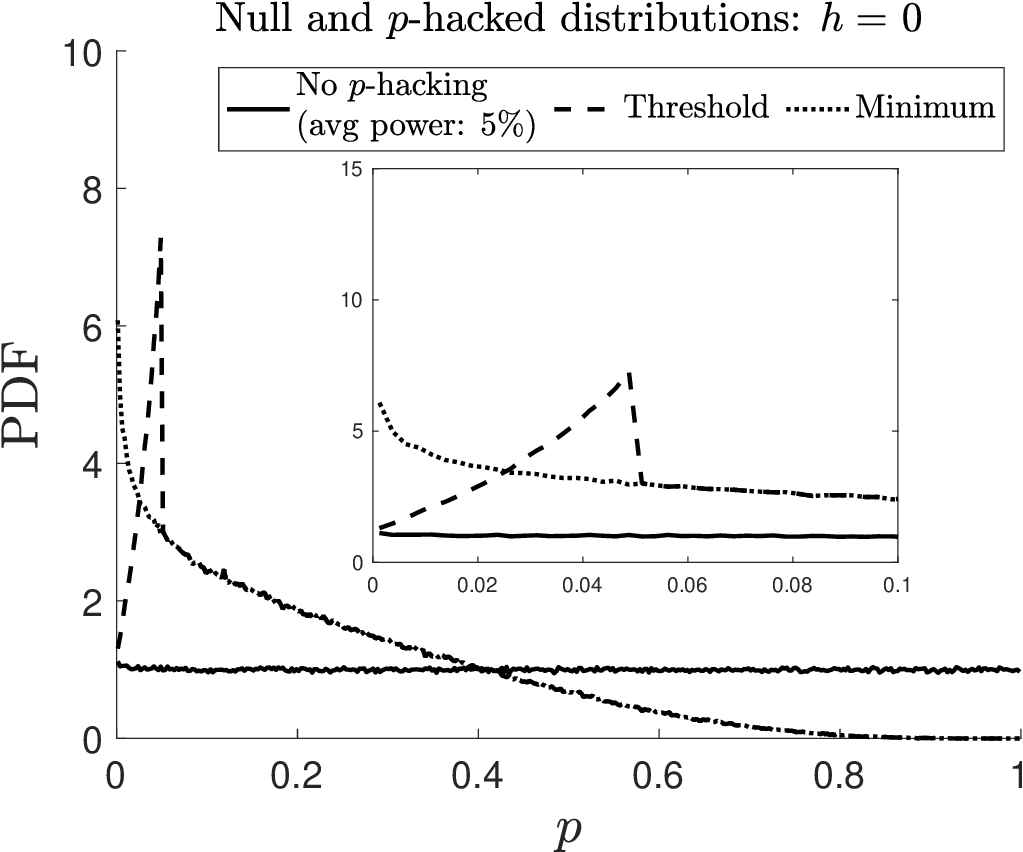}
\includegraphics[width=0.24\textwidth]{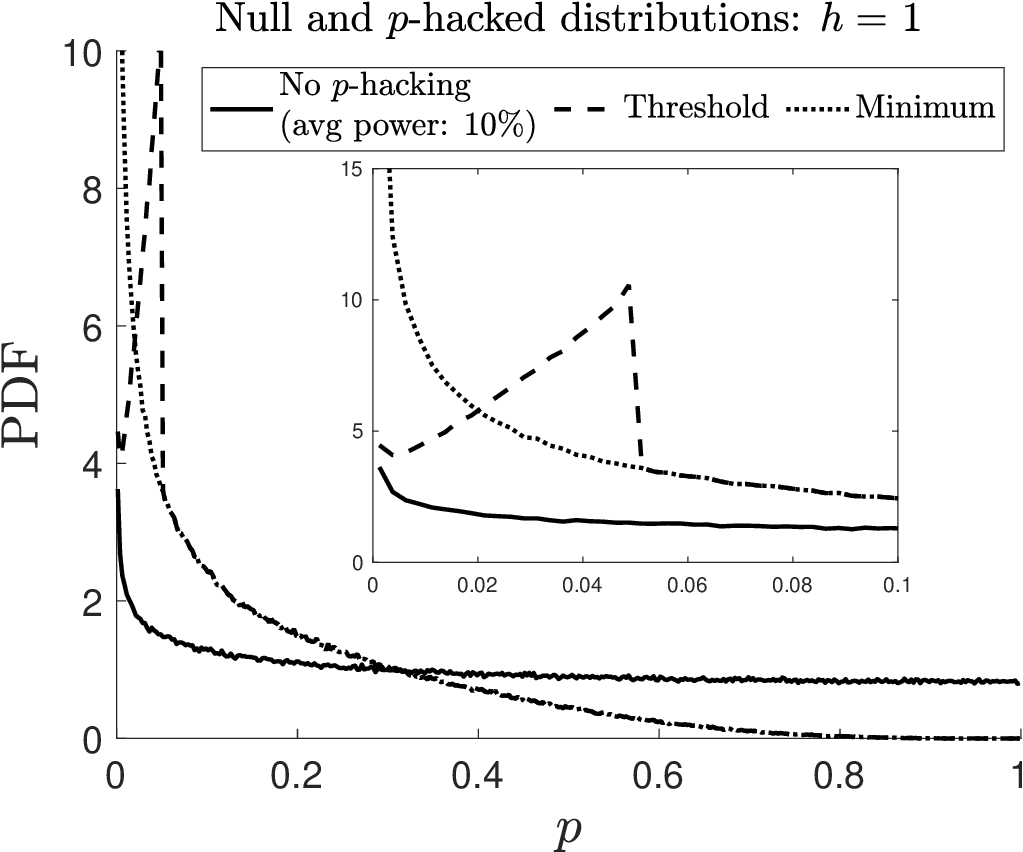}
\includegraphics[width=0.24\textwidth]{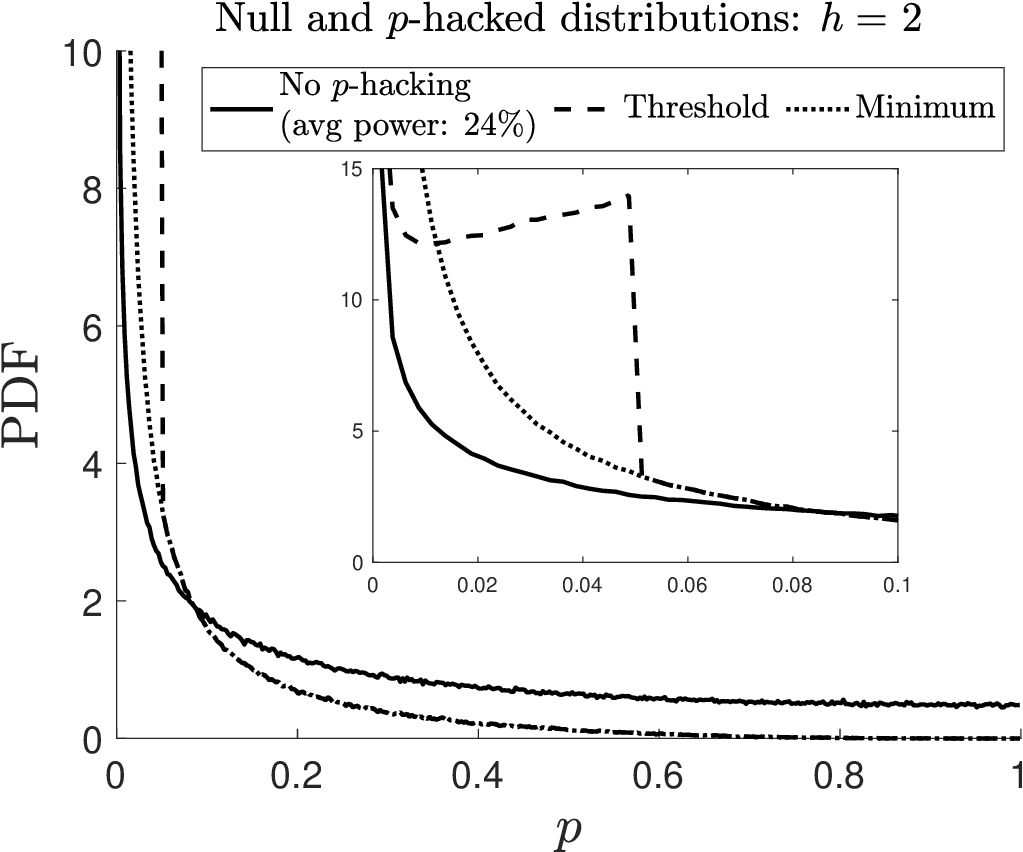}
\includegraphics[width=0.24\textwidth]{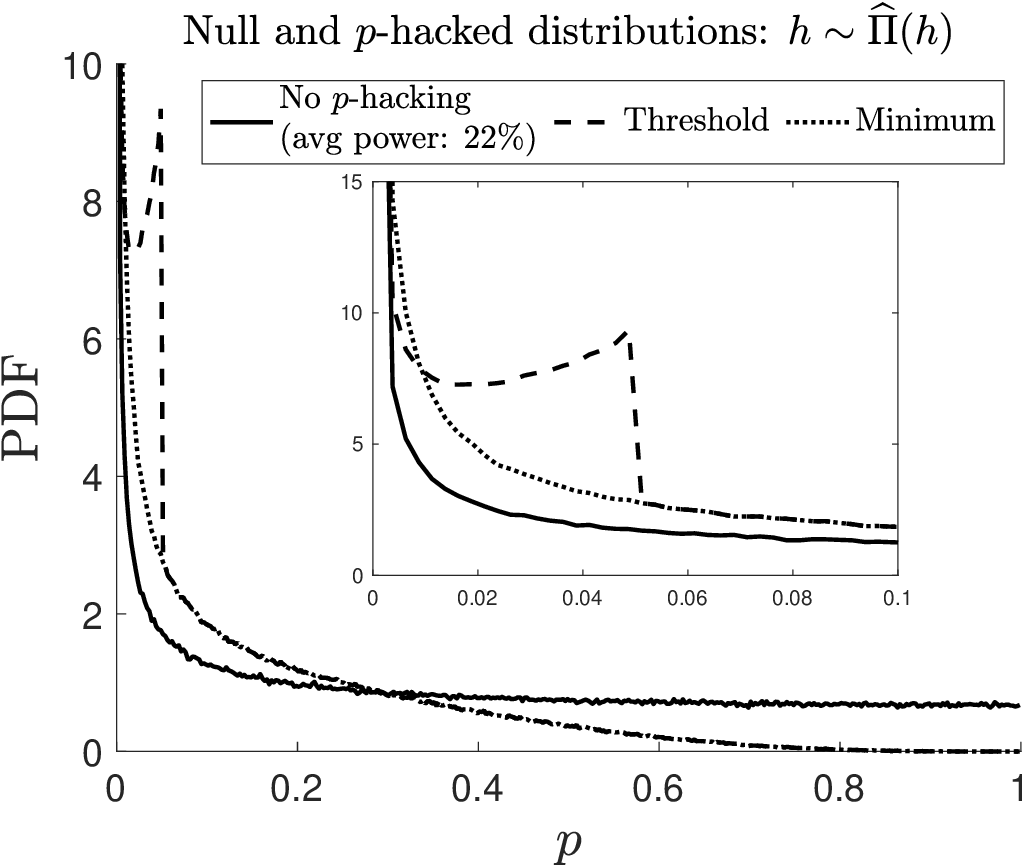}
\vspace{-7.5mm}
$$K=7$$

\vspace{-5mm}
{\small{Two-sided, general-to-specific}}

\includegraphics[width=0.24\textwidth]{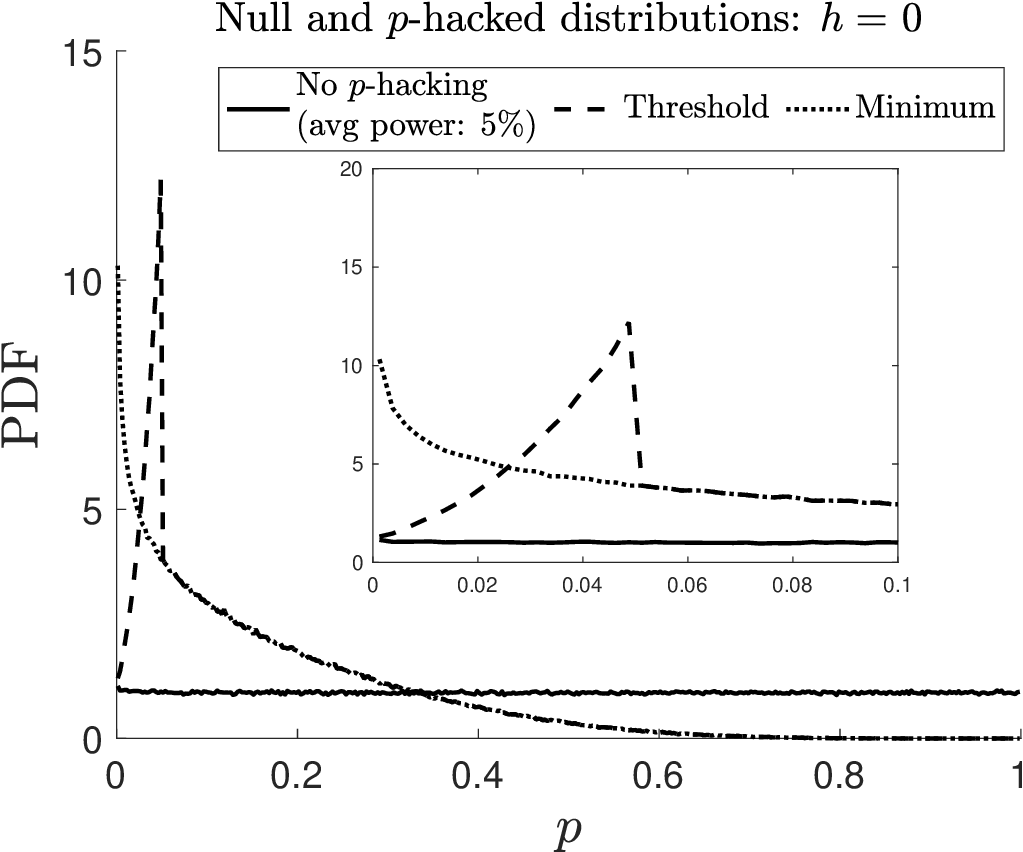}
\includegraphics[width=0.24\textwidth]{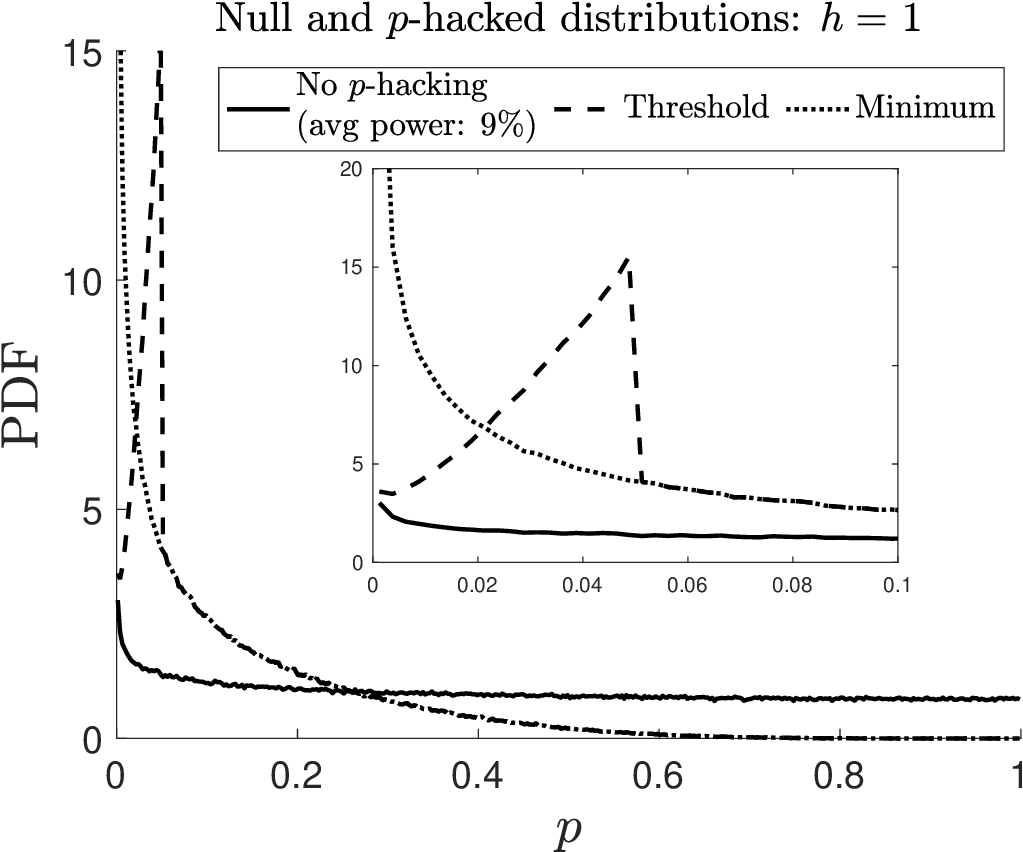}
\includegraphics[width=0.24\textwidth]{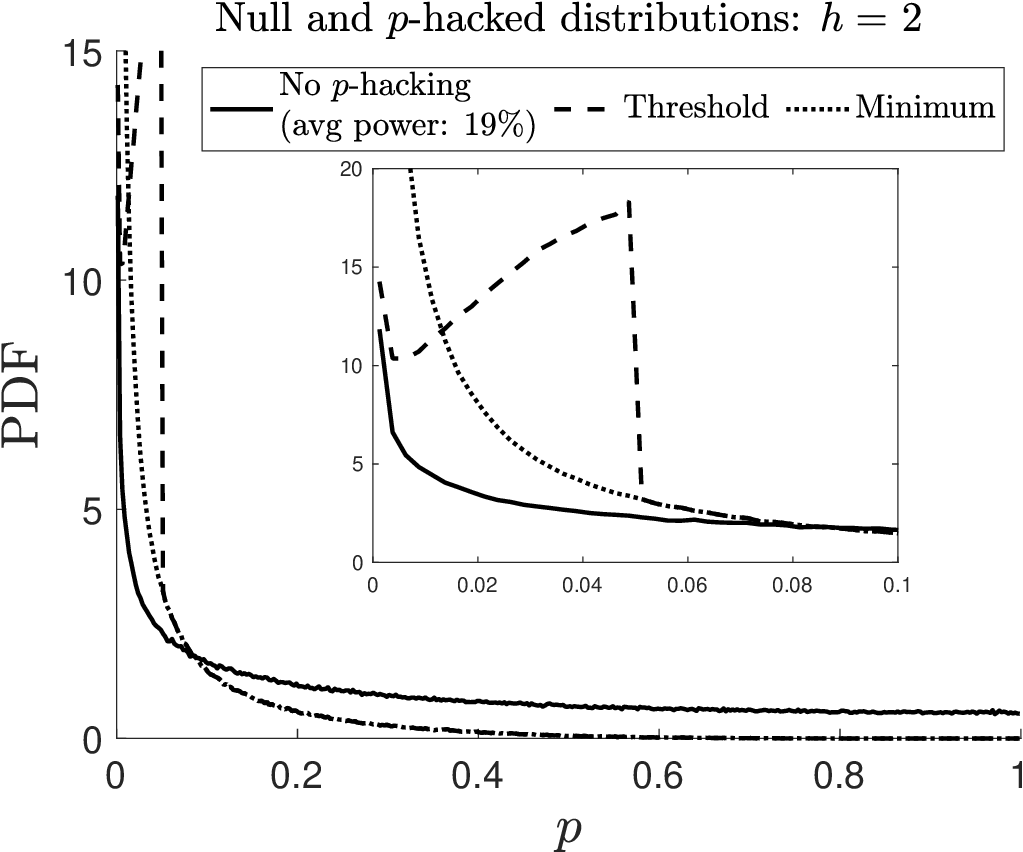}
\includegraphics[width=0.24\textwidth]{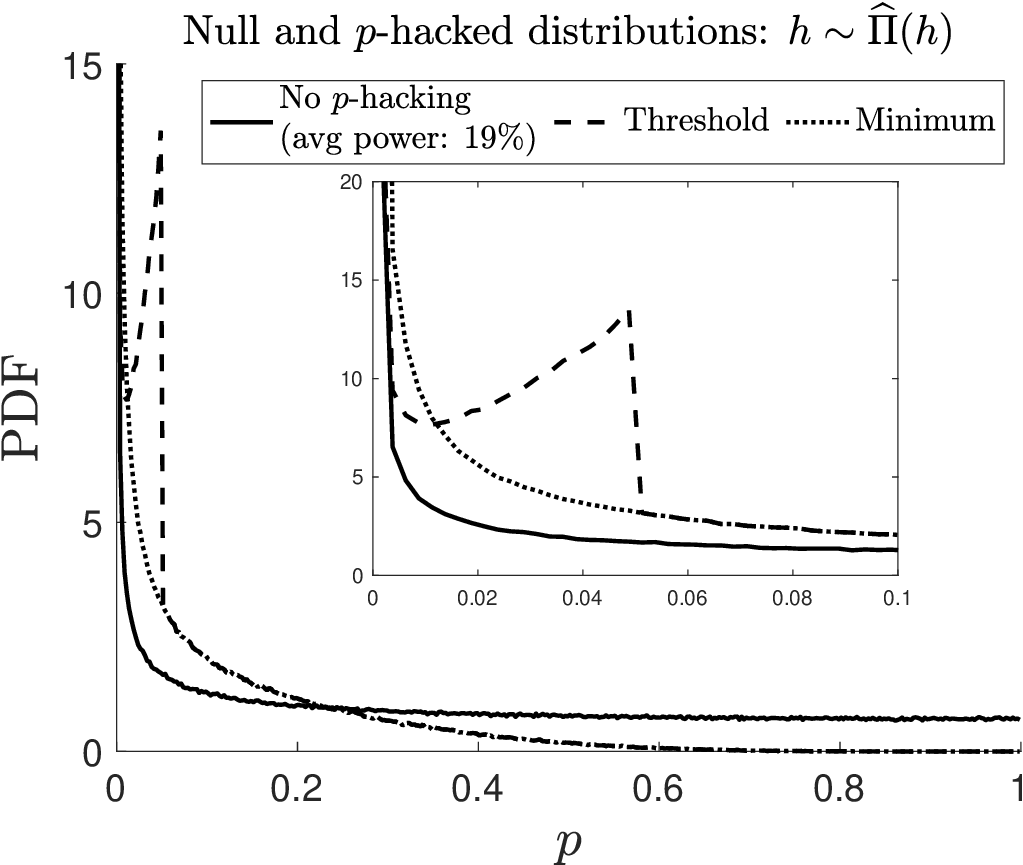}
\end{center}

\end{center}

\vspace{-2mm}

\footnotesize{\textit{Notes}: Figures show the null and alternative ($p$-hacked) distributions for covariate selection. They also report the average power of the underlying studies, which is equal to the probability mass below 0.05 under no $p$-hacking.}
                
\end{figure}

\begin{figure}[H]

\caption{IV selection with $K=3$: null and p-hacked distributions}
				\label{fig:hists_IV3}

\begin{center}
\includegraphics[width=0.24\textwidth]{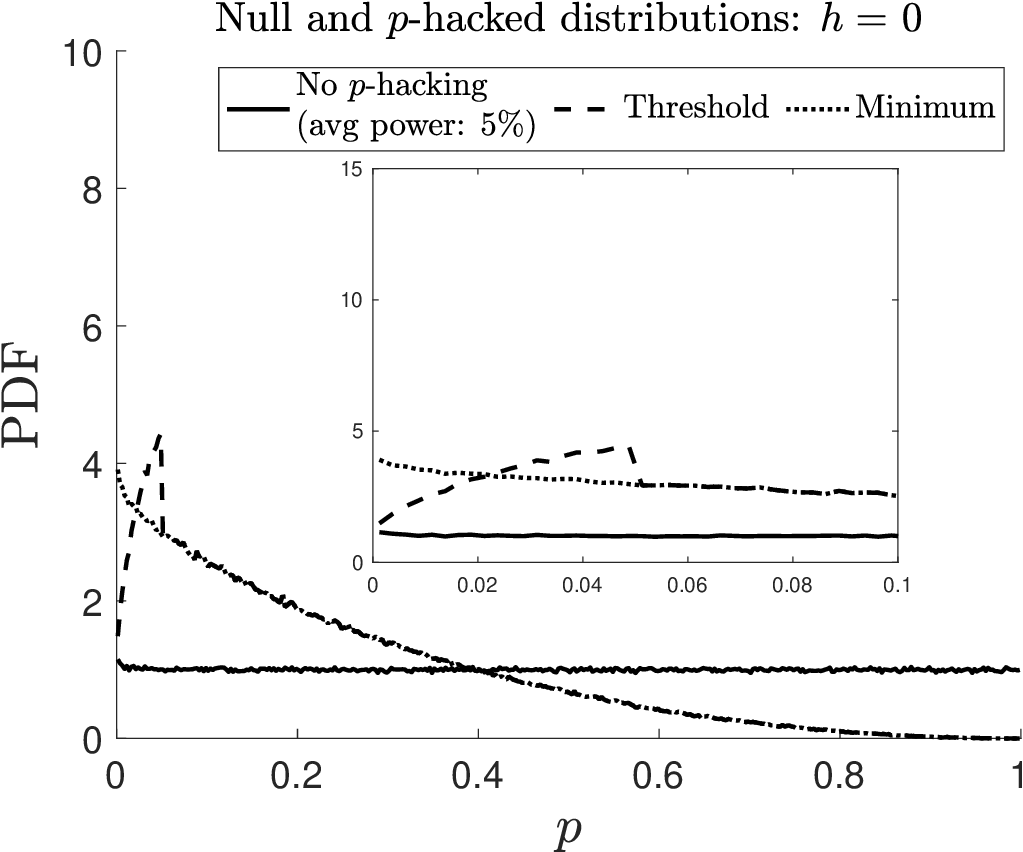}
\includegraphics[width=0.24\textwidth]{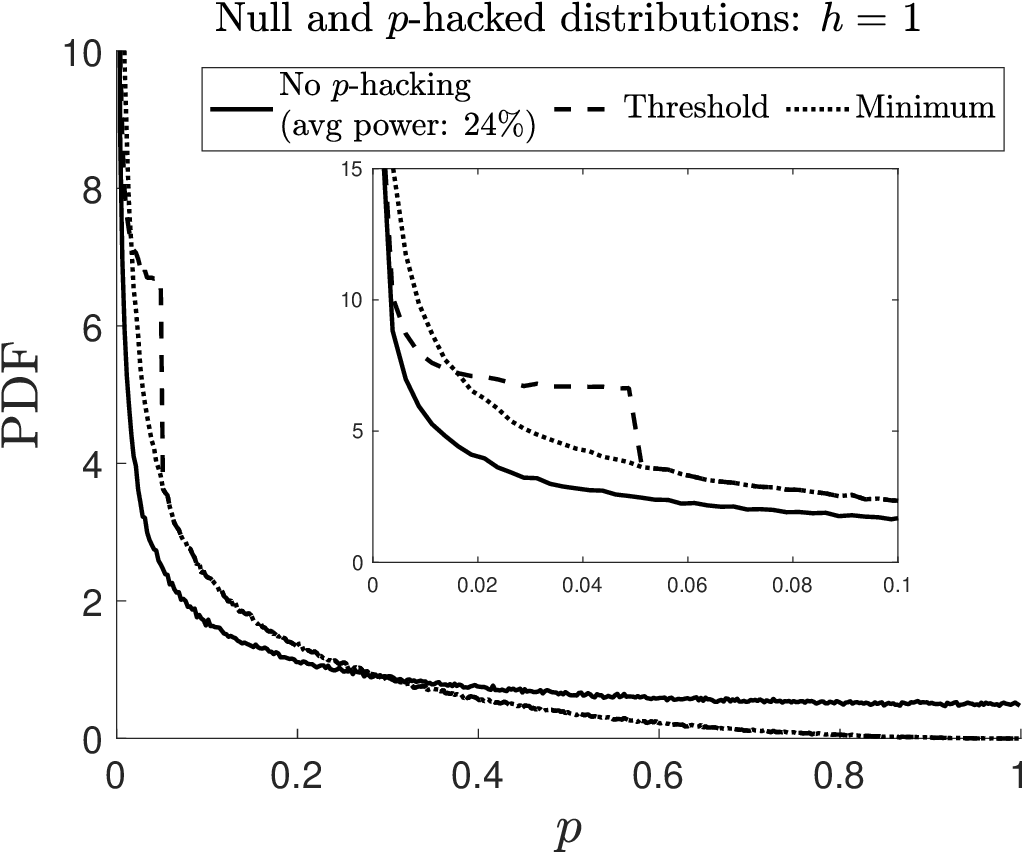}
\includegraphics[width=0.24\textwidth]{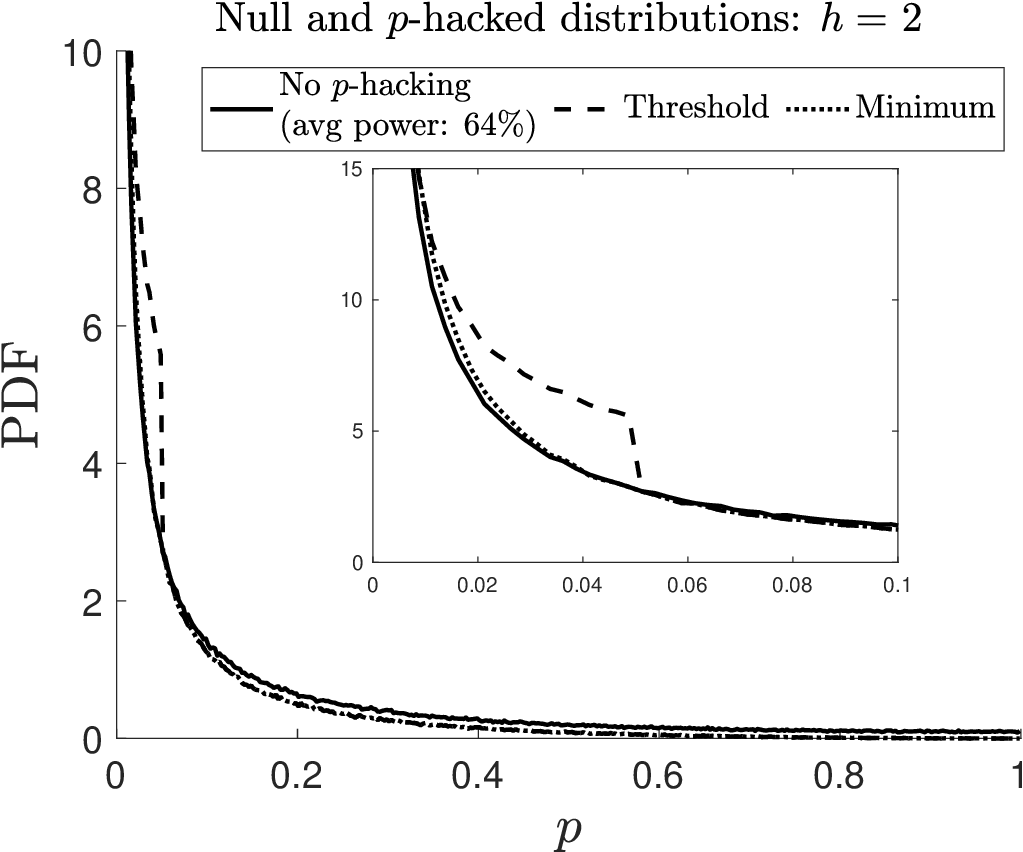}
\includegraphics[width=0.24\textwidth]{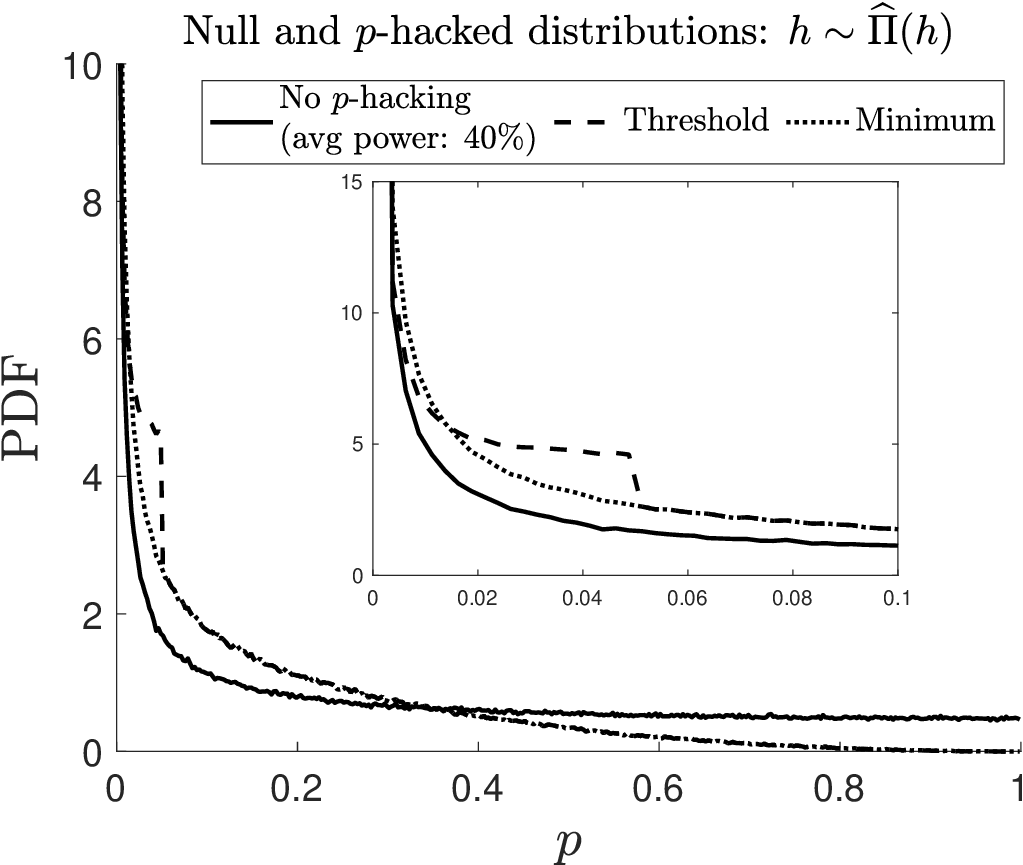}

\includegraphics[width=0.24\textwidth]{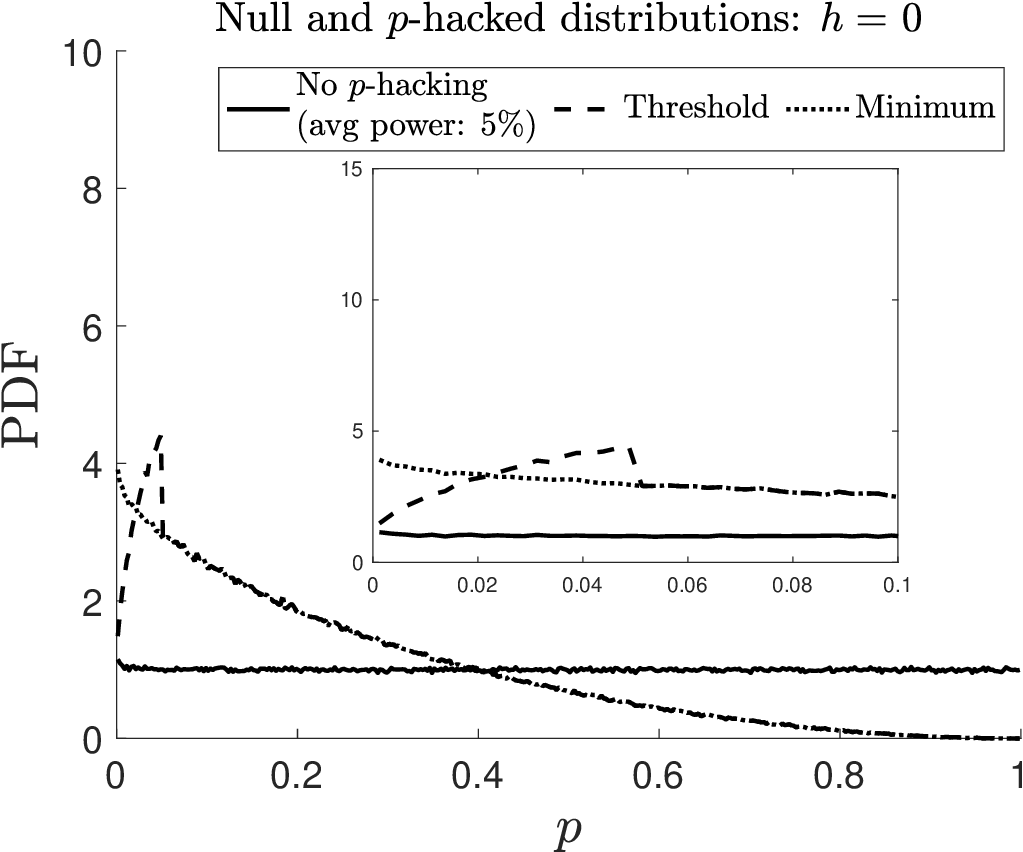}
\includegraphics[width=0.24\textwidth]{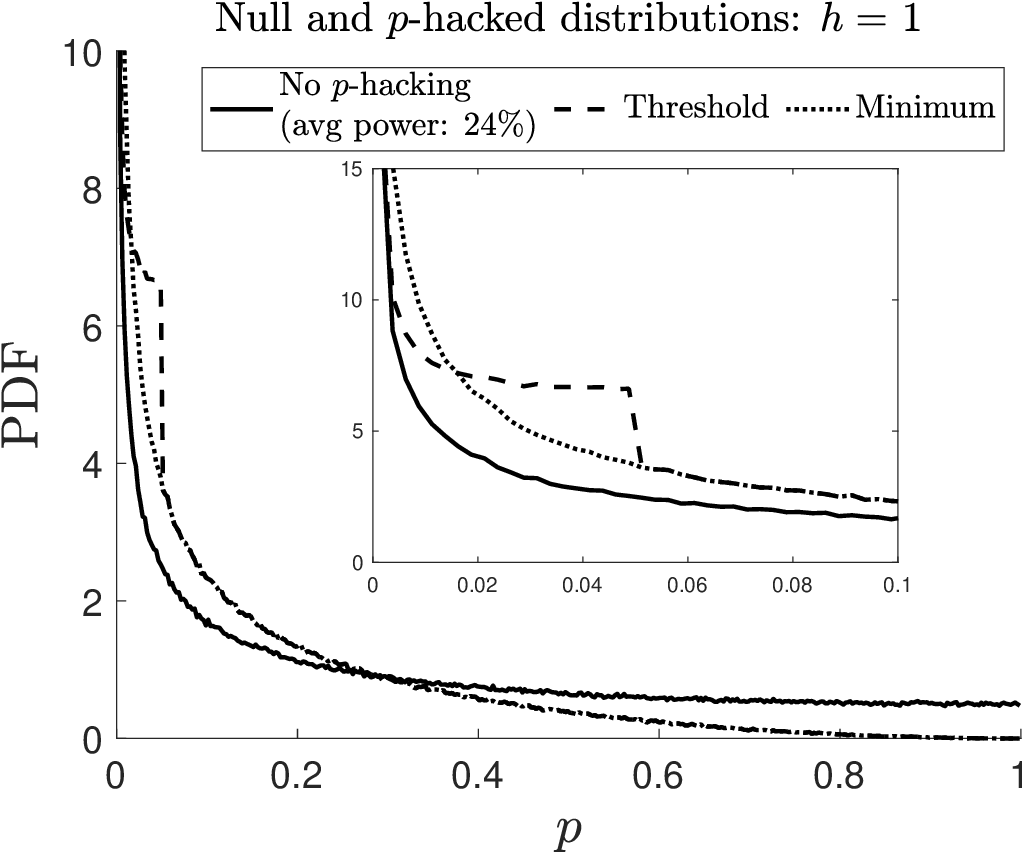}
\includegraphics[width=0.24\textwidth]{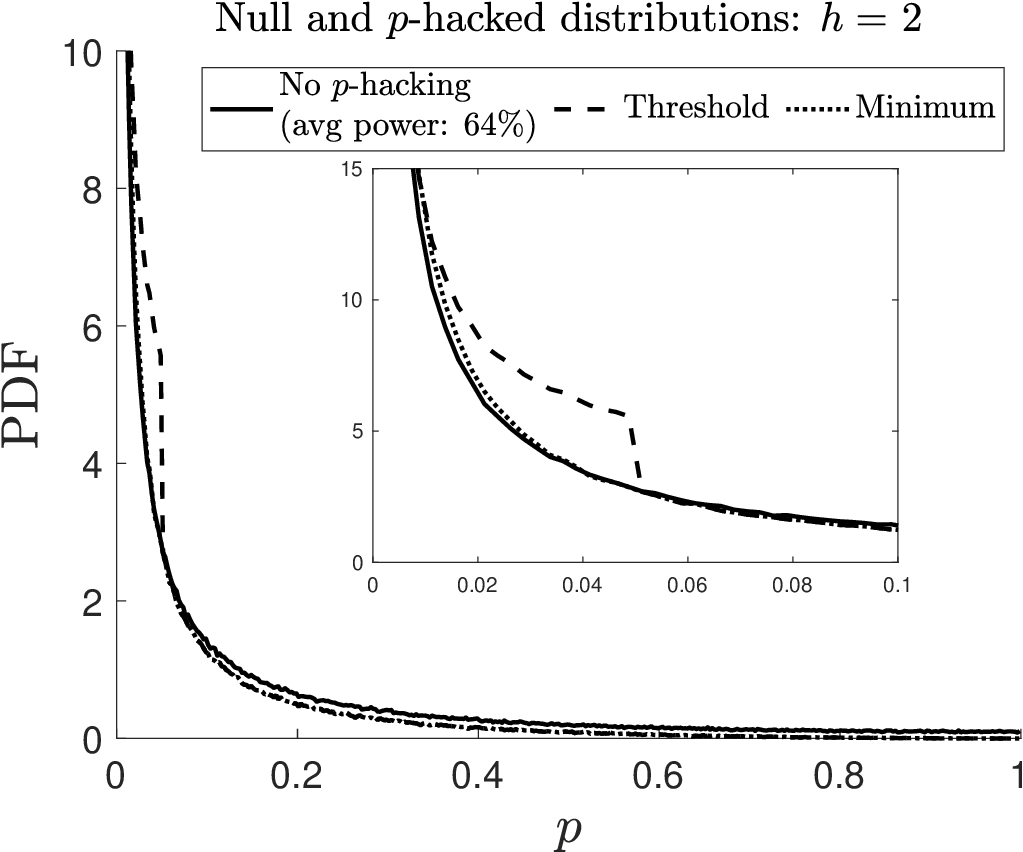}
\includegraphics[width=0.24\textwidth]{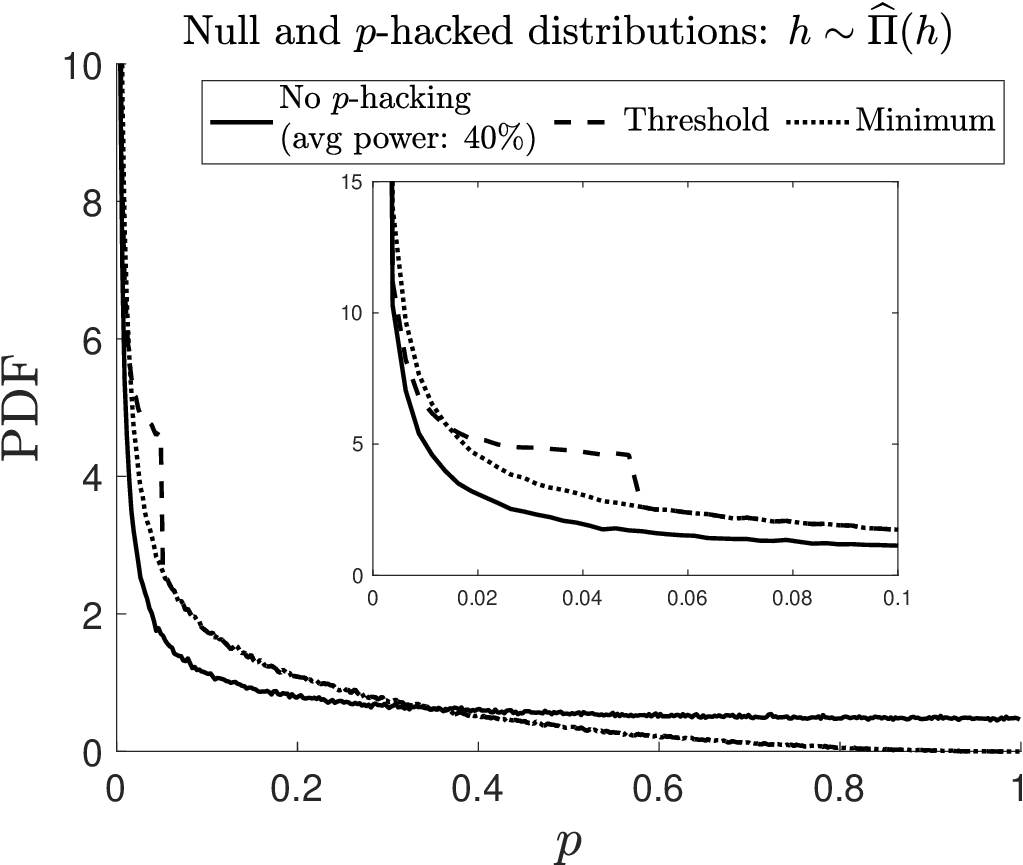}
\end{center}

\vspace{-2mm}

\footnotesize{\textit{Notes:} Figures show the null and alternative ($p$-hacked) distributions for IV selection with $K=3$. They also report the average power of the underlying studies, which is equal to the probability mass below 0.05 under no $p$-hacking. Row 1: two-sided tests, using all specifications; row 2: two-sided tests, using only specifications with $F>10$.}

\end{figure}

\begin{figure}[H]
\caption{IV selection with $K=5$: null and $p$-hacked distributions}
				\label{fig:hists_IV5}
\begin{center}
\includegraphics[width=0.24\textwidth]{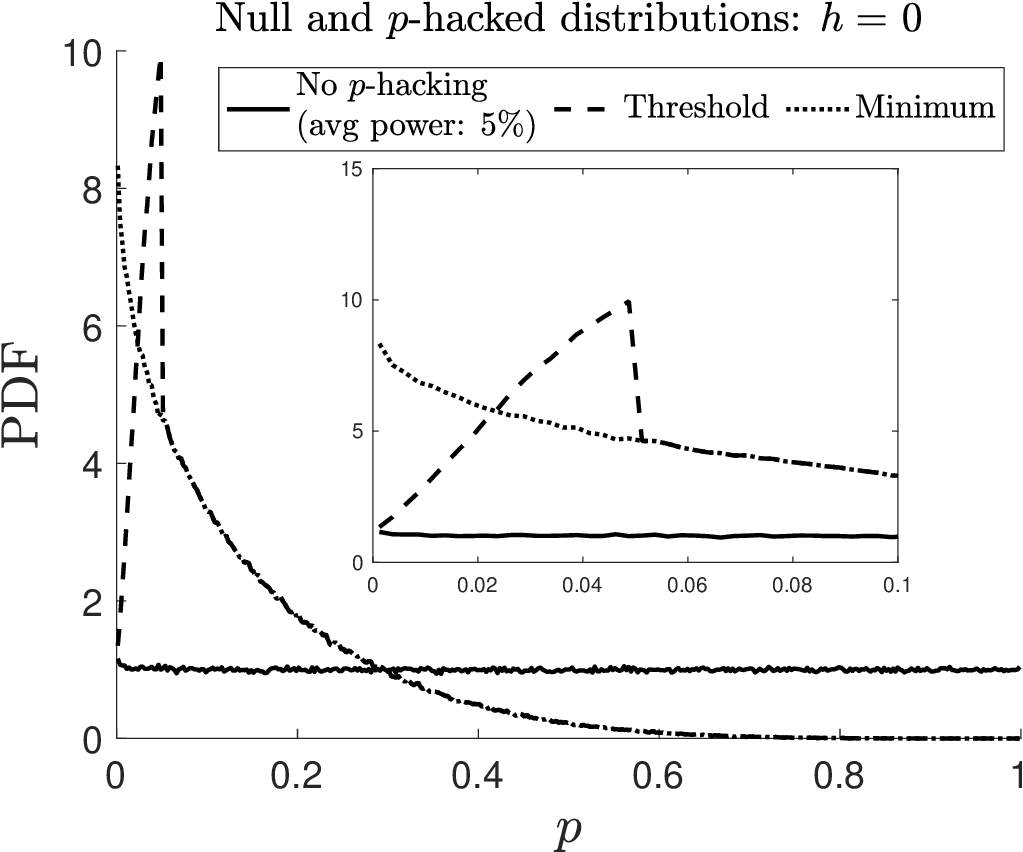}
\includegraphics[width=0.24\textwidth]{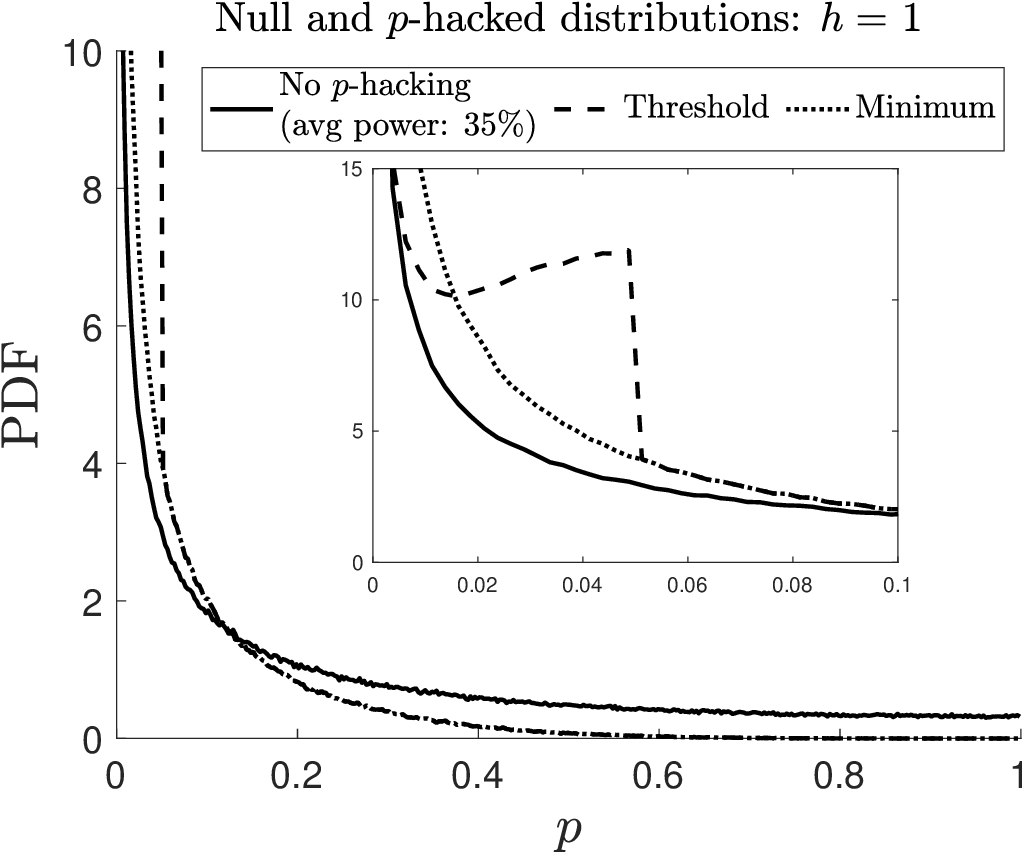}
\includegraphics[width=0.24\textwidth]{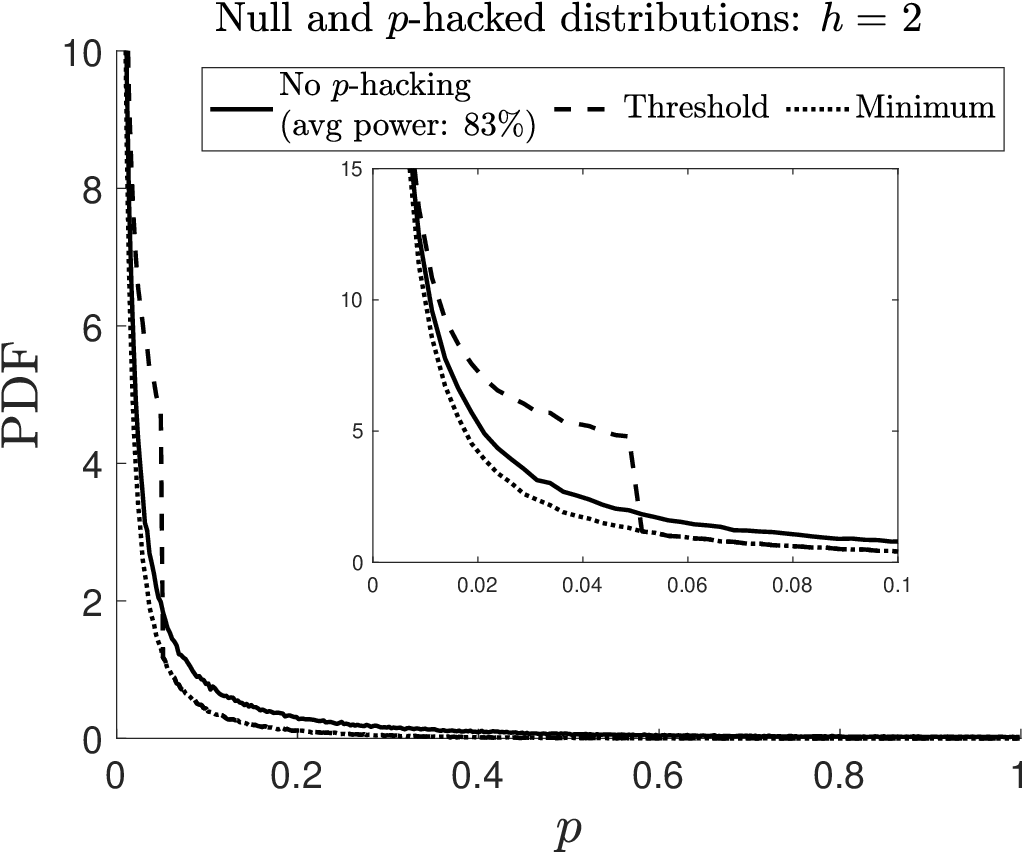}
\includegraphics[width=0.24\textwidth]{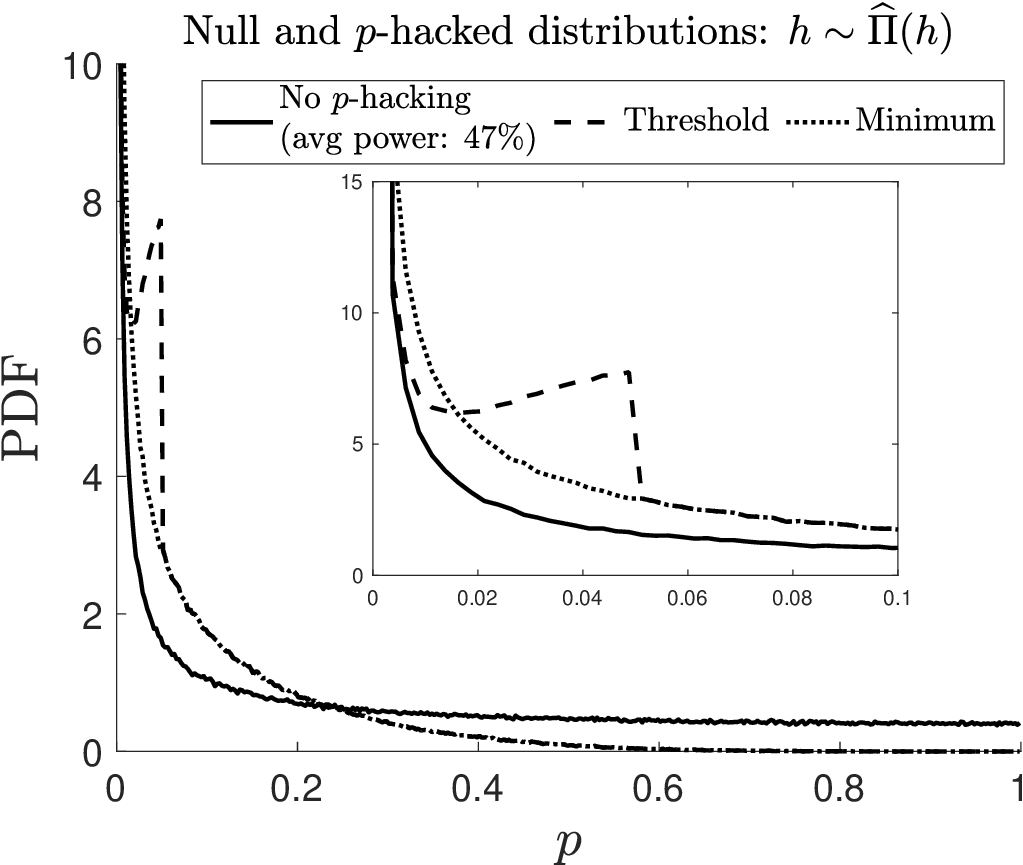}

\includegraphics[width=0.24\textwidth]{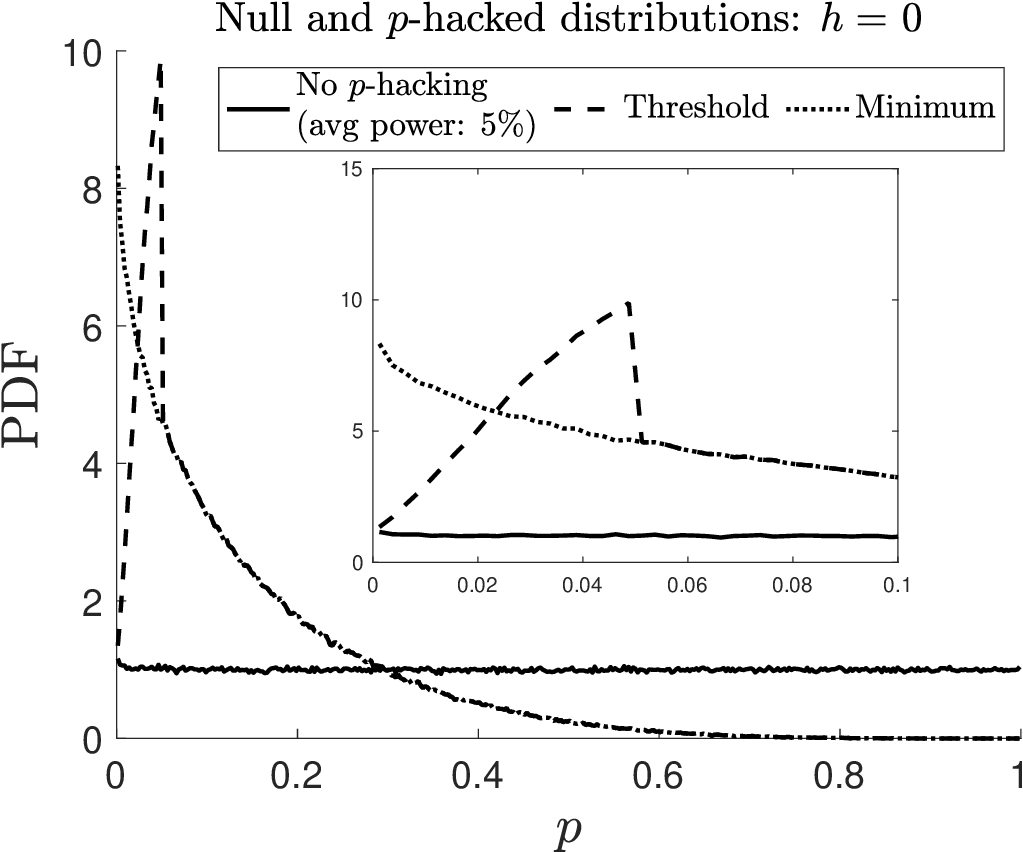}
\includegraphics[width=0.24\textwidth]{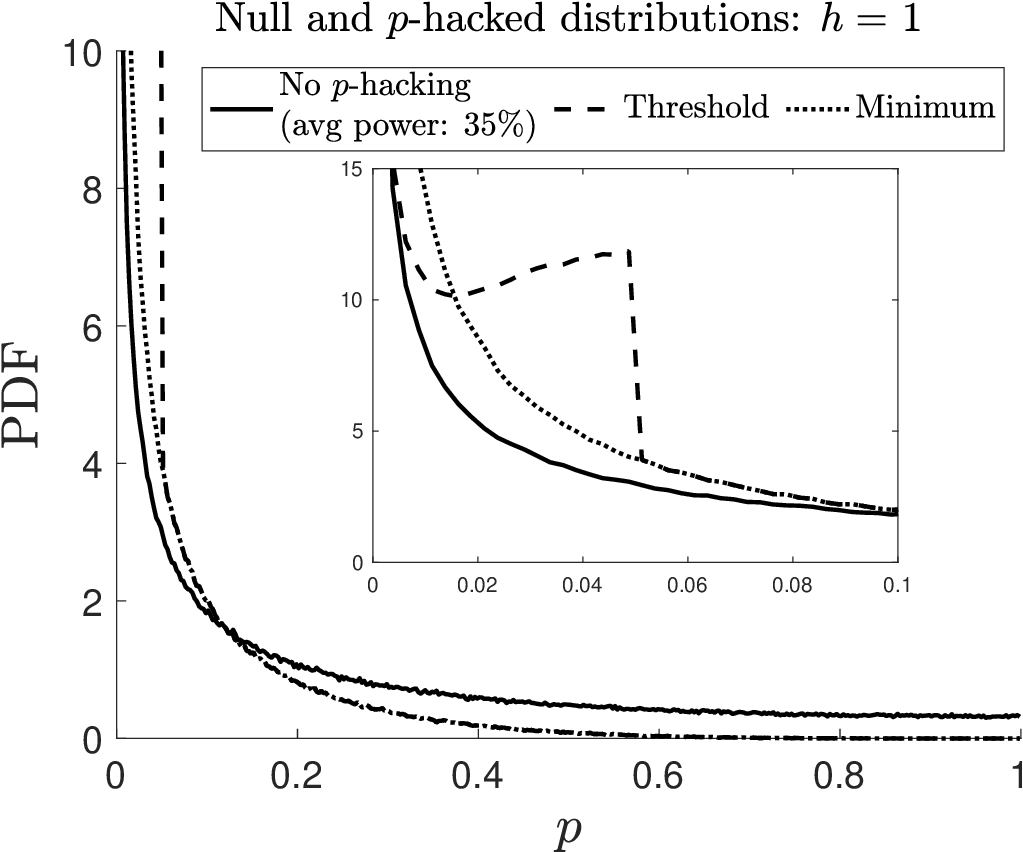}
\includegraphics[width=0.24\textwidth]{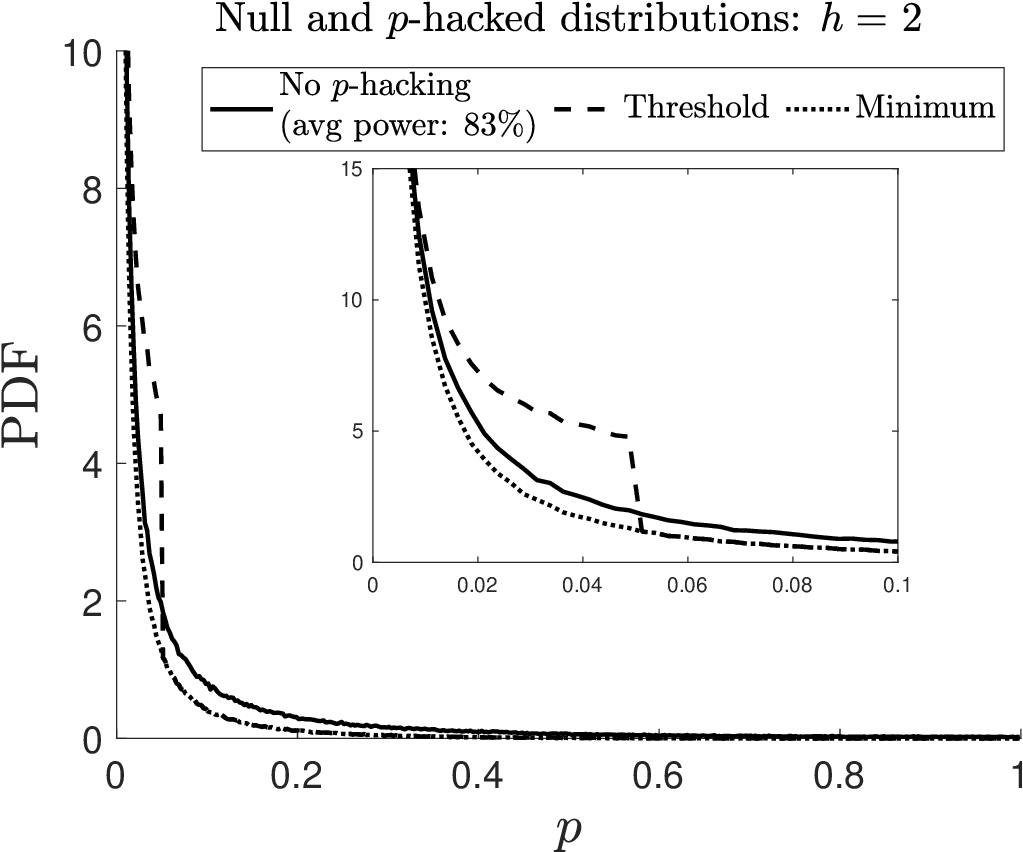}
\includegraphics[width=0.24\textwidth]{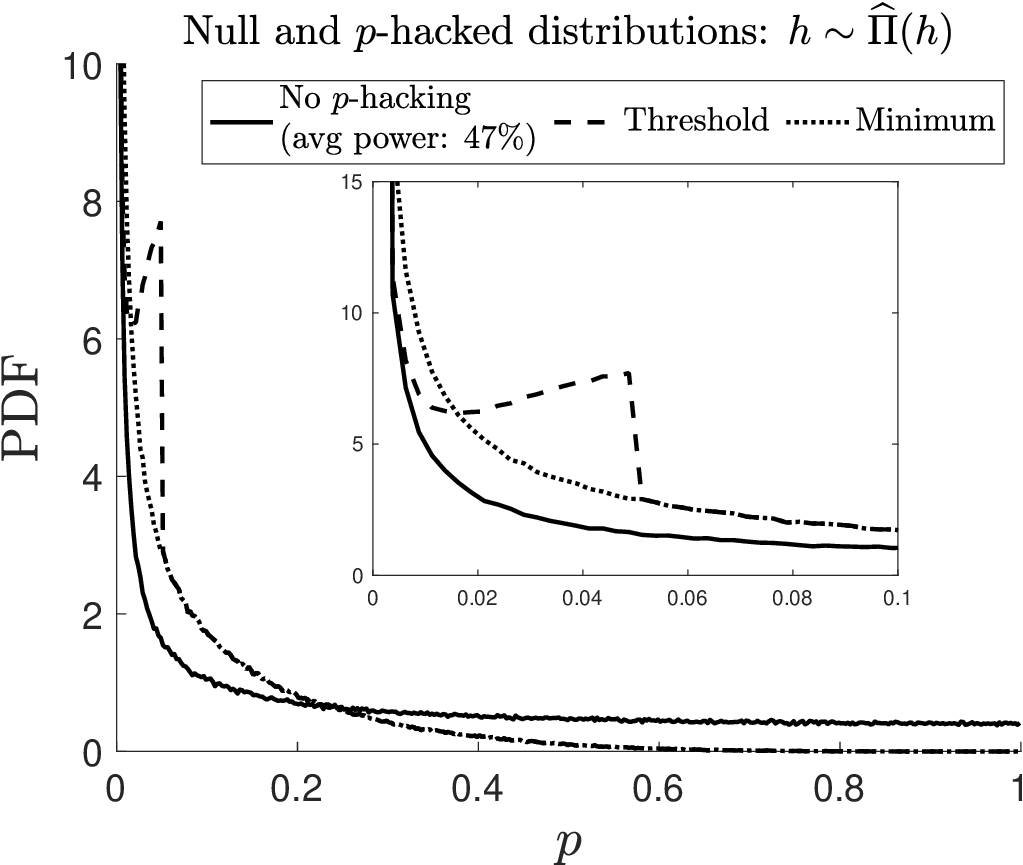}
\end{center}

\vspace{-2mm}

\footnotesize{\textit{Notes:} Figures show the null and alternative ($p$-hacked) distributions for IV selection with $K=5$. They also report the average power of the underlying studies, which is equal to the probability mass below 0.05 under no $p$-hacking. Row 1: two-sided tests, using all specifications; row 2: two-sided tests, using only specifications with $F>10$.}
\end{figure}

\begin{figure}[H]
\caption{Lag length selection: null and $p$-hacked distributions}
\label{fig:hists_var}
\begin{center}
\includegraphics[width=0.24\textwidth]{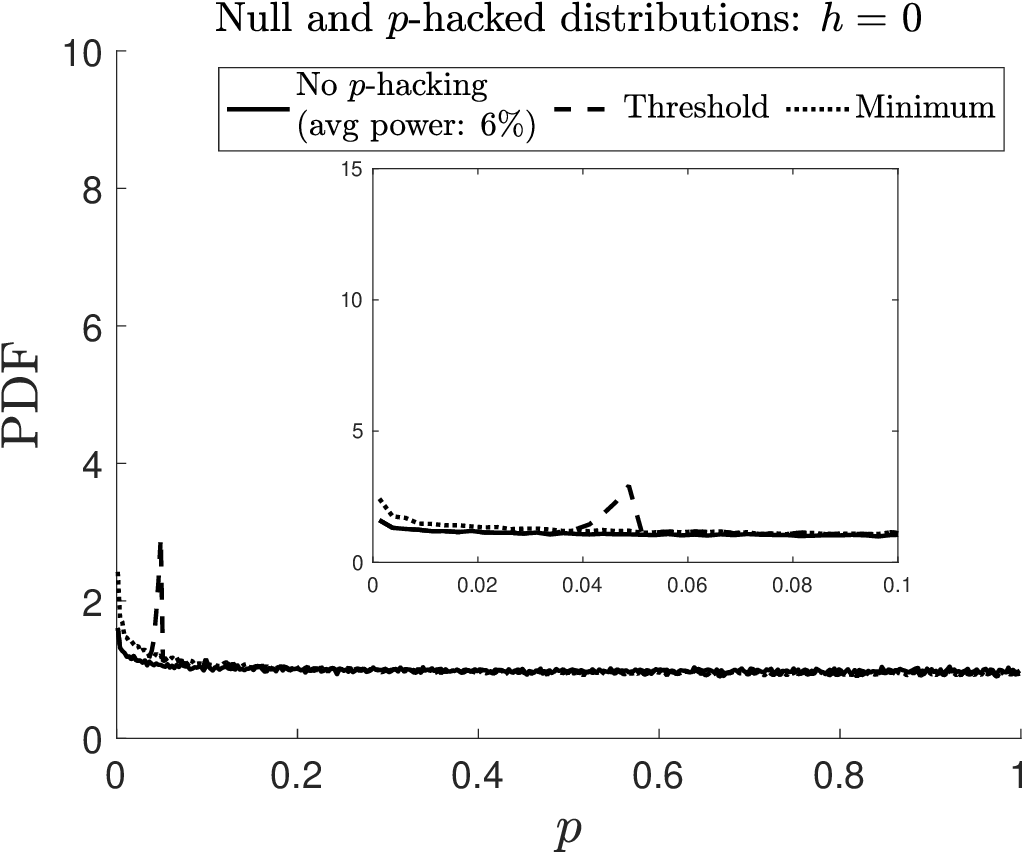}
\includegraphics[width=0.24\textwidth]{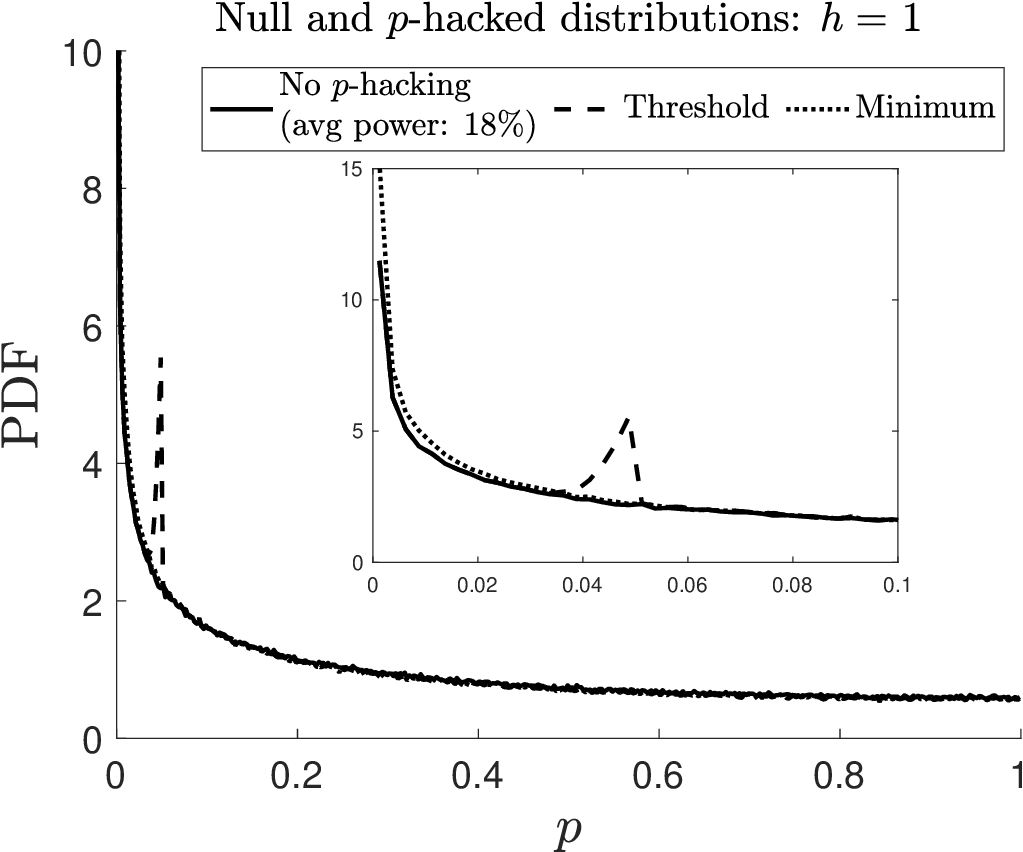}
\includegraphics[width=0.24\textwidth]{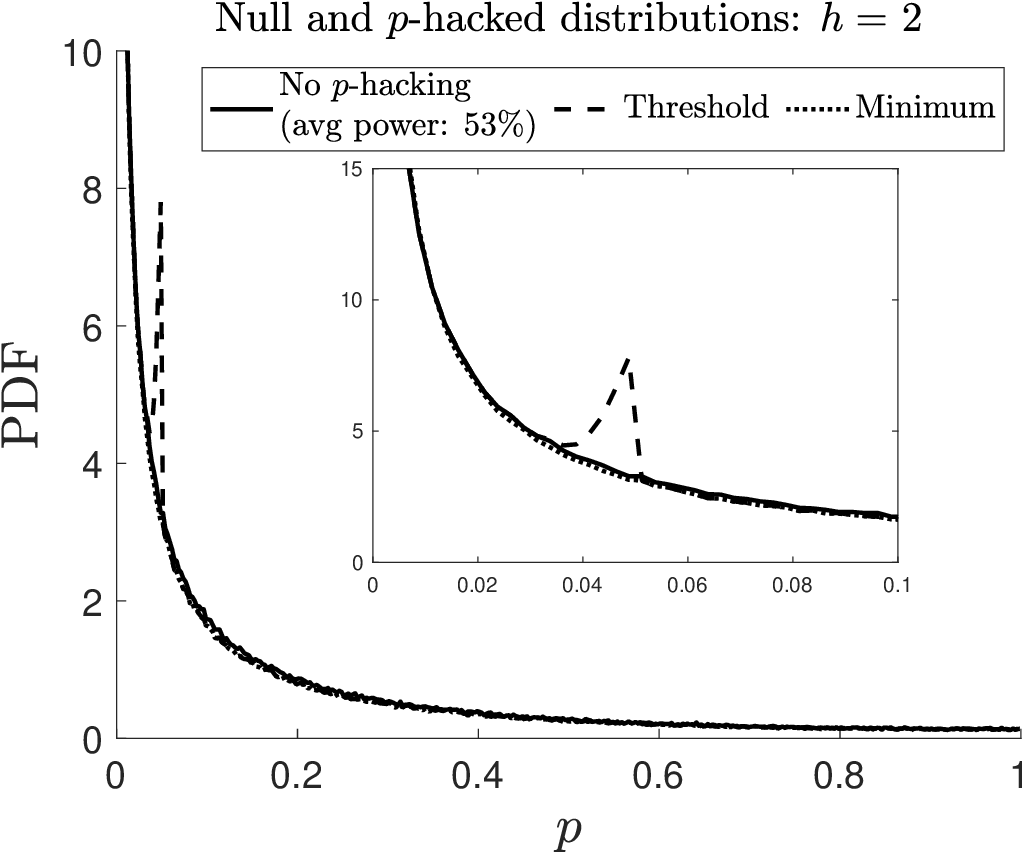}
\includegraphics[width=0.24\textwidth]{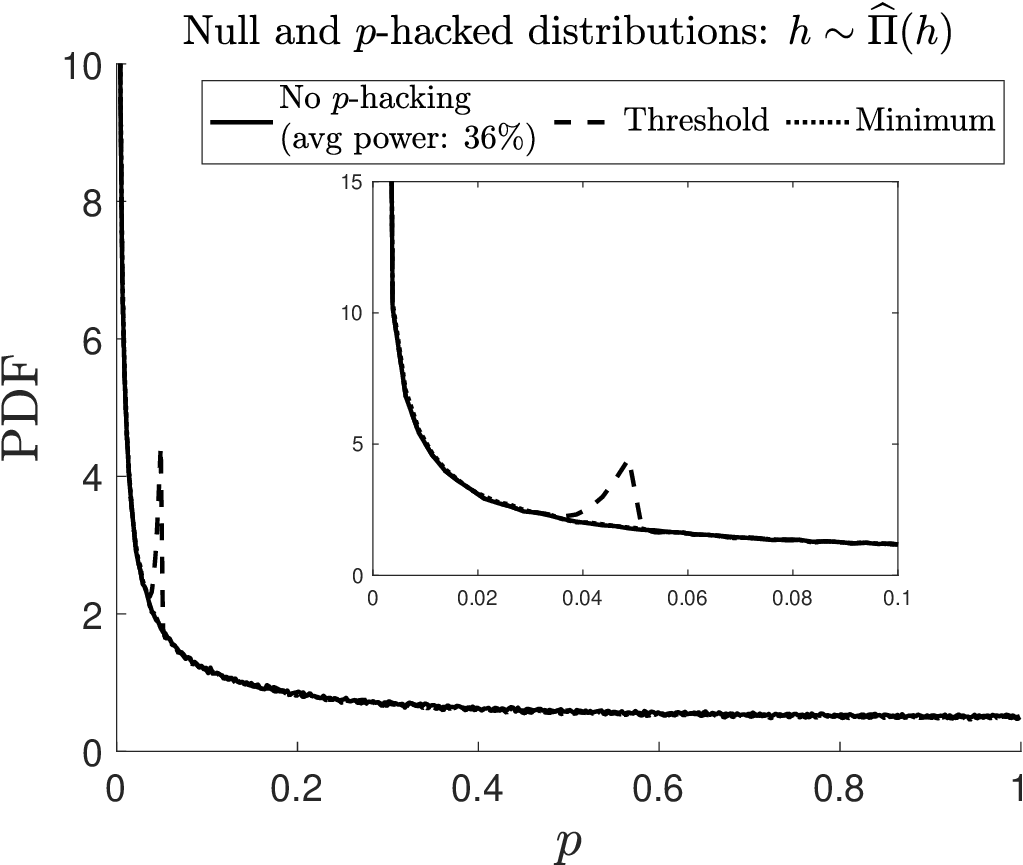}

\end{center}

\vspace{-2mm}

\footnotesize{\textit{Notes:} Figures show the null and alternative ($p$-hacked) distributions for lag length selection. They also report the average power of the underlying studies, which is equal to the probability mass below 0.05 under no $p$-hacking.}

\end{figure}

\begin{figure}[H]
\caption{Cluster level selection: null and $p$-hacked distributions}
\label{fig:hists_clust}
\begin{center}
\includegraphics[width=0.24\textwidth]{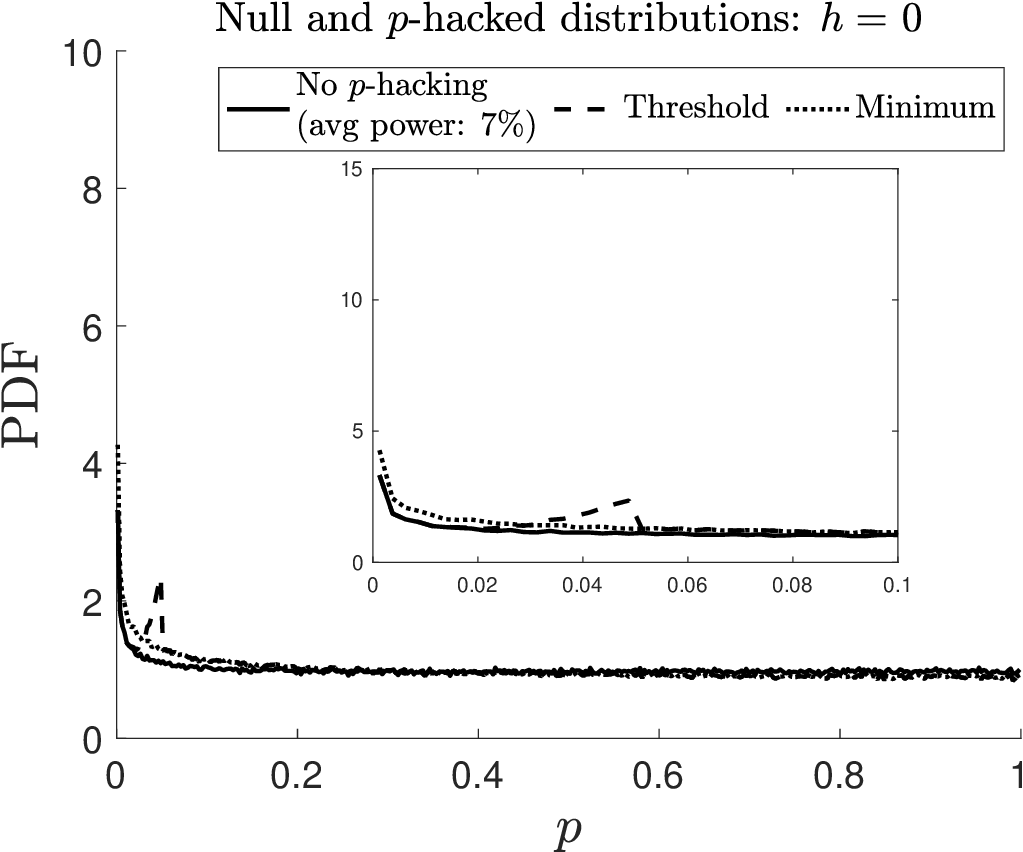}
\includegraphics[width=0.24\textwidth]{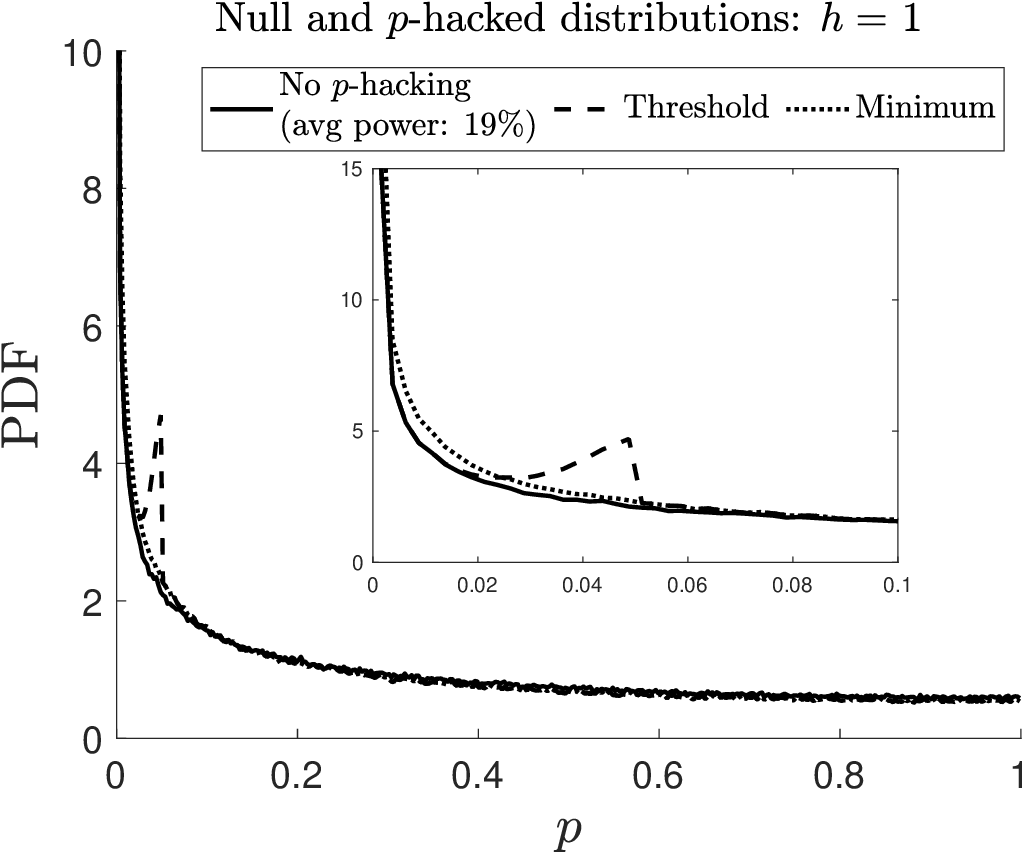}
\includegraphics[width=0.24\textwidth]{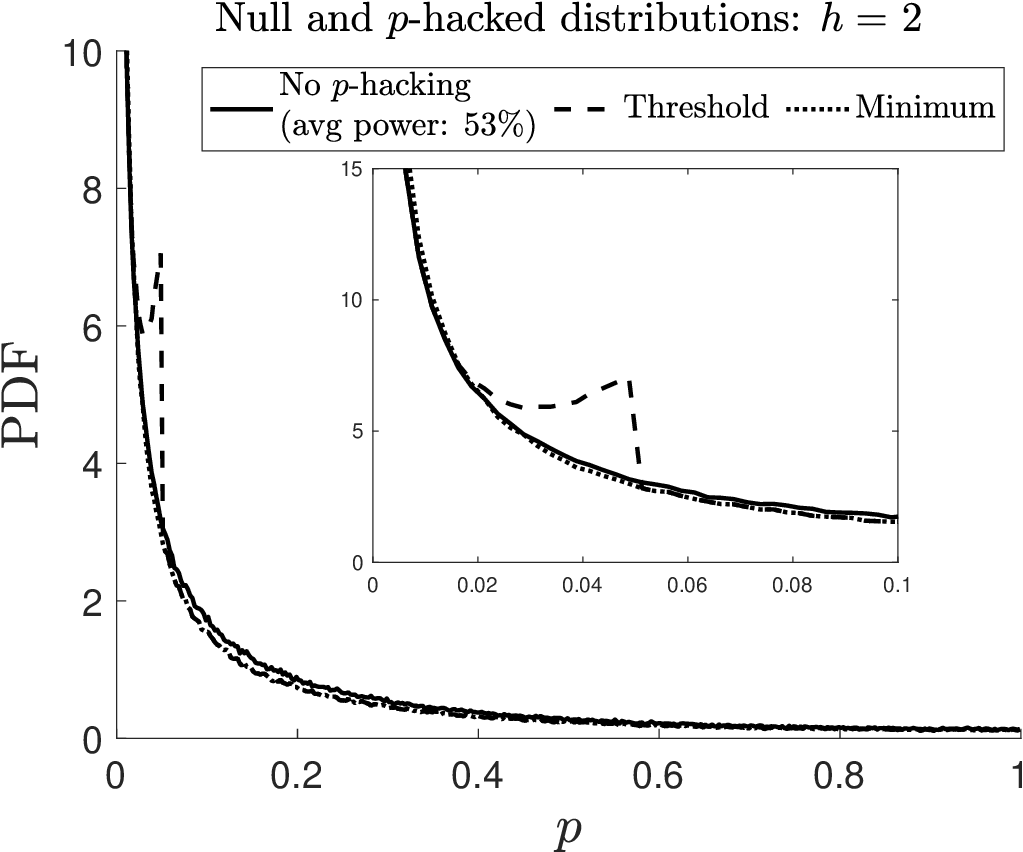}
\includegraphics[width=0.24\textwidth]{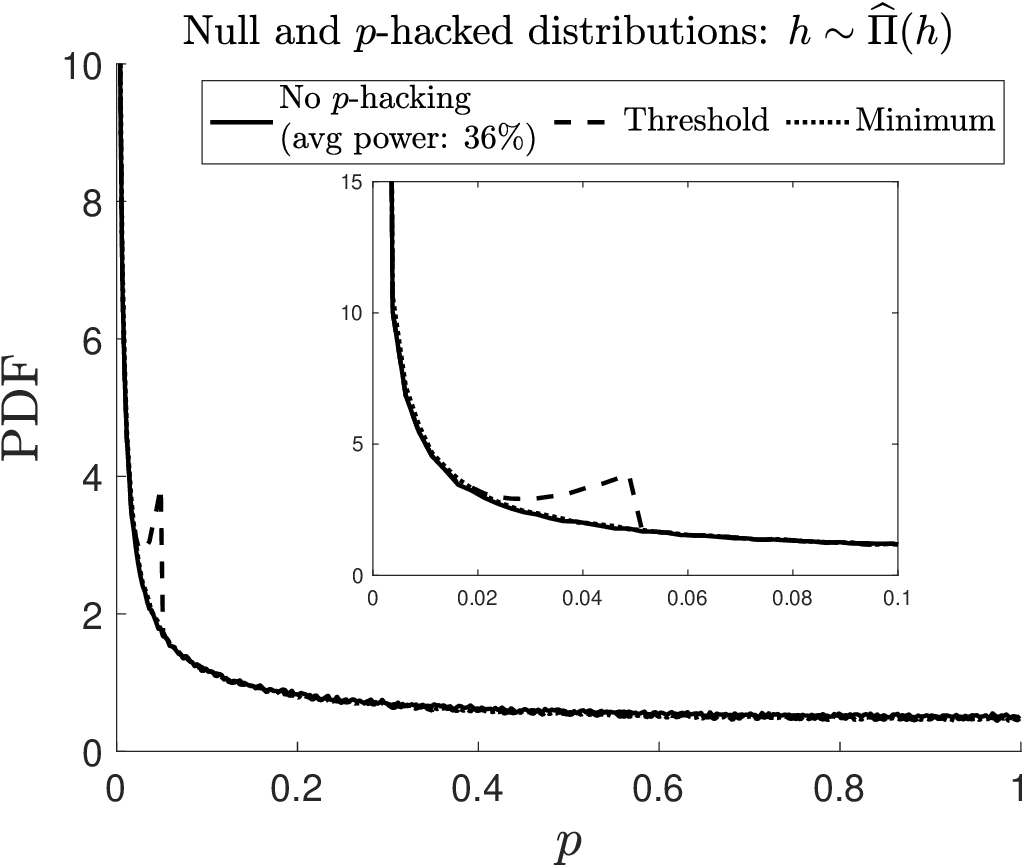}

\end{center}
\vspace{-2mm}

\footnotesize{\textit{Notes:} Figures show the null and alternative ($p$-hacked) distributions for cluster selection. They also report the average power of the underlying studies, which is equal to the probability mass below 0.05 under no $p$-hacking.}
\end{figure}

\section{The Impact of the Degree of Dependence between Tests on Power}
\label{app:impact_rho}

In this appendix, we study the impact of the degree of dependence between the test statistics used for $p$-hacking and the power for detecting $p$-hacking. Figure \ref{fig:power_rho} shows how power changes as a function of $\rho$ in the analytical example of Section \ref{sec:specification_search_regression}. The impact of changing $\rho$ is quite nuanced, especially under the threshold approach: it depends on the particular tests and the underlying testable implications, and it can even be non-monotonic. This is because $\rho$ has a complicated nonlinear impact on the $p$-curve, as can be seen from the analytical formulas for the $p$-curves under $p$-hacking in Section \ref{sec:specification_search_regression}.

\begin{figure}[H]

 \caption{Power as a function of $\rho$}\label{fig:power_rho}

 \vspace{-5mm}
 
 \begin{center}

 \includegraphics[width=0.24\textwidth]{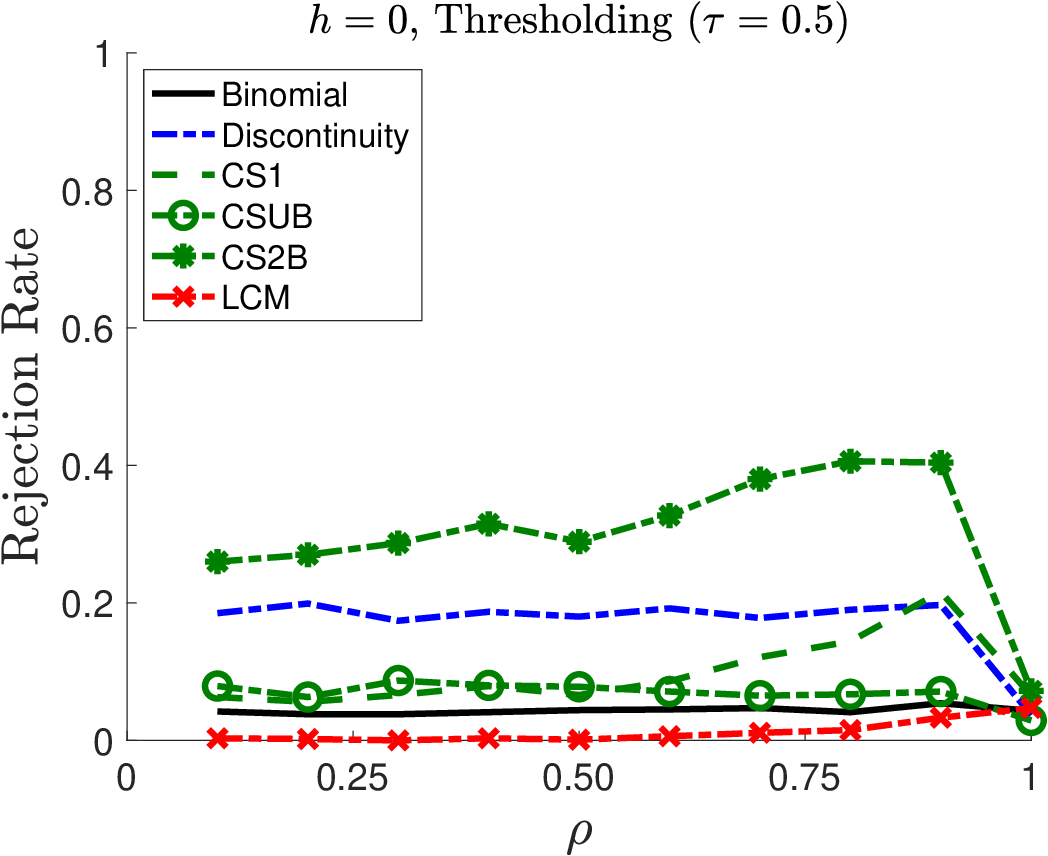}
 \includegraphics[width=0.24\textwidth]{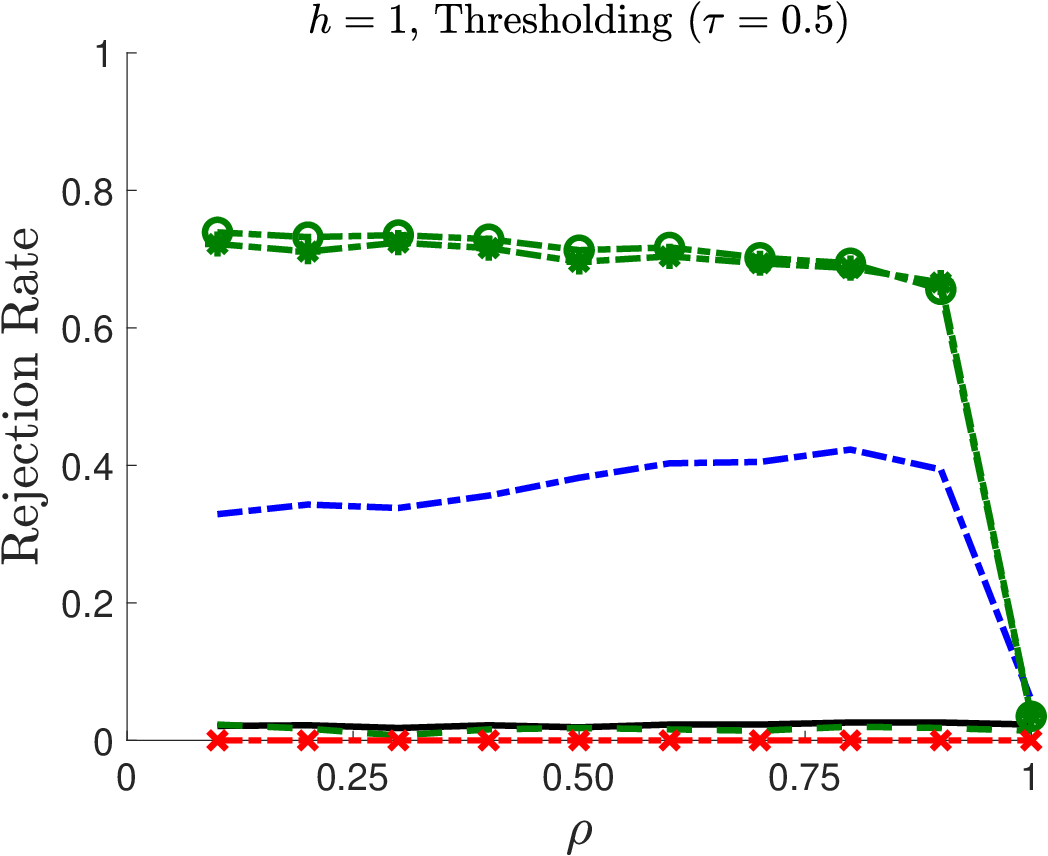}
 \includegraphics[width=0.24\textwidth]{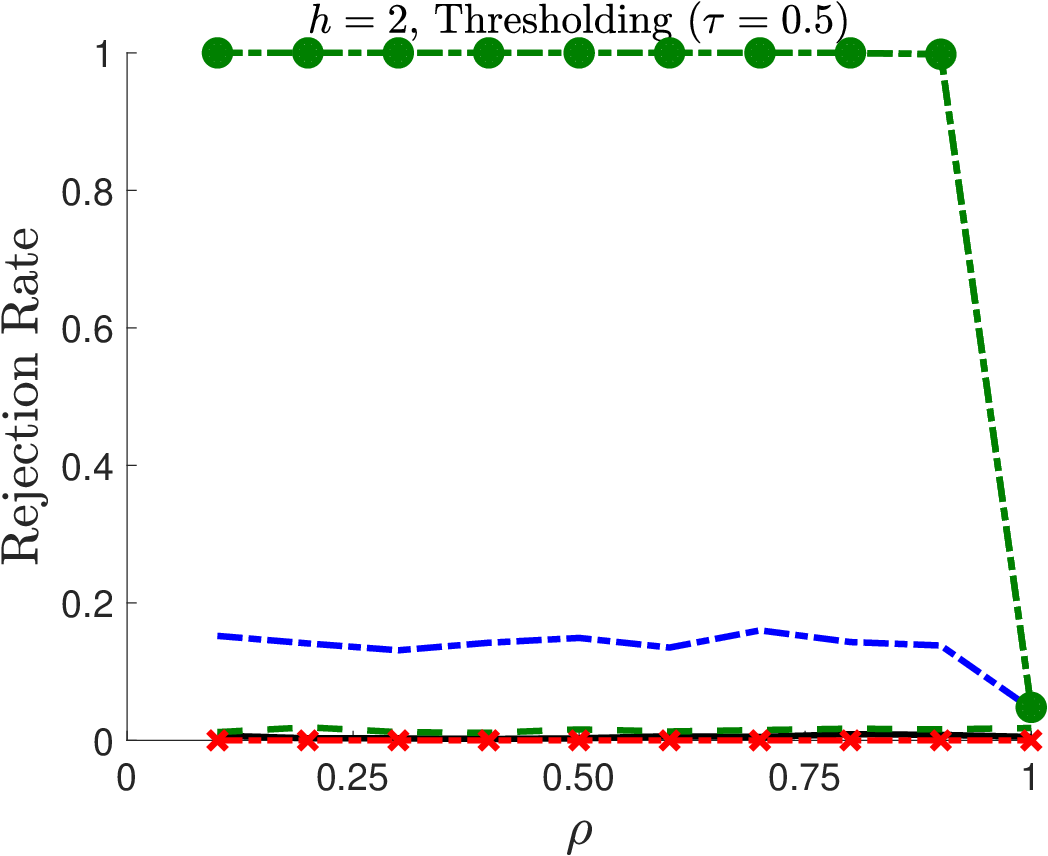}
 \includegraphics[width=0.24\textwidth]{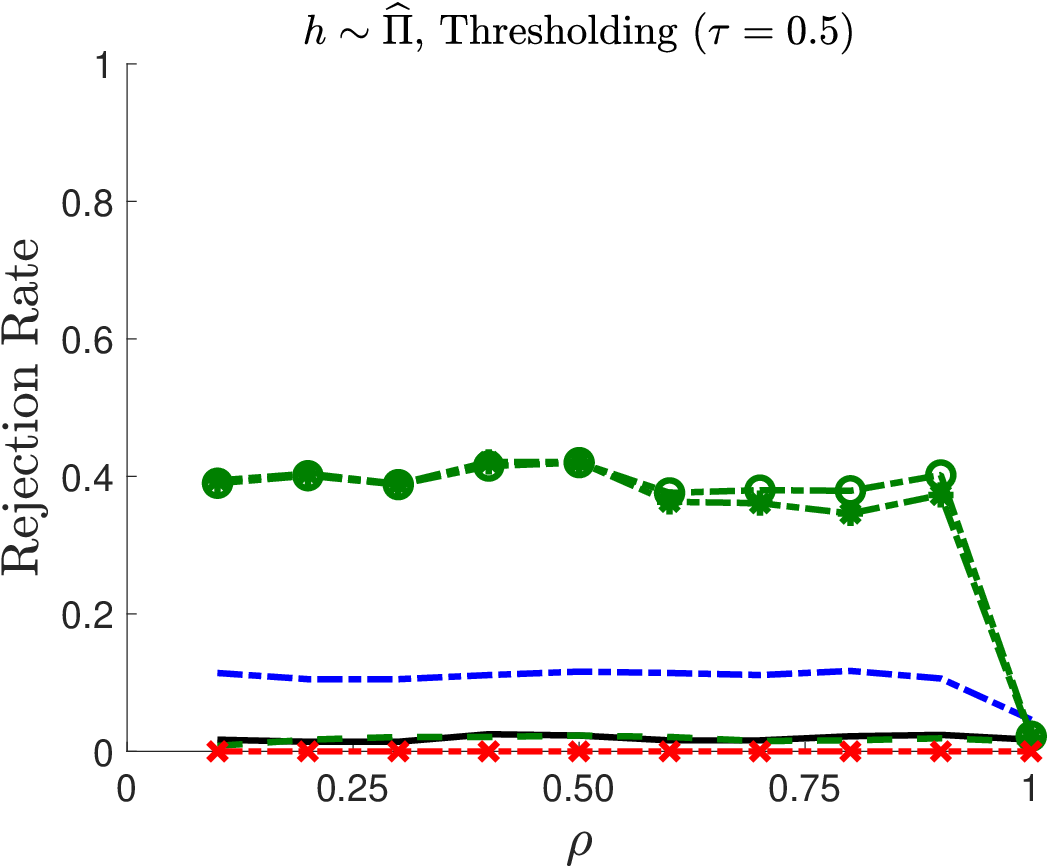}

 \includegraphics[width=0.24\textwidth]{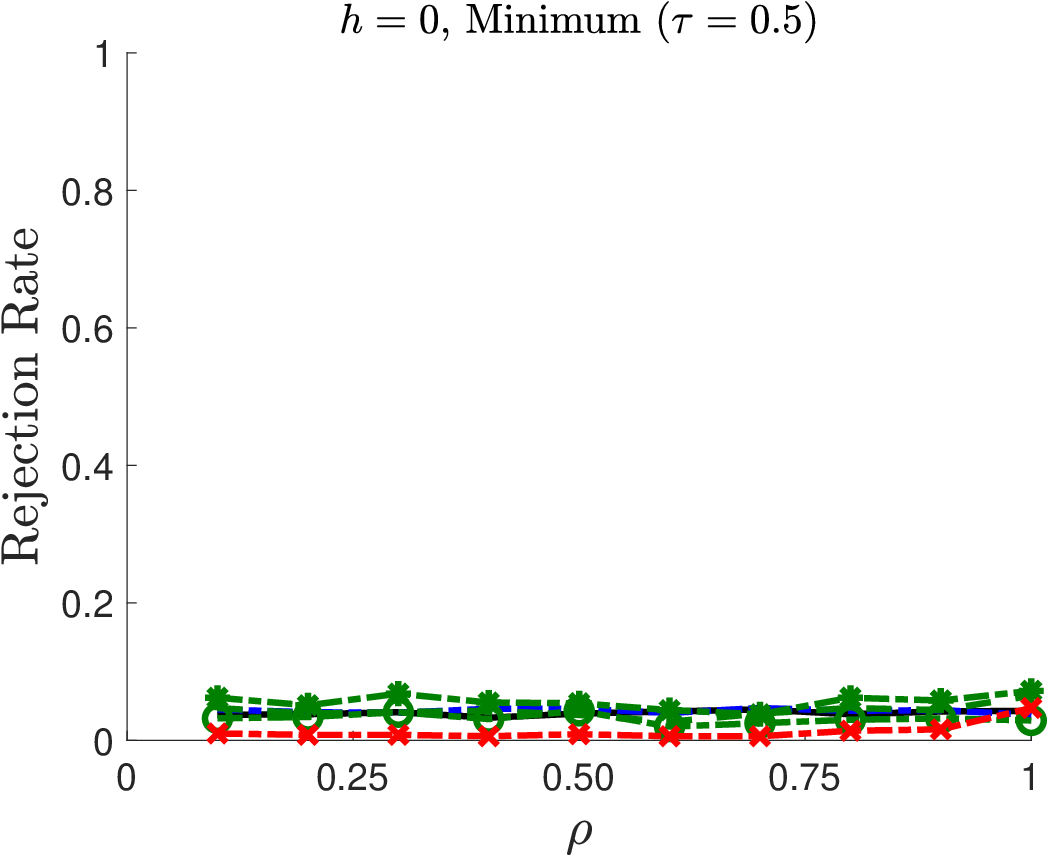}
 \includegraphics[width=0.24\textwidth]{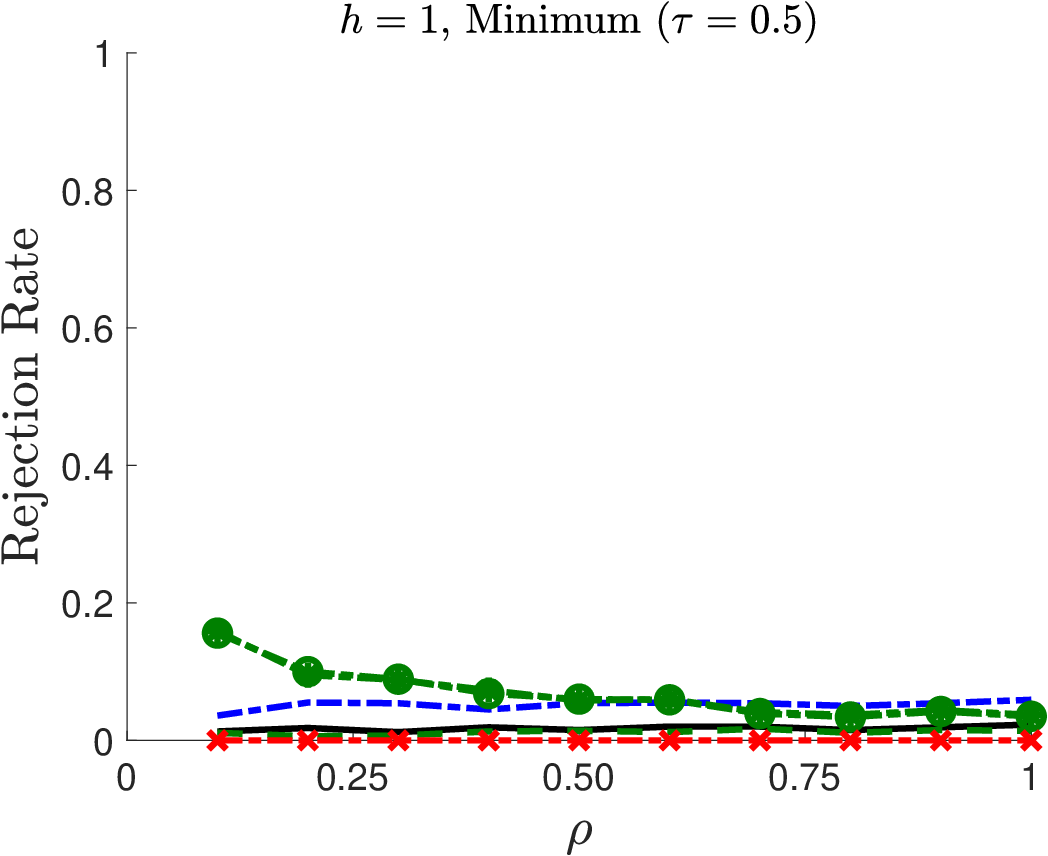}
 \includegraphics[width=0.24\textwidth]{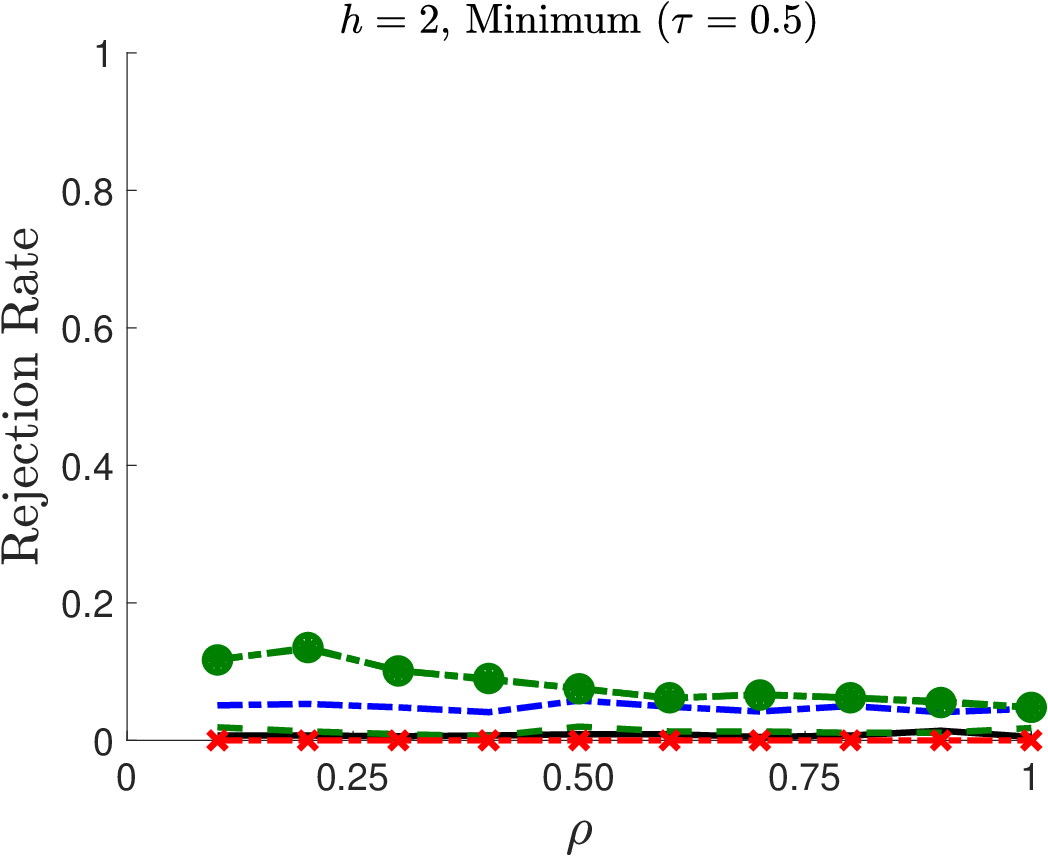}
 \includegraphics[width=0.24\textwidth]{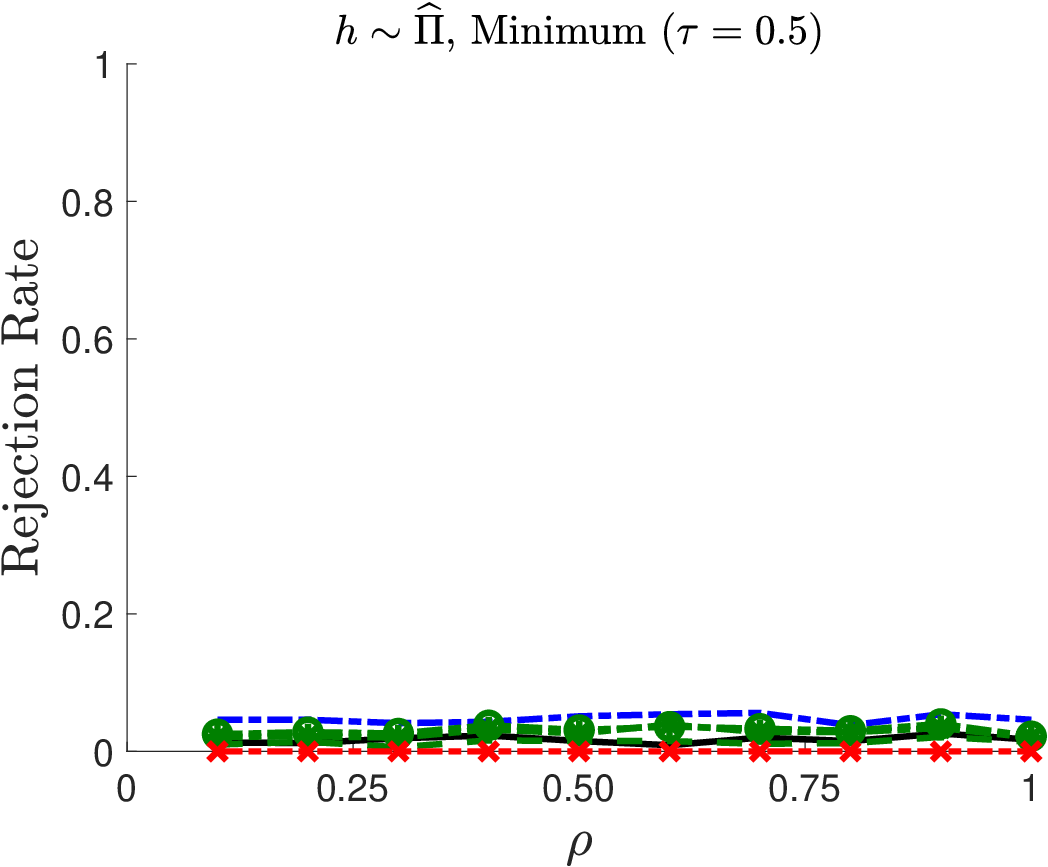}
 \end{center}
 \vspace*{-3mm}
 \footnotesize{\textit{Notes:} Figures show how power changes as a function of $\rho$ for covariate selection, as described in the analytical example in Section \ref{sec:specification_search_regression}, with one-sided tests. The results are based on 1,000 simulation repetitions.}
 \end{figure}

\section{The Impact of Selecting Incorrect $p$-Values on Power}
\label{app:sample_selection}

In this appendix, we analyze the importance of including $p$-values corresponding to the ``correct'' tests into the meta-analytic dataset, where, by ``correct'' test, we mean the main tests of interest that the researcher targets when $p$-hacking. Specifically, we augment the analytical example in Section \ref{sec:specification_search_regression} by letting the meta-analyst mistakenly record --- instead of the $p$-value for the coefficient of interest --- the $p$-value for the coefficient on the control variable $Z_j$ in the specification $j\in\{1,2\}$ reported by the researchers post $p$-hacking. The mistake occurs with probability $p_{\text{mis}}\in[0,1]$. 

Figure \ref{fig:power_cov_p_mis} plots the power of the tests as a function of $p_{\text{mis}}$. It shows that selection mistakes can substantially reduce power, even for the best tests. When $p_{\text{mis}}$ is large enough, no tests has any power for detecting $p$-hacking. The findings in this appendix underscore the importance of data quality when testing for $p$-hacking.

 \begin{figure}[H]

\caption{Power as a function of $p_{\text{mis}}$}\label{fig:power_cov_p_mis}

\vspace{-5mm}
 
\begin{center}

 \includegraphics[width=0.24\textwidth]{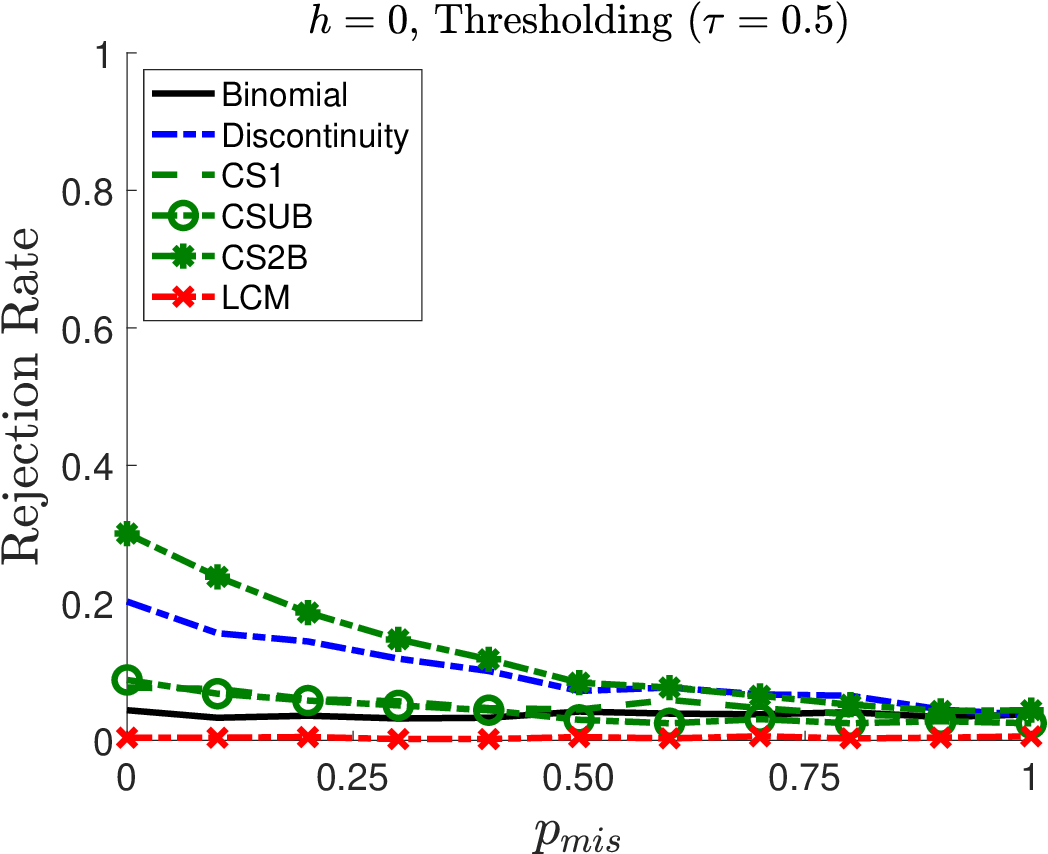}
 \includegraphics[width=0.24\textwidth]{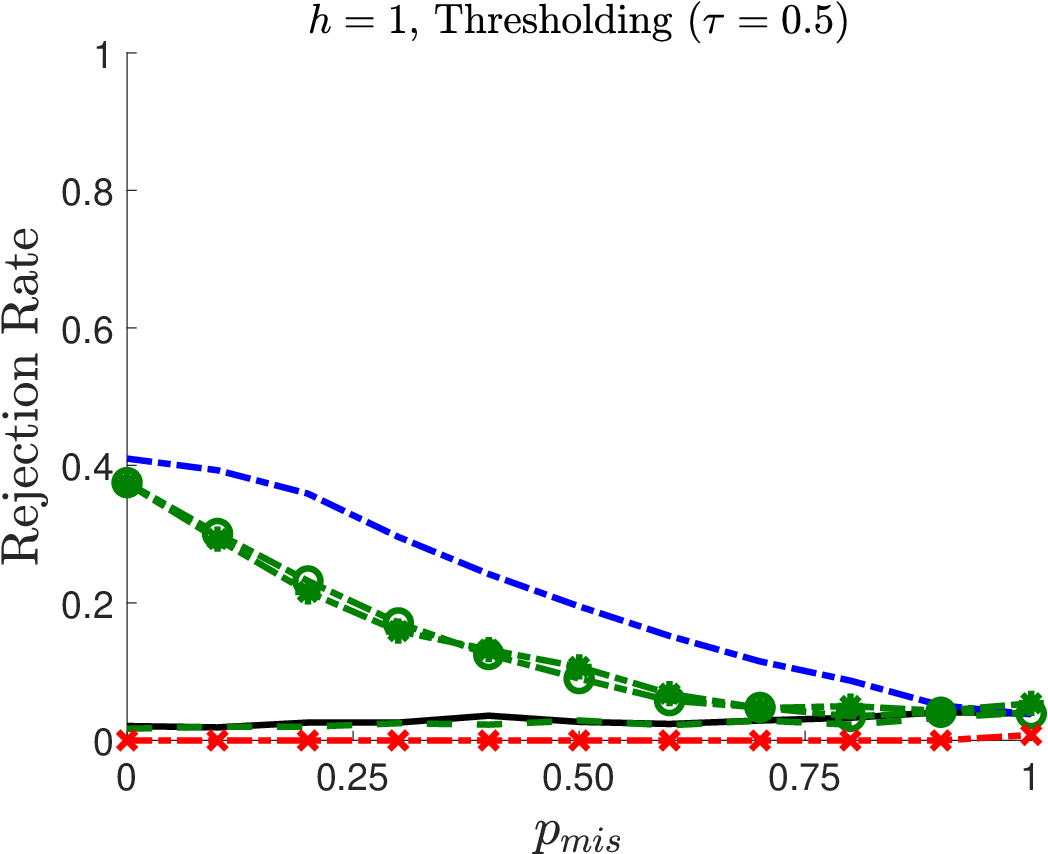}
 \includegraphics[width=0.24\textwidth]{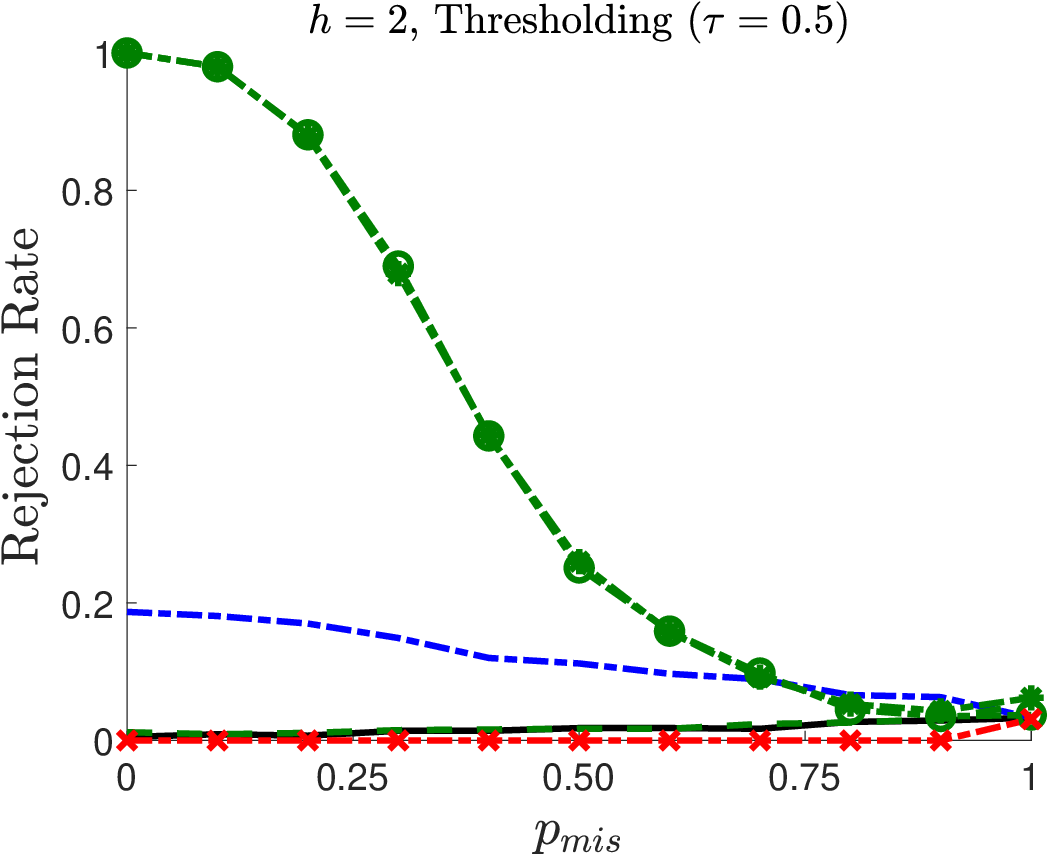}
 \includegraphics[width=0.24\textwidth]{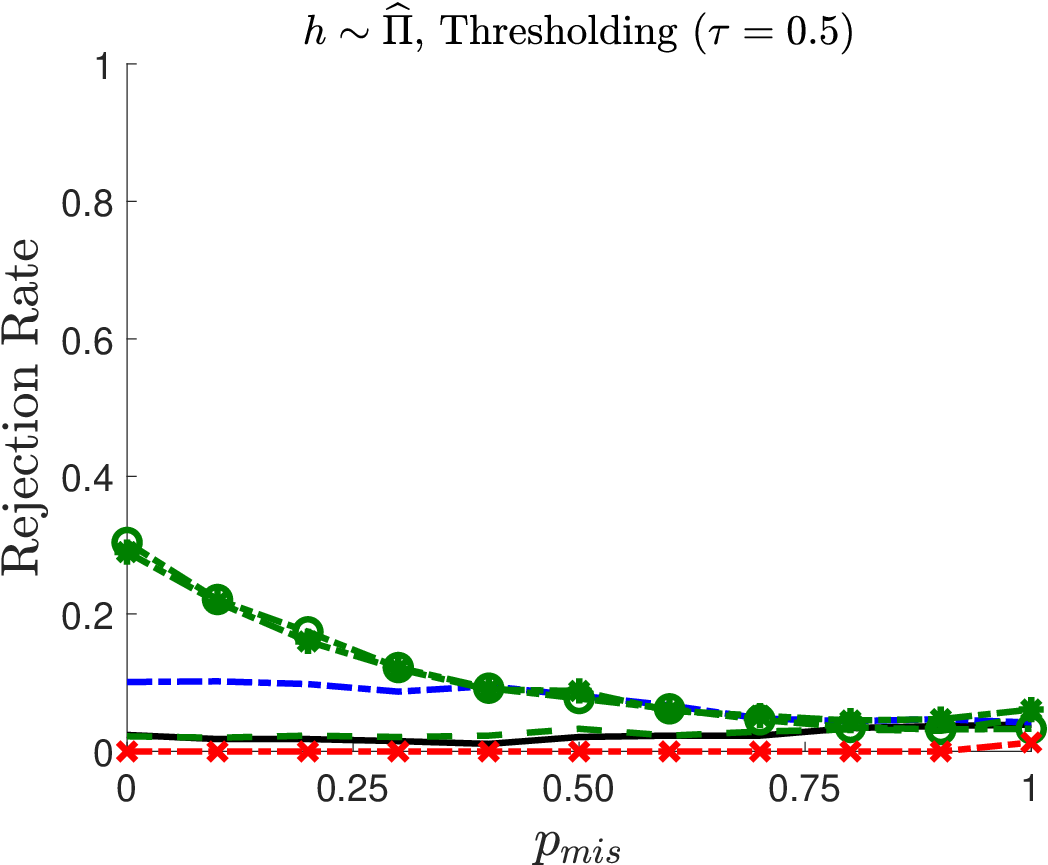}

 \includegraphics[width=0.24\textwidth]{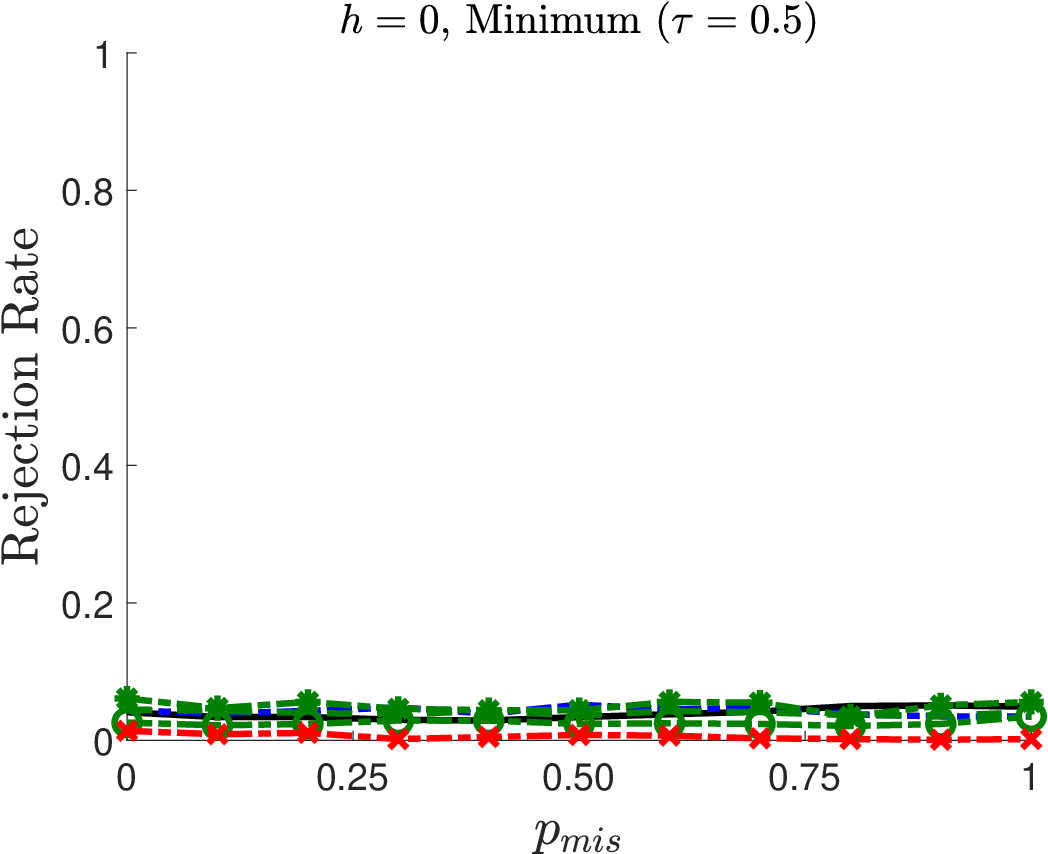}
 \includegraphics[width=0.24\textwidth]{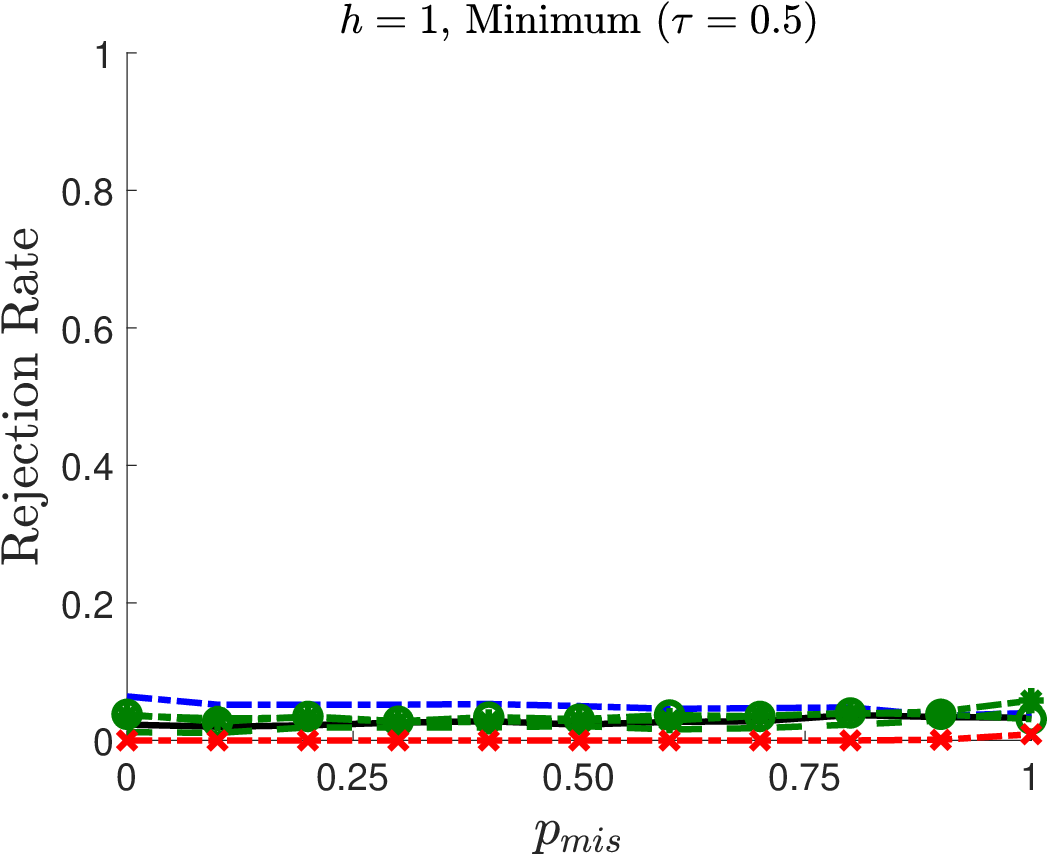}
 \includegraphics[width=0.24\textwidth]{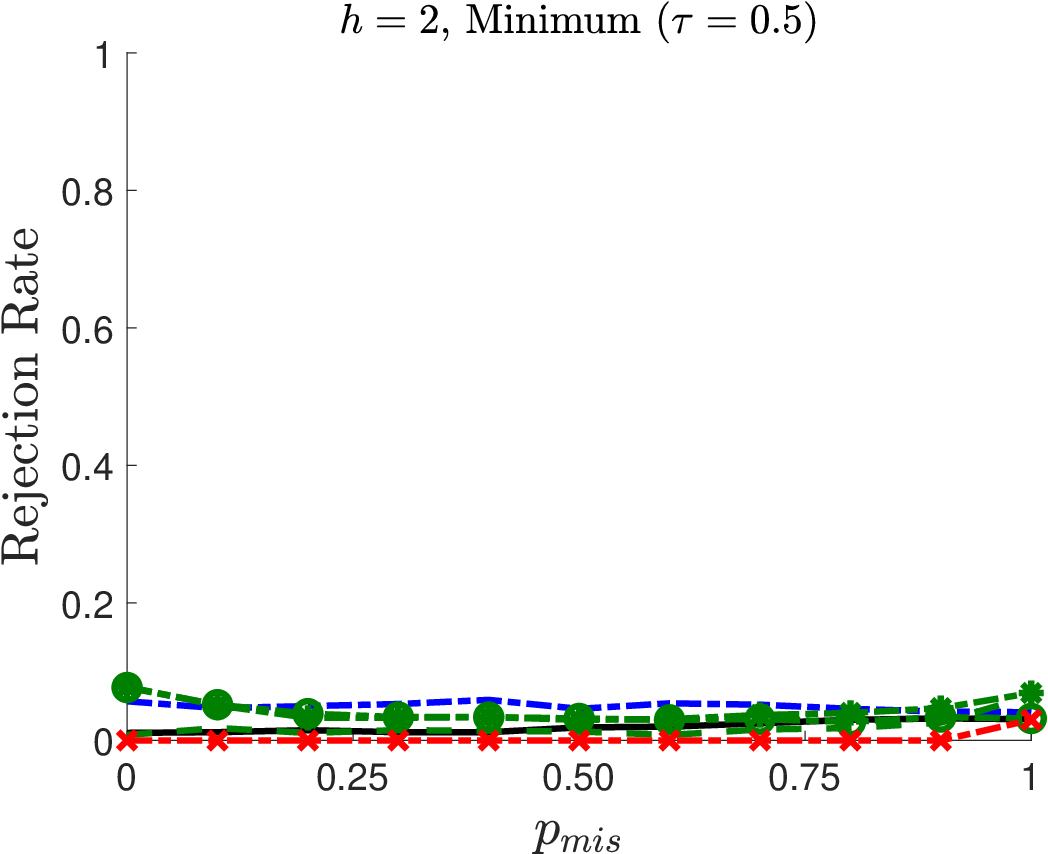}
 \includegraphics[width=0.24\textwidth]{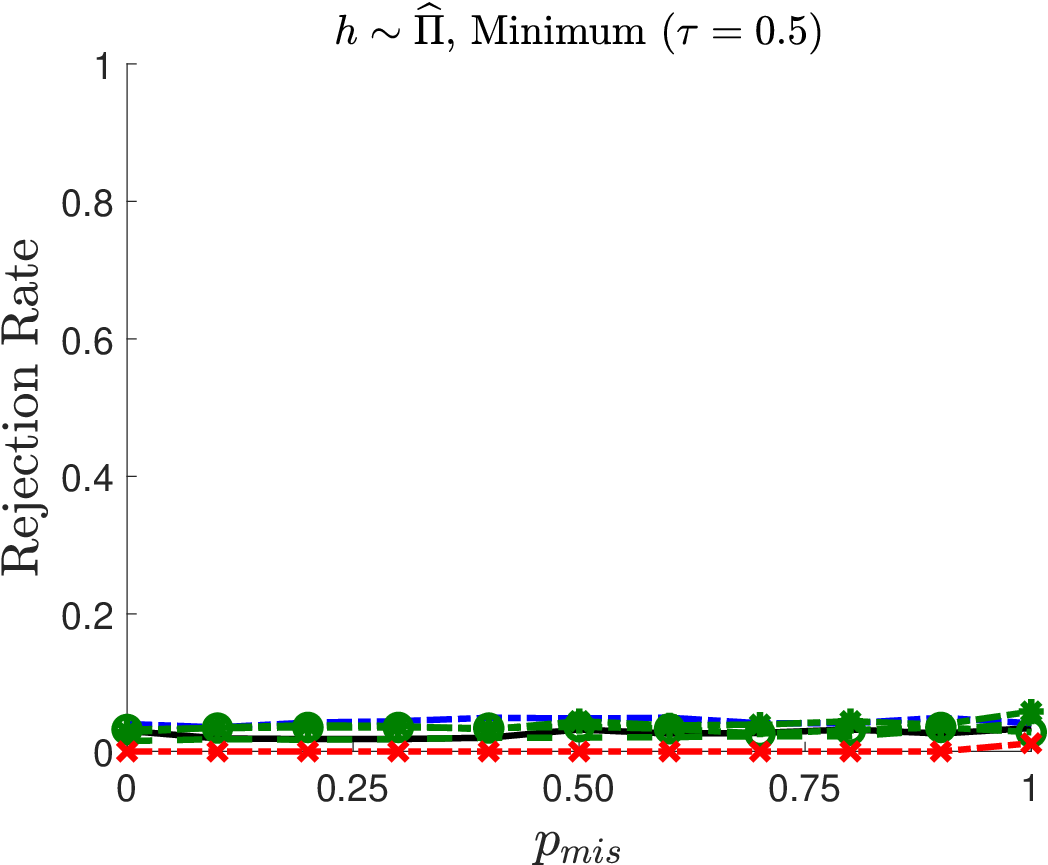}
 
\end{center}
\vspace*{-3mm}
\footnotesize{\textit{Notes:} Figures show how power changes as a function of the probability of selecting incorrect $p$-values, $p_{\text{mis}}$, for covariate selection, as described in the analytical example in Section \ref{sec:specification_search_regression}, with two-sided tests. The results are based on 1,000 simulation repetitions.}
\end{figure}

\section{The Impact of the Sample Size on Power}
\label{app:sample_size}

In this appendix, we analyze the relationship between the sample size of the meta-analytic dataset and the power of the tests for detecting $p$-hacking. Specifically, we redo the analysis corresponding to Figure \ref{fig:power_cov_combined_a} for different sample sizes ranging between 400 and 5,000 observed $p$-values. For illustration, we focus on the case where 50\% of the observed results are $p$-hacked ($\tau=0.5$). 

Figure \ref{fig:power_cov_sample_size} plots the power of the tests as a function of the sample size. It shows that the power can be very low even for detecting $p$-hacking based on the threshold approach when the sample size is insufficiently large. As expected, the power increases monotonically with the sample size, except for the cases where the power is trivial.

 \begin{figure}[H]

\caption{Power as a function of the sample size}\label{fig:power_cov_sample_size}

\vspace{-5mm}
 
\begin{center}

 \includegraphics[width=0.24\textwidth]{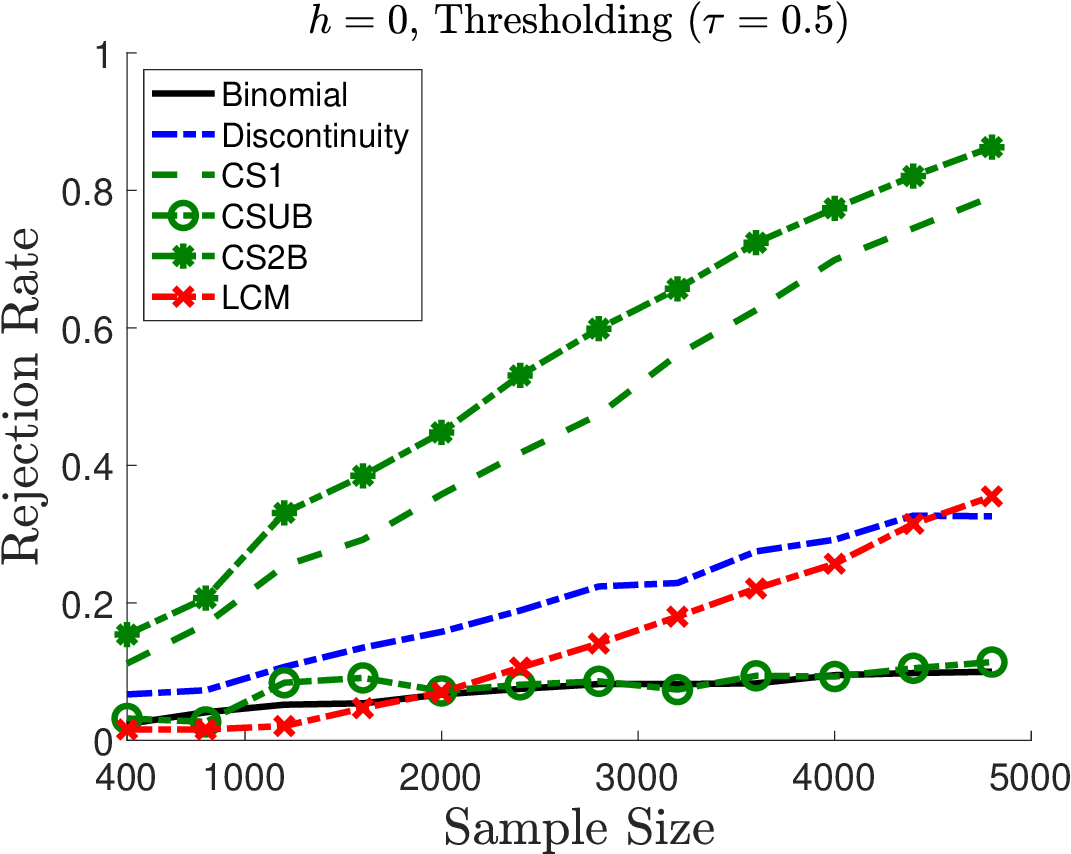}
 \includegraphics[width=0.24\textwidth]{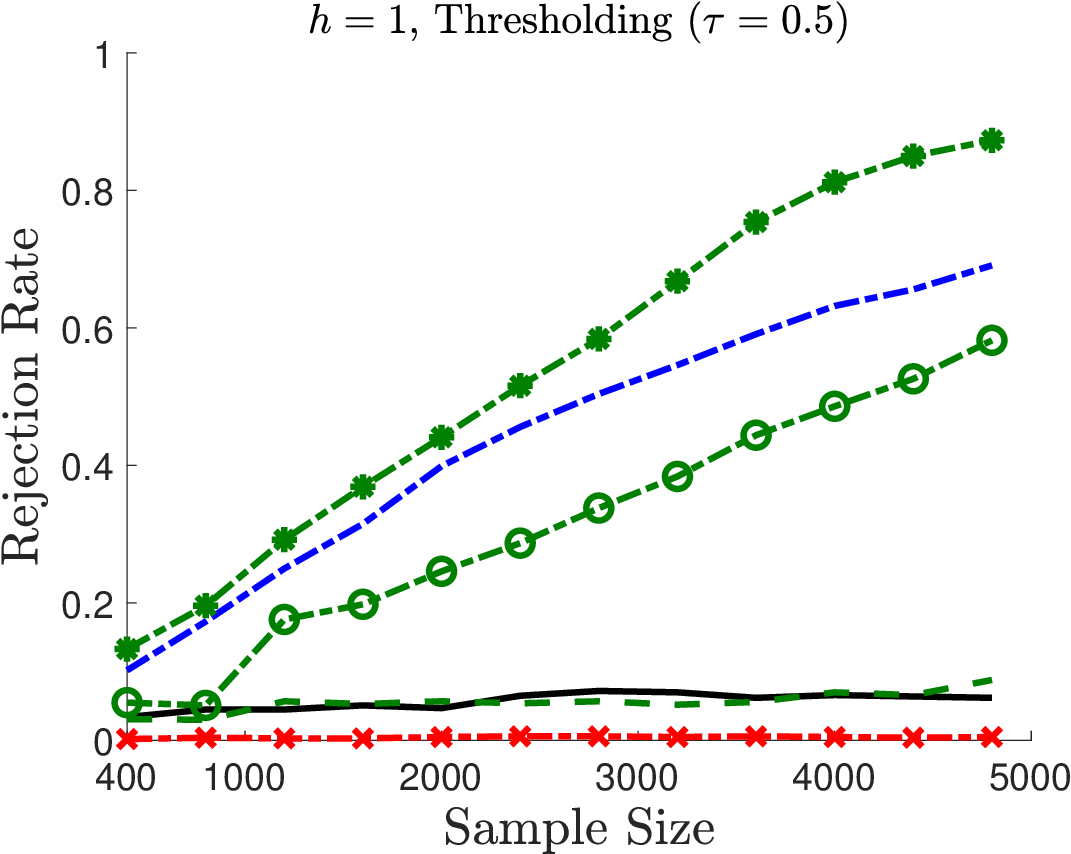}
 \includegraphics[width=0.24\textwidth]{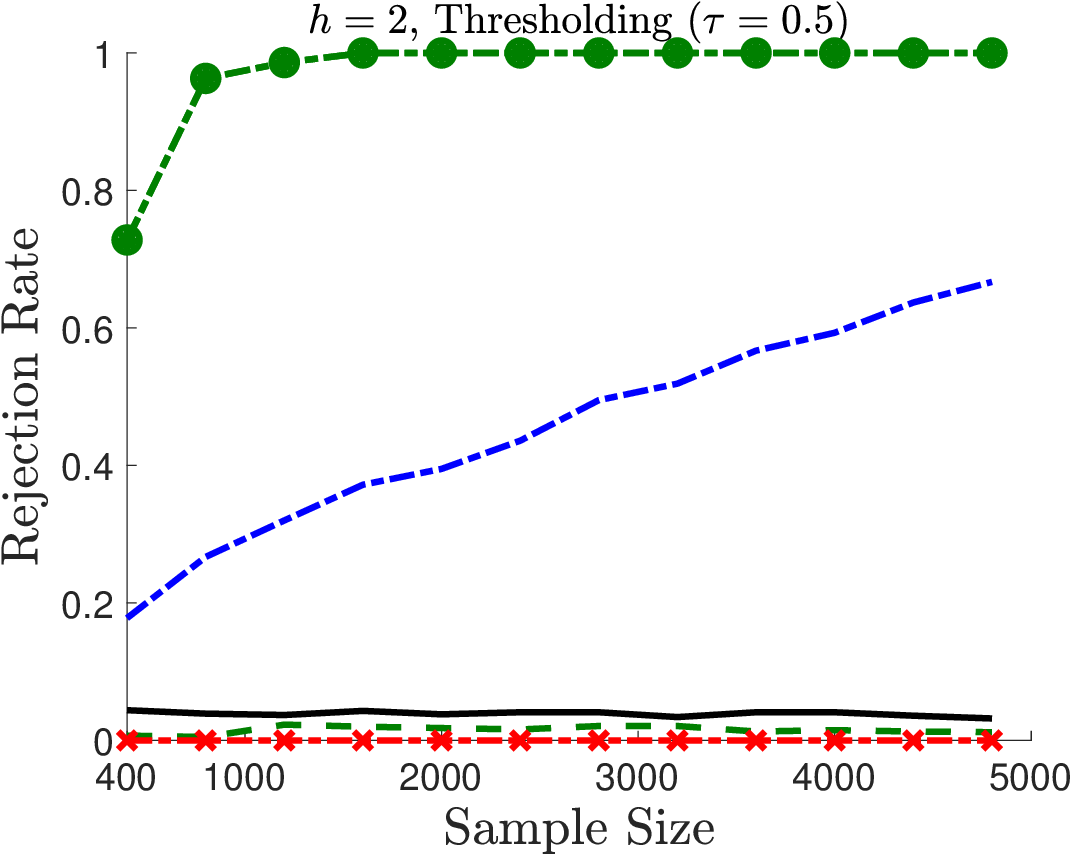}
 \includegraphics[width=0.24\textwidth]{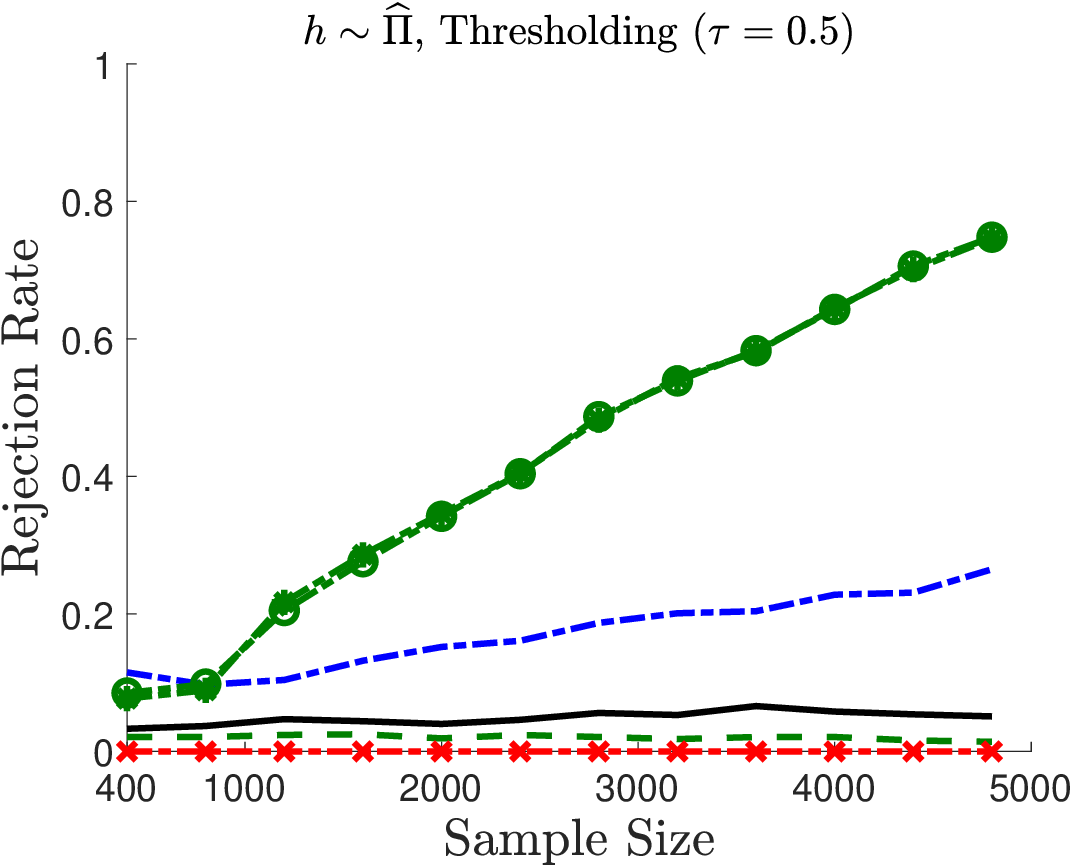}

 \includegraphics[width=0.24\textwidth]{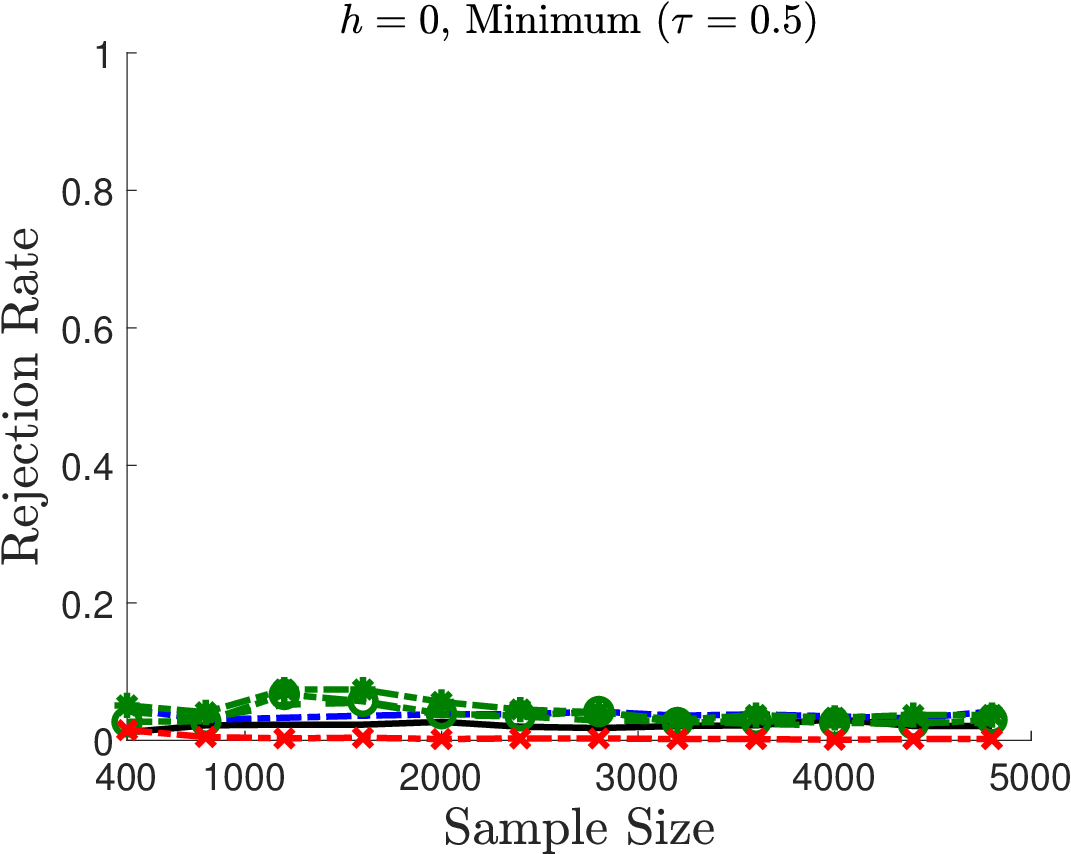}
 \includegraphics[width=0.24\textwidth]{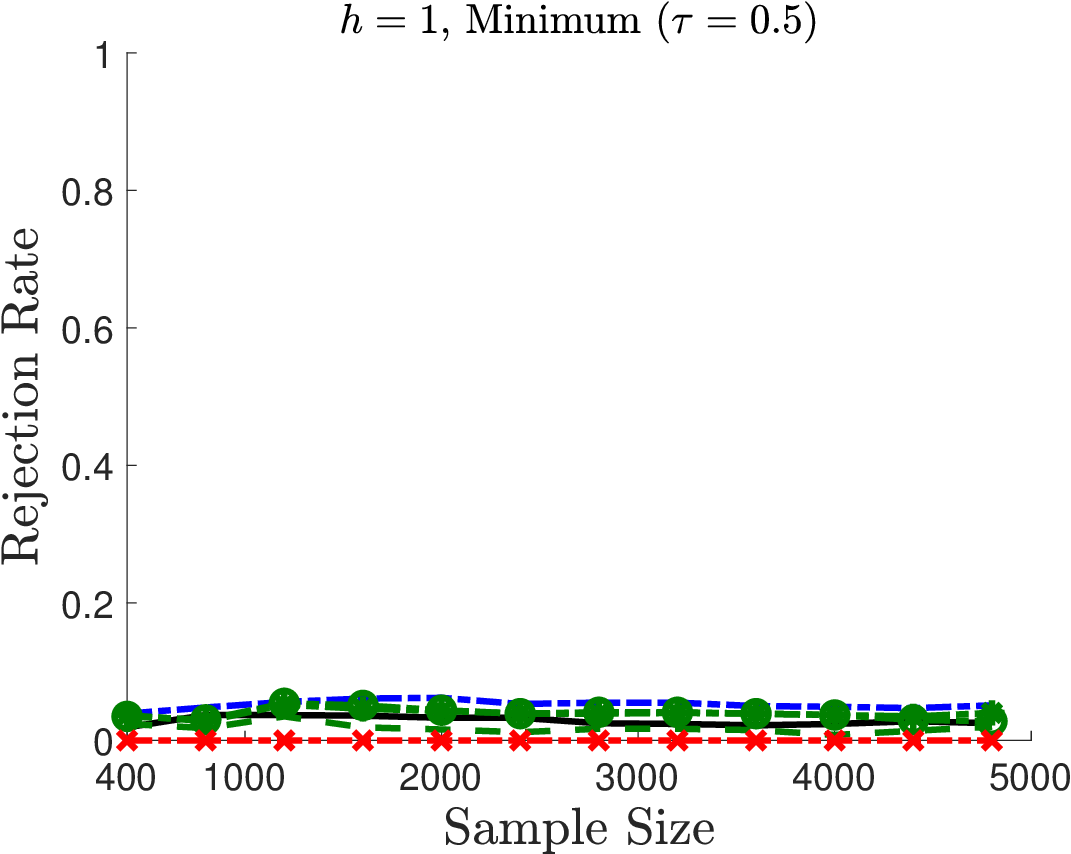}
 \includegraphics[width=0.24\textwidth]{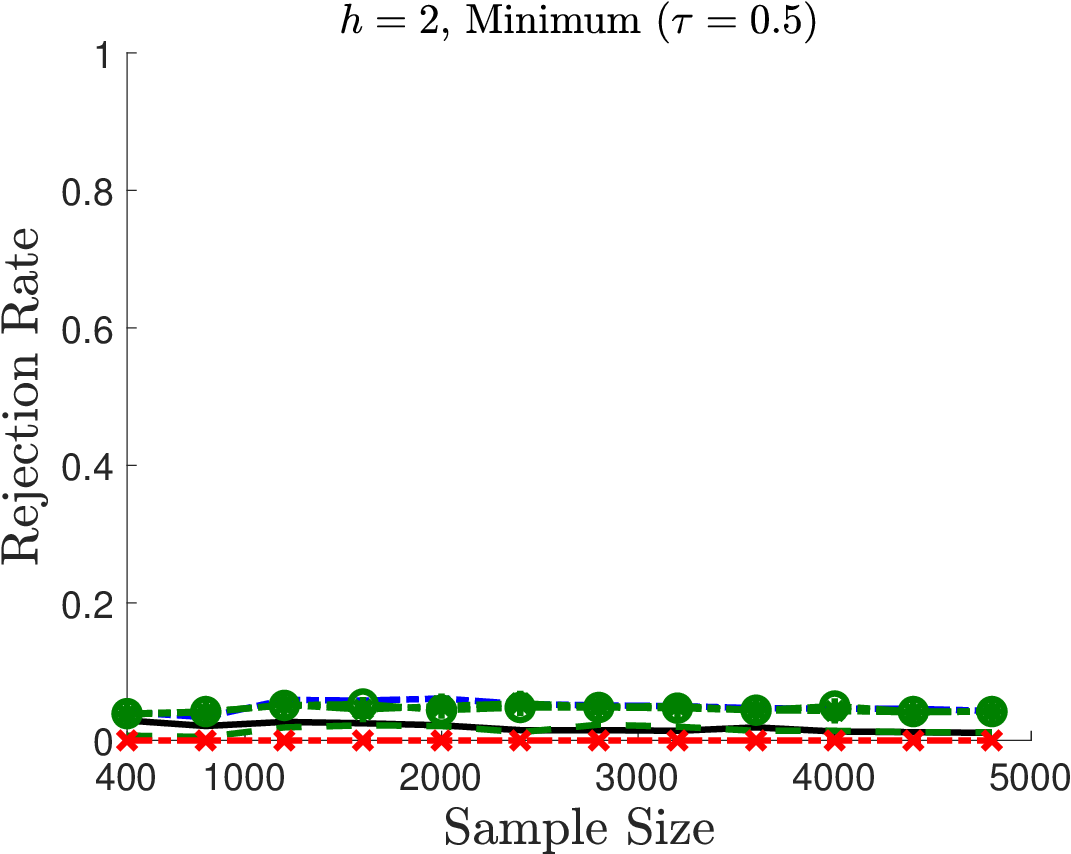}
 \includegraphics[width=0.24\textwidth]{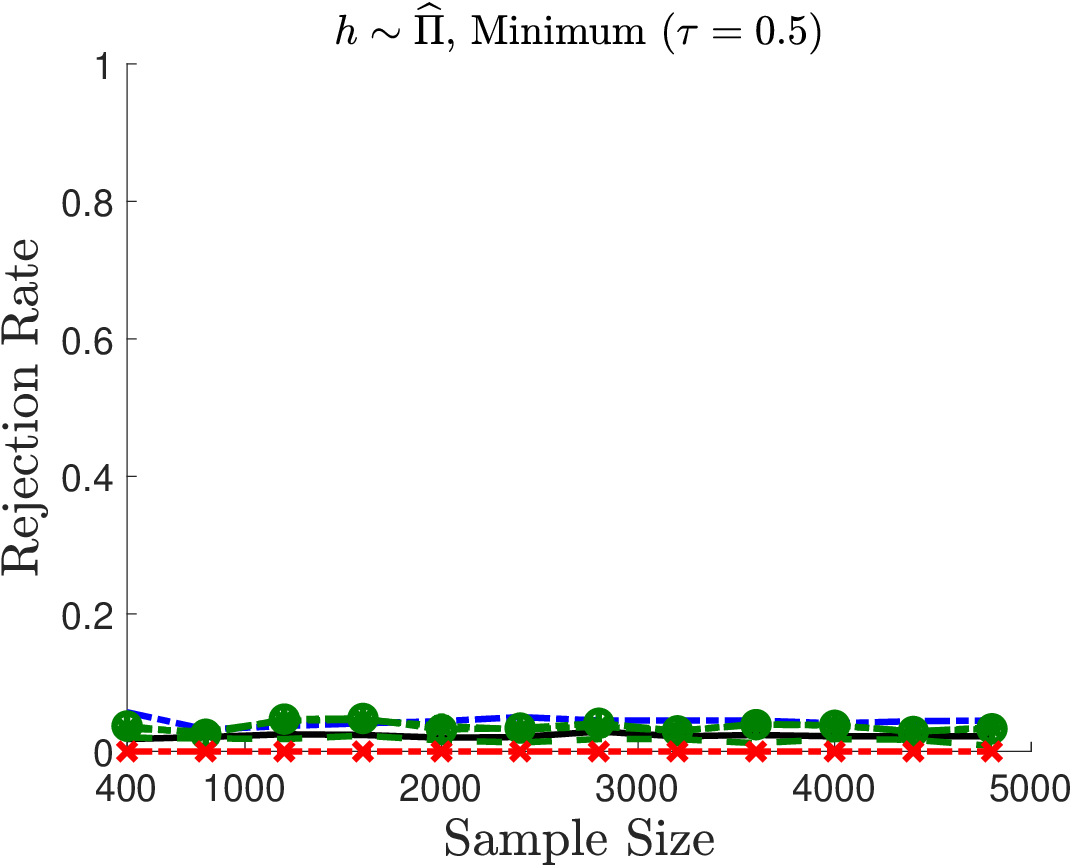}
 
\end{center}
\vspace*{-3mm}
\footnotesize{\textit{Notes:} Figures show how power changes as a function of the sample size for covariate selection, as described in Sections \ref{sec:covariate_selection_MC} (two-sided tests, general-to-specific, $K=3$) and \ref{sec:simulations_setup}. To ensure size control, we use 6 evenly spaced bins (0, 0.025], (0.025, 0.05],$\dots$, (0.125, 0.15] when the sample size is below 1,000 instead of 15 bins when the sample size is at least 1,000. The results are based on 1,000 simulation repetitions.}
\end{figure}

\section{The Relationship between Power and Bias}
\label{app:power_bias}

As discussed in Appendix \ref{app:theory_bias_size_distortions}, $p$-hacking has two types of costs: size distortions when $h=0$ and biased estimates. Here we relate the power of the tests to the costs of $p$-hacking, measured in terms of the relative bias of the estimates (i.e., the bias rescaled by the true effect). We do not show results relating the degree of size distortion to power because the degree of size distortion is linear in $\tau$, so that the power curves for $h=0$ already capture this relationship.

Figure \ref{fig:power_bias} plots the relationship between the relative bias and the power of the tests for covariate selection with $K=3$. We focus on values of $\tau$ for which the power curves of the best tests in Figure \ref{fig:power_cov_combined} are increasing to avoid trivial results. Specifically, we show results for $\tau=0.25$ for the threshold approach and $\tau=0.9$ for the minimum approach, reflecting the much lower power of even the best tests for detecting $p$-hacking based on the minimum approach. We restrict attention to the relative bias implied by $h\in \{0.25,\dots,4.5\}$. We exclude values of $h$ closer to zero (which would lead to a higher relative bias) because the Monte Carlo error makes it increasingly difficult to estimate the relative bias accurately as $h\rightarrow 0$.

The left panel of Figure \ref{fig:power_bias} shows the results for the threshold approach. For the two most powerful tests (CS2B and CSUB), the relationship between the relative bias and power is hump-shaped, with the maximum power being achieved when the relative bias is between 10\% and 20\%. The power of the discontinuity test is increasing below 30\% relative bias and then decreasing. For CS1, the power is initially low and then increases sharply for values of the relative bias above 30\%; the power of the Binomial test is slightly increasing, but the test has no or very low power for all values of the relative bias; the LCM test has no power for the values of the relative bias considered.

The right panel of Figure \ref{fig:power_bias} shows the results for the minimum approach. Only the CSUB and CS2B have power exceeding size for some values of the relative bias. The relationship between relative bias and power is hump-shaped and quite similar for both of these tests. The maximum power is lower and achieved at higher values of the relative bias than under thresholding. 

\begin{figure}[H]

\caption{Power vs. relative bias covariate selection  with $K=3$}\label{fig:power_bias}

\vspace{-5mm}

\begin{center}

\includegraphics[width=0.49\textwidth]{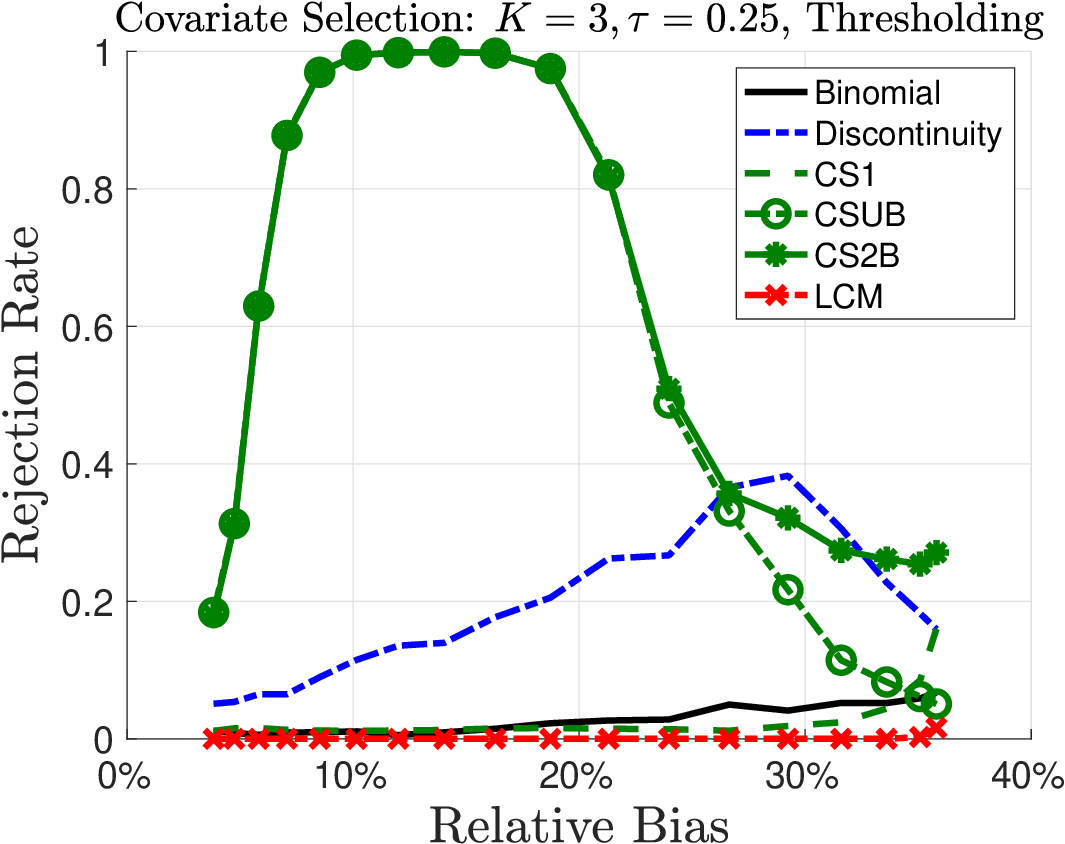}
\includegraphics[width=0.49\textwidth]{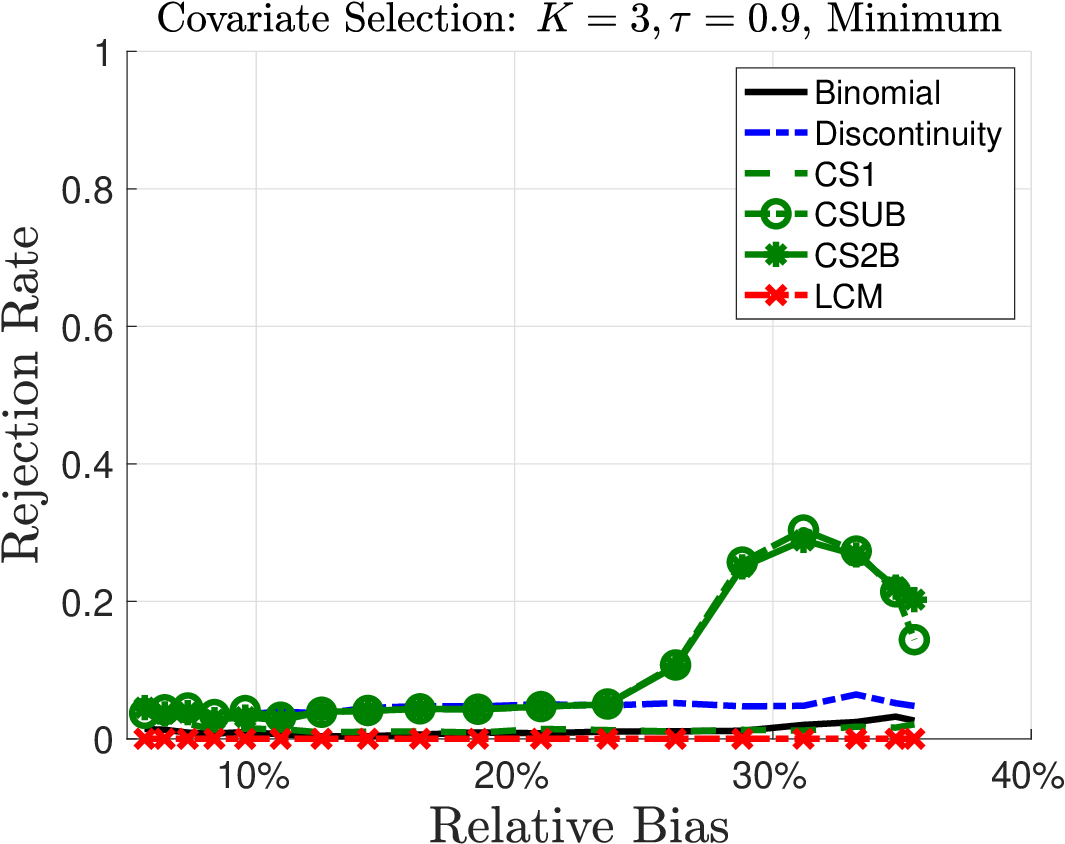}

\end{center}
\vspace*{-3mm}
\footnotesize{\textit{Notes:} Figures show rejection rates of the tests in Table \ref{tab:tests} as a function of the relative bias implied by $h\in\{0.25,0.50,\dots,4.5\}$ and a $p$-hacking approach (thresholding or minimum). The simulation design is described in Sections \ref{sec:covariate_selection_MC} (two-sided tests, general-to-specific, $K=3$) and \ref{sec:simulations_setup}. The results are based on 5,000 simulation repetitions.}
\end{figure}

\section{Additional Simulation Results}
\label{app:additional_simulation_results}

 \begin{figure}[H]

\caption{Power curves covariate selection with $K=3$ (one-sided tests, general-to-specific)}\label{fig:power_cov_K3}

\vspace{-5mm}
 
\begin{center}
\textbf{Thresholding}


\includegraphics[width=0.24\textwidth]{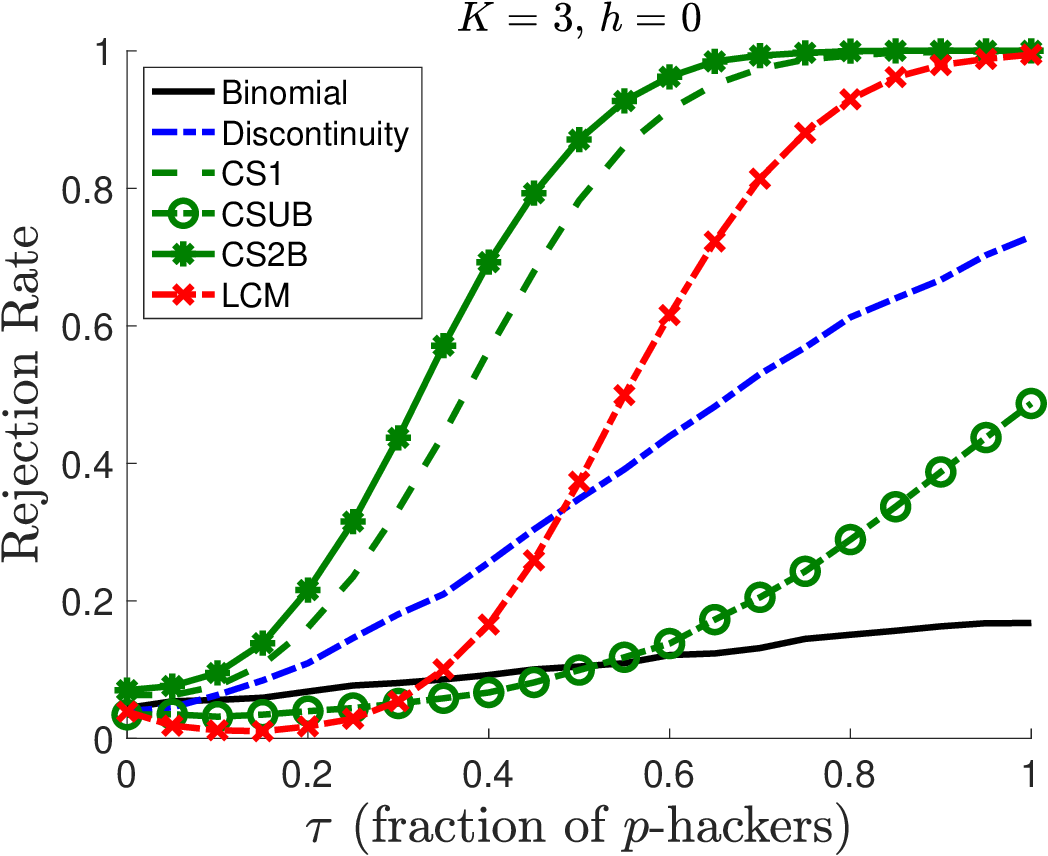}
\includegraphics[width=0.24\textwidth]{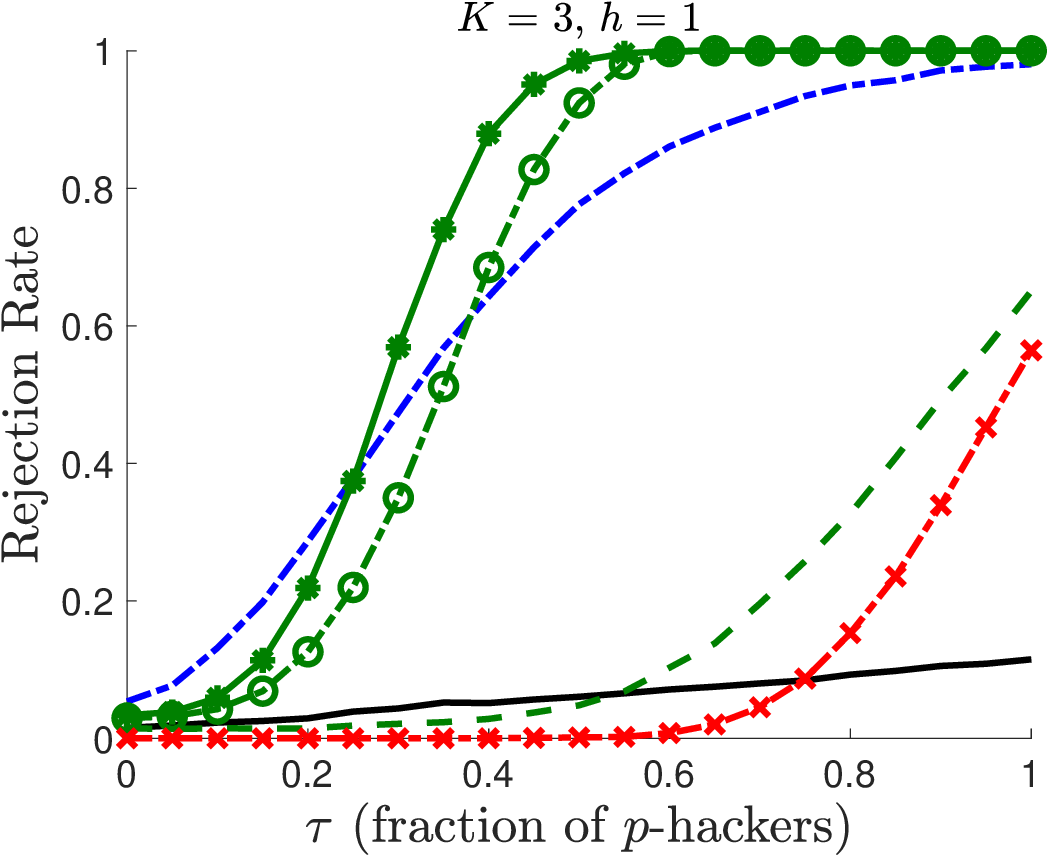}
\includegraphics[width=0.24\textwidth]{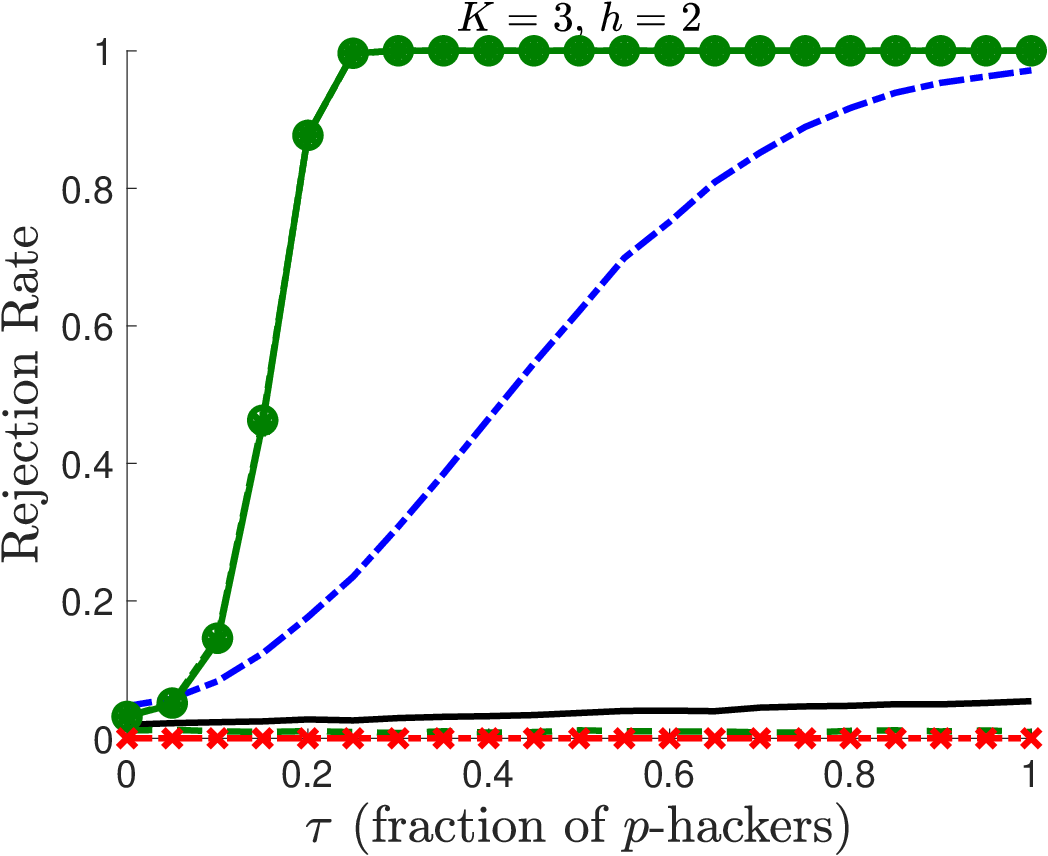}
\includegraphics[width=0.24\textwidth]{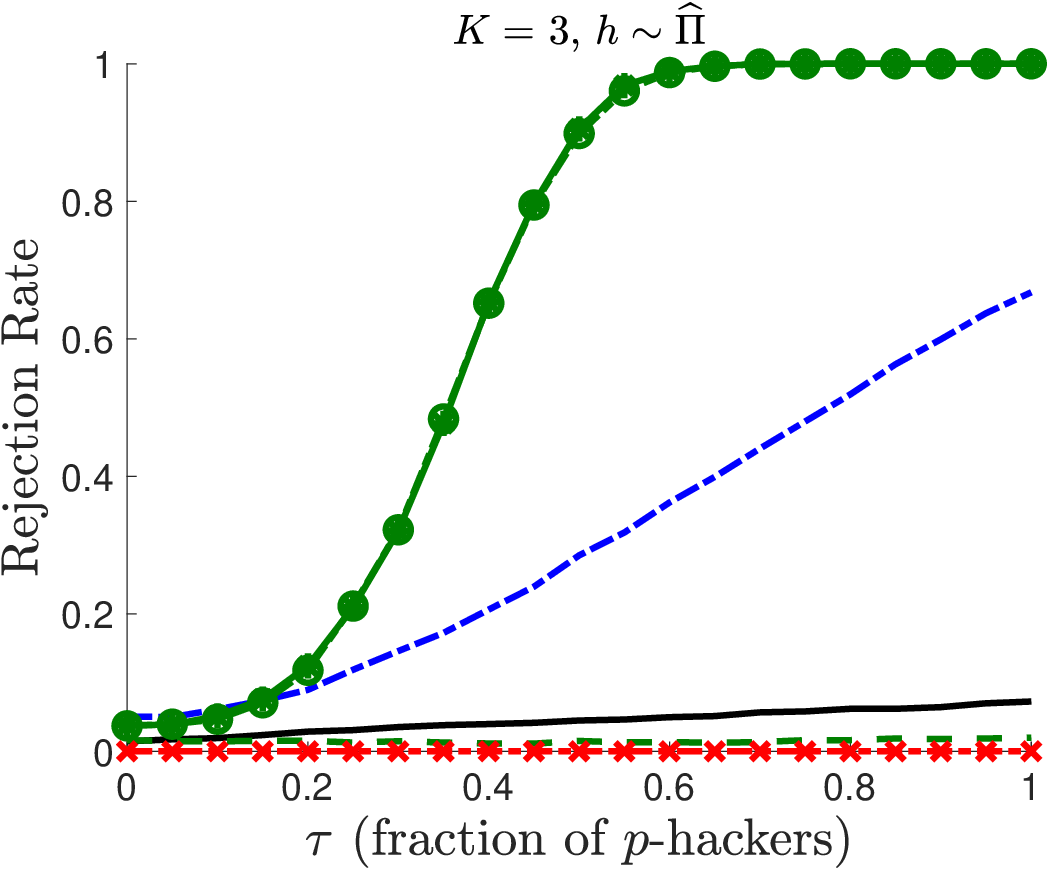}

\textbf{Minimum}
\smallskip

\includegraphics[width=0.24\textwidth]{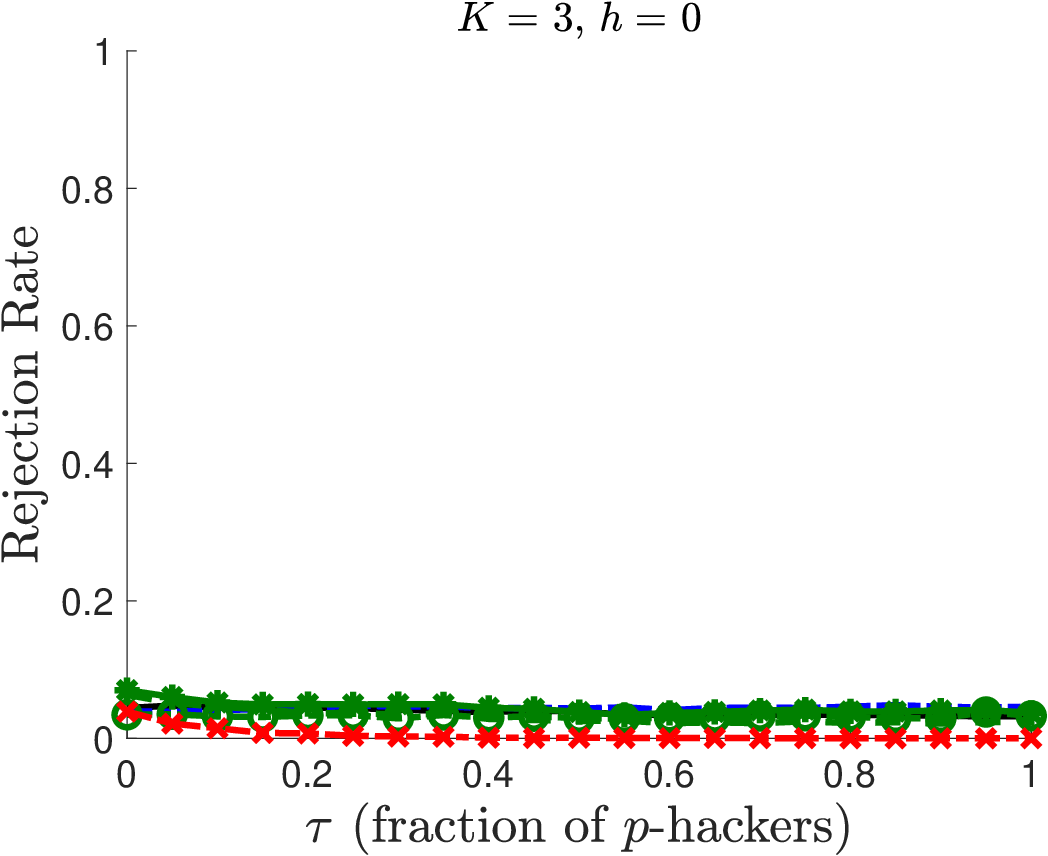}
\includegraphics[width=0.24\textwidth]{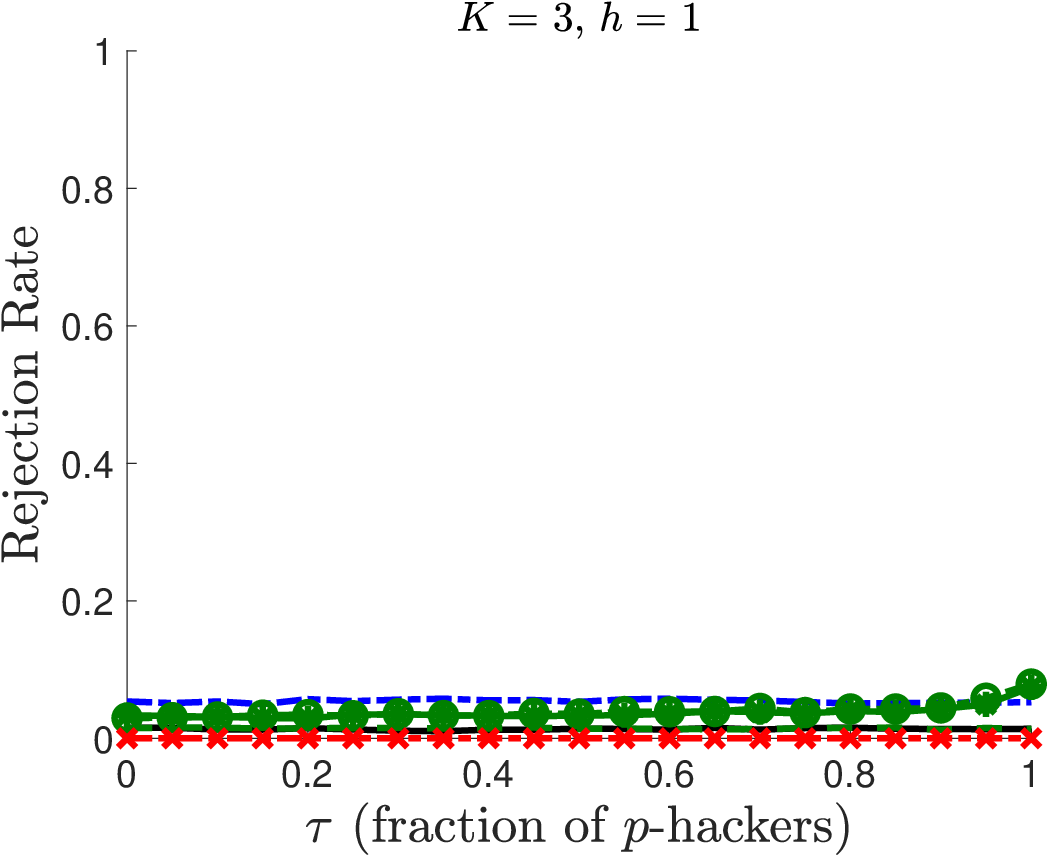}
\includegraphics[width=0.24\textwidth]{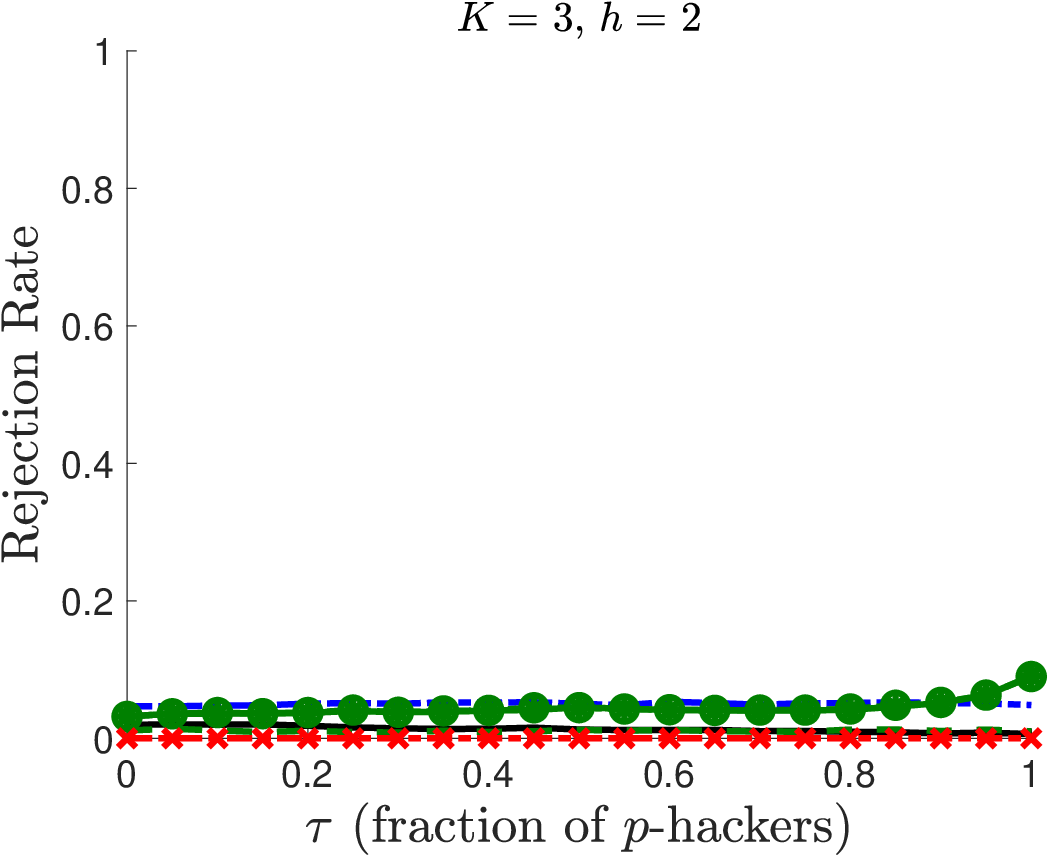}
\includegraphics[width=0.24\textwidth]{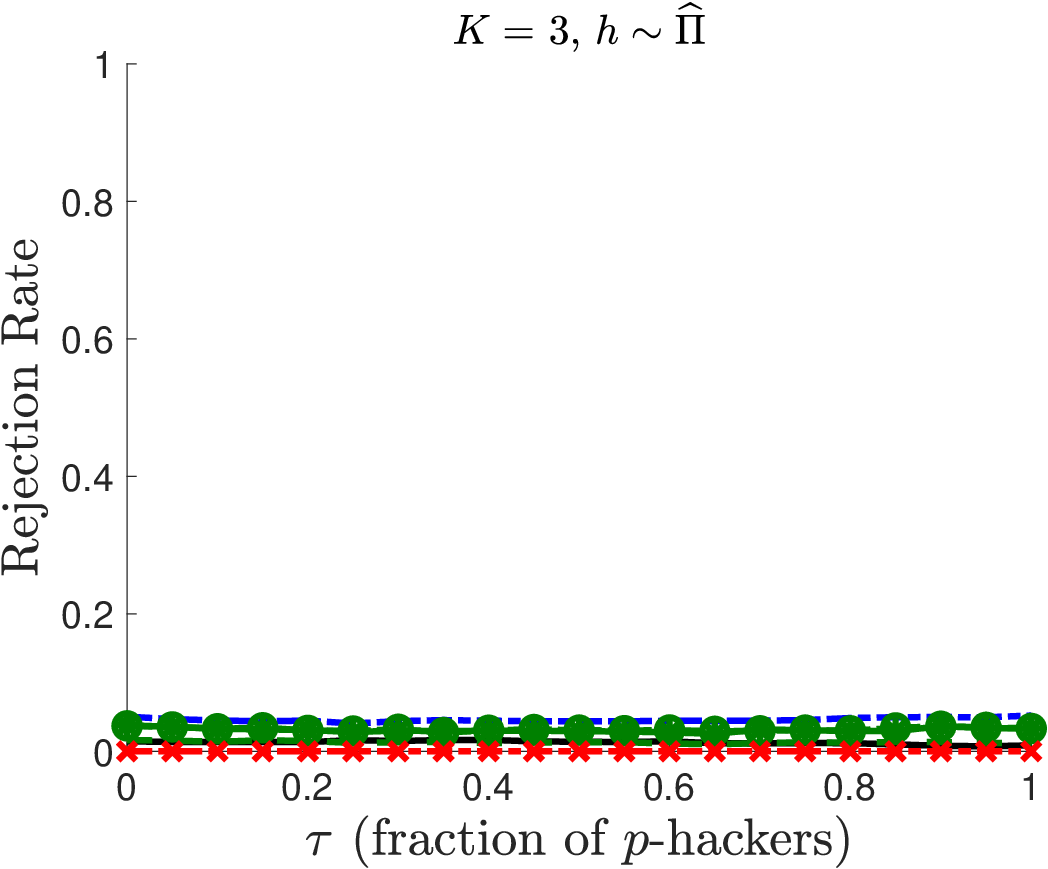}
\end{center}
\vspace*{-3mm}

\footnotesize{\textit{Notes:} Figures show rejection rates of the tests in Table \ref{tab:tests} as a function of $\tau$ for the threshold and minimum approach to covariate selection (general-to-specific) with $K=3$ and one-sided tests. The simulation design is described in Sections \ref{sec:covariate_selection_MC} and \ref{sec:simulations_setup}. The results are based on 5,000 simulation repetitions.}

\end{figure}

 \begin{figure}[H]
\caption{Power curves covariate selection with $K=5$ (two-sided tests, general-to-specific)}\label{fig:power_cov_K5_2sided}

\vspace{-5mm}
 
\begin{center}
\textbf{Thresholding}

\smallskip

\includegraphics[width=0.24\textwidth]{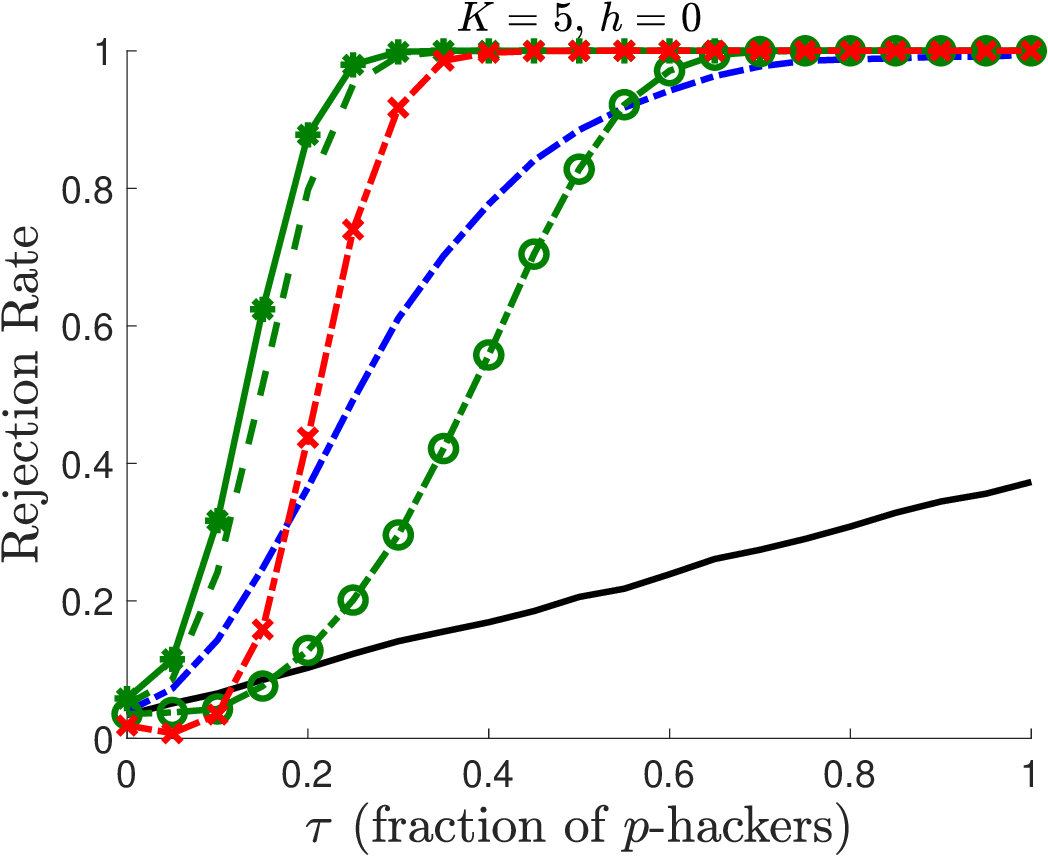}
\includegraphics[width=0.24\textwidth]{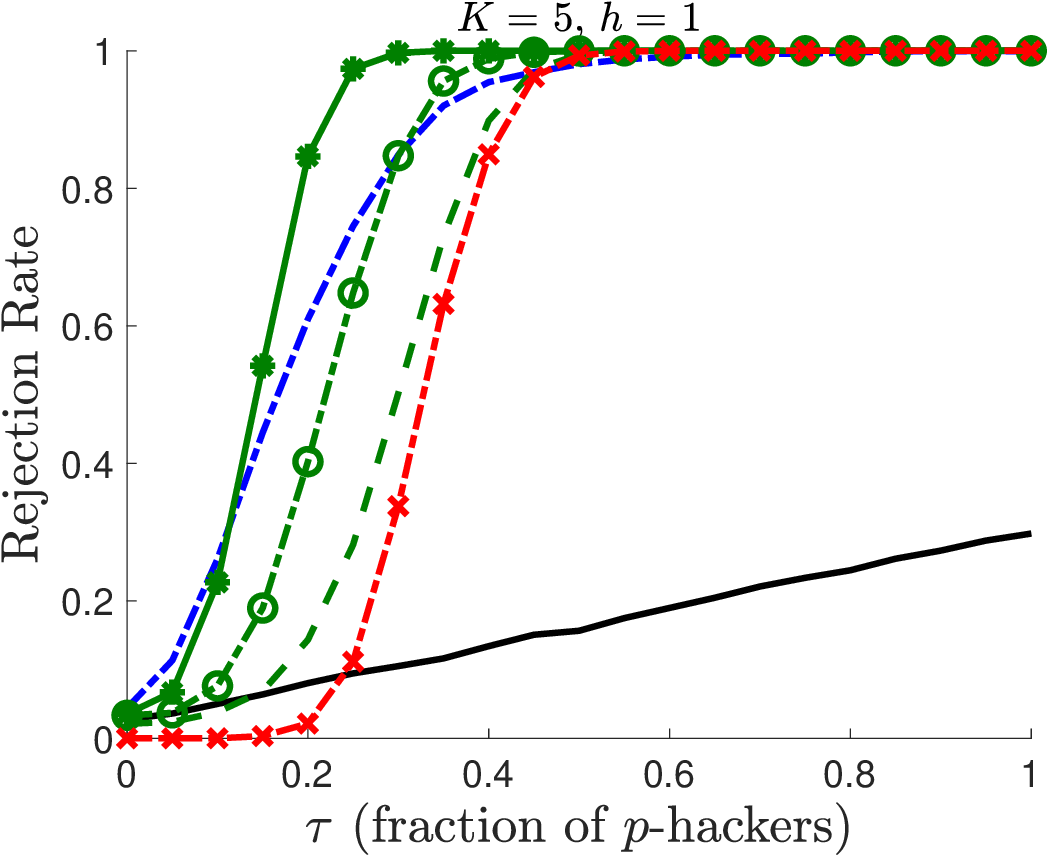}
\includegraphics[width=0.24\textwidth]{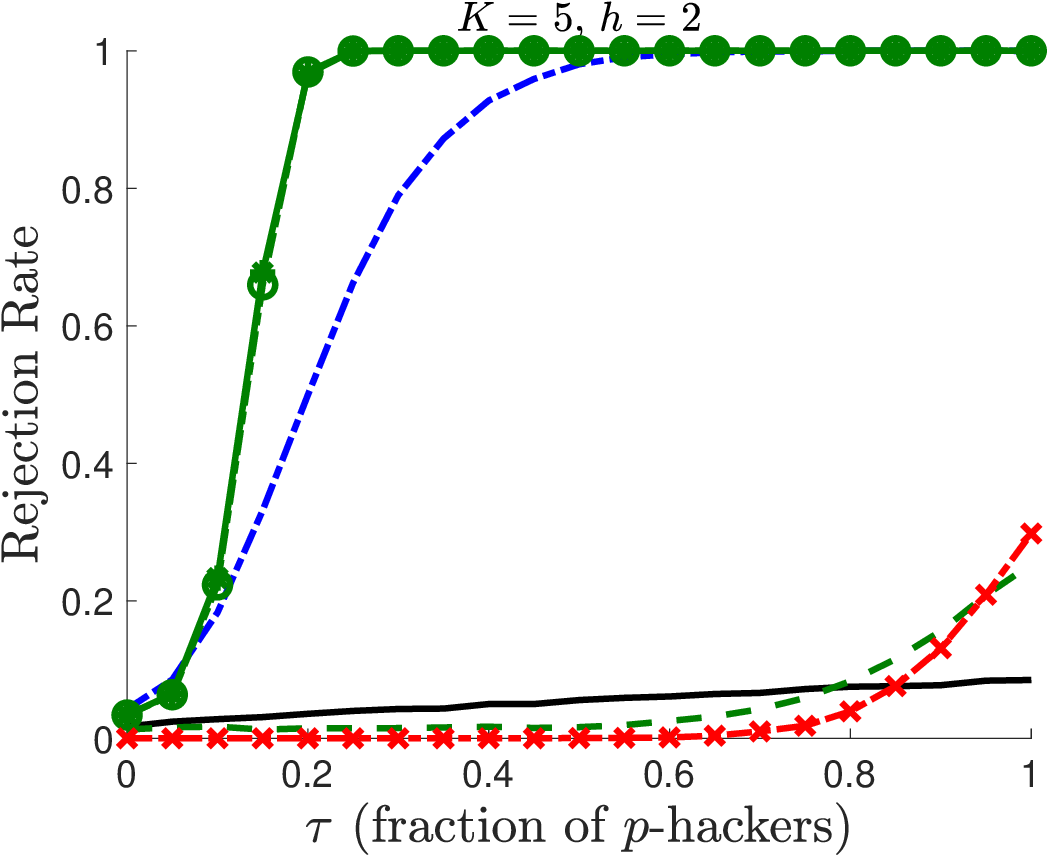}
\includegraphics[width=0.24\textwidth]{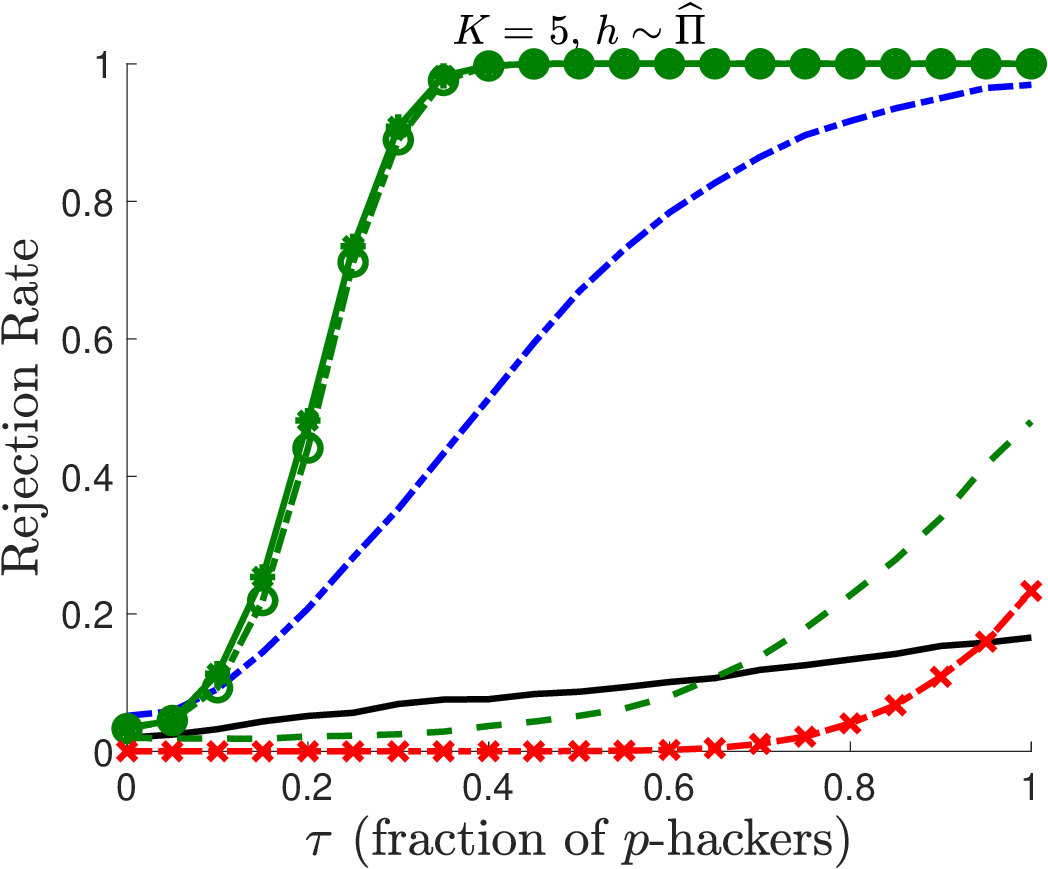}

\textbf{Minimum}
\smallskip

\includegraphics[width=0.24\textwidth]{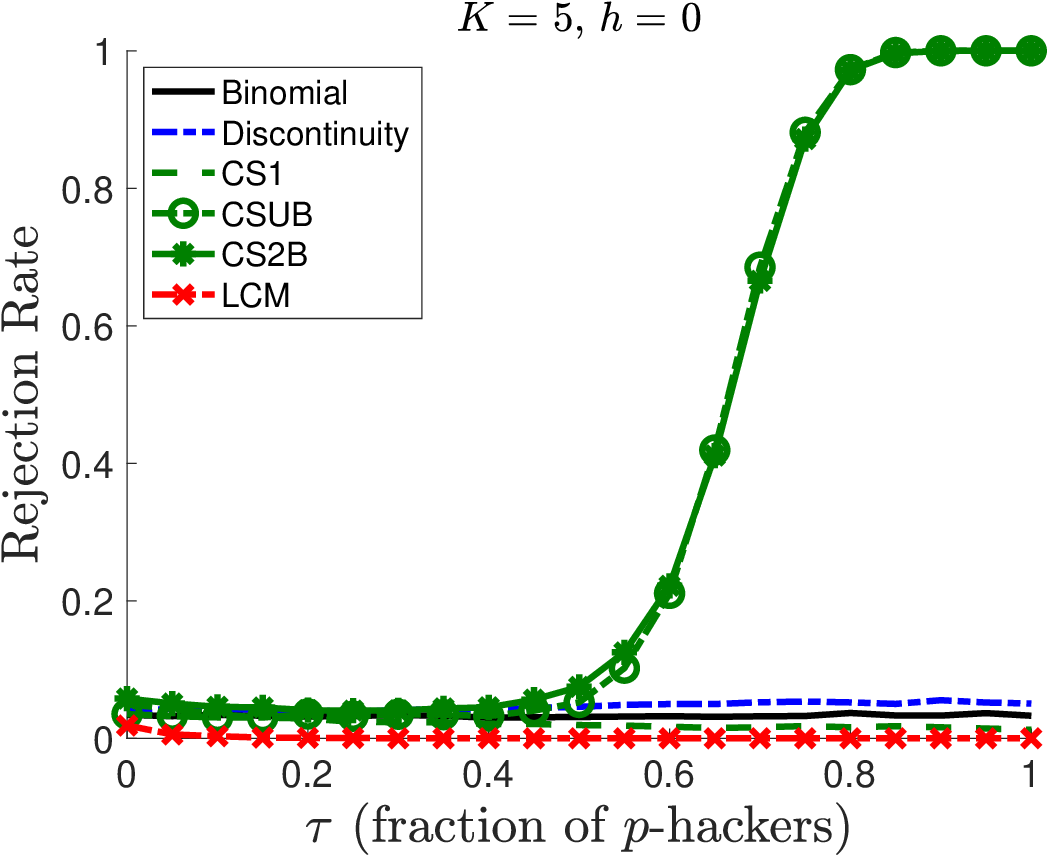}
\includegraphics[width=0.24\textwidth]{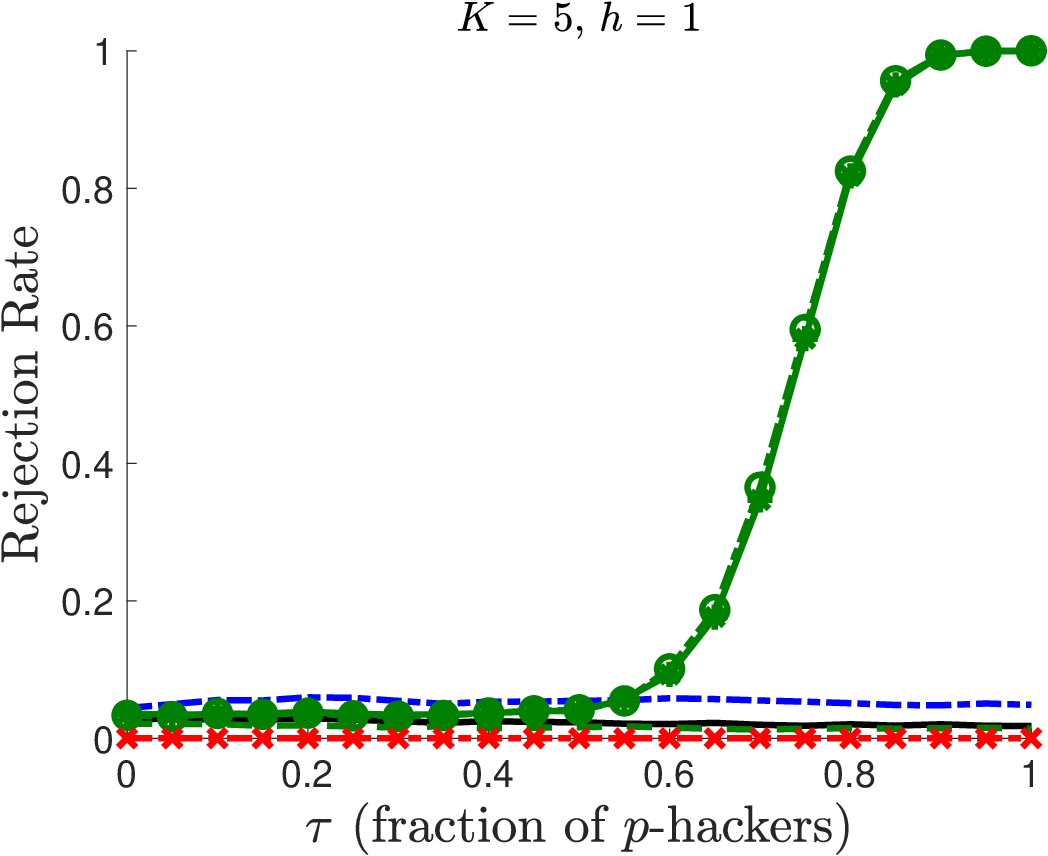}
\includegraphics[width=0.24\textwidth]{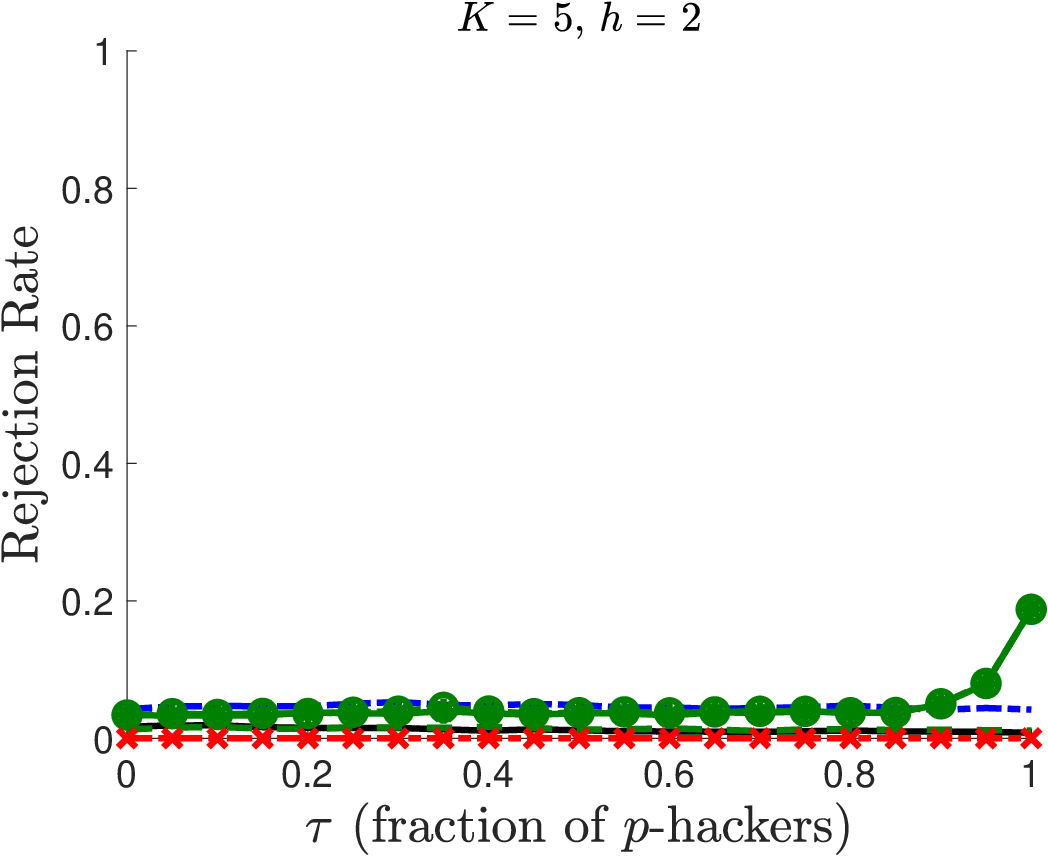}
\includegraphics[width=0.24\textwidth]{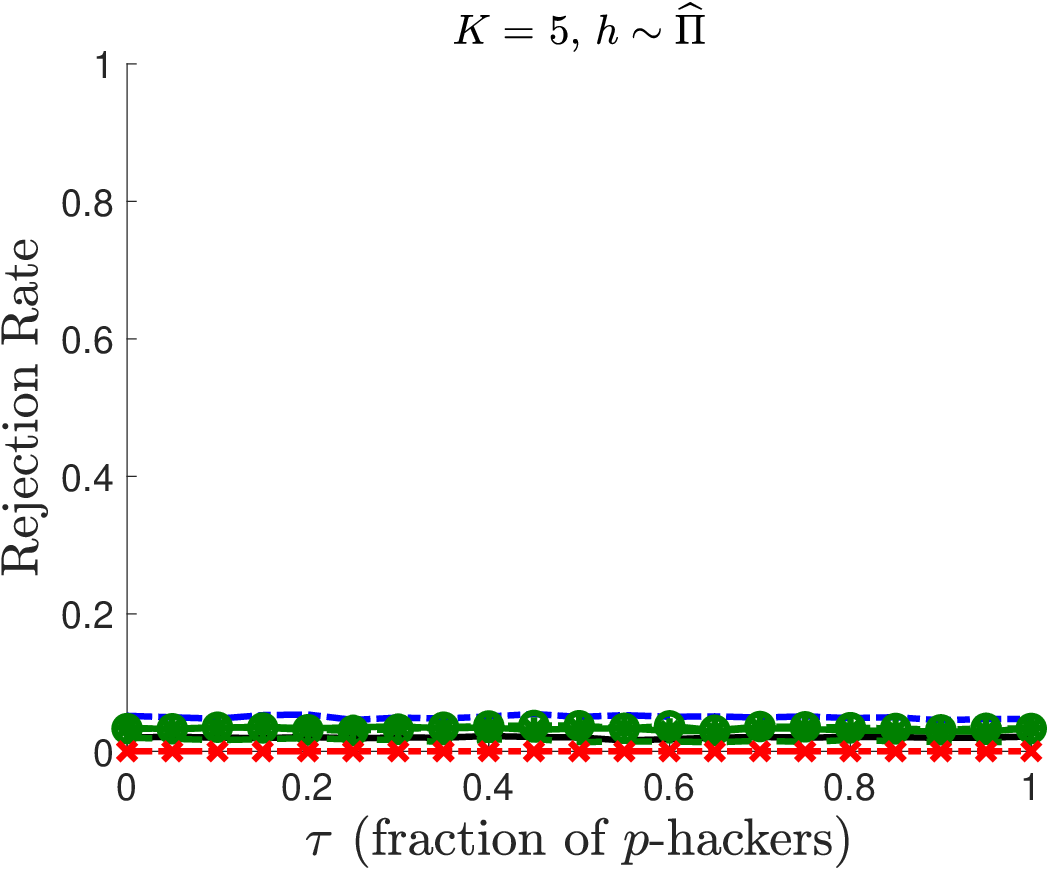}
\end{center}
\vspace*{-3mm}

\footnotesize{\textit{Notes:} Figures show rejection rates of the tests in Table \ref{tab:tests} as a function of $\tau$ for the threshold and minimum approach to covariate selection (general-to-specific) with $K=5$ and two-sided tests. The simulation design is described in Sections \ref{sec:covariate_selection_MC} and \ref{sec:simulations_setup}. The results are based on 5,000 simulation repetitions.}

\end{figure}

 \begin{figure}[H]
\begin{center}

\caption{Power curves covariate selection with $K=7$ (two-sided tests, general-to-specific)}\label{fig:power_cov_K7_2sided}


\textbf{Thresholding}

\smallskip

\includegraphics[width=0.24\textwidth]{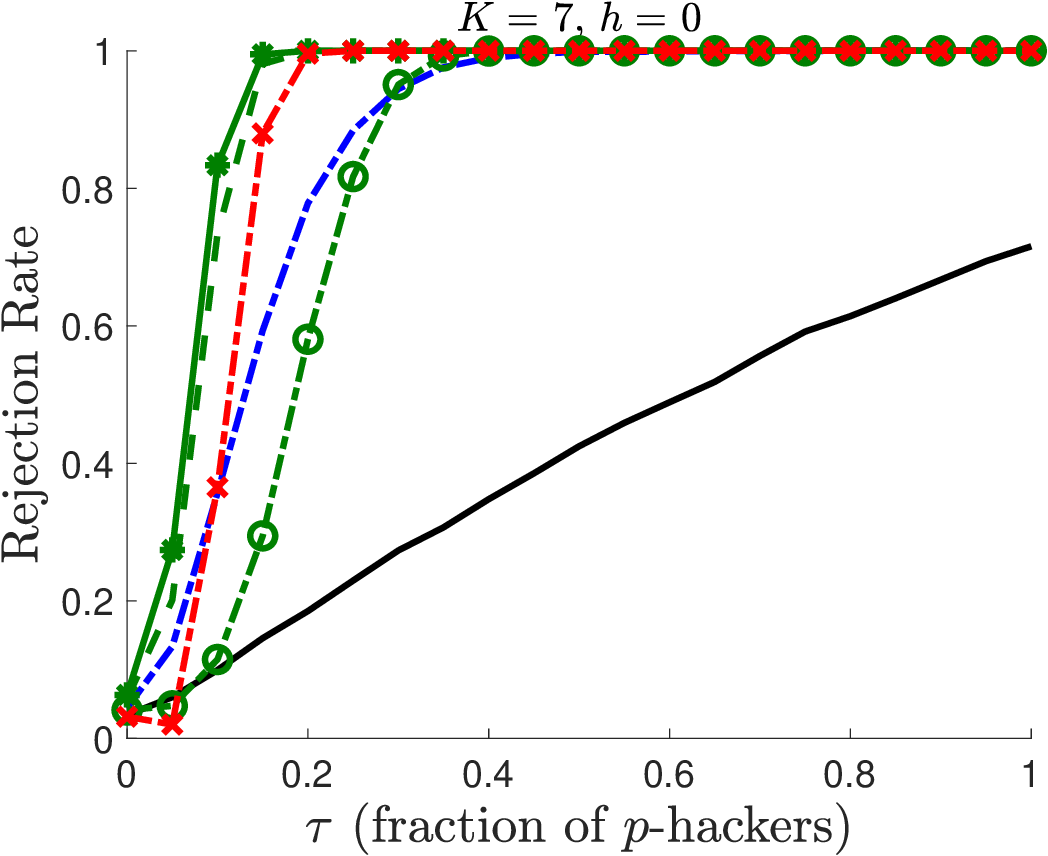}
\includegraphics[width=0.24\textwidth]{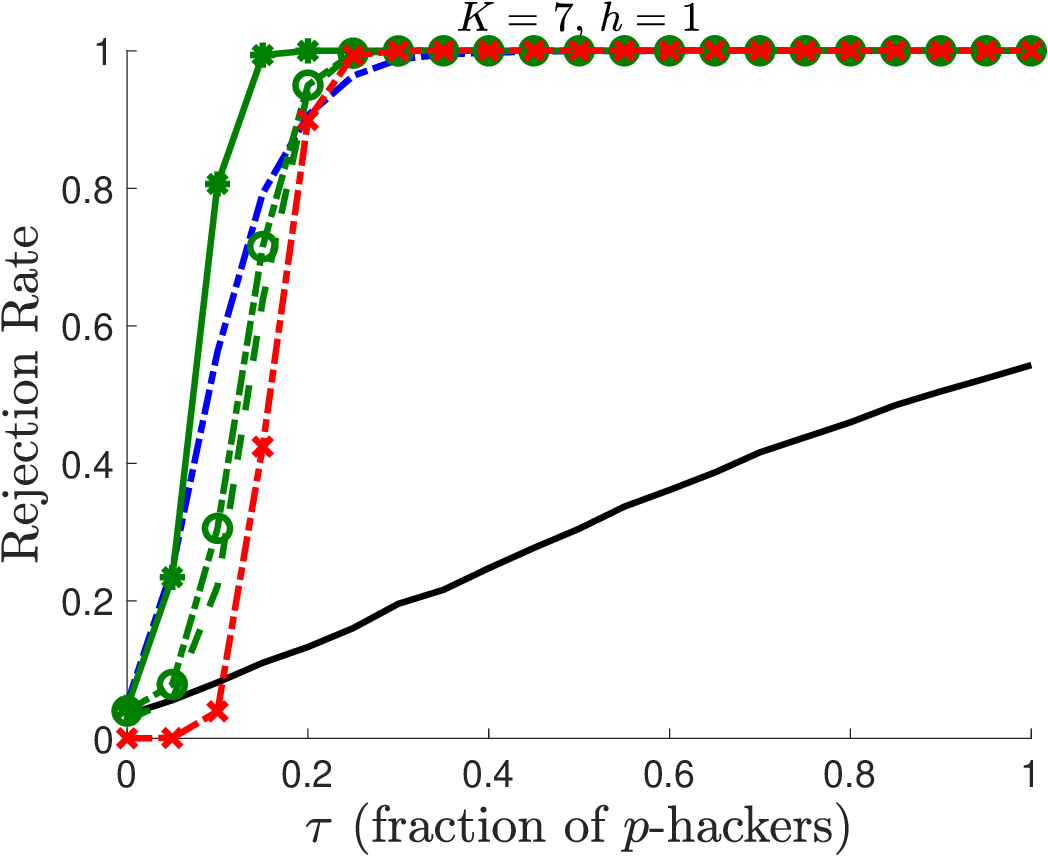}
\includegraphics[width=0.24\textwidth]{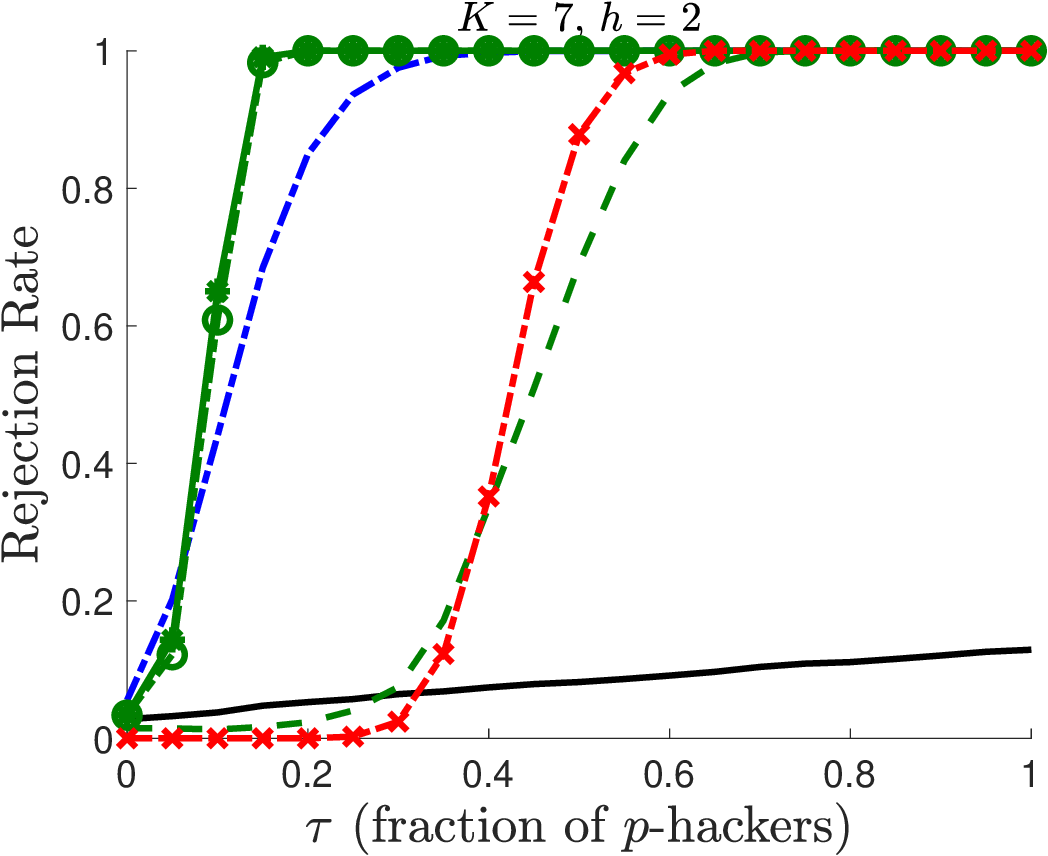}
\includegraphics[width=0.24\textwidth]{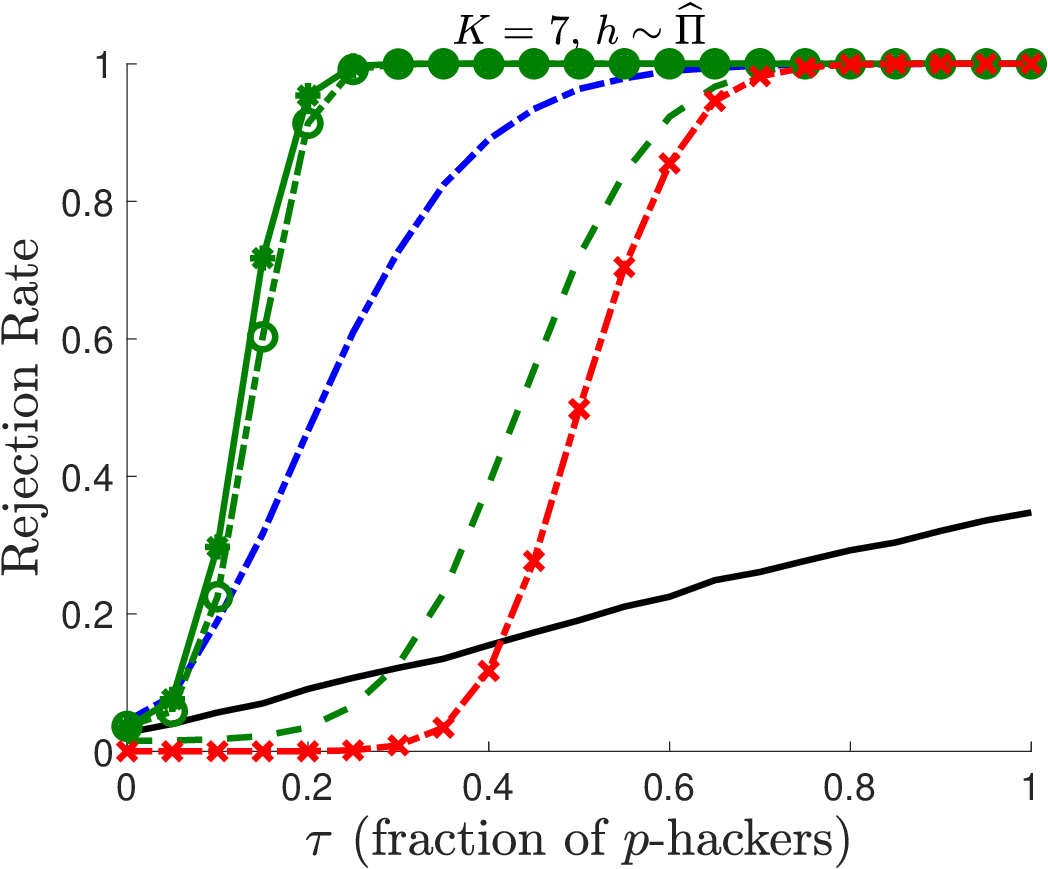}

\textbf{Minimum}
\smallskip

\includegraphics[width=0.24\textwidth]{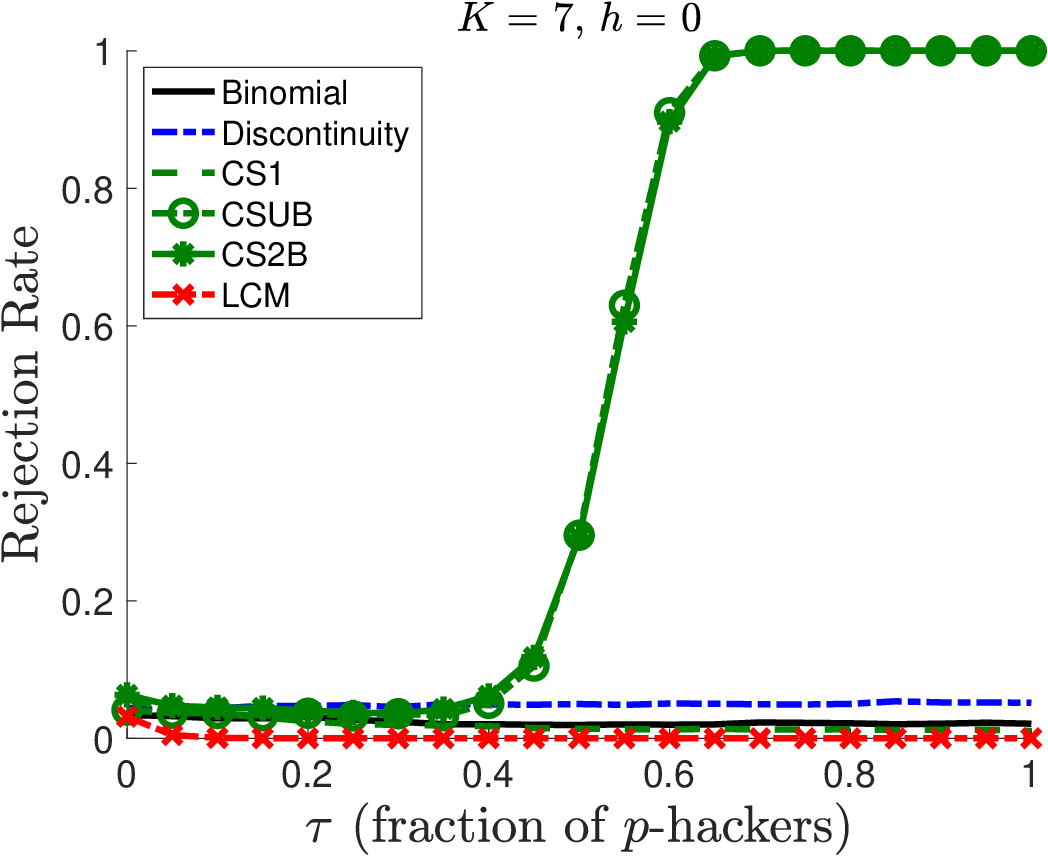}
\includegraphics[width=0.24\textwidth]{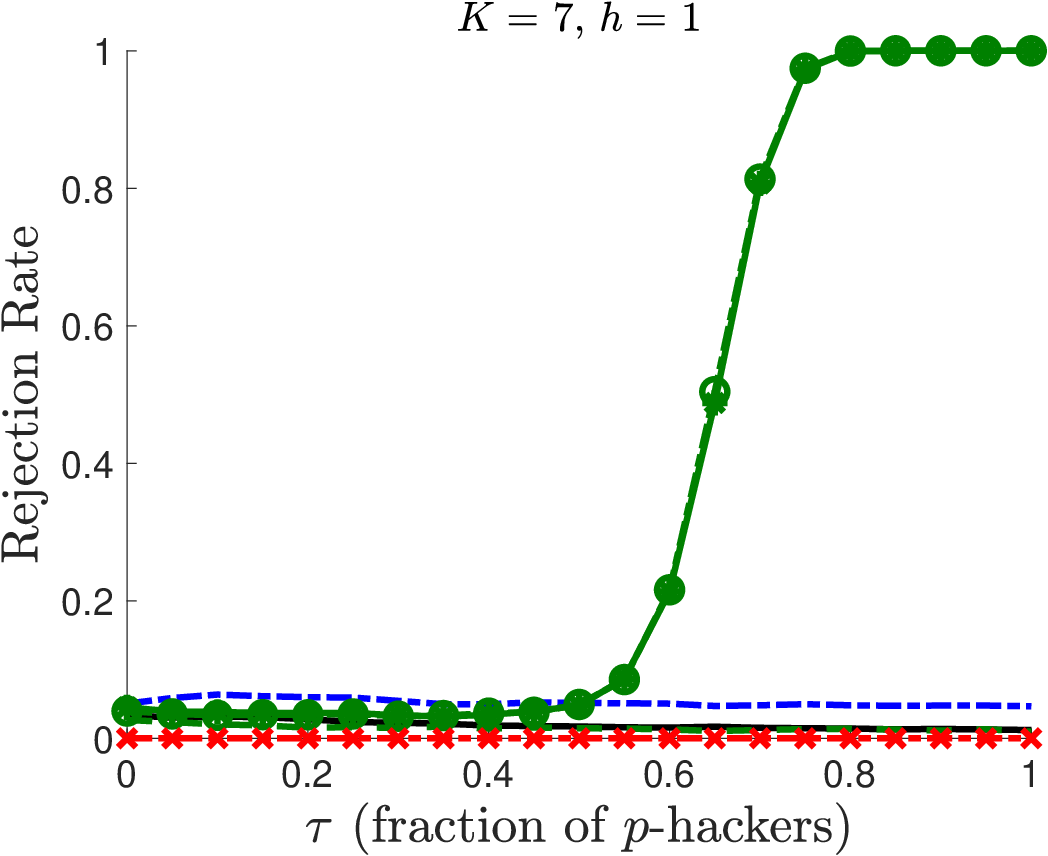}
\includegraphics[width=0.24\textwidth]{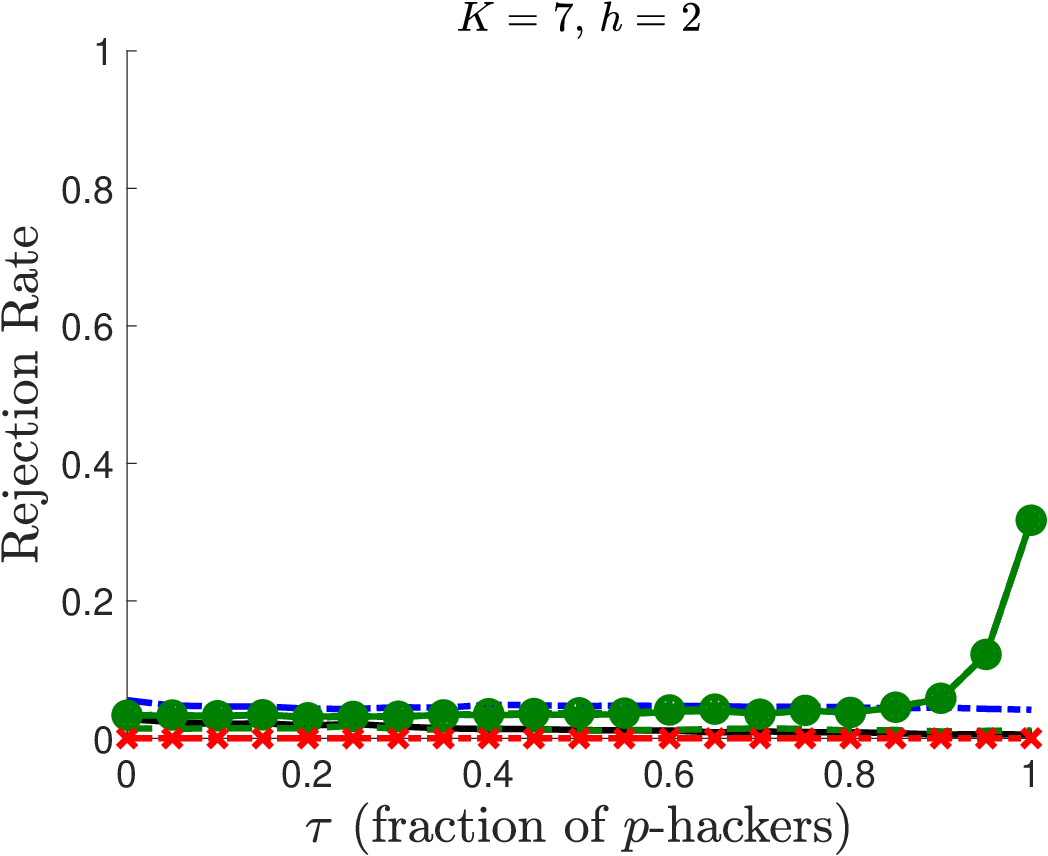}
\includegraphics[width=0.24\textwidth]{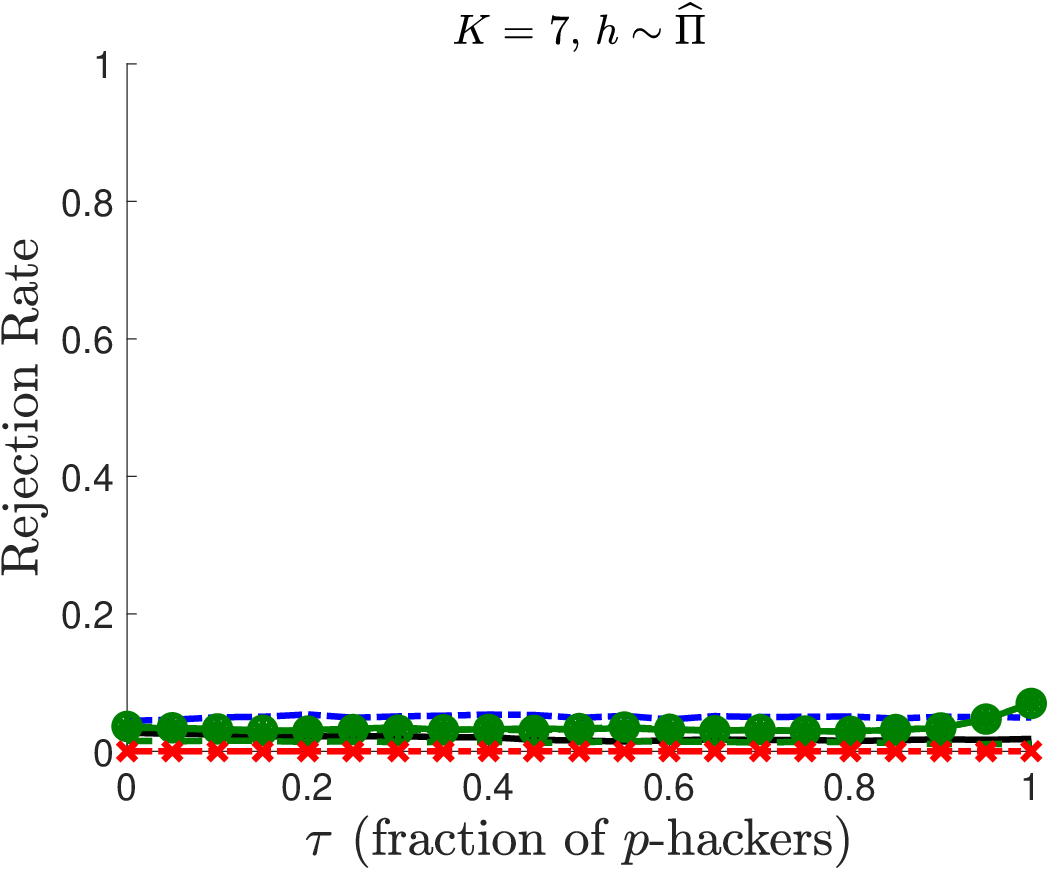}
\end{center}
\vspace*{-3mm}

\footnotesize{\textit{Notes:} Figures show rejection rates of the tests in Table \ref{tab:tests} as a function of $\tau$ for the threshold and minimum approach to covariate selection (general-to-specific) with $K=7$ and two-sided tests. The simulation design is described in Sections \ref{sec:covariate_selection_MC} and \ref{sec:simulations_setup}. The results are based on 5,000 simulation repetitions.}

\end{figure}

\begin{figure}[H]
\caption{Power curves IV selection with $K=5$ (all specifications)}\label{fig:power_iv_K5}

\vspace{-5mm}

\begin{center}
\textbf{Thresholding}


\includegraphics[width=0.24\textwidth]{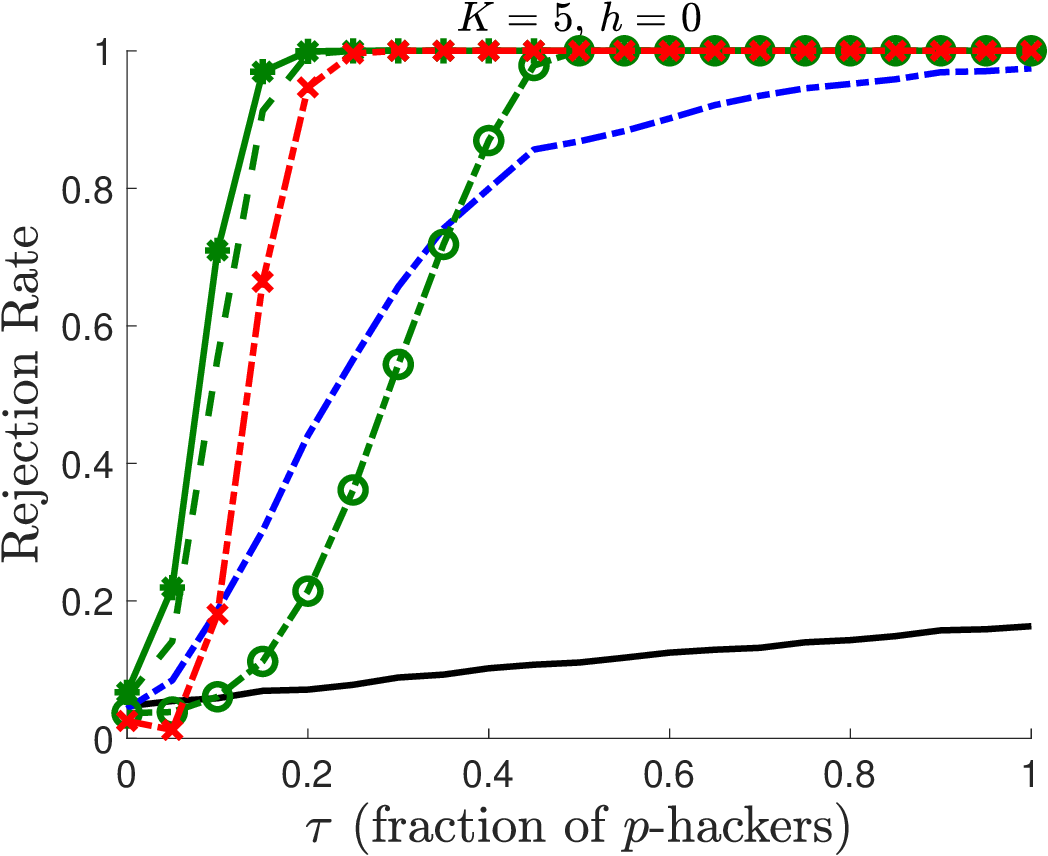}
\includegraphics[width=0.24\textwidth]{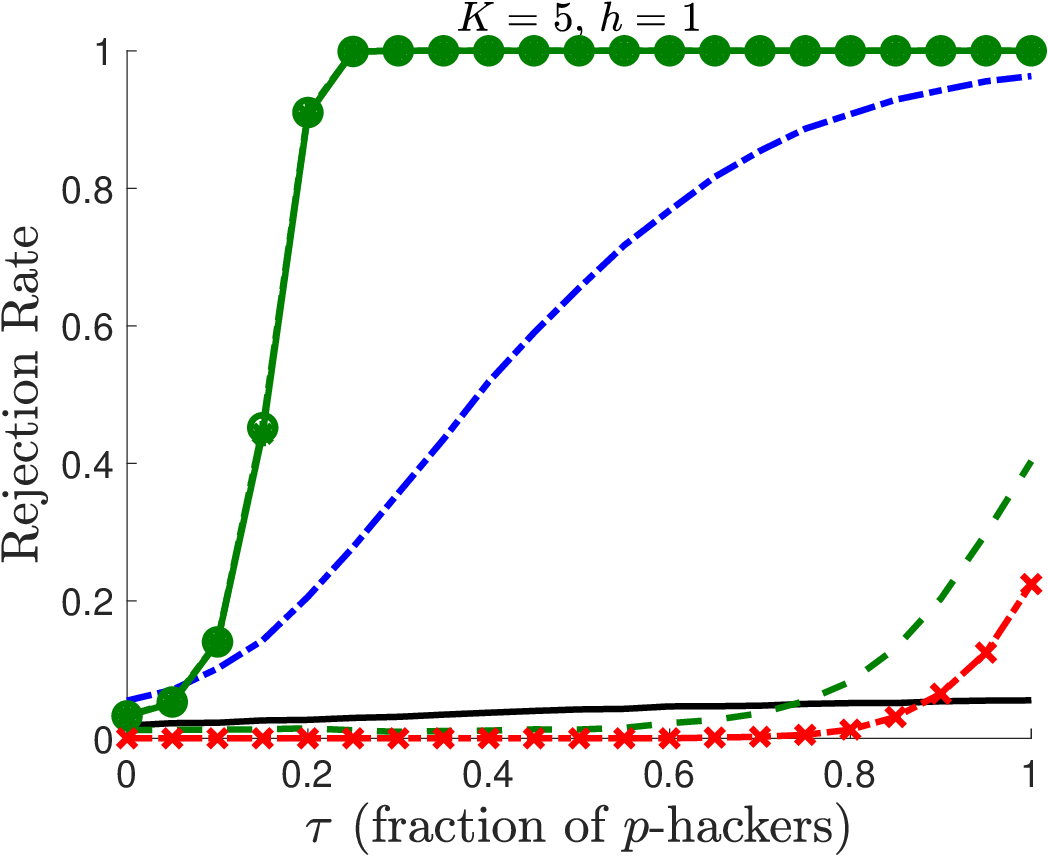}
\includegraphics[width=0.24\textwidth]{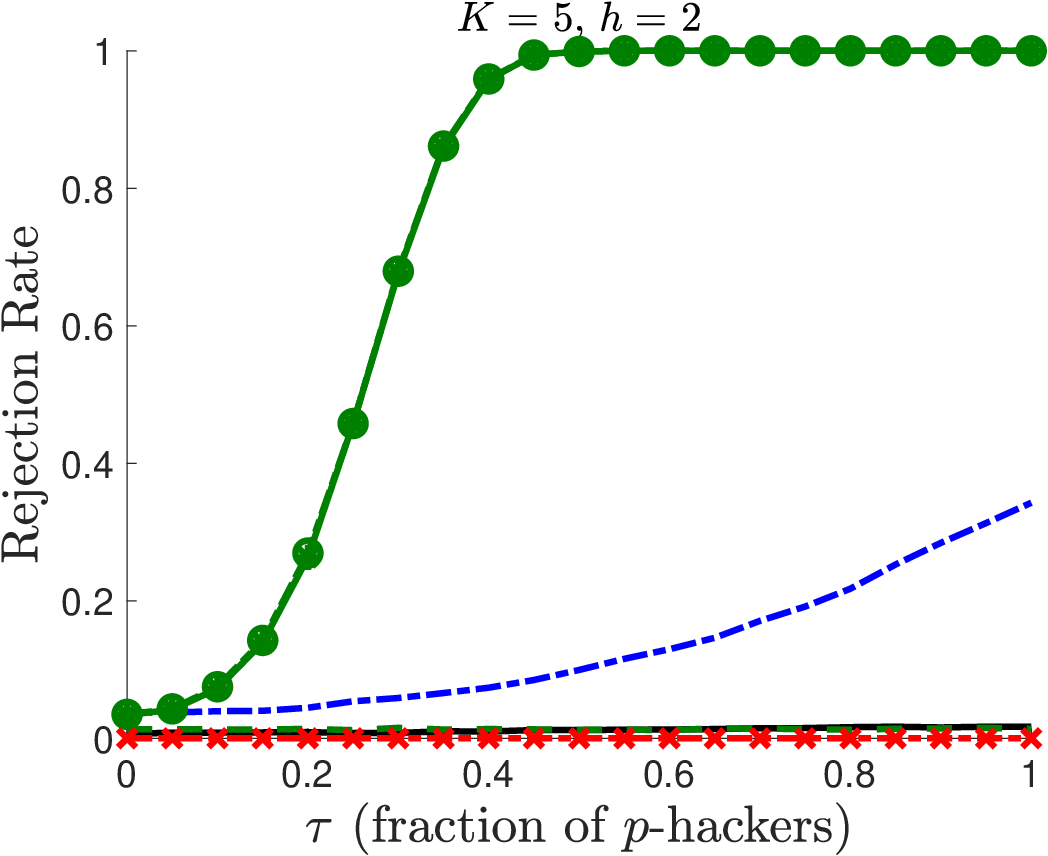}
\includegraphics[width=0.24\textwidth]{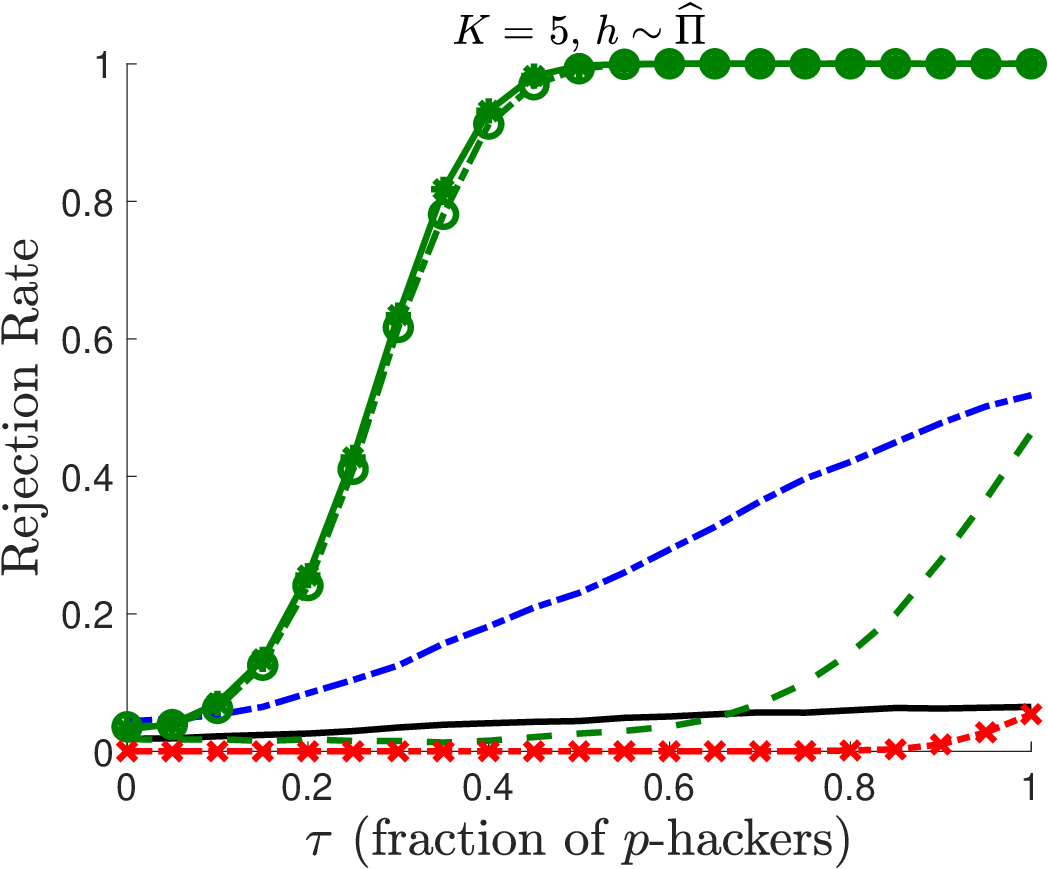}

\textbf{Minimum}

\vspace{0.15em}

\includegraphics[width=0.24\textwidth]{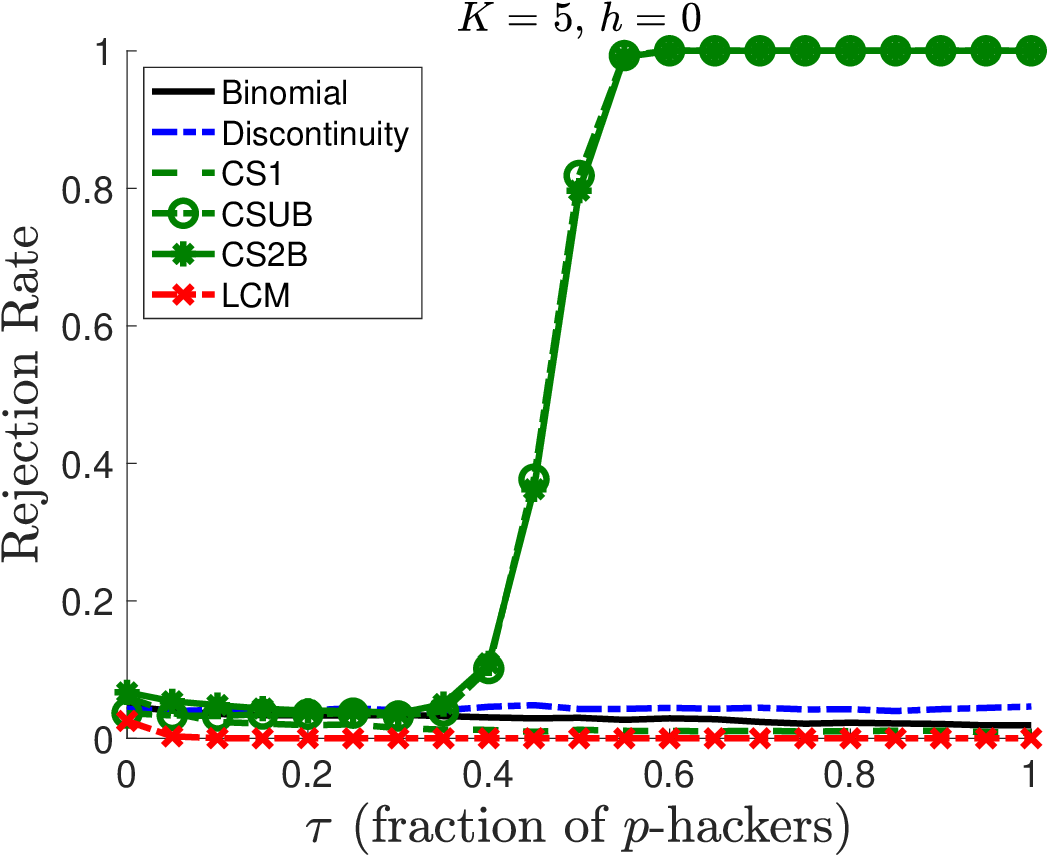}
\includegraphics[width=0.24\textwidth]{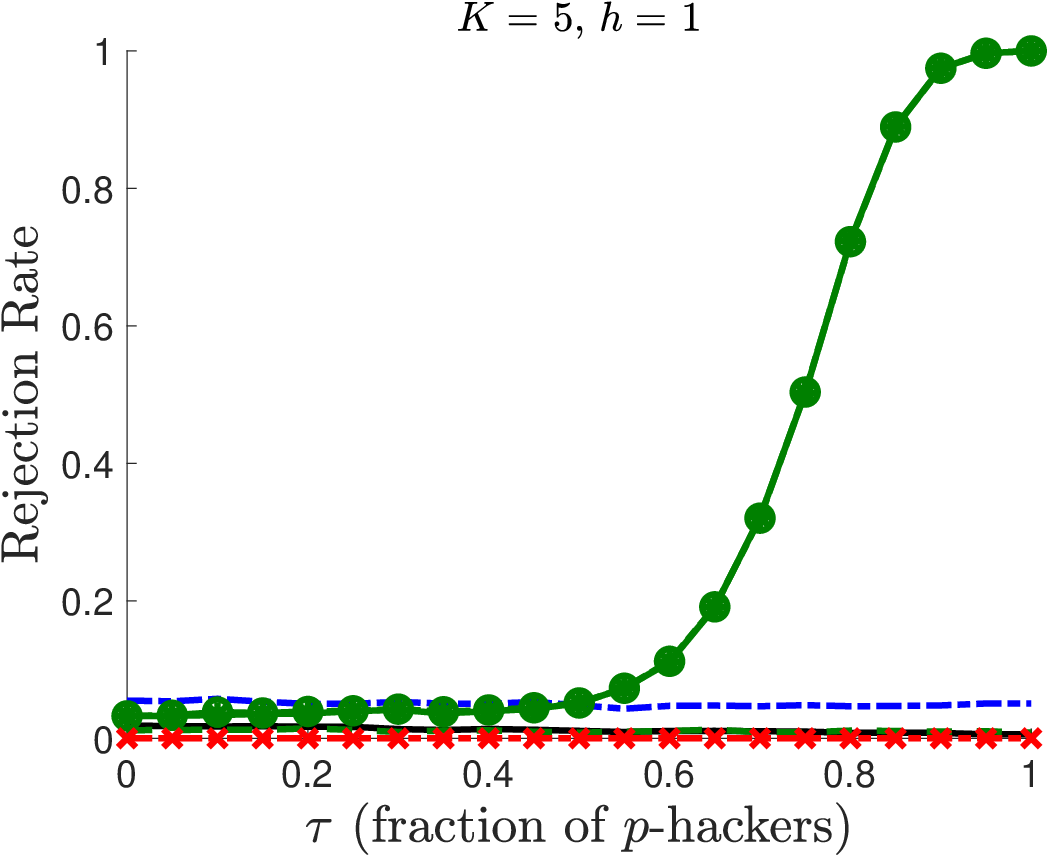}
\includegraphics[width=0.24\textwidth]{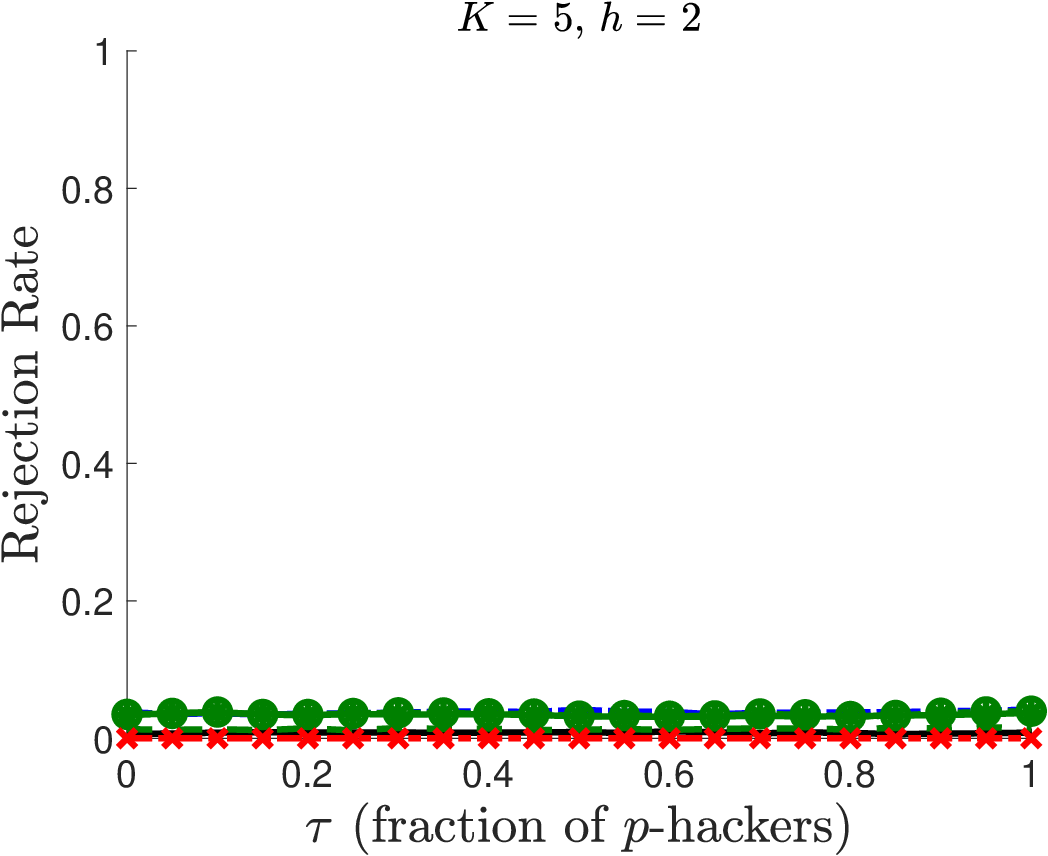}
\includegraphics[width=0.24\textwidth]{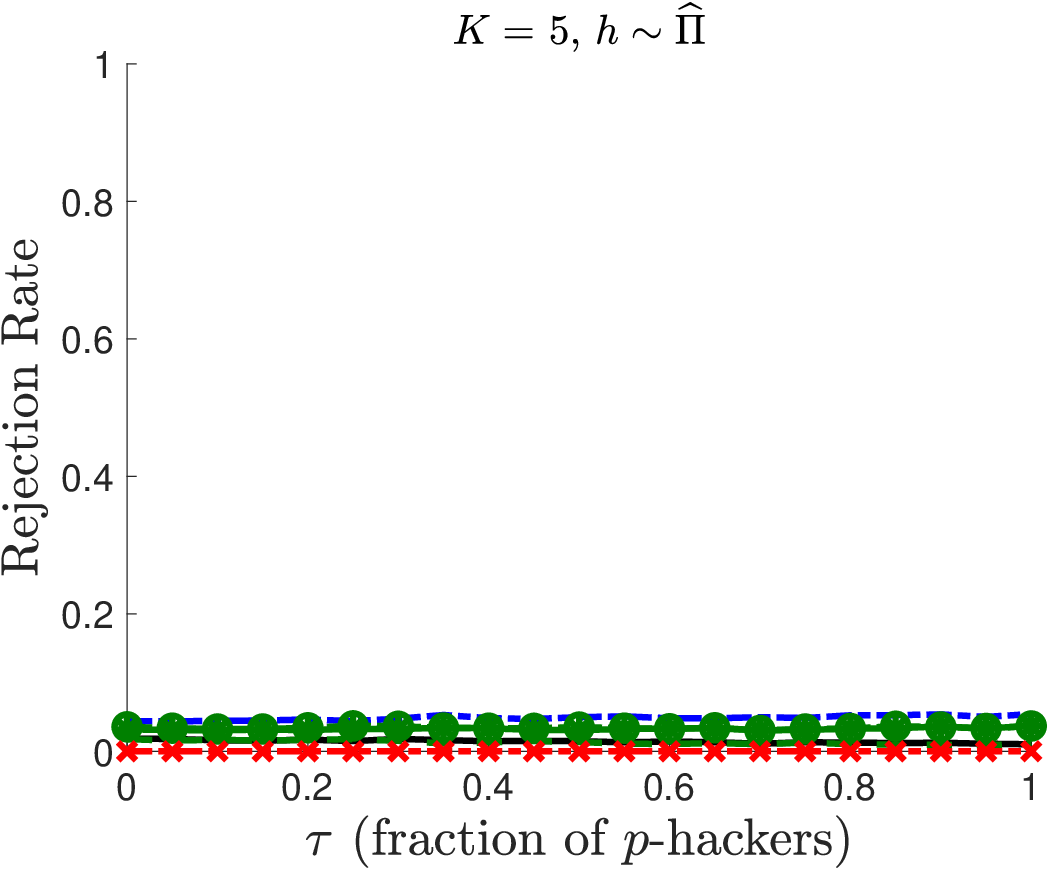}
\end{center}
\vspace*{-3mm}

\footnotesize{\textit{Notes:} Figures show rejection rates of the tests in Table \ref{tab:tests} as a function of $\tau$ for the threshold and minimum approach to IV selection with $K=5$ (using all specifications). The simulation design is described in Sections \ref{sec:iv_selection_MC} and \ref{sec:simulations_setup}. The results are based on 5,000 simulation repetitions.}

\end{figure}

\begin{figure}[H]

\caption{Power curves IV selection with $K=5$ (specifications with $F>10$).}\label{fig:power_iv_F_K5}

\vspace{-5mm}

\begin{center}
\textbf{Thresholding}


\includegraphics[width=0.24\textwidth]{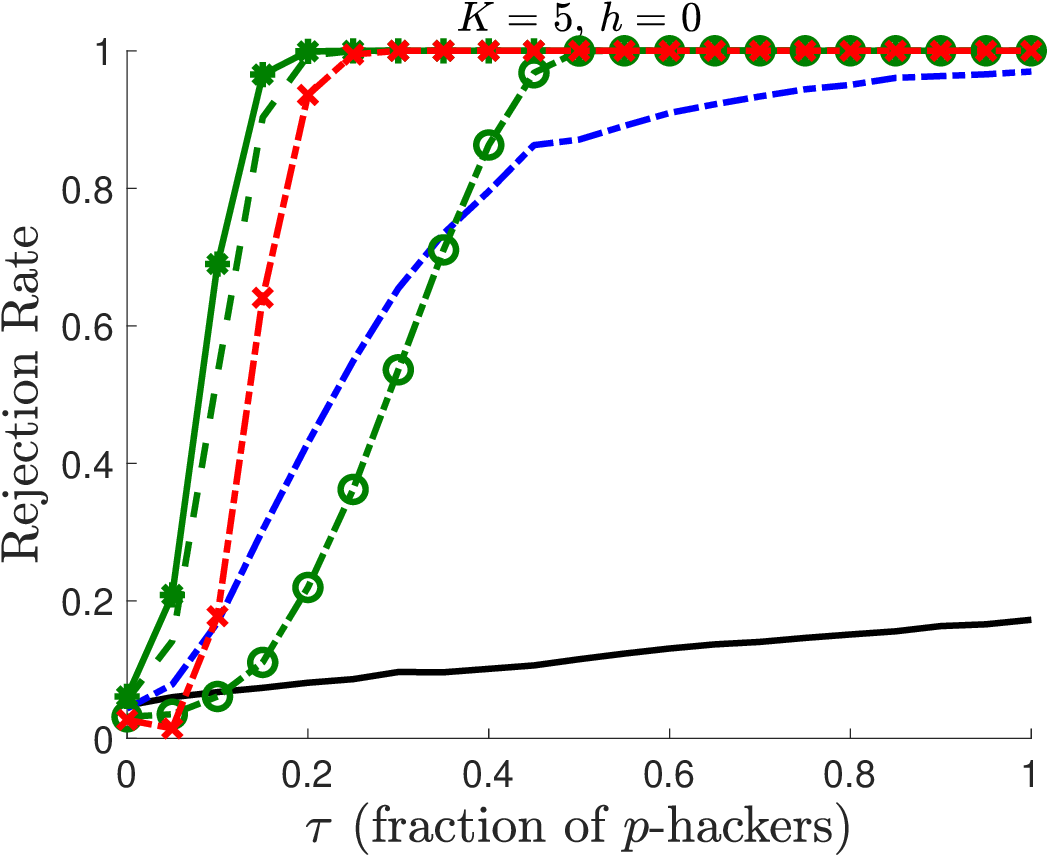}
\includegraphics[width=0.24\textwidth]{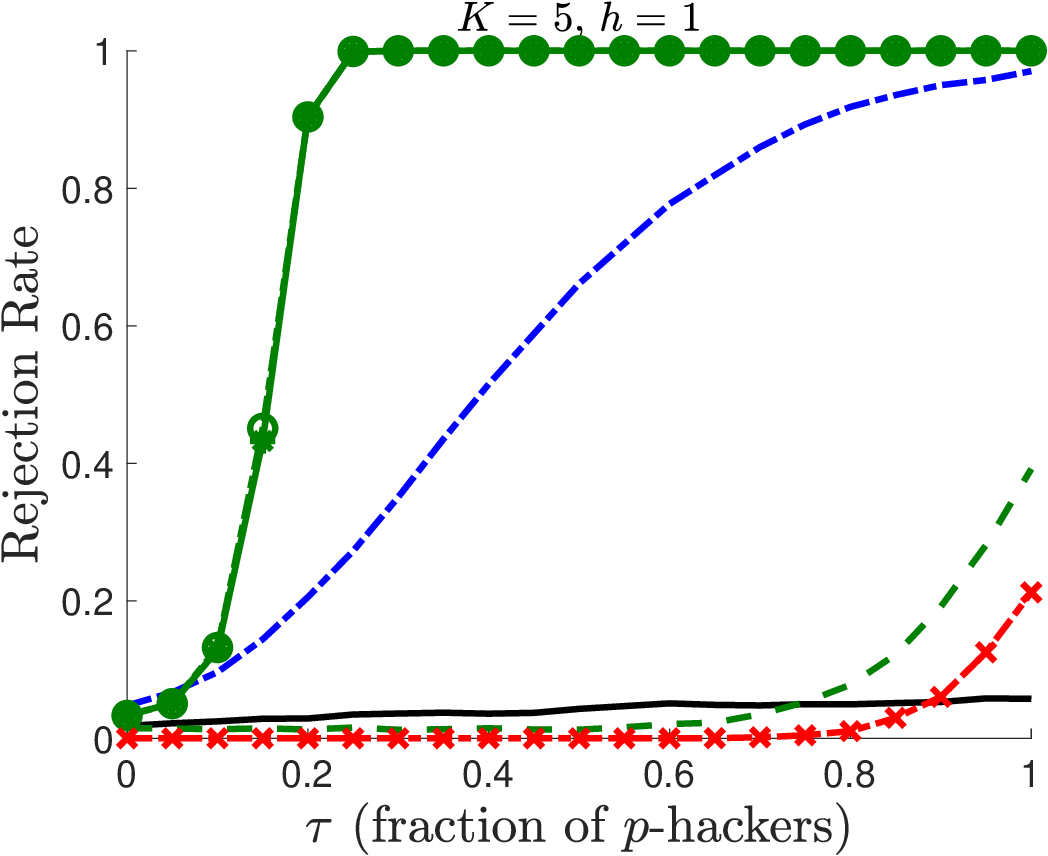}
\includegraphics[width=0.24\textwidth]{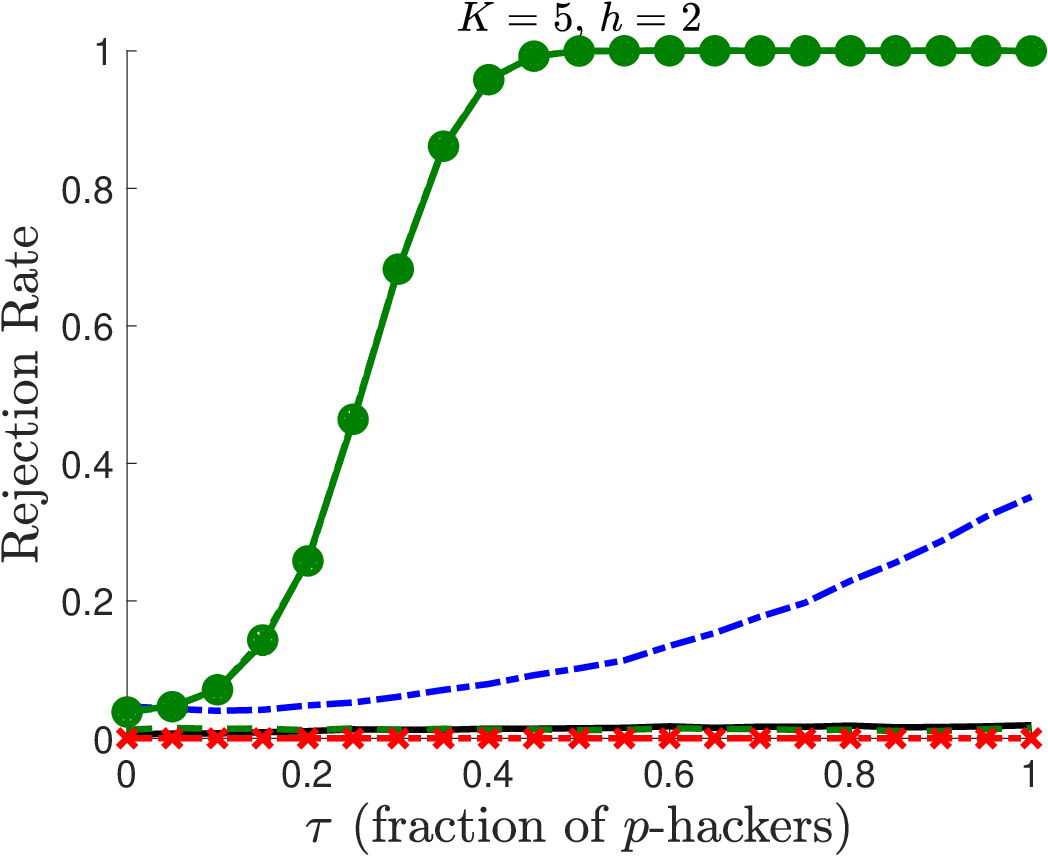}
\includegraphics[width=0.24\textwidth]{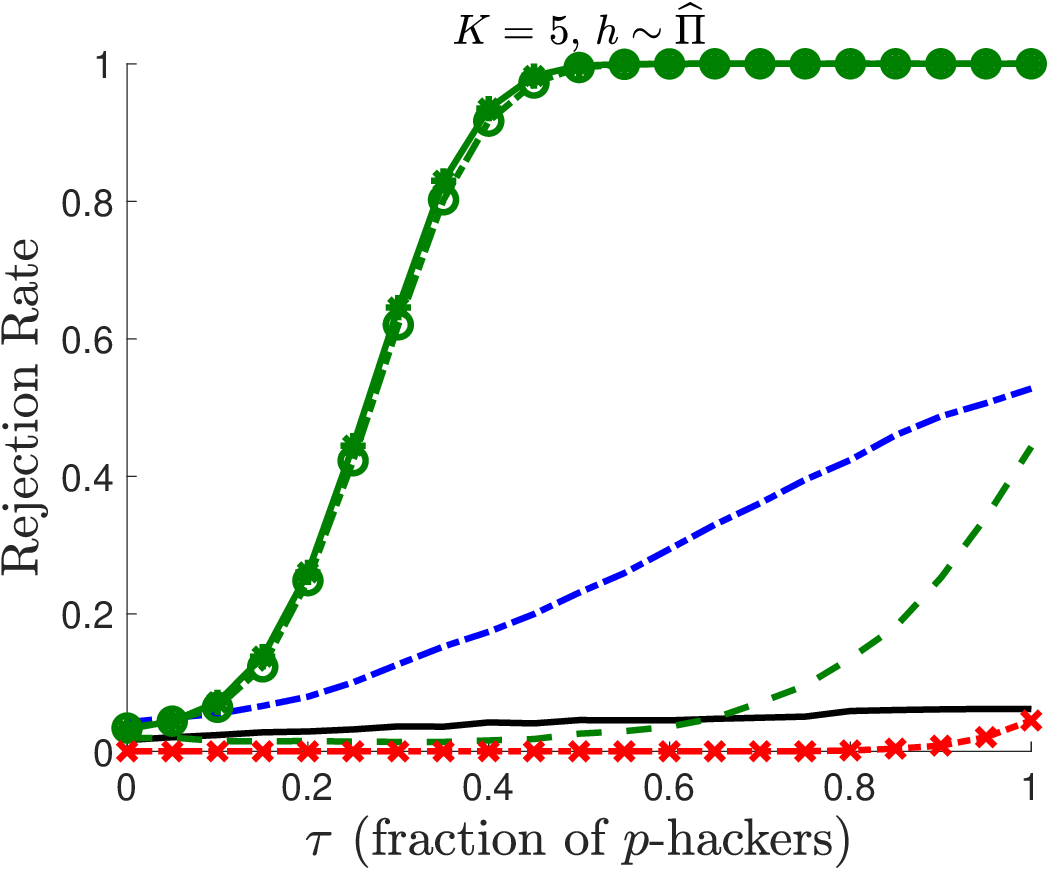}

\textbf{Minimum}

\vspace{0.15em}

\includegraphics[width=0.24\textwidth]{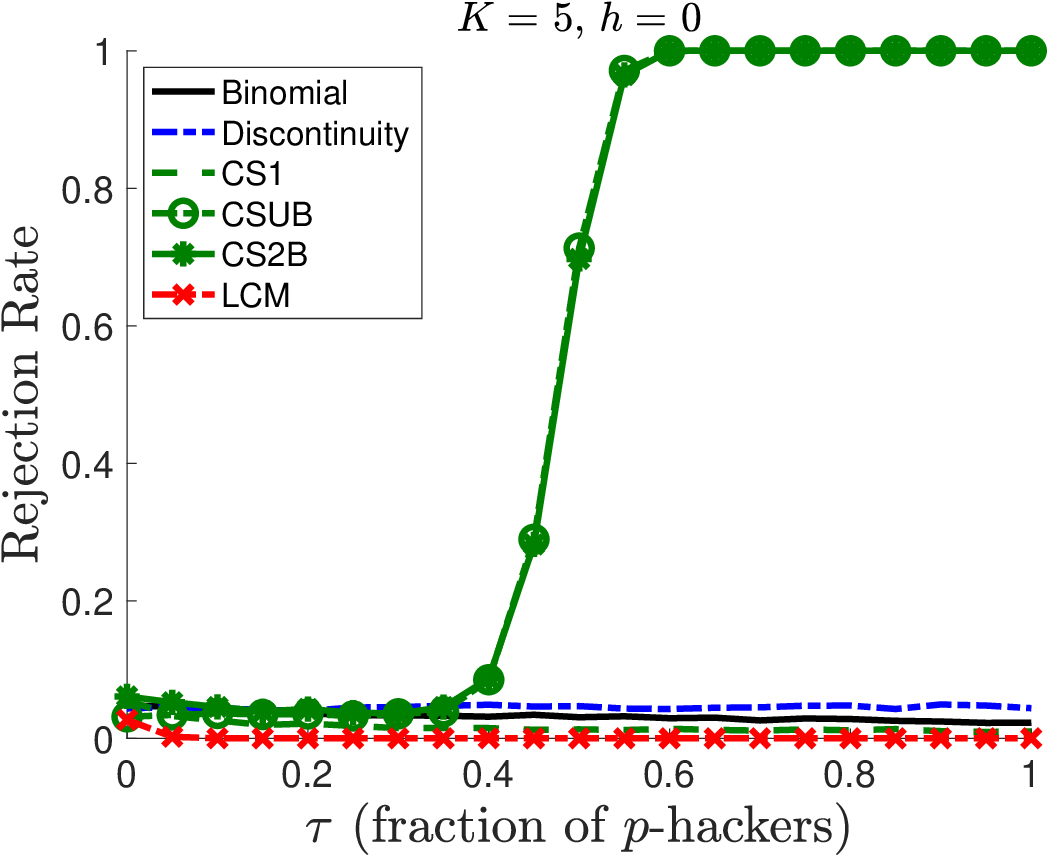}
\includegraphics[width=0.24\textwidth]{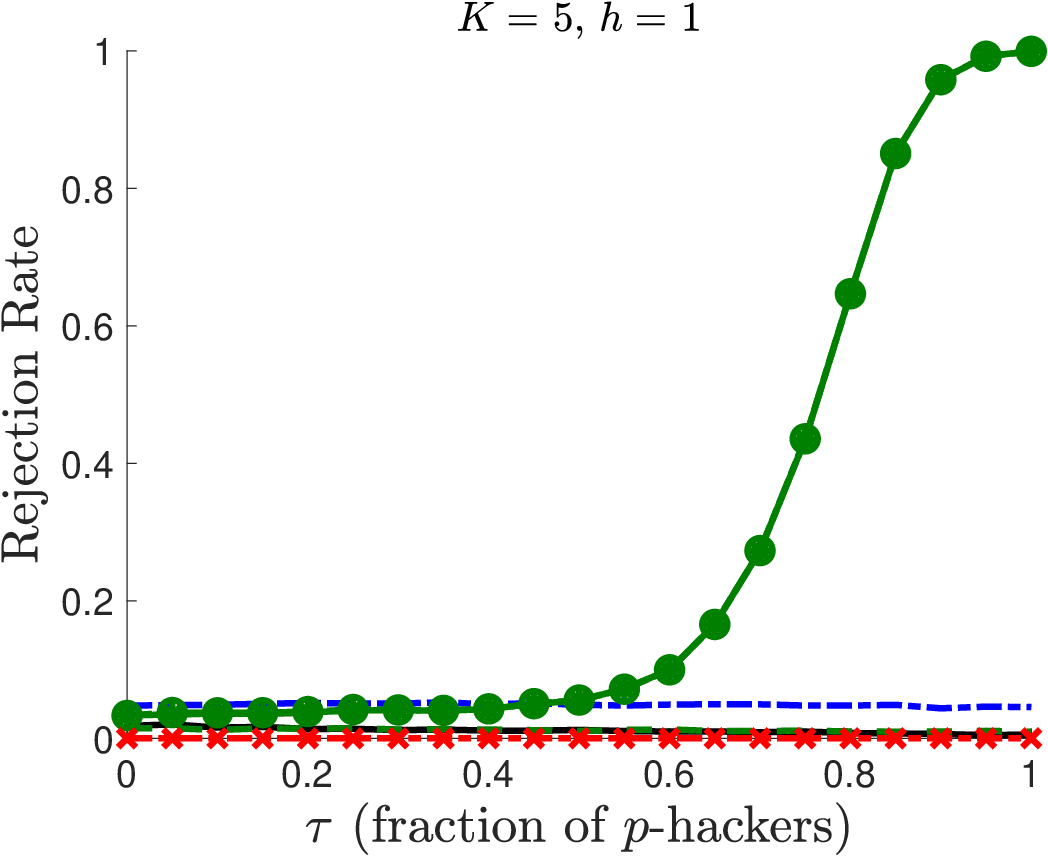}
\includegraphics[width=0.24\textwidth]{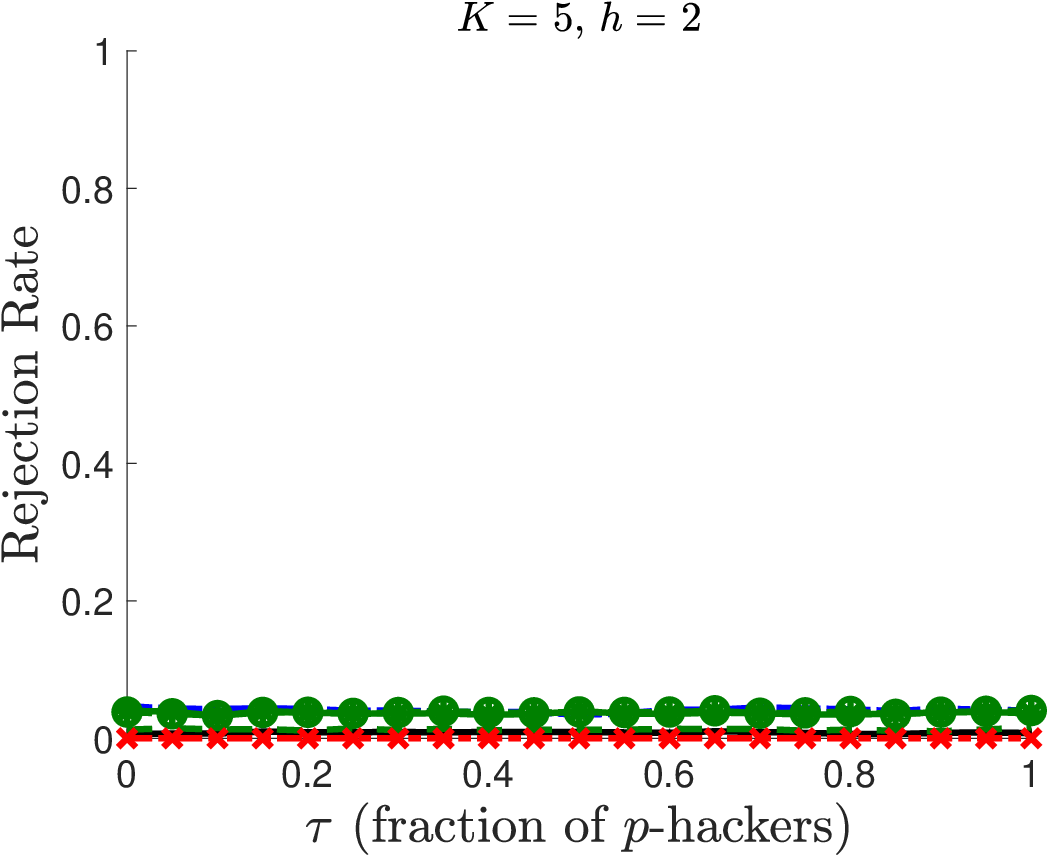}
\includegraphics[width=0.24\textwidth]{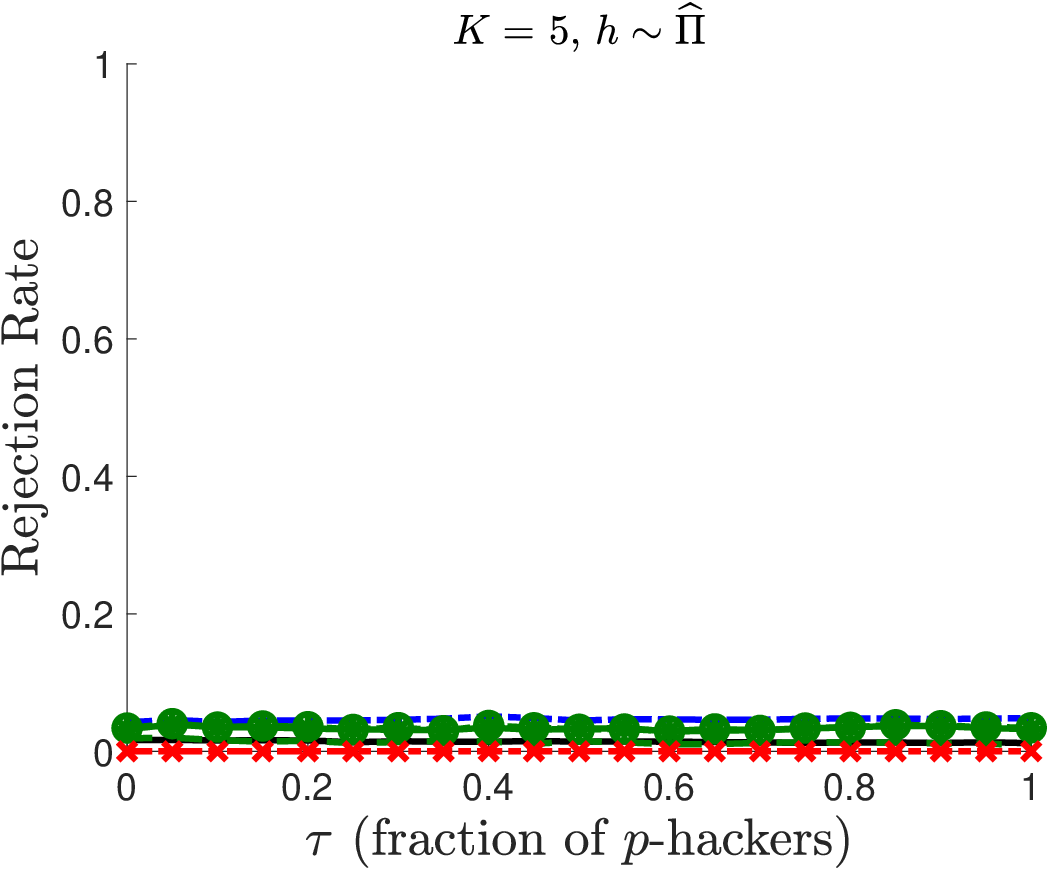}
\end{center}
\vspace*{-3mm}
\footnotesize{\textit{Notes:} Figures show rejection rates of the tests in Table \ref{tab:tests} as a function of $\tau$ for the threshold and minimum approach to IV selection with $K=5$ (using only specifications with $F>10$). The simulation design is described in Sections \ref{sec:iv_selection_MC} and \ref{sec:simulations_setup}. The results are based on 5,000 simulation repetitions.}
\end{figure}

\begin{table}[H]
\begin{center}
\caption{The effect of publication bias for $h=1$, $h=2$, and $h\sim \widehat{\Pi}$}
\label{tab:publication_bias_appendix}
\footnotesize
\begin{tabular}{lccccccccccccccccccccccc}
\toprule
 & \multicolumn{18}{c}{Test} \\ \cline{2-19}
 & \multicolumn{3}{c}{Binomial} & \multicolumn{3}{c}{Discontinuity} & \multicolumn{3}{c}{CS1} & \multicolumn{3}{c}{CSUB} & \multicolumn{3}{c}{CS2B} & \multicolumn{3}{c}{LCM} \\
Frac. of $p$-hackers & 0 & 0.5 & 1 & 0 & 0.5 & 1 & 0 & 0.5 & 1 & 0 & 0.5 & 1 & 0 & 0.5 & 1 & 0 & 0.5 & 1 \\ \hline
\multicolumn{1}{c}{} & \multicolumn{18}{c}{$h=1$} \\
\textit{} & \multicolumn{18}{c}{Thresholding} \\ \cline{2-19}
No Pub Bias & 0.03 & 0.07 & 0.12 & 0.05 & 0.72 & 0.95 & 0.02 & 0.08 & 0.63 & 0.03 & 0.59 & 1.00 & 0.04 & 0.89 & 1.00 & 0.00 & 0.00 & 0.54 \\
Sharp Pub Bias & 0.03 & 0.07 & 0.12 & 0.94 & 1.00 & 1.00 & 0.06 & 0.10 & 0.61 & 1.00 & 1.00 & 1.00 & 1.00 & 1.00 & 1.00 & 0.00 & 0.09 & 0.89 \\
Smooth Pub Bias & 0.02 & 0.04 & 0.06 & 0.05 & 0.42 & 0.77 & 0.01 & 0.01 & 0.03 & 1.00 & 1.00 & 1.00 & 1.00 & 1.00 & 1.00 & 0.00 & 0.00 & 0.00 \\
\textit{} & \multicolumn{18}{c}{Minimum} \\ \cline{2-19}
No Pub Bias & 0.03 & 0.02 & 0.02 & 0.05 & 0.05 & 0.06 & 0.02 & 0.01 & 0.01 & 0.03 & 0.03 & 0.60 & 0.04 & 0.04 & 0.59 & 0.00 & 0.00 & 0.00 \\
Sharp Pub Bias & 0.03 & 0.02 & 0.02 & 0.94 & 0.93 & 0.93 & 0.06 & 0.05 & 0.04 & 1.00 & 1.00 & 1.00 & 1.00 & 1.00 & 1.00 & 0.00 & 0.00 & 0.00 \\
Smooth Pub Bias & 0.02 & 0.02 & 0.01 & 0.05 & 0.05 & 0.05 & 0.01 & 0.01 & 0.01 & 1.00 & 1.00 & 1.00 & 1.00 & 1.00 & 1.00 & 0.00 & 0.00 & 0.00 \\
\hline
\multicolumn{1}{c}{} & \multicolumn{18}{c}{$h=2$} \\
\textit{} & \multicolumn{18}{c}{Thresholding} \\ \cline{2-19}
No Pub Bias & 0.02 & 0.04 & 0.06 & 0.05 & 0.66 & 0.97 & 0.02 & 0.01 & 0.01 & 0.04 & 1.00 & 1.00 & 0.04 & 1.00 & 1.00 & 0.00 & 0.00 & 0.00 \\
Sharp Pub Bias & 0.02 & 0.04 & 0.06 & 0.93 & 1.00 & 1.00 & 0.03 & 0.03 & 0.03 & 1.00 & 1.00 & 1.00 & 1.00 & 1.00 & 1.00 & 0.00 & 0.00 & 0.00 \\
Smooth Pub Bias & 0.01 & 0.02 & 0.03 & 0.05 & 0.38 & 0.83 & 0.01 & 0.01 & 0.01 & 1.00 & 1.00 & 1.00 & 1.00 & 1.00 & 1.00 & 0.00 & 0.00 & 0.00 \\
\textit{} & \multicolumn{18}{c}{Minimum} \\ \cline{2-19}
No Pub Bias & 0.02 & 0.01 & 0.01 & 0.05 & 0.05 & 0.05 & 0.02 & 0.01 & 0.01 & 0.04 & 0.04 & 0.08 & 0.04 & 0.04 & 0.08 & 0.00 & 0.00 & 0.00 \\
Sharp Pub Bias & 0.02 & 0.01 & 0.01 & 0.93 & 0.92 & 0.92 & 0.03 & 0.03 & 0.03 & 1.00 & 1.00 & 1.00 & 1.00 & 1.00 & 1.00 & 0.00 & 0.00 & 0.00 \\
Smooth Pub Bias & 0.01 & 0.01 & 0.00 & 0.05 & 0.05 & 0.05 & 0.01 & 0.01 & 0.01 & 1.00 & 1.00 & 1.00 & 1.00 & 1.00 & 1.00 & 0.00 & 0.00 & 0.00 \\
\hline
\multicolumn{1}{c}{} & \multicolumn{18}{c}{$h\sim\widehat{\Pi}$} \\
\textit{} & \multicolumn{18}{c}{Thresholding} \\ \cline{2-19}
No Pub Bias & 0.02 & 0.05 & 0.07 & 0.05 & 0.28 & 0.63 & 0.01 & 0.02 & 0.02 & 0.03 & 0.78 & 1.00 & 0.03 & 0.79 & 1.00 & 0.00 & 0.00 & 0.00 \\
Sharp Pub Bias & 0.02 & 0.05 & 0.07 & 0.91 & 1.00 & 1.00 & 0.06 & 0.05 & 0.06 & 0.99 & 1.00 & 1.00 & 0.99 & 1.00 & 1.00 & 0.00 & 0.00 & 0.00 \\
Smooth Pub Bias & 0.01 & 0.03 & 0.03 & 0.04 & 0.13 & 0.33 & 0.01 & 0.01 & 0.01 & 0.13 & 1.00 & 1.00 & 0.13 & 1.00 & 1.00 & 0.00 & 0.00 & 0.00 \\
\textit{} & \multicolumn{18}{c}{Minimum} \\ \cline{2-19}
No Pub Bias & 0.02 & 0.02 & 0.02 & 0.05 & 0.05 & 0.05 & 0.01 & 0.02 & 0.01 & 0.03 & 0.04 & 0.03 & 0.03 & 0.03 & 0.03 & 0.00 & 0.00 & 0.00 \\
Sharp Pub Bias & 0.02 & 0.02 & 0.02 & 0.91 & 0.92 & 0.93 & 0.06 & 0.05 & 0.05 & 0.99 & 1.00 & 1.00 & 0.99 & 1.00 & 1.00 & 0.00 & 0.00 & 0.00 \\
Smooth Pub Bias & 0.01 & 0.01 & 0.01 & 0.04 & 0.04 & 0.04 & 0.01 & 0.01 & 0.01 & 0.13 & 0.09 & 0.07 & 0.13 & 0.08 & 0.07 & 0.00 & 0.00 & 0.00 \\
\bottomrule
\end{tabular}
\end{center}
\footnotesize{\textit{Notes:} Table shows the impact of publication bias on the power of the tests when $p$-hacking is based on covariate selection (general-to-specific, two-sided tests) with $K=3$ and $h=1$, $h=2$, and $h\sim \widehat{\Pi}$. The results are based on 5,000 simulation repetitions.}
\end{table}

\end{document}